\documentclass[12pt]{article}
\usepackage[utf8]{inputenc}
\usepackage[margin=0.8in]{geometry}
\usepackage{multirow,multicol}
\usepackage{amsfonts}
\usepackage{amsmath,amssymb,amsthm}
\usepackage[mathscr]{euscript}
\usepackage{comment}
\usepackage{graphicx,color}
\usepackage{mathtools}
\usepackage{mathabx} 
\usepackage[linesnumbered,ruled,vlined]{algorithm2e}
\usepackage{setspace}
\usepackage{sectsty}
\usepackage{makecell}
\usepackage{chngcntr}
\usepackage{apptools}
\usepackage{booktabs}
\usepackage{caption}
\usepackage{subcaption}
\usepackage{tabularx}
\newcolumntype{Y}{>{\raggedleft\arraybackslash}X}
\newcolumntype{Z}{>{\centering\arraybackslash}X}
\usepackage{enumitem}
\usepackage{underscore} 
\usepackage{array} 
\usepackage{arydshln}
\usepackage{soul}
\usepackage{tikz}
\usepackage{hyperref} 
\usetikzlibrary{matrix}
\usepackage[english]{babel}
\usepackage[toc,page]{appendix}
\usepackage{multirow} 

\usepackage{lipsum}
\usepackage{calrsfs} 
\usepackage{pdflscape}
\usepackage{float} 
\usetikzlibrary{shapes.misc}
\usetikzlibrary{shapes}

\DeclareMathAlphabet{\pazocal}{OMS}{zplm}{m}{n} 

\DeclareFontFamily{OT1}{pzc}{}
\DeclareFontShape{OT1}{pzc}{m}{it}{<-> s * [1.10] pzcmi7t}{}
\DeclareMathAlphabet{\mathpzc}{OT1}{pzc}{m}{it}

\usepackage{natbib}
\usepackage{bibunits}
\usepackage{bookmark}

\newtheorem{assumption}{Assumption}

\newtheorem{definition}{Definition}
\newtheorem{lemma}{Lemma}
\newtheorem{proposition}{Proposition}
\newtheorem{theorem}{Theorem}

\newtheorem{remark}{Remark}

\newcolumntype{P}[1]{>{\centering\arraybackslash}p{#1}}

\newcommand{\bm}{\boldsymbol}
\newcommand{\cm}[1]{\mbox{\boldmath$\mathscr{#1}$}}
\newcommand{\cmt}[1]{\mbox{\boldmath\scriptsize$\mathscr{#1}$}}
\newcommand{\cmtt}[1]{\mbox{\boldmath\tiny$\mathscr{#1}$}}
\newcommand{\Fr}{{\mathrm{F}}}
\newcommand{\op}{{\mathrm{op}}}

\newcommand{\ma}{{\mathrm{MA}}}

\newcommand{\colsp}{\mathrm{colsp}}
\newcommand{\stcomp}[1]{{#1}^\complement} 
\definecolor{c_rate40}{RGB}{102,153,221}
\definecolor{c_rate20}{RGB}{58,76,101}
\definecolor{c_rate10}{RGB}{37,38,41}
\definecolor{c_order60}{RGB}{237,209,203}
\definecolor{c_order65}{RGB}{200,135,158}
\definecolor{c_order70}{RGB}{131,76,125}
\definecolor{c_order75}{RGB}{45,30,62}
\definecolor{white}{RGB}{255,255,255}

\newcommand{\Lratefourty}{\raisebox{0pt}{\tikz{\draw[c_rate40,solid,line width = 1.0pt,fill=c_rate40](2mm,0) rectangle ++(1.8mm,1.8mm);\draw[-,c_rate40,densely dotted,line width = 1.5pt](0.,0.8mm) -- (5.5mm,0.8mm)}}}
\newcommand{\Lratetwenty}{\raisebox{-2pt}{\tikz{\node[draw=c_rate20,scale=0.35,cross out,minimum size=5mm,rotate=0,line width=1.mm] at (2.75mm,0.65mm) {};\draw[-,c_rate20,dashed,line width = 1.5pt](0.,0.65mm) -- (5.5mm,0.65mm)}}}
\newcommand{\Lrateten}{\raisebox{0pt}{\tikz{\draw[c_rate10,solid,line width = 1.0pt,fill=c_rate10](2.75mm,0.75mm) circle (0.9mm);\draw[-,c_rate10,solid,line width = 1.5pt](0.,0.8mm) -- (5.5mm,0.8mm)}}}
\newcommand{\Lorderone}{\raisebox{0pt}{\tikz{\draw[-,c_order60,solid,line width = 1.5pt](0.,0.8mm) -- (5.5mm,0.8mm); \draw[white,solid,line width = 1.0pt,fill=c_order60](2.75mm,0.75mm) circle (1mm);\draw[c_order60,solid,line width = 1.0pt,fill=c_order60](2.75mm,0.75mm) circle (0.9mm);}}}
\newcommand{\Lordertwo}{\raisebox{-1pt}{\tikz{\draw[-,c_order65,solid,line width = 1.5pt](0mm,0.65mm) -- (5.5mm,0.65mm); \node[draw=white,scale=0.36,cross out,minimum size=5mm,rotate=0,line width=1.1mm] at (2.75mm,0.65mm) {}; \node[draw=c_order65,scale=0.35,cross out,minimum size=5mm,rotate=0,line width=1mm] at (2.75mm,0.65mm) {};  }}}
\newcommand{\Lorderthree}{\raisebox{0pt}{\tikz{\draw[-,c_order75,solid,line width = 1.5pt](0.,0.8mm) -- (5.5mm,0.8mm); \draw[white,solid,line width = 1.0pt,fill=c_order75](1.95mm,0) rectangle ++(1.6mm,1.6mm); \draw[c_order75,solid,line width = 1.0pt,fill=c_order75](2mm,0) rectangle ++(1.5mm,1.5mm);}}}

\def\HH{{\mathrm{\scriptscriptstyle\mathsf{H}}}}

\DeclareMathOperator*{\vect}{vec}
\DeclareMathOperator*{\rank}{rank}
\DeclareMathOperator*{\trace}{tr}
\DeclareMathOperator*{\argmin}{arg\,min}

\DeclareMathOperator*{\diag}{diag}

\DeclareMathOperator*{\var}{var}
\DeclareMathOperator*{\stk}{stack}

\renewcommand{\arraystretch}{0.85}

\numberwithin{equation}{section}

\title{SARMA:  Scalable Low-Rank High-Dimensional Autoregressive Moving Averages via Tensor Decomposition}
\author{Feiqing Huang,  Kexin Li, and Yao Zheng\footnote{Corresponding Author. Assistant Professor, Department of Statistics, University of Connecticut, Storrs, CT 06269. (Email: yao.zheng@uconn.edu).}
	\\ \textit{University of Hong Kong and University of Connecticut}}
\date{}

\begin{document}
	
\setlength{\parindent}{16pt}
	
\maketitle
	
\begin{abstract}
	Existing models for high-dimensional time series are overwhelmingly developed within the  finite-order vector autoregressive (VAR) framework. However, the more flexible vector autoregressive moving averages (VARMA)  have been much less considered. This paper introduces a Tucker-low-rank framework to efficiently capture  VARMA-type dynamics for high-dimensional time series, named the Scalable ARMA (SARMA) model. 
It generalizes the Tucker-low-rank finite-order VAR model to the infinite-order case via flexible parameterizations of the AR coefficient tensor along the temporal dimension. The resulting model enables dynamic factor extraction across  response and predictor variables, facilitating interpretation of group patterns.
Additionally, we consider sparsity assumptions on the factor loadings to accomplish automatic variable selection and greater estimation efficiency. Both rank-constrained and sparsity-inducing estimators are developed for the proposed model, along with algorithms and model selection methods.
The validity of our theory and empirical advantages of our approach are confirmed by simulation studies and real data examples.
\end{abstract}


\textit{MSC2020 subject classifications}: Primary 62M10; secondary 62H12, 60G10

\newpage

\begin{bibunit}[apalike]
	
\section{Introduction}\label{sec:intro}

Consider the vector autoregressive (VAR) model of order $1\leq P\leq \infty$ for an $N$-dimensional time series $\{\bm{y}_t\}_{t=1}^T$, 
\begin{equation}\label{eq:infVARmodel}
	\bm{y}_t=\sum_{j=1}^{P} \bm{A}_j \bm{y}_{t-j}+\bm{\varepsilon}_t,
\end{equation}
where $\bm{A}_j\in\mathbb{R}^{N\times N}$ are the AR coefficient matrices, and 
$\bm{\varepsilon}_t\in\mathbb{R}^{N}$ is the innovation term at time $t$.  
VAR models are fundamental in multivariate time series modeling \citep{Lutkepohl2005, Tsay14}. Their versatility have led to numerous extensions to accommodate more complex dynamic structures and high dimensions \citep[e.g.,][]{BM21, Wang2021High,bai2023multiple}. 
However, the VAR model is typically employed with a small lag order  $P$ in practice; or theoretically, $P$ is fixed as the sample size and dimension increase. This amounts to assuming that the present observation $\bm{y}_t$ is  influenced by only a few small lags, $\bm{y}_{t-j}$. By excluding data from further in the past, this significantly limits the flexibility of the VAR model. In this paper, we address this limitation by developing a dimension reduction approach to accommodate $P=\infty$. Note that it  includes the finite-order VAR as a special case, yet eliminates the need for order selection.


Let $\cm{A}\in\mathbb{R}^{N\times N\times \infty}$ be the AR coefficient tensor defined by stacking $\bm{A}_j$ for $1\leq j\leq P=\infty$. To address the large dimension $N$, this paper  considers a low-Tucker-rank approach for dimension reduction in the first two modes of $\cm{A}$.  When the cross-sectional dependency among the $N$ component series is sparse,  it is natural to impose sparsity  on AR coefficient matrices  \citep[e.g.,][]{basu2015regularized, sparseARMA}. However, in many applications, especially in economics and finance, the component series often exhibit co-movements and are believed to be driven by a small number of common factors. In such scenarios, the cross-sectional dependency would be non-sparse. For  VAR models with a fixed order $P$, \cite{Wang2021High} imposes a low-Tucker-rank assumption on $\cm{A}\in\mathbb{R}^{N\times N\times P}$ for simultaneous dimension reduction across all three modes of the coefficient tensor. However, it remains an open question how to extend this approach to  $P=\infty$. 

To accommodate both large $N$ and $P=\infty$, this paper introduces a novel solution, which can be encapsulated by the following decomposition of $\cm{A}\in\mathbb{R}^{N\times N\times \infty}$:
\begin{equation}\label{eq:SUL}
	\cm{A}=\cm{S}\times_1 \bm{U}_1\times_2 \bm{U}_2 \times_3\bm{L}(\bm{\omega}),
\end{equation}
where 
$\cm{S}\in\mathbb{R}^{\pazocal{R}_1\times \pazocal{R}_2\times \pazocal{R}_3}$ and $\bm{U}_i\in\mathbb{R}^{N\times \pazocal{R}_i}$ for $i=1,2$ are unknown parameter tensor and matrices, and $\bm{L}(\bm{\omega})\in\mathbb{R}^{\infty\times \pazocal{R}_3}$ is a predetermined function with an unknown low-dimensional parameter vector $\bm{\omega}$.  Similar to \cite{Wang2021High}, this implies that $\cm{A}\in\mathbb{R}^{N\times N\times \infty}$ has low Tucker ranks, $(\pazocal{R}_1, \pazocal{R}_2, \pazocal{R}_3)$. However, instead of factoring out an unrestricted matrix $\bm{U}_3\in\mathbb{R}^{\infty\times \pazocal{R}_3}$ at the third mode,  we deal with the infinite dimension via further parameterization. As a result,  the number of parameters is reduced from infinity to $O(N)$. It is worth noting that our theoretical analysis will differ substantially from  \cite{Wang2021High}  due to the possible nonlinearity of the general parameterization for $\bm{L}$. The conjunction of \eqref{eq:infVARmodel} and \eqref{eq:SUL} leads to a VAR process with simultaneous scalability across response, predictor, and temporal dimensions, allowing large $N$ and $P=\infty$. Moreover, we discover that the factorization of $\bm{L}(\bm{\omega})$ is closely related to vector autoregressive moving average (VARMA) models; in other words, VARMA models are essentially low-rank   VAR($\infty$) models. 
Therefore, we call the proposed high-dimensional low-Tucker-rank VAR($\infty$) model the \textit{scalable ARMA (SARMA)} model.

Recently, there has been growing interest in the estimation of large and high-dimensional VARMA models \citep[e.g.,][]{AV08,CEK16,Dias18,WBBM21}. However,  VARMA models require complex identification constraints, which often obscure the interpretation of the parameters. On the other hand,  \cite{sparseARMA} develops a sparse VAR($\infty$) model which is motivated by the VARMA model but resolves the identification issue. Interestingly, we can show that the model in \cite{sparseARMA} is a special case of \eqref{eq:SUL} without factorizations of $\bm{U}_1$ and $\bm{U}_2$, i.e., there is no low-rank structure for the first two modes. The parametric approach in \cite{sparseARMA} indeed constitutes an easy-to-implement specification for  the user-defined function $\bm{L}(\bm{\omega})$. Furthermore,  we prove the identifiability of the model for the first time in this paper.


Additionally, in the  ultra-high-dimensional setup where $N$ may grow exponentially with the sample size, we further consider a  sparse low-Tucker-rank (SLTR) structure for $\cm{A}$ by imposing  entrywise-sparsity  on the loadings of the response and predictor factors. This results in a more substantial dimension reduction and can be interpreted as an automatic selection of important variables into the response and predictor factors. For the proposed SARMA model, we introduce two estimators: (i) the rank-constrained estimator for the case with non-sparse factor loadings, and (ii) the SLTR estimator for the case with sparse factor loadings.  For both estimators, we derive nonasymptotic error bounds.  A consistent estimator for the Tucker ranks, method for model order selection, and algorithms for implementing the proposed methods are detailed in the supplementary file.


The rest of this paper is organized as follows. Section \ref{section:model} introduces the proposed model, including a general formulation, representative examples, the high-dimensional SARMA model, and the dynamic factor interpretations of the latter. Section \ref{sec:HDmodel} develops estimation methods in both non-sparse and sparse cases, together with theoretical properties. Empirical studies are provided in Section \ref{sec:empirical}. Section \ref{sec:conclusion} concludes with a brief discussion.  Due to the page limit, methods for rank and model order selection, simulation studies,  algorithms, and technical details are given in a separate supplementary file.


%



\subsection{Notations}

Unless otherwise specified, we denote scalars by lowercase letters $x, y, \dots$, vectors by  boldface lowercase letters $\bm{x}, \bm{y}, \dots$, and matrices by boldface capital letters $\bm{X}, \bm{Y}, \dots$.
For any $a,b\in\mathbb{R}$, denote $a\vee b=\max\{a,b\}$ and $a\wedge b=\min\{a,b\}$. 
For any vector $\bm{x}$, denote its $\ell_2$ norm by $\|\bm{x}\|_2$. 
For any matrix $\bm{X}\in\mathbb{R}^{d_1\times d_2}$, let $\sigma_{1}(\bm{X})\geq \sigma_{2}(\bm{X})\geq  \cdots \geq \sigma_{d_1\wedge d_2}(\bm{X})\geq 0$ be its singular values in descending order.  Let $\bm{X}^\prime$, $\sigma_{\max}(\bm{X})$ (or $\sigma_{\min}(\bm{X})$),  $\lambda_{\max}(\bm{X})$ (or $\lambda_{\min}(\bm{X})$), and $\rank(\bm{X})$ denote its transpose, largest (or smallest) singular value, largest (or smallest) eigenvalue, and rank, respectively. 
Its vectorization $\vect(\bm{X})$ is the long vector obtained by stacking all its columns.
In addition, its operator norm, Frobenius norm, and nuclear norm are $\|\bm{X}\|_\op=\sigma_{\max}(\bm{X})$,  $\|\bm{X}\|_\Fr=\sqrt{\sum_{i,j}\bm{X}_{ij}^2}=\sqrt{\sum_{k=1}^{d_1\wedge d_2}\sigma_{k}^2(\bm{X})}$, and $\|\bm{X}\|_*=\sum_{k=1}^{d_1\wedge d_2}\sigma_{k}(\bm{X})$, respectively.  For any two sequences $x_n$ and $y_n$, denote $x_n\lesssim y_n$ (or $x_n\gtrsim y_n$) if there exists an absolute constant $C>0$ such that $x_n\leq C y_n$ (or $x_n\geq C y_n$). Write $x_n\asymp y_n$ if $x_n\lesssim y_n$ and $x_n\gtrsim y_n$. Let $\mathbb{I}_{\{\cdot\}}$ be the indicator function taking value one when the condition is true and zero otherwise. The capital letters $C, C_{\cmtt{G}}, \dots$ and lowercase letters $c, c_{\cmtt{G}}, \dots$ represent generic large and small positive absolute constants, respectively, whose values may vary from place to place.

This paper involves third-order tensors, a.k.a. three-way arrays, which are denoted by calligraphic capital letters. For example, a  $d_1\times d_2\times d_3$ tensor is $\cm{X}=(\cm{X}_{i_1 i_2 i_3})_{1\leq i_1\leq d_1, 1\leq i_2\leq d_2, 1\leq i_3\leq d_3}$. It has three modes, with dimension $d_i$ for mode $i$, for $1\leq i\leq 3$. The  Frobenius norm of the tensor is defined as $\|\cm{X}\|_{\Fr} = \sqrt{\sum_{i_1,i_2, i_3}\cm{X}_{i_1 i_2 i_3}^2}$.  The mode-3 product of $\cm{X}$ and a $K\times d_3$ matrix $\bm{Y}$ is the $d_1\times d_2\times K$ tensor given by $\cm{X} \times_3 \bm{Y} = (\sum_{i_3=1}^{d_3}\cm{X}_{i_1 i_2 i_3}\bm{Y}_{k i_3})_{1\leq i_1\leq d_1, 1\leq i_2 \leq d_2,1\leq k\leq K}$. Similarly, the mode-$i$ multiplication $\times_i$ between $\cm{X}$ and a $K\times d_i$ matrix can be defined for $i=1,2$.
The matricization along mode $i$ of $\cm{X}$ results in a matrix where the mode $i$ becomes the rows of the matrix, and the other modes are collapsed into the columns. The mode-$i$  matricization is denoted by $\cm{X}_{(i)}$, and it can be shown that $\cm{X}_{(1)}=(\bm{X}_1,\dots,\bm{X}_{d_3})$,
$\cm{X}_{(2)}=(\bm{X}_1', \dots,\bm{X}_{d_3}') \in\mathbb{R}^{d_2\times d_1d_3}$, and $\cm{X}_{(3)}=(\text{vec}(\bm{X}_1),\dots,\text{vec}(\bm{X}_{d_3}))'\in\mathbb{R}^{d_3\times d_1d_2}$.
The Tucker rank of $\cm{X}$  at mode $i$ is the rank of $\cm{X}_{(i)}$, i.e., $\pazocal{R}_i=\textrm{rank}(\cm{X}_{(i)})$ for $1\leq i\leq 3$ \citep{tucker1966some,delathauwer2000multilinear}. Unlike row and column ranks of a matrix, $\pazocal{R}_1, \pazocal{R}_2$ and $\pazocal{R}_3$ in general are not identical.

\section{Proposed model} \label{section:model}
\subsection{General formulation}
We propose a general VAR($\infty$) framework as follows:
\begin{equation}\label{eq:genmodel}
	\bm{y}_t=\sum_{j=1}^\infty \bm{A}_j(\bm{\omega},\cm{G}) \bm{y}_{t-j}+\bm{\varepsilon}_t, \quad \text{with}\quad  	\cm{A}(\bm{\omega},\cm{G})=\cm{G}\times_3\bm{L}(\bm{\omega}),
\end{equation}
where $\cm{A}=\cm{A}(\bm{\omega},\cm{G})\in\mathbb{R}^{N\times N\times \infty}$ is the AR coefficient tensor formed by stacking  $\bm{A}_j=\bm{A}_j(\bm{\omega},\cm{G})$ for $j\geq 1$. In addition,  $\cm{G}\in\mathbb{R}^{N\times N\times d}$ is the  tensor formed by  stacking $\bm{G}_k\in\mathbb{R}^{N\times N}$ for $1\leq k\leq  d$, and $\bm{L}(\bm{\omega})\in\mathbb{R}^{\infty\times d}$  is a user-defined function dependent on $\bm{\omega}\in \bm{\Omega}$, with a finite-dimensional parameter space $\bm{\Omega}$. Here
the unknown parameters are $\cm{G}$ and $\bm{\omega}$.

Note that $\cm{A}$ is factorized into $\cm{G}$ and $\bm{L}(\bm{\omega})$ along the third mode. It implies that the Tucker rank of $\cm{A}$ at the third mode is $\pazocal{R}_3=d$. This tensor factorization is equivalent to
\begin{equation}\label{eq:Afactor}
	\bm{A}_j=
	\bm{A}_j(\bm{\omega},\cm{G})=\sum_{k=1}^{d}\ell_{j,k}(\bm{\omega})\bm{G}_k,\quad \text{for}\quad j\geq1.
\end{equation}
If $\cm{G}$ is Tucker-low-rank at the first two modes, we further have the factorization in \eqref{eq:SUL}, where $\pazocal{R}_3 =d$. The low-Tucker-rank VAR($\infty$) model to be developed in this paper, the sparse VAR($\infty$) model in \cite{sparseARMA}, and even the VARMA model, can all be regarded as special cases of \eqref{eq:genmodel}; we discuss this further in the next subsection.

We establish the stationarity condition for model \eqref{eq:genmodel} under the following assumption on $\bm{L}(\bm{\omega})=(\ell_{j,k}(\bm{\omega}))_{j\geq1, 1\leq k\leq d}$.

\begin{assumption}\label{assum:decay1}
	There exists $0<\rho<1$ and $C>0$ such that $\max_{1\leq k\leq K}  |\ell_{j,k}(\bm{\omega})|\leq C\rho^j$ for $j\geq 1$.
\end{assumption}

\begin{theorem}[Weak stationarity]\label{thm:sol}
	Suppose that   Assumption \ref{assum:decay1} holds and $\{\bm{\varepsilon}_t\}$ is an $i.i.d.$ sequence  with $E(\|\bm{\varepsilon}_t\|_2)<\infty$. If $\sum_{k=1}^{d}\|\bm{G}_k\|_{\op}<1/\rho-1$,
	then there exists a unique weakly stationary solution to model  \eqref{eq:genmodel},  and it has the form of $\bm{y}_t = \bm{\varepsilon}_t+ \sum_{j=1}^{\infty}\bm{\Psi}_j\bm{\varepsilon}_{t-j}$, where 
	$\bm{\Psi}_j=\sum_{k=1}^{\infty}\sum_{j_1+\cdots+ j_k=j}\bm{A}_{j_1}\cdots \bm{A}_{j_k}$, and
	$\bm{A}_j$ are given by \eqref{eq:Afactor}.
\end{theorem}



\subsection{Motivation and representative examples of model \eqref{eq:genmodel}}\label{subsec:repar}

The formulation in \eqref{eq:genmodel} is intrinsically motivated by an interesting discovery: VARMA models are essentially Tucker-low-rank VAR($\infty$) processes, where the low-rankness occurs  at  the third mode of the AR coefficient tensor $\cm{A}\in\mathbb{R}^{N\times N\times \infty}$ and can be parameterized in the form of $\bm{L}(\bm{\omega})$.

To show this, consider an invertible VARMA($p,q$) model in the form of  $\bm{y}_t =\sum_{i=1}^{p} \bm{\Phi}_i\bm{y}_{t-i}+\bm{\varepsilon}_t - \sum_{j=1}^{q}\bm{\Theta}_j\bm{\varepsilon}_{t-j}$, where $\bm{\Phi}_i, \bm{\Theta}_j\in\mathbb{R}^{N\times N}$ are the coefficient matrices. Denote the MA companion matrix  \citep{Lutkepohl2005} by
\[
\underline{\bm{\Theta}}=
\left(\begin{matrix}
	\bm{\Theta}_1&\bm{\Theta}_2&\cdots&\bm{\Theta}_{q-1}&\bm{\Theta}_q\\
	\bm{I}&\bm{0}&\cdots&\bm{0}&\bm{0}\\
	\bm{0}&\bm{I}&\cdots&\bm{0}&\bm{0}\\
	\vdots&\vdots&\ddots&\vdots&\vdots\\
	\bm{0}&\bm{0}&\cdots&\bm{I}&\bm{0}
\end{matrix} \right ). 
\]
Let $\underline{\bm{\Theta}}$ have $R+2S$ nonzero eigenvalues, where $R=\sum_{k=1}^{r} n_j$ and $S=\sum_{k=1}^{s}m_k$ account for real and complex eigenvalues, respectively. Specifically, let the distinct nonzero real eigenvalues be $\lambda_j$ for $1\leq j\leq r$, with algebraic multiplicity $n_j$. Let the distinct conjugate pairs of complex eigenvalues be $(\gamma_ke^{i\theta_k},\gamma_ke^{-i\theta_k})$  for $1\leq k\leq s$, with algebraic multiplicity  $m_k$, where $\gamma_k\in(0,1)$ and $\theta_k \in (-\pi/2, \pi/2)$. 

\begin{proposition}\label{prop:VARMAgen}
	If the geometric multiplicities of all nonzero eigenvalues of $\underline{\bm{\Theta}}$ are one, then for all $j\geq1$, any invertible VARMA($p,q$) model can be written in the form of \eqref{eq:infVARmodel} with $P=\infty$  and 
	\begin{align}\label{eq:linearcomb2gen}
		\begin{split}
			\bm{A}_{j} 
			&=\sum_{k=1}^{p}\mathbb{I}_{\{j=k\}}\bm{G}_{k}
			+ \sum_{k=1}^{r}  \sum_{i=1}^{n_k} \mathbb{I}_{\{j\geq p+(i-1)\vee1\}} \lambda_k^{j-p - i + 1} \binom{ j-p}{ i - 1}  \bm{G}_{k,i}^{I} \\
			&\hspace{5mm}+\sum_{k=1}^{s}  \sum_{i=1}^{m_k}  \mathbb{I}_{\{j\geq p+(i-1)\vee1\}} \gamma_k^{j-p-i+1} \binom{ j-p }{i- 1 }  \\
			&\hspace{22mm} \cdot \left [ \cos \{ (j-p-i+1) \theta_k\} \bm{G}_{k,i}^{II,1}
			+\sin \{( j-p-i+1) \theta_k\} \bm{G}_{k,i}^{II,2}\right ],
		\end{split}
	\end{align}
	where 
	$\bm{G}_{k,i}^{I}, \bm{G}_{k,i}^{II,1},\bm{G}_{k,i}^{II,2}\in\mathbb{R}^{N\times N}$ are all determined jointly by $\bm{\Phi}_i$'s and $\bm{\Theta}_j$'s, and the corresponding sum will be suppressed if $p, r$ or $s$ is zero. Moreover, 
	for any fixed $k$ and $i$, $\bm{G}_{k,i}^{II,h}$ for $h=1,2$ have the same row and column spaces, and 
	$\rank(\bm{G}_{j,l}^{I})\leq n_k$ and $\rank(\bm{G}_{k,i}^{II,h})\leq 2m_k$
	for all $1\leq j\leq r$, $1\leq k\leq s$, $1\leq l\leq n_k$, $1\leq i\leq m_k$, and $h=1,2$.
\end{proposition}

While \eqref{eq:linearcomb2gen} looks complicated, it can be simply written as \eqref{eq:Afactor} with $d=p+R+2S$. By tensor algebra, we can rewrite it in a surprisingly succinct form, i.e., the second equation in \eqref{eq:genmodel}. 
This reveals that the VARMA model in general can be  understood as a low-rank dimension reduction approach for VAR($\infty$) processes.

A special case of \eqref{eq:linearcomb2gen} is considered in \cite{sparseARMA}, which corresponds to   $n_1=\cdots=n_r=m_1=\cdots=m_s=1$, i.e., $R=r$ and $S=s$. In this case,  $d=p+r+2s=\pazocal{R}_3$, and
\begin{equation}\label{eq:Lform}
	\ell_{j,k}(\bm{\omega}) = \begin{dcases}\mathbb{I}_{\{j = k\}}& \text{if }1\leq k\leq p,\\
		\mathbb{I}_{\{j \geq p+1\}}\lambda_{m_k}^{j-p} & \text{if } p+1\leq k\leq p+r\\\mathbb{I}_{\{j \geq p+1\}}\gamma_{n_k}^{j-p} [\cos(j-p) \theta_{n_k} + \sin(j-p) \theta_{n_k}], & \text{if }p+r+1\leq k \leq d,\end{dcases} 
\end{equation}  
where $m_k=k-p$, $n_k=\lceil\frac{k-p-r}{2}\rceil$, and $\bm{\omega}=(\lambda_1,\dots,\lambda_r,\gamma_1,\theta_1,\dots,\gamma_s,\theta_s)^\prime$.

\begin{remark}[Infinite distributed lags]
	Based on \eqref{eq:linearcomb2gen},  $\lim_{j\rightarrow\infty}\ell_{j,k}(\bm{\omega})= 0$  for each $k$, which ensures that the coefficient matrix $\bm{A}_j$ diminishes to zero as the lag index  $j\rightarrow\infty$. This feature is reminiscent of the infinite distributed lag  model \citep{Judge1991}, which captures the gradually diminishing impact of past values of the variable $x_t\in\mathbb{R}$ on the response  $y_t\in\mathbb{R}$. For example, in the case with geometric lags, the model is 
	$y_t= \sum_{j=0}^{\infty}\beta_j x_{t-j} + e_t$, where $\beta_j= \lambda^j \alpha.$
	Here $\lambda\in(-1,1)$ and $\alpha\in\mathbb{R}$  are unknown parameters, which are analogous to  $\lambda_k$ and $\bm{G}_{k}$, respectively.  Obviously, $\alpha$ is extended to a matrix as the response and predictor in VAR($\infty$) models are  both $N$-dimensional.  However, \eqref{eq:linearcomb2gen} also extends geometric lags by allowing for various decay patterns.
\end{remark}

Other types of decay patterns can be found in the literature on infinite distributed lags, such as the Pascal lags, discounted polynomial lags, Gamma lags and exponential lags \citep{Judge1991}. Thus, for $\bm{A}_j=\sum_k\ell_{j,k}(\bm{\omega})\bm{G}_k$, a modeler may consider alternative functional forms of $\ell_{j,k}(\bm{\omega})$ based on their prior knowledge or  preferred interpretation, rather than relying solely on   \eqref{eq:linearcomb2gen}.

In  the sequel, for succinctness, we will restrict our discussion to the special case  in \eqref{eq:Lform}.
However, all theoretical results can be readily extended to general functions $\bm{L}$ under regularity conditions on the decay rate of $\ell_{j,k}(\cdot)$ as $j\rightarrow\infty$ and their smoothness.

\subsection{Identifiability}

In what follows, we will focus on the special case in \eqref{eq:Lform}. That is, given orders $(p,r,s)$, the model is written as
\begin{equation}\label{eq:model-scalar}
	\bm{y}_t=\sum_{j=1}^\infty  \bm{A}_j \bm{y}_{t-j}+\bm{\varepsilon}_t, \quad\text{with}\quad \bm{A}_j=
	\bm{A}_j(\bm{\omega},\cm{G})=\sum_{k=1}^{d}\ell_{j,k}(\bm{\omega})\bm{G}_k,
\end{equation}
where $d=p+r+2s$, $\bm{G}_k\in\mathbb{R}^{N\times N}$ for $1\leq k\leq d$,  and the parameter space of $\bm{\omega}$ is
\begin{equation}\label{eq:Omega}
	\bm{\Omega} =  \{\bm{\omega}\in\mathbb{R}^{r+2s} \mid  |\lambda_k|, \gamma_h \in(0,1), \theta_h \in(0,\pi) \text{ for }1\leq k\leq r \text{ and } 1\leq h\leq s\}.
\end{equation}
Note that under the specification in \eqref{eq:Lform},  $\ell_{j,k}(\bm{\omega})$ is the $(j,k)$-th entry of the matrix 
\begin{equation}\label{eq:Lfunc}
	\bm{L}(\bm{\omega})
	=\left (\begin{matrix}
		\bm{I}_p &\bm{0}&\cdots&\bm{0}&\bm{0}&\cdots&\bm{0}\\
		\bm{0}& \bm{\ell}^{I}(\lambda_1) & \cdots & \bm{\ell}^{I}(\lambda_r) & \bm{\ell}^{II}(\gamma_1,\theta_1) & \cdots & \bm{\ell}^{II}(\gamma_s,\theta_s)
	\end{matrix}\right )\in \mathbb{R}^{\infty\times d},
\end{equation}
where 
\begin{equation*}
	\bm{\ell}^{I}(\lambda)= (\lambda, \lambda^2, \lambda^3, \dots)^\prime\quad \text{and}\quad
	\bm{\ell}^{II}(\gamma,\theta) = 
	\left (\begin{array}{cccc}
		\gamma \cos(\theta)& \gamma^2\cos(2\theta)& \gamma^3\cos(3\theta)&\cdots\\
		\gamma \sin(\theta)& \gamma^2\sin(2\theta)& \gamma^3\sin(3\theta)&\cdots\\
	\end{array}\right )^\prime,
\end{equation*}	
for any $\lambda$ and $(\gamma,\theta)$.
Given the model orders $(p,r,s)$, the following theorem implies that the parameters $\bm{\omega}$ and $\bm{G}_1,\ldots,\bm{G}_d$ for this model are identifiable.
\begin{theorem}[Identifiability]\label{thm:identifiable}
	Suppose that $\bm{G}_1,\dots, \bm{G}_d\neq \bm{0}$ and $\bm{\omega}\in\bm{\Omega}$. If  $\lambda_1<\cdots <\lambda_r$, and the pairs $(\gamma_m,\theta_m)$'s are distinct and sorted in ascending order of $\gamma_m$'s and $\theta_m$'s, then there is a one-to-one correspondence between matrices $\{\bm{A}_1,\bm{A}_2,\dots\}$ and $\{\bm{\omega},  \bm{G}_1,\ldots,\bm{G}_d\}$, where $\bm{A}_j$'s are defined as in \eqref{eq:model-scalar}.
\end{theorem}

Since any VAR($\infty$) process is uniquely defined by its AR coefficient matrices $\{\bm{A}_1,\bm{A}_2,\dots\}$,  Theorem \ref{thm:identifiable} establishes the identifiability of $\bm{\omega}$ and $\bm{G}_1,\ldots,\bm{G}_d$ up to a permutation. Thus, unlike the VARMA model, no additional parameter constraint is needed for the identification of the parameters $\{\bm{\omega},  \bm{G}_1,\ldots,\bm{G}_d\}$. Moreover, with $\bm{A}_j$'s parameterized as linear combinations of matrices,  the computation for this model avoids any high-order matrix polynomials, which substantially reduces the computational cost compared to the VARMA model.



\subsection{The high-dimensional SARMA model} \label{subsec:HDmodel}
The general formulation in \eqref{eq:genmodel} addresses the third mode of $\cm{A}$, i.e., the temporal mode, by imposing $\pazocal{R}_3 = \rank(\cm{A}_{(3)})=d$.
When $N$ is large, we need to further conduct the dimension reduction for the first two modes of $\cm{A}$, i.e.,  the \textit{response} and \textit{predictor} modes. Specifically, we impose the  low-Tucker-rank assumption on $\cm{A}$ for its first two modes as follows: 
\[
\pazocal{R}_i = \rank(\cm{A}_{(i)})\ll N, \quad i=1,2.
\]
Note that $\rank(\cm{A}_{(i)})=\rank(\cm{G}_{(i)})$ for $i=1,2$, as the factorizations along different modes of the tensor do not interfere with each other. Thus, this is also equivalent to assuming that $\cm{G}$ has low Tucker ranks at its first two modes:
\begin{equation}\label{eq:Glowrank}
	\pazocal{R}_i = \rank(\cm{G}_{(i)})\ll N, \quad i=1,2.
\end{equation}
Then, under this assumption, there exist a small tensor $\cm{S}\in\mathbb{R}^{\pazocal{R}_1\times \pazocal{R}_2\times d}$  and full-rank matrices $\bm{U}_i\in\mathbb{R}^{N\times \pazocal{R}_i}$ for $i=1,2$ such that $\cm{G}=\cm{S}\times_1 \bm{U}_1\times_2\bm{U}_2$. Thus,
\begin{equation}\label{eq:tuckerL}
	\cm{A}=\underbrace{\cm{S}\times_1 \bm{U}_1\times_2\bm{U}_2}_{\cm{G}}\times_3 \bm{L}(\bm{\omega}):=[\![\cm{S};\bm{U}_1,\bm{U}_2,\bm{L}(\bm{\omega})]\!],
\end{equation}
In tensor algebra, \eqref{eq:tuckerL} is called the Tucker decomposition of the tensor $\cm{A}$, with $\cm{S}$ termed the core tensor, and $\bm{U}_1, \bm{U}_2$ and $\bm{L}(\bm{\omega})$ termed the factor matrices.
Note that the factorization of $\cm{G}$ is written mainly to facilitate the understanding of  low-Tucker-rank assumption; see Section \ref{subsec:df}. The unknown parameters to be estimated are still $\bm{\omega}$ and $\bm{G}_1,\dots, \bm{G}_d$ (i.e., the tensor $\cm{G}$). 

Similar to the low-rankness of matrices, the low-Tucker-rank assumption enables a reduction in the number of parameters for the coefficient tensor: it reduces the effective dimension of $\cm{G}$ from $N^2d$ to $O(N(\pazocal{R}_1+\pazocal{R}_2)+\pazocal{R}_1\pazocal{R}_2d)$. From the viewpoint of VAR($\infty$) modelling,   \eqref{eq:tuckerL} reveals that a simultaneous dimension reduction is conducted across the response, predictor, and temporal modes of the AR coefficient tensor $\cm{A}$.
To emphasize the resulting scalability across all three directions, we name model \eqref{eq:model-scalar} with the low-Tucker-rank assumption in \eqref{eq:Glowrank} for $\cm{G}$ the Scalable ARMA (SARMA) model.

\subsection{Dynamic factor interpretation}\label{subsec:df}

We discuss the interpretation of the low-Tucker-rank assumption in \eqref{eq:Glowrank} for the SARMA model and show that it implies low-dimensional dynamic factor structures underlying  both the response  and  lagged predictor series.

As the model is parameterized by $\bm{\omega}$ and  $\cm{G}$, it is not necessary to construct estimators for the components $\cm{S}$, $\bm{U}_1$ and $\bm{U}_2$ in the factorization of $\cm{G}$. Nonetheless, \eqref{eq:tuckerL} facilitates our understanding of the low-Tucker-rank assumption on $\cm{G}$ for the VAR($\infty$) model. 
It reveals that while $\bm{L}(\bm{\omega})$ extracts essential patterns from the temporal mode of the coefficient tensor $\cm{A}$, the matrices  $\bm{U}_1$ and $\bm{U}_2$ summarize information along the cross-sectional dimension of the response and lagged predictors, respectively. Note that 
\begin{equation}\label{eq:rotation}
	\cm{G}=\cm{S}\times_1 \bm{U}_1\times_2\bm{U}_2=(\cm{S}\times_1\bm{O}_1 \times_2\bm{O}_2)\times_1 (\bm{U}_1\bm{O}_1^{-1})\times_2(\bm{U}_2\bm{O}_2^{-1})
\end{equation}
for any invertible matrices $\bm{O}_i$ with $i=1,2$, indicating the  rotational and scale indeterminacies of the components. Without loss of generality, the normalization constraint
$\bm{U}_i'\bm{U}_i=\bm{I}_{\pazocal{R}_i}$ for $i=1,2$ can be imposed to facilitate interpretations.  

Moreover,   the SARMA model can be interpreted from the factor modelling perspective, with $\bm{U}_1$ and $\bm{U}_2$ representing loading matrices for the response factors and lagged predictor factors, respectively. To see this, first consider the simple example with Tucker ranks $\pazocal{R}_1=\pazocal{R}_2=1$. In this case, $\cm{S}=(s_1,\dots, s_d)^\prime\in\mathbb{R}^d:=\bm{s}$ and $\bm{U}_i:=\bm{u}_i\in\mathbb{R}^N$ for $i=1,2$ all reduce to  vectors, hence denoted by bold lowercase letters.  Then, \eqref{eq:tuckerL} implies $\bm{G}_k=s_k \bm{u}_1\bm{u}_2^\prime$ for $1\leq k\leq d$, which  are  rank-one matrices. As a result, 	$\bm{A}_j=\sum_{k=1}^{d}\ell_{j,k}(\bm{\omega})s_k \bm{u}_1\bm{u}_2^\prime$ for $j\geq1$. Note that $\bm{u}_1$ and $\bm{u}_2$ capture patterns from the rows and columns of $\bm{A}_j$'s, respectively. Consequently, with the normalization $\bm{u}_i^\prime \bm{u}_i=1$ for $i=1,2$, a single-factor model is implied:
\[
\underbrace{\bm{u}_1^\prime \bm{y}_t}_{\text{single response factor}} = \sum_{j=1}^{\infty} \sum_{k=1}^{d}\ell_{j,k}(\bm{\omega})s_k  \underbrace{\bm{u}_2^\prime \bm{y}_{t-j}}_{\text{single  predictor factor}} + \bm{e}_t,
\]
where $\bm{e}_t=\bm{u}_1^\prime\bm{\varepsilon}_t$. For instance, suppose that $\bm{y}_t$ contains  realized volatilities  of $N$ stocks in a market. Then $\bm{u}_1^{\prime} \bm{y}_t$ and $\bm{u}_2^\prime \bm{y}_{t-j}$ can be viewed as  latent \textit{response} and \textit{lagged predictor factors}, respectively, which can also be  regarded as two different  market volatility indices. The predictor factor loading $\bm{u}_2$ encapsulates how the the \textit{past} signals from various stocks are absorbed into the market, while the response factor loading  $\bm{u}_1$ summarizes the overall response of the \textit{present} market to these signals; see also Section \ref{sec:empirical} for an empirical example. 

For general Tucker ranks $\pazocal{R}_1$ and $\pazocal{R}_2$, analogously  we have 
\begin{equation}\label{eq:factors}
	\bm{U}_1^{\prime}\bm{y}_t
	=  \sum_{j=1}^\infty \sum_{k=1}^{d}\ell_{j,k}(\bm{\omega})\bm{S}_k \bm{U}_2^{\prime} \bm{y}_{t-j} + \bm{e}_t.
\end{equation}
Here, $\bm{U}_1^{\prime} \bm{y}_t$   represents  $\pazocal{R}_1$ response factors, while
$\bm{U}_2^{\prime} \bm{y}_{t-j}$ represents $\pazocal{R}_2$  lagged predictor factors. with the loading matrices being $\bm{U}_i\in\mathbb{R}^{N\times \pazocal{R}_i}$ for $i=1$ and 2, respectively. Thus, by imposing the low-Tucker-rank assumption on $\cm{G}$ in \eqref{eq:Glowrank},  simultaneous dimension reduction is achieved by extracting factors across both the response  and lagged predictors. For convenience, we call $\pazocal{R}_1$ and $\pazocal{R}_2$ the \textit{response} and \textit{predictor ranks}, respectively.

In addition, when $N$ is extremely large,  we may further assume that  $\bm{U}_1$ and $\bm{U}_2$ are sparse matrices for more efficient dimension reduction. This implies that each factor contains only a small subset  of variables. Take $\bm{U}_1^{\prime} \bm{y}_t$ as an example. For  $1\leq i\leq N$ and $1\leq k\leq \pazocal{R}_1$, if the $(i,k)$th entry of $\bm{U}_1$ is nonzero, then it implies that the $i$th variable in $\bm{y}_t$ is selected into the $k$th response factor. This sparsity assumption, which is embedded in the Tucker decomposition, will make the estimation of the SARMA model more challenging; see  Section \ref{subsec:hdest} for details.



\section{High-dimensional estimation}\label{sec:HDmodel}

\subsection{Rank-constrained estimator}\label{subsec:est}

We first introduce a rank-constrained   approach to estimate the  parameter vector $\bm{\omega}$ and the   low-Tucker-rank parameter tensor $\cm{G}$. As will be shown in Section \ref{subsec:theory}, this estimator is consistent under $N=o(T)$, where $T$ is the sample size; another estimation method applicable to  the ultra-high-dimensional case which allows  $\log (N)/ T \rightarrow0$ is introduced in Section \ref{subsec:hdest}.

Let $\bm{x}_{t}=(\bm{y}_{t-1}^\prime,\bm{y}_{t-2}^\prime,\dots)^\prime$. Then the squared error loss function is $\mathbb{L}_T(\bm{\omega},\cm{G})=\sum_{t=1}^{T} \| \bm{y}_t - \cm{A}_{(1)}\bm{x}_{t}\|_2^2$, where $\cm{A}_{(1)} = (\bm{A}_1,\bm{A}_2,\dots)$ with $\bm{A}_j =\bm{A}_j(\bm{\omega}, \cm{G})= \sum_{k=1}^d\ell_{j,k}(\bm{\omega})\bm{G}_k$  for $j\geq 1$.  Since the loss depends on observations in the  infinite past, initial values for $\{\bm{y}_t, t\leq 0\}$ are needed in practice. We set them to zero for simplicity, that is, let $\bm{\widetilde{x}}_{t}= (\bm{y}_{t-1}^\prime,\dots,\bm{y}_1^\prime,0,0,\dots)^\prime$ be the initialized version of $\bm{x}_t$, and define the feasible squared loss function:
\begin{equation}\label{eq:ls}
	\widetilde{\mathbb{L}}_T(\bm{\omega},\cm{G})=\sum_{t=1}^{T}\| \bm{y}_t - \cm{A}_{(1)}\bm{\widetilde{x}}_{t}\|_2^2 =\sum_{t=1}^{T} \Big \|	\bm{y}_t-\sum_{j=1}^{t-1}  \bm{A}_j(\bm{\omega},\cm{G}) \bm{y}_{t-j}\Big\|_2^2.
\end{equation}
The initialization effect will be accounted for in our theoretical analysis.

Suppose that the response and predictor ranks $(\pazocal{R}_1, \pazocal{R}_2)$ are known; see Section \ref{sec:selection} for a data-driven selection procedure.
When $N$ is moderately large compared to $T$,   we  propose the rank-constrained estimator as follows: 
\begin{equation}\label{eq:lse}
	(\bm{\widehat{\omega}},\cm{\widehat{G}})=\argmin_{\bm{\omega}\in\bm{\Omega}, \cmt{G}\in \bm{\Gamma}(\pazocal{R}_1,\pazocal{R}_2)}\widetilde{\mathbb{L}}_T(\bm{\omega},\cm{G}), 
\end{equation}
where  $\bm{\Omega}$ is defined in \eqref{eq:Omega}, and the parameter space of $\cm{G}$ is
\[
\bm{\Gamma}(\pazocal{R}_1,\pazocal{R}_2) = \{\cm{G}\in\mathbb{R}^{N\times N\times d}\mid \rank(\cm{G}_{(1)})\leq \pazocal{R}_1, \rank(\cm{G}_{(2)})\leq \pazocal{R}_2 \}.
\] 
Then based on the results from \eqref{eq:lse}, we can obtain $\cm{\widehat{A}} = \cm{\widehat{G}}\times_3\bm{L}(\bm{\widehat{\omega}})
$; i.e., the corresponding AR coefficient matrices are estimated by	$\bm{\widehat{A}}_j=\sum_{k=1}^d\ell_{j,k}(\bm{\widehat{\omega}})\bm{\widehat{G}}_k$ for $j\geq1$.

\begin{remark}
	Note that \eqref{eq:lse} does not require estimation of $\cm{S}$, $\bm{U}_1$ and $\bm{U}_2$, i.e., the components in the Tucker decomposition of $\cm{G}$. Thus, the rotational and scale indeterminacies in \eqref{eq:rotation} are not an issue. However, to interpret the underlying dynamic factor structure presented in \eqref{eq:factors},  it is beneficial to conduct the Tucker decomposition of $\cm{\widehat{G}}$ to obtain the corresponding estimated loading matrices $\bm{\widehat{U}}_1$ and  $\bm{\widehat{U}}_2$ after the rank-constrained estimation in \eqref{eq:lse}. A common approach to ensure the uniqueness of the Tucker decomposition is to employ the higher-order singular value decomposition (HOSVD), which is the special Tucker  decomposition as follows \citep{delathauwer2000multilinear}. Specifically, to get the HOSVD, $\cm{\widehat{G}}=\cm{\widehat{S}}\times_1\bm{\widehat{U}}_1\times_2\bm{\widehat{U}}_2$, the matrix $\bm{\widehat{U}}_i$ is defined  as the top $\pazocal{R}_i$ left singular vectors of $\cm{\widehat{G}}_{(i)}$ with the first element in each column of $\bm{\widehat{U}}_i$ being positive, for $i=1,2$. This rules out both rotational and sign indeterminacies. In addition, by the orthonormality of $\bm{\widehat{U}}_i$'s, we can compute $\cm{\widehat{S}}=\cm{\widehat{G}}\times_1\bm{\widehat{U}}_1'\times_2\bm{\widehat{U}}_2'$. Thus, the factor representation in \eqref{eq:factors} for the fitted model can be obtained.
	This will allow us to clearly interpret the dynamic factor structure based on the uniquely defined loading matrices $\bm{\widehat{U}}_1$ and $\bm{\widehat{U}}_2$.
\end{remark}

\subsection{Sparse low-Tucker-rank estimator}\label{subsec:hdest}

When $N$ is very large relative to the sample size $T$, the rank-constrained estimator can be inefficient, and a more substantial dimension reduction is needed. Motivated by the dynamic factor structure in \eqref{eq:factors}, we additionally assume that the loadings  $\bm{U}_1$ and $\bm{U}_2$  are sparse, and develop a high-dimensional  estimator that simultaneously enforces the low-Tucker-rank and sparse structures. This not only improves the estimation efficiency but enhances the interpretability as it  automatically   selects only important variables into the factors.

However, unlike the rank-constrained estimator in  \eqref{eq:lse}, explicit factorization of $\cm{G}$  must be incorporated into the sparse estimation.  Moreover, to ensure the identifiability of the sparsity patterns, we  assume that $\bm{U}_i$ is the orthonormal matrix consisting of the top $\pazocal{R}_i$ left singular vectors of $\cm{G}_{(i)}$, for $i=1,2$. This implies that $\cm{S}=\cm{G}\times_1\bm{U}_1'\times_2\bm{U}_2'$. Note that since $\bm{U}_i$ is orthonormal, it can be shown that $\cm{S}_{(i)}$ is row-orthogonal, for $i=1,2$.

We consider the following $\ell_1$-regularized sparse 	low-Tucker-rank (SLTR) estimator:
\begin{equation}\label{eq:sparse_lse}
	(\bm{\widetilde{\omega}},\cm{\widetilde{S}}, \bm{\widetilde{U}}_1,  \bm{\widetilde{U}}_2)=\argmin_{\bm{\omega}\in\bm{\Omega}, \, \cmt{S}\in\text{RO}(\pazocal{R}_1,\pazocal{R}_2), \,\bm{U}_i'\bm{U}_i=\bm{I}_{\pazocal{R}_i}, \, i=1,2}
	\left \{\widetilde{\mathbb{L}}_T(\bm{\omega},\cm{S}\times_{1}\bm{U}_1 \times_{2}\bm{U}_2) + \lambda\sum_{i=1}^{2}\|\bm{U}_i\|_{1} \right \}, 
\end{equation}
where 
\[
\text{RO}(\pazocal{R}_1,\pazocal{R}_2)=\{\cm{S}\in\mathbb{R}^{\pazocal{R}_1\times \pazocal{R}_2\times d}: \cm{S}_{(i)}\text{ is row-orthogonal},~i=1,2\}.
\]
Then it is straightforward to estimate $\cm{G}$ and $\cm{A}$ by $\cm{\widetilde{G}}=\cm{\widetilde{S}}\times_1\bm{\widetilde{U}}_1\times_2\bm{\widetilde{U}}_2$ and $\cm{\widetilde{A}} = \cm{\widetilde{G}}\times_3\bm{L}(\bm{\widetilde{\omega}})$, respectively; i.e., the estimated coefficient matrices $\bm{\widetilde{G}}_1,\dots, \bm{\widetilde{G}}_d$ and $\bm{\widetilde{A}}_j$ for $j\geq1$ can be obtained.



\subsection{Nonasymptotic error bounds}\label{subsec:theory}

This section provides nonasymptotic error bounds for the proposed rank-constrained and SLTR estimators, in the non-sparse and sparse cases, respectively. We assume that the observed time series  $\{\bm{y}_t\}_{t=1}^T$ is generated from a stationary SARMA model with  response and predictor ranks $(\pazocal{R}_1, \pazocal{R}_2)$. 


Let $\bm{\omega}^* \in\bm{\Omega}$ and $\cm{G}^*  \in\bm{\Gamma}(\pazocal{R}_1,\pazocal{R}_2)$ denote the true values of $\bm{\omega}$ and $\cm{G}$, respectively.
Similarly,  $\cm{A}^*$, $\lambda_k^*$'s, $\gamma_k^*$'s, $\theta_k^*$'s, etc., denote the true values of the corresponding parameters. 
To prove the consistency of the rank-constrained estimator, we make the following assumptions.

\begin{assumption}[Sub-Gaussian error]\label{assum:error}
	Let $\bm{\varepsilon}_t = \bm{\Sigma}_\varepsilon^{1/2}\bm{\xi}_t$, where  $\bm{\xi}_t$ is a sequence of i.i.d. random vectors with zero mean and $\var(\bm{\xi}_t) = \bm{I}_{N}$, and $\bm{\Sigma}_\varepsilon$ is a positive definite covariance matrix.
	In addition, the coordinates $(\bm{\xi}_{it})_{1\leq i\leq N}$ within $\bm{\xi}_t$ are mutually independent and $\sigma^2$-sub-Gaussian.
\end{assumption}

\begin{assumption}[Parameters]\label{assum:statn}
	(i) There exists an absolute constant $0<\bar{\rho}<1$ such that for all $\bm{\omega}\in\bm{\Omega}$, $|\lambda_1|,\ldots,|\lambda_r|, \gamma_1,\ldots,\gamma_s\in \Lambda$, where $\Lambda$ is  a compact subset of $(0,\bar{\rho})$;
	(ii) all $\lambda_k^*$'s are bounded away from each other, and all pairs  $(\gamma_m^*,\theta_m^*)$'s are bounded away from each other, for $1\leq k\leq r$ and $1\leq m\leq s$;
	and (iii) $\max_{1\leq k\leq d}\|\bm{G}_k^*\|_{\op}\leq C_{\cmtt{G}}$ for some absolute constant $C_{\cmtt{G}}>0$, and $ \|\bm{G}_k^*\|_{\Fr} \asymp \alpha$ for $p+1 \leq k \leq d$, where $\alpha=\alpha(N)>0$ may depend on the dimension $N$.
\end{assumption}


Assumption \ref{assum:error} is weaker than the commonly imposed Gaussian assumption in the literature on high-dimensional time series; see, e.g., \cite{basu2015regularized} and  \cite{WBBM21}. Assumption \ref{assum:statn}(i) requires $|\lambda_k|$'s and $\gamma_m$'s to be bounded away from   one. Assumption  \ref{assum:statn}(ii)  ensures that different elements of $\bm{\omega}^*$ can be distinguished in the estimation. While Assumption \ref{assum:statn}(iii) requires that $ \|\bm{G}_k^*\|_{\Fr}$ for $p+1\leq k\leq d$ have the same order of magnitude $\alpha$, it is allowed to vary with $N$. This condition can be readily relaxed through a slightly more involved proof. In this case, the lower and upper bounds of $ \|\bm{G}_k^*\|_{\Fr}$ will affect the error bounds.


While the proposed model is linear in $\cm{A}$,  the loss function in \eqref{eq:lse} is nonconvex with respect to  $\bm{\omega}$ and $\cm{G}$ jointly. As an intermediate step to prove the consistency of the proposed estimators, the following lemma allows us to linearize  $\cm{A}$ with respect to  $\bm{\omega}$ and $\cm{G}$ within a constant-radius neighborhood of $\bm{\omega}^*$; see Remark \ref{remark:radius} for more details about the radius $c_{\bm{\omega}}$. 

\begin{lemma}\label{lemma:delnorm}
	Under Assumption \ref{assum:statn}, for any  $\cm{A}=\cm{G} \times_3 \bm{L}(\bm{\omega})$ with $\cm{G}\in\mathbb{R}^{N\times N\times d}$ and  $\bm{\omega}\in\bm{\Omega}$,  if  $\|\bm{\omega} - \bm{\omega}^*\|_2\leq c_{\bm{\omega}}$, then
	$\|\cm{A}-\cm{A}^*\|_{\Fr} \asymp
	\|\cm{G}-\cm{G}^*\|_{\Fr} +\alpha\|\bm{\omega} - \bm{\omega}^*\|_2$,
	where $c_{\bm{\omega}}>0$ is a non-shrinking radius.
\end{lemma}


Note that  any stationary VAR($\infty$) process admits the VMA($\infty$)  representation, $\bm{y}_t =\bm{\Psi}_*(B)\bm{\varepsilon}_{t}$, where $B$ is the backshift operator, and $\bm{\Psi}_*(B) = \bm{I}_N+\sum_{j=1}^{\infty}\bm{\Psi}_j^* B^j$; see Theorem \ref{thm:sol} for a sufficient condition for the stationarity  of the SARMA model. Here we suppress the dependence of $\bm{\Psi}_j^*$'s on $\bm{A}^*_j$'s and hence $\bm{\omega}^*$ and $\bm{G}_k^*$'s for brevity. Let $\mu_{\min}(\bm{\Psi}_*) = \min_{|z|=1}\lambda_{\min}(\bm{\Psi}_*(z)\bm{\Psi}_*^{\HH}(z))$ and $\max_{|z|=1}\lambda_{\max}(\bm{\Psi}_*(z)\bm{\Psi}_*^{\HH}(z))$,
where   $\bm{\Psi}_*^{\HH}(z)$ is the conjugate transpose of $\bm{\Psi}_*(z)$ for $z\in\mathbb{C}$, and it can be verified that $\mu_{\min}(\bm{\Psi}_*) > 0$; see also \cite{basu2015regularized}. Then let
$\kappa_1=	\lambda_{\min}(\bm{\Sigma}_\varepsilon)\mu_{\min}(\bm{\Psi}_*) 	\min\{1, c_{\bar{\rho}}^2\}$ and  
$\kappa_2=\lambda_{\max}(\bm{\Sigma}_\varepsilon)\mu_{\max}(\bm{\Psi}_*)\max\{1, C_{\bar{\rho}}^2\}$,
where $c_{\bar{\rho}}, C_{\bar{\rho}}>0$ are absolute constants defined in Lemma \ref{lemma:fullrank} in the supplementary file. 


\begin{theorem}[Rank-constrained estimator]\label{thm:parametric}
	Let $d_{\pazocal{R}}=
	\pazocal{R}_1\pazocal{R}_2 d  + (\pazocal{R}_1+ \pazocal{R}_2)N$. Suppose that  $\|\bm{\widehat{\omega}} - \bm{\omega}^*\|_{2}\leq c_{\bm{\omega}}$ and 
	$T \gtrsim (\kappa_2 / \kappa_1)^2  d_{\pazocal{R}} \log(\kappa_2/\kappa_1)$. Then under  Assumptions \ref{assum:error} and \ref{assum:statn},   with probability at least $1-4e^{-cd_{\pazocal{R}}\log(\kappa_2/\kappa_1)}-8e^{-cN}-\{ 2+\sqrt{\kappa_2/\lambda_{\max}(\bm{\Sigma}_{\varepsilon})} \} \sqrt{N/\{ (\pazocal{R}_1+ \pazocal{R}_2) T\}}$, we have the following estimation and prediction error bounds:
	\begin{equation*}
		\|\cm{\widehat{A}} - \cm{A}^*\|_{\Fr} \lesssim  \sqrt{\frac{\kappa_2 \lambda_{\max}(\bm{\Sigma}_{\varepsilon}) d_{\pazocal{R}}}{\kappa_1^2 T}}
		\hspace{5mm}\text{and}\hspace{5mm}
		\frac{1}{T}\sum_{t=1}^{T} \left \| (\cm{\widehat{A}} - \cm{A}^*)_{(1)}\bm{\widetilde{x}}_{t} \right \|_2^2 \lesssim \frac{\kappa_2 \lambda_{\max}(\bm{\Sigma}_{\varepsilon})  d_{\pazocal{R}} }{\kappa_1 T}.
	\end{equation*} 
\end{theorem}

Combining Theorem \ref{thm:parametric} with Lemma \ref{lemma:delnorm}, we immediately have that with the same probability,   $\|\cm{\widehat{G}}-\cm{G}^*\|_{\Fr}\lesssim \sqrt{\kappa_2 \lambda_{\max}(\bm{\Sigma}_{\varepsilon}) d_{\pazocal{R}}/(\kappa_1^2 T)}$ and  $\|\bm{\widehat{\omega}} - \bm{\omega}^*\|_2 \lesssim \sqrt{\kappa_2 \lambda_{\max}(\bm{\Sigma}_{\varepsilon}) d_{\pazocal{R}}/(\alpha^2 \kappa_1^2 T)}$.

For the SLTR estimator, we make the following additional assumptions.

\begin{assumption}[Sparsity]\label{assum:sparse}
	Each column of the matrix $\bm{U}^*_i$ has at most $s_i$ nonzero entries, where $i=1, 2$.
\end{assumption}

\begin{assumption}[Restricted parameter space]\label{assum:para_add}
	The parameter spaces for $\cm{S}$ and $\bm{U}_i$ with $i=1$ or $2$ are $\bm{\Omega}_{\cmtt{S}}=\{\cm{S}\in\mathrm{AO}(\pazocal{R}_1, \pazocal{R}_2):\sigma_{1}(\cm{S}_{(i)}) \leq C_{\cmtt{S}}<\infty, i=1,2\}$ and $\pazocal{U}_i = \{\bm{U}\in\mathbb{R}^{N\times \pazocal{R}_i} \mid \bm{U}^\prime \bm{U} = \bm{I}_{\pazocal{R}_i}, \text{and }\bm{U}_{j,m}^2 \geq \underline{u}>0\text{ or }\bm{U}_{j,m}^2 = 0, \forall1\leq j\leq N, 1\leq m\leq \pazocal{R}_i\}$, respectively, where $\underline{u}$ is a uniform lower threshold, and $\bm{U}_{j,m}$ is the $(j,m)$-th entry of the matrix $\bm{U}$.
\end{assumption}

\begin{assumption}[Relative spectral gap]\label{assum:spec_gap}
	The nonzero singular values of $\cm{G}_{(i)}$ satisfy that $\sigma_{j-1}^2(\cm{G}_{(i)}) - \sigma_{j}^2(\cm{G}_{(i)}) \geq \beta \sigma_{j-1}^2(\cm{G}_{(i)})$ for $2 \leq j \leq \pazocal{R}_i$ and $i=1,2$, where $\beta>0$ is a constant. 
\end{assumption}

Assumption \ref{assum:sparse} defines the entrywise sparsity of $\bm{U}^*_i$'s. In Assumption \ref{assum:para_add}, the upper bound condition on $\cm{S}$ is mild since large singular values in $\cm{S}$ could cause nonstationarity of the process. The lower threshold $\underline{u}$ for  $\bm{U}_i$'s is needed to establish the restricted eigenvalue condition \citep{Bickel2009}. Since $\underline{u}$  may shrink to zero as the dimension increases, this is not a stringent condition.  Assumption \ref{assum:spec_gap} requires that the singular values of $\cm{G}_{(i)}$'s are well separated to ensure identifiability. See \cite{Wang2021High} for similar assumptions.
The consistency of the SLTR estimator is established as follows.

\begin{theorem}[SLTR estimator]\label{thm:sparse}
	Let $d_{\pazocal{S}} = \pazocal{R}_1\pazocal{R}_2 d + \sum_{i=1}^{2} s_i \pazocal{R}_i \log (N\pazocal{R}_i)$. Suppose that $\|\bm{\widetilde{\omega}} - \bm{\omega}^*\|_2\leq c_{\omega}$ and $T \gtrsim d_{\pazocal{S}} + \underline{u}^{-1}\sum_{i=1}^{2}\pazocal{R}_i \log (N\pazocal{R}_i) + \underline{u}^{-2}(s_1\vee s_2)(\pazocal{R}_1\vee \pazocal{R}_2)(d+\log N)$. Then under Assumptions \ref{assum:error}--\ref{assum:spec_gap}, if $\lambda \gtrsim \sqrt{\kappa_2 \lambda_{\max}(\bm{\Sigma}_{\varepsilon})d_{\pazocal{S}}/{T}}$, 	with probability at least $1-5e^{-cd_{\pazocal{S}}-c\underline{u}^{-1}(\pazocal{R}_1 + \pazocal{R}_2 + \log N\pazocal{R}_1 + \log N\pazocal{R}_2)}-7e^{-cs_2\log N(\pazocal{R}_1 \wedge \pazocal{R}_2)} - c\sqrt{(s_2+\underline{u}^{-1})\pazocal{R}_2/T}(1+\sqrt{(s_2+\underline{u}^{-1})\pazocal{R}_2/d_{\pazocal{S}}})$, it holds
	\[
	\|\cm{\widetilde{A}} - \cm{A}^*\|_{\Fr} \lesssim \frac{(\eta_1 + \eta_2) \sqrt{s_1 + s_2} \lambda}{\beta\kappa_1}\hspace{2mm}\text{and}\hspace{2mm}  \frac{1}{T}\sum_{t=1}^{T}\|(\cm{\widetilde{A}} - \cm{A}^*)_{(1)}\bm{\widetilde{x}}_{t}\|_2^2 \lesssim \frac{(\eta_1 + \eta_2)^2 (s_1 + s_2) \lambda^2}{\beta^2\kappa_1},
	\]
	where $\eta_i = \sum_{j=1}^{\pazocal{R}_i}\sigma_{1}^2(\cm{G}_{(i)}^*) / \sigma_{j}^2(\cm{G}_{(i)}^*)$ for $i=1, 2$.
\end{theorem}

Taking $\lambda  \asymp \sqrt{\kappa_2 \lambda_{\max}(\bm{\Sigma}_{\varepsilon})d_{\pazocal{S}}/{T}}$, the estimation and prediction error bounds in Theorem \ref{thm:sparse}  become $(\eta_1 + \eta_2) \sqrt{(s_1 + s_2)	\kappa_2 \lambda_{\max}(\bm{\Sigma}_{\varepsilon})d_{\pazocal{S}} /(\beta^2\kappa_1^2 T)}$ and $(\eta_1 + \eta_2) (s_1 + s_2)	\kappa_2 \lambda_{\max}(\bm{\Sigma}_{\varepsilon})d_{\pazocal{S}} /(\beta^2\kappa_1 T)$, respectively. Then, in view of Lemma \ref{lemma:delnorm}, the high probability bounds for $\|\cm{\widetilde{G}}-\cm{G}^*\|_{\Fr}$ and  $\|\bm{\widetilde{\omega}} - \bm{\omega}^*\|_2$ can be easily obtained.

In practice, the ranks  $\pazocal{R}_1, \pazocal{R}_2$ and model orders $p,r,s$ are usually small. Then by fixing the constants $\lambda_{\min}(\bm{\Sigma}_\varepsilon)$,  $\lambda_{\max}(\bm{\Sigma}_\varepsilon)$,  $\mu_{\min}(\bm{\Psi}_*)$, $\mu_{\max}(\bm{\Psi}_*)$, $\eta_1, \eta_2$ and $\beta$, the estimation error bound for the rank-constrained estimator $\cm{\widehat{A}}$ can be simplified to $\sqrt{N/T}$, while that for the SLTR estimator $\cm{\widetilde{A}}$  reduces to $\sqrt{(s_1+s_2)^2\log (N) /T}$. 

\begin{remark}\label{remark:radius}
	We give more details about the non-shrinking radius $c_{\bm{\omega}}$ in Lemma \ref{lemma:delnorm}. The result of Lemma \ref{lemma:delnorm} comes from the following first-order Taylor expansion:
	$\bm{\Delta}(\bm{\omega}, \cm{G}) =\cm{{A}}(\bm{\omega}, \cm{G})-\cm{A}^* = \cm{M}(\bm{\omega}-\bm{\omega}^*, \cm{G}-\cm{G}^*) \times_3\bm{L}_{\rm{stack}}(\bm{\omega}^*) + \text{remainder}$,
	where  $\cm{M}:  \mathbb{R}^{r+2s} \times \mathbb{R}^{N\times N\times d} \to \mathbb{R}^{N \times N\times (d+r+2s)}$ is a bilinear function, and $\bm{L}_{\rm{stack}}(\bm{\omega}^*)$ is a $\infty\times (d+r+2s)$ constant matrix; see the proof of Lemma \ref{lemma:delnorm} in the  supplementary file. The negligibility of the remainder term  requires that  $\bm{\omega}$ lies  within a constant radius  of $\bm{\omega}^*$. In our proof, we derive the radius
	$c_{\bm{\omega}}= \min \left \{2, \frac{c_{\cmtt{G}}(1-\bar{\rho})\sigma_{\min, L} }{8\sqrt{2}C_L}\right \}$,
	where $\sigma_{\min, L}:=\sigma_{\min}(\bm{L}_{\rm{stack}}(\bm{\omega}^*))$, $c_{\cmtt{G}}:=\min_{p+1\leq k\leq d}\|\bm{G}_k^*\|_{\Fr} / \max_{p+1\leq k\leq d}\|\bm{G}_k^*\|_{\Fr}$, and $C_L>0$ is an absolute constant given in Lemma \ref{cor1} in the supplementary file. Note that Assumption \ref{assum:statn} implies that $\sigma_{\min, L}>0$  and  $c_{\cmtt{G}}>0$ are both absolute constants: the former is shown by Lemma \ref{lemma:fullrank} in the supplementary file, and the latter is a direct consequence of Assumption \ref{assum:statn}(iii). Thus, the radius $c_{\bm{\omega}}$ is  non-shrinking. 
\end{remark}

\begin{remark}
	In the proofs of  Theorems \ref{thm:parametric} and \ref{thm:sparse},  we show that the effect of initial values for $\{\bm{y}_t,t\leq 0\}$ has no contribution to the final estimation error rates; see  the quantities $|S_i(\widehat{\bm{\Delta}})|$ for $1\leq i\leq 3$ in the supplementary file.  We bound the initialization error terms by Markov's inequality, resulting in a nonexponential tail probability, which  may be sharpened by employing more sophisticated concentration inequalities. 
\end{remark}

\begin{remark}
	The  ranks  $\pazocal{R}_1, \pazocal{R}_2$ and  model orders $p,r,s$ are unknown in practice, but can be selected by data-driven methods. Due the page limit, we discuss them in Section \ref{sec:selection} of the supplementary material. We present an easy-to-implement estimator of the ranks, which is proved to be consistent. Model order selection  via the  Bayesian information criterion  is also discussed therein.
\end{remark}

\section{Two empirical examples}\label{sec:empirical}

\subsection{Macroeconomic dataset}\label{subsec:macro}
This dataset contains observations of 20 quarterly macroeconomic variables from June 1959 to December 2019, with $T=243$, retrieved from FRED-QD \citep{MN16}. These variables come from four categories: (i) stock market, (ii) exchange rates, (iii) money and credit, and (iv) interest rates. These categories are usually considered in the construction of financial condition index, since they reflect important factors that can affect the stance of monetary policy and aggregate demand conditions
\citep{goodhart2001asset,bulut2016financial,hatzius2010financial}. All series are transformed to be stationary, and standardized to have zero mean and unit variance; see the supplementary file for more details of the variables and their transformations.

We first explore the factor structures of this dataset. As discussed in Section \ref{subsec:HDmodel}, $\bm{U}_1$ and  $\bm{U}_2$  capture response factor and predictor factor spaces, respectively.  By the rank selection method in Section \ref{sec:selection}, we obtain $(\widehat{\pazocal{R}}_1, \widehat{\pazocal{R}}_2) = (3,3)$. 
Figure~\ref{fig:sparse} displays $\widehat{\bm{U}}_1$ and $\widehat{\bm{U}}_2$ based on the proposed SLTR estimation in Section \ref{subsec:HDmodel}, where the regularization parameter is selected by cross-validation. Overall, it can be observed that the response factors (RFs) are mainly influenced by variables in categories (i) and (iv) and business loan indicator from category (iii), while the influence from categories (ii) is relatively weak. On the other hand, only the S\&P 500 index from category (i) and category (iv) contributes significantly to the predictor factors (PFs).

\begin{figure}[t]
	\centering
	\includegraphics[width=1.\linewidth]{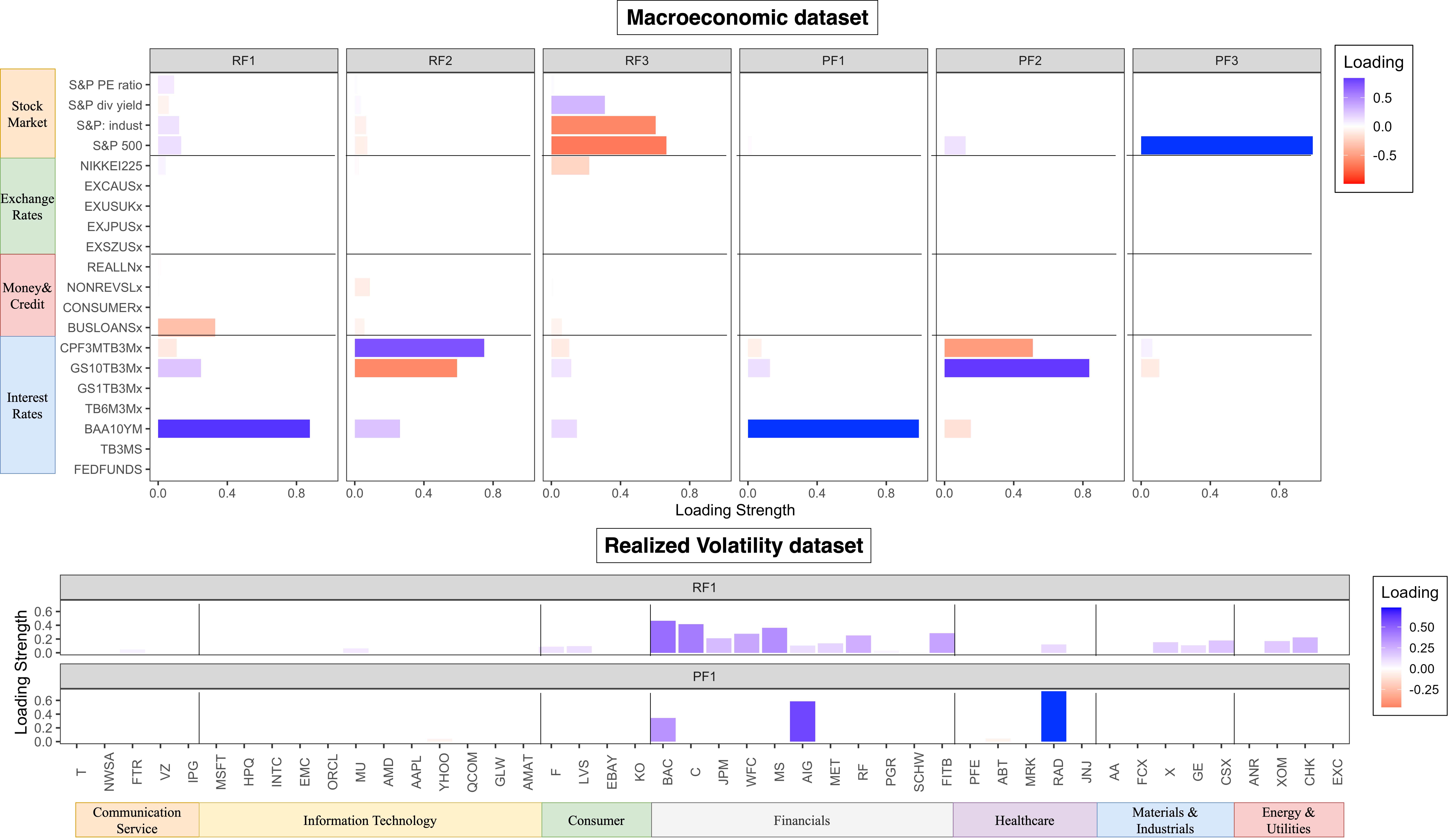}
	\caption{SLTR estimates of factor loadings in the proposed SARMA model for the macroeconomic and realized volatility datasets. Response factors (``RF''s) correspond to columns of $\bm{U}_1$ while predictor factors (``PF''s)  correspond to columns of $\bm{U}_2$.}
	\label{fig:sparse}
\end{figure}

We evaluate the performance of our method based on out-of-sample forecast accuracy. The following rolling forecast procedure is adopted: we first fit the models using  historical data with the end point rolling from the fourth quarter of 2015 to the third quarter of 2019, and then conduct one-step-ahead forecasts based on the fitted models. In addition to the proposed rank-constrained (RC) and SLTR estimators, we consider five other existing methods, including three based on the VAR model and two based on the VARMA model.  Specifically, for the VAR model, we consider (a) the Lasso method \citep{basu2015regularized} and two methods in \cite{Wang2021High}: (b) the multilinear low-rank (MLR) method and (c) the sparse higher-order reduced-rank (SHORR) method,  which  further imposes sparsity on the factor matrices in (b) using a slightly different regularizer than the method in this paper. 
For the VARMA model, we apply the method in \cite{WBBM21} with (d) the $\ell_1$-penalty or (e) the HLag penalty. Note that (a) is used as the Phase-I estimator for the estimators in (d) and (e), and the AR order is selected according to  \cite{WBBM21}.  The AR order for (b) and (c) is chosen as in \cite{Wang2021High}. For the proposed low-Tucker-rank SARMA model,  the estimated model orders are  $(\widehat{p}, \widehat{r}, \widehat{s}) = (0,1,0)$.
Throughout the rolling forecast procedure, the same model orders and ranks are used.

Table~\ref{tab:real} reports the mean squared forecast error (MSFE) and mean absolute forecast error (MAFE) for all methods. It can be observed that the proposed methods achieve the smallest forecast errors among all competing ones. Compared to sparse but non-low-rank models, i.e., (a), (d) and (e), the proposed model can better capture the factor structure which is prominent in this dataset. Meanwhile, its higher flexibility than the VAR model is supported by its  better forecasting performance than (b) and (c). In addition, note that imposing sparsity on the factor matrices generally results in smaller forecast errors for both VAR and SARMA models; see Figure \ref{fig:sparse}.

\begin{table}[t]\small
	\caption{Forecast errors for macroeconomic and realized volatility datasets. The smallest numbers in each row are marked in bold.}
	\label{tab:real}
	\centering
	\resizebox{\textwidth}{!}{\begin{tabular}{@{}rrrrrrrrr@{}}
			\hline
			&&\multicolumn{3}{c}{\makecell{VAR}}&\multicolumn{2}{c}{\makecell{VARMA }} & \multicolumn{2}{c}{\makecell{SARMA }} \\
			\cmidrule(r){3-5}\cmidrule(r){6-7}\cmidrule(r){8-9}
			&&(a) Lasso& (b) MLR& (c) SHORR & (d) $\ell_1$& (e) HLag &  RC &  SLTR\\
			\hline
			\multirow{2}{1.1cm}{\makecell{Macro}}&MSFE&2.78&2.77&2.71&2.80&2.79&2.67&\textbf{2.62}\\
			&MAFE&9.26&9.27&8.99&9.28&9.24&8.75&\textbf{8.45}\\
			\hline
			\multirow{2}{1cm}{\makecell{RV}}&MSFE&5.17&4.93&4.87&5.19&5.19&4.78&\textbf{4.74}\\
			&MAFE&21.58&19.02&18.22&21.70&21.70&\textbf{16.45}&16.99\\
			\hline
	\end{tabular}}
\end{table}


\subsection{Realized volatility}

As another example, we study  daily realized volatilities for 46 stocks from January 2, 2012 to December 31, 2013, with $T=495$. These are the  stocks of top S\&P 500 companies ranked by trading volumes on the first day of 2013.
Specifically, we obtain the tick-by-tick data  from WRDS (\url{https://wrds-www.wharton.upenn.edu}) and compute the daily realized volatility from five-minute returns \citep{andersen2006volatility}. By examining the sample autocorrelation functions, we have confirmed the stationarity of all series. Each series is then standardized to have zero mean and unit variance. More information about the stocks is given in Table S.2 in the supplementary file.
We conduct the same rolling forecast procedure as in Section \ref{subsec:macro}, where the last 10\% of the sample is used as the forecast period. As shown in Table~\ref{tab:real}, the proposed methods considerably outperform the other ones in terms of forecast accuracy. 


The estimated ranks and model orders  are  $(\widehat{\pazocal{R}}_1, \widehat{\pazocal{R}}_2, \widehat{p}, \widehat{r}, \widehat{s}) = (1,1,0,1,0)$. As a result, the fitted model has the following factor structure: $\widehat{\bm{u}}_1^{\prime} \bm{y}_t= 0.336 \sum_{j=1}^\infty 0.872^j \widehat{\bm{u}}_2^{\prime} \bm{y}_{t-j}  + \bm{e}_t$, where the loadings $\bm{\widehat{u}}_1$ and $\bm{\widehat{u}}_2$ are displayed in Figure \ref{fig:sparse}. We have several interesting findings. First, $\widehat{\lambda}=0.872$ indicates that the influence of the past on the present decays quite slowly. 	 This lends  support to the well-established fact that the volatility of asset returns is highly persistent, that is, the AR process of the volatility is nearly unit-root; see, e.g., \cite{ABDL03}. Second, it can be observed that the weights in $\bm{\widehat{u}}_1$ are more evenly spread out across four sectors, including Financials, Healthcare, Material \& Industrials, and Energy \& Utilities. However, the weights in $\bm{\widehat{u}}_2$ are more concentrated on a few stocks. As discussed in Section \ref{subsec:df}, $\widehat{\bm{u}}_1^{\prime} \bm{y}_t$ and $\widehat{\bm{u}}_2^{\prime} \bm{y}_{t}$  can be regarded  as  two different market volatility indices, with the loadings $\widehat{\bm{u}}_1$  and $\widehat{\bm{u}}_2$  capturing how the market responds to and  picks up risks across stocks, respectively. Lastly, the estimated slope $0.336$ signifies the overall association between $\bm{y}_t$ and its lags, after summarizing the information across all stocks into market indices, while taking into account the decaying temporal dependence over lags. It shows that the present and past volatilities are positively correlated, a phenomenon commonly known as the volatility clustering in the literature of financial time series \citep{Tsay2010}.


\section{Conclusion and discussion}\label{sec:conclusion}

This paper contributes to the underdeveloped literature on high-dimensional  VARMA models. First, the originally unwieldy VARMA form is turned into a much more tractable infinite-order VAR form. Second, building on the close connection between this form and the tensor decomposition for the AR coefficient tensor $\cm{A}$, a low-Tucker-rank structure is naturally considered, so that dimension reduction can be simultaneously performed across all time lags and variables. 
In summary, by combining the reparameterization and tensor decomposition techniques, this paper expands the available model family for high-dimensional time series from finite-order VAR to VARMA processes. 

Moreover, a comprehensive high-dimensional estimation procedure is developed, together with theoretical properties and efficient   algorithms that leverage the tractable form of the model. The limited literature on high-dimensional VARMA and VAR($\infty$) models focuses on sparse coefficient matrices. To the best of our knowledge, this is the first approach to accommodate low-rank structures across both response and predictor dimensions in high-dimensional VAR($\infty$) processes.

There are  many worthwhile questions  that remain to be explored. Firstly, the  convergence theory developed for the estimators in this paper focuses on the statistical error, whereas the optimization error of the algorithm is not studied. For the nonconvex estimation of low-rank tensor models, \cite{Han2021} establishes both the statistical error bound and the linear rate of computational convergence of their proposed algorithm. For our model, the main difficulty in conducting such an algorithmic analysis  lies in the nonconvexity of the coefficient tensor $\cm{A}$ with respect to $\bm{\omega}$. Second, it is important to develop high-dimensional statistical inference procedures for the proposed model. So far there have been limited studies on inference for low-rank tensor regression models. A recent work is \cite{Xia2022} which, however, focuses on $i.i.d.$ low-Tucker-rank models with non-sparse factor matrices. Extensions of such asymptotic distributional results to the time series setting can be challenging. Moreover, when the factor matrices are sparse, the corresponding inference will be even more difficult, and debiasing techniques are likely inevitable. We leave these interesting problems to future research.

\section*{Supplementary Materials}

The Supplementary Material contains methods for rank and model order selection, algorithms for the proposed estimators, simulation studies, all technical details,  and  additional details for empirical studies in this paper.

\putbib[SARMA]
\end{bibunit}

\newpage
\renewcommand{\arraystretch}{0.85}

\begin{bibunit}[apalike]
\vspace*{10pt}	
\begin{center}
	{\Large \bf Online Supplement for ``SARMA:  Scalable Low-Rank High-Dimensional Autoregressive Moving Averages via Tensor Decomposition"}
\end{center}
\vspace{10pt}

\begin{abstract}
This supplementary file contains six sections. Section \ref{sec:selection} develops a method for selecting the response and predictor ranks, which is further proved to be consistent. Section \ref{section:algo} presents algorithms for the proposed rank-constrained and SLTR estimators. Section \ref{sec:sim} presents simulation studies for verifying the theoretical results.
Section \ref{sec5} provides descriptions of datasets  for the
empirical examples in the main paper. All technical proofs for the theoretical results in Sections \ref{section:model} and \ref{sec:HDmodel} in  main paper are provided in Sections \ref{sec:proof1} and \ref{sec:proof2}, respectively.
\end{abstract}

\renewcommand\thetable{S.\arabic{table}}
\renewcommand{\thesection}{S\arabic{section}}
\renewcommand{\thesubsection}{S\arabic{section}.\arabic{subsection}}
\renewcommand{\thetheorem}{S\arabic{theorem}}
\renewcommand{\theproposition}{S\arabic{proposition}}
\renewcommand{\theequation}{S\arabic{equation}}
\renewcommand{\thefigure}{S\arabic{figure}}
\renewcommand{\thelemma}{S\arabic{lemma}}
\setcounter{lemma}{0}
\setcounter{section}{0}

\section{Model selection}\label{sec:selection}
As the response and predictor ranks  are unknown in practice, we provide a data-driven method to select them and prove  the consistency of the estimated ranks.

Denote the true values of the ranks by $(\pazocal{R}_1^*, \pazocal{R}_2^*)$.
Suppose that $\cm{\widehat{A}}^{\text{init}}$ is a consistent initial estimator of $\cm{A}^*$; see Remark \ref{remark:init} for a detailed discussion on its choice. Denote by $\widehat{\sigma}_j(i)$ and $\sigma_j^*(i)$ the $j$th largest singular value of $\cm{\widehat{A}}^{\text{init}}_{(i)}$ and $\cm{A}^*_{(i)}$, respectively, for $i=1$ or $2$.  
We adopt the ridge-type ratio estimator \citep{Xia2015Consistently, Wang2021High}:
\[
\widehat{\pazocal{R}}_{i} = \argmin_{1\leq j \leq N-1} \frac{\widehat{\sigma}_{j+1}(i)+\tau}{\widehat{\sigma}_j(i)+\tau}, \quad i = 1, 2,
\] 
where $\tau$ is a parameter to be chosen such that Assumption \ref{assum:rank-select} below is satisfied.

Let
\[
\zeta_i =  \frac{1}{\sigma^*_{\min}(i)} \cdot \max_{1\leq j \leq \pazocal{R}_i^* -1}\frac{\sigma_j^*(i)}{ \sigma^*_{j+1}(i) }, \quad \text{for}\hspace{2mm}i=1,2,
\]
where $\sigma^*_{\min}(i)$ is the minimum singular value of $\cm{A}^*_{(i)}$.
The following assumption is needed for the consistency of the rank selection method.

\begin{assumption}[Signal strength]\label{assum:rank-select}
	The parameter $\tau>0$ is specified such that (i) $\|\cm{\widehat{A}}^{\text{init}} - \cm{A}^*\|_{\Fr}/\tau = o_p(1)$; and (ii) $\tau\max \{\zeta_1, \zeta_2\} = o(1)$.
\end{assumption}

In Assumption \ref{assum:rank-select}, condition (i) requires that the estimation error of $\cm{\widehat{A}}^{\text{init}}$ is dominated by $\tau$, and condition (ii) can be regarded as the minimal signal assumption which will simply reduce to $\tau=o(1)$ if $\sigma_{j}^*$ for $1\leq j\leq \pazocal{R}_i^*$ and $i=1,2$ are bounded above and away from zero by some absolute constant. Following \cite{Wang2021High}, it is straightforward to establish the consistency of the estimator.

\begin{theorem}\label{thm:rankselection}
	Under Assumption \ref{assum:rank-select}, $\mathbb{P}( \widehat{\pazocal{R}}_1 = \pazocal{R}_1^*, \widehat{\pazocal{R}}_2 = \pazocal{R}_2^*) \to 1$ as $T\to\infty$. 
\end{theorem} 

\begin{remark}\label{remark:init}
	We can obtain the initial estimator $\cm{\widehat{A}}^{\text{init}}$ through a VAR($P$) approximation of the VAR($\infty$) process, where $P$ scales with the sample size $T$ \citep{Lutkepohl2005}. Let $\cm{A}_{\mathrm{trim}}$ be a truncated form of $\cm{A}$ such that $(\cm{A}_{\mathrm{trim}})_{(1)} = (\bm{A}_1,\ldots,\bm{A}_P)$. We begin by estimating $\cm{\widehat{A}}^{\text{init}}_{\mathrm{trim}}$, and then append infinitely many zero matrices to $\cm{\widehat{A}}^{\text{init}}_{\mathrm{trim}}$ to obtain $\cm{\widehat{A}}^{\text{init}}$ with  $\cm{\widehat{A}}^{\text{init}}_{(1)}=((\cm{\widehat{A}}^{\text{init}}_{\mathrm{trim}})_{(1)}, \bm{0}_{N\times N}, \bm{0}_{N\times N}, \dots )$.  Following Proposition 4.2 in \cite{WBBM21}, under regularity conditions, the approximation error due to the truncation after lag $P$  can be shown to be negligible if $P\asymp T^{1/2-\epsilon}$, where $\epsilon\in(0,1/2)$. Some possible choices for $\cm{\widehat{A}}^{\text{init}}$ are as follows: (a) The nuclear norm regularized estimator $\cm{\widehat{A}}_{\pazocal{R},\mathrm{trim}} = \argmin_{\cmtt{A}\in\mathbb{R}^{N\times N\times P}}\sum_{t=P+1}^{T}\|\bm{y}_t - \sum_{j=1}^{P}\bm{A}_j\bm{y}_{t-j}\|_2^2 /(T-P) + \lambda_{\text{nuc}}\sum_{i=1}^{2} \|(\cm{A}_{\mathrm{trim}})_{(i)}\|_*$,  where $\lambda_{\text{nuc}}>0$, and the low-rankness of $(\cm{A}_{\mathrm{trim}})_{(i)}$ for $i=1, 2$ is enforced via the nuclear norm penalty; see, e.g., \cite{gandy2011tensor} and \cite{Raskutti17}; (b) the group-lasso estimator $\cm{\widehat{A}}_{\pazocal{S},\mathrm{trim}} = \argmin_{\cmtt{A}\in\mathbb{R}^{N\times N\times P}}\sum_{t=P+1}^{T}\|\bm{y}_t - \sum_{j=1}^{P}\bm{A}_j\bm{y}_{t-j}\|_2^2 /(T-P) + \lambda_{\text{lag}} \sum_{j=1}^{P}\|\bm{A}_{j}\|_{\Fr}$, which 	corresponds to the lag-sparse estimator  in \cite{nicholson2017varx}; and (c) the spectral estimator in \cite{Han2021} which captures the low-Tucker-rank structure of $\cm{A}$.
	In practice, we suggest setting $P = T^{1/3}$ and choosing the regularization parameters $\lambda_{\text{nuc}}$ and $\lambda_{\text{lag}}$ chosen by the time series cross-validation method similar to that in \cite{WBBM21}.  For the non-sparse case, we employ (a) to obtain the initialization for the rank-constrained estimator. Along the lines of the proofs of Theorem 2 in \cite{WZL21}, under some regularity conditions, it can be shown that $\|\cm{\widehat{A}}^{\text{init}} - \cm{A}^*\|_{\Fr}=O_p\{\sqrt{(\pazocal{R}_1^*+\pazocal{R}_2^*)NP/(T-P)}\}$. For the sparse case, we recommend (b) for initializing the SLTR estimator, and it can be shown that $\|\cm{\widehat{A}}^{\text{init}} - \cm{A}^*\|_{\Fr}=O_p\{\sqrt{N^2\log P/(T-P)}\}$.
\end{remark}

\begin{remark}\label{remark:order}
	In practice,  the model orders $(p,r,s)$ also need to be chosen. Suppose the Tucker ranks are consistently estimated, we can then select the model orders by minimizing the  Bayesian information criterion (BIC), $\textup{BIC}(p,r,s) = \log\{ T^{-1}\sum_{t=1}^{T} \|	\bm{y}_t-\sum_{j=1}^{t-1}  \bm{A}_j(\breve{\bm{\omega}},\breve{\cm{G}}) \bm{y}_{t-j}\|_2^2\} +T^{-1} c d_{\pazocal{M}} \log T$, where $(p,r,s)$ is searched over the range $0\leq p \leq p_{\max}$, $0\leq r \leq r_{\max}$, and $0\leq s \leq s_{\max}$, for some predetermined upper bounds, $c>0$ is a constant, and $\bm{\breve{\omega}}$ and $\cm{\breve{G}}$ are the estimates obtained  by fitting the model with orders $(p, r, s)$ using either the rank-constrained estimator or the SLTR estimator. In addition, 
	$d_{\pazocal{M}} = 
	\pazocal{R}_1\pazocal{R}_2 d  + (\pazocal{R}_1+ \pazocal{R}_2)N$ for the former, and  $d_{\pazocal{M}} = \pazocal{R}_1\pazocal{R}_2 d + \sum_{i=1}^{2} \pazocal{R}_i \log (N\pazocal{R}_i)$ for the latter. Then, the consistency of the selected model orders via the BIC can be established along the lines of \cite{sparseARMA}.
\end{remark}

\section{Algorithms}\label{section:algo}

\subsection{Algorithm for the rank-constrained estimator}
We first consider the algorithm for the rank-constrained estimator. By the factorization $\cm{G} = \cm{S}\times_1\bm{U}_1\times_2\bm{U}_2$, the  rank-constrained estimation in \eqref{eq:lse} can be rewritten as the unconstrained problem,
\begin{equation}\label{eq:als}
	(\bm{\widehat{\omega}},\cm{\widehat{S}},\bm{\widehat{U}}_1,\bm{\widehat{U}}_2)=\argmin \widetilde{\mathbb{L}}_T(\bm{\omega},\cm{S},\bm{U}_1,\bm{U}_2).
\end{equation}
Then  we have $\cm{\widehat{G}} = \cm{\widehat{S}}\times_1\bm{\widehat{U}}_1\times_2\bm{\widehat{U}}_2$.  Note that $\bm{U}_i$'s need not be subject to any orthogonality constraint in this minimization.

To implement \eqref{eq:als}, we adopt an alternating minimization algorithm; see Algorithm \ref{alg:sarma}. Note that the optimization for $\lambda_k$'s and $(\gamma_k, \theta_k)$'s in lines 3--6 is efficient due to the following property:
\[
\cm{A}_{(1)}\bm{\widetilde{x}}_t- \sum_{k=1}^p\bm{G}_{k}\bm{y}_{t-k}=\sum_{k=1}^r f^I(\widetilde{\bm{x}}_t; \lambda_{k})+\sum_{k=1}^s f^{II}(\widetilde{\bm{x}}_t; \gamma_{k}, \theta_{k}),
\]
where $f^I(\widetilde{\bm{x}}_t; \lambda_{k})=\sum_{j=1}^{t-p-1} \lambda_k^j \bm{G}_{p+k} \bm{y}_{t-p-j}$, and  
$f^{II}(\widetilde{\bm{x}}_t; \gamma_{k}, \theta_{k})=\sum_{j=1}^{t-p-1}\gamma_{k}^{j} [ \cos(j\theta_{k}) \bm{G}_{p+r+2k-1}+\sin(j\theta_{k}) \bm{G}_{p+r+2k}] \bm{y}_{t-p-j}$. Note that  each  $\lambda_k$ or $(\gamma_{k}, \theta_{k})$ appears in only one summand. Thus, fixing all other parameters, the optimization problem for each  $\lambda_k$ or $(\gamma_{k}, \theta_{k})$ will be only one- or two-dimensional,  where the irrelevant summands will be treated as the intercept and absorbed into the response. These problems can be solved efficiently by the  Newton-Raphson method or even in parallel.

In lines 7--9 of Algorithm \ref{alg:sarma}, the updates for $\bm{U}_1, \bm{U}_2$ and $\cm{S}$ are simple linear least squares problems. To see this, we can write 
\begin{align*}
	\cm{A}_{(1)}\bm{\widetilde{x}}_t &=\left [ \{\bm{z}_t^\prime(\bm{\omega}) (\bm{I}_d\otimes \bm{U}_2) \cm{S}_{(1)}^\prime\}\otimes\bm{I}_N \right ] \vect(\bm{U}_1)
	=\bm{U}_1\cm{S}_{(1)} \{\bm{Z}_t^\prime(\bm{\omega}) \otimes \bm{I}_{R_2}\}\vect(\bm{U}_2^\prime)\\
	&=\left [\{\bm{z}_t^\prime(\bm{\omega})(\bm{I}_d\otimes \bm{U}_2)\}\otimes\bm{U}_1\right ] \vect(\cm{S}_{(1)}),
\end{align*}
where $\bm{z}_t(\bm{\omega})=(\bm{z}_{t,1}^\prime(\bm{\omega}), \dots, \bm{z}_{t,d}^\prime(\bm{\omega}))^\prime=\{\bm{L}^\prime(\bm{\omega})\otimes \bm{I}_N\}\bm{\widetilde{x}}_t$, with $\bm{z}_{t,k}(\bm{\omega})=\sum_{j=1}^{t-1} \ell_{j,k}(\bm{\omega}) \bm{y}_{t-j}$, and  $\bm{Z}_t(\bm{\omega})=(\bm{z}_{t,1}(\bm{\omega}), \dots, \bm{z}_{t,d}(\bm{\omega}))\in\mathbb{R}^{N\times d}$. Alternatively, when $N$ is large and the computation of closed-form solutions is time-consuming, the gradient descent method can be used to further speed up the algorithm. 
In addition,  Algorithm \ref{alg:sarma} can be applied to  the basic SARMA model without any low-Tucker-rank constraint on $\cm{G}$; see the supplementary file for a simulation study, which demonstrates its computational advantage over the VARMA model. In this case, lines 7, 8 and 10 will be dropped, and line 9 will be the update of $\cm{G}^{(i+1)}$, where both $\bm{U}_1$ and $\bm{U}_2$ are set to the $N\times N$ identity matrix.

\begin{remark}
	We initialize Algorithm \ref{alg:sarma} in practice as follows. First, we apply the data-driven procedure in Section \ref{sec:selection} to select  the ranks and model orders. Then, given the selected  $(\pazocal{R}_1,\pazocal{R}_2, p,r,s)$, we initialize the parameters  $\bm{\omega}$, $\bm{U}_1$, $\bm{U}_2$, $\cm{S}$, and $\cm{G}$ by the method described in Section \ref{sec:init}, which exhibits reliable numerical performance in our simulations. 
\end{remark}

\begin{algorithm}[H]
	\caption{Alternating minimization algorithm}
	\label{alg:sarma}
	\textbf{Input:} ranks $(\pazocal{R}_1,\pazocal{R}_2)$, model orders $(p,r,s)$,  initialization $\bm{\omega}^{(0)}, \bm{U}_1^{(0)},  \bm{U}_2^{(0)}, \cm{S}^{(0)}$ and $\cm{G}^{(0)}$.\\
	\textbf{repeat} $i=0,1,2,\dots$\\\vspace{2mm}
	\hspace*{5mm}\textbf{for} $k=1,\dots,r$:\\\vspace{2mm}
	\hspace*{8mm}$\displaystyle \lambda_k^{(i+1)} \leftarrow \underset{\lambda \in (-1,1)}{\argmin} \; \widetilde{\mathbb{L}}_T(\lambda_1^{(i+1)},\dots,\lambda_{k-1}^{(i+1)},\lambda,\lambda_{k+1}^{(i)},\dots, \theta_s^{(i)}, \cm{G}^{(i)})	$\\\vspace{2mm}
	\hspace*{5mm}\textbf{for} $k=1,\dots,s$:\\\vspace{2mm}
	\hspace*{8mm}	$\displaystyle(\gamma_k^{(i+1)}, \theta_k^{(i+1)}) \leftarrow \underset{\gamma\in(0,1), \theta\in(0,\pi)}{\argmin} \widetilde{\mathbb{L}}_T(\lambda_1^{(i+1)},\dots, \theta_{k-1}^{(i+1)}, \gamma,\theta, \gamma_{k+1}^{(i)}, \dots, \theta_{s}^{(i)},\cm{G}^{(i)})$\\\vspace{2mm}
	\hspace*{5mm}$\bm{U}_1^{(i+1)}\leftarrow\underset{\bm{U}_1}{\argmin}\sum_{t=1}^{T}\|\bm{y}_t- [ \{\bm{z}_t^\prime(\bm{\omega}^{(i+1)}) (\bm{I}_d\otimes \bm{U}_2^{(i)}) \cm{S}_{(1)}^{(i)\prime}\}\otimes\bm{I}_N  ] \vect(\bm{U}_1)\|_2^2$\\\vspace{2mm}
	\hspace*{5mm}$\bm{U}_2^{(i+1)} \leftarrow \underset{\bm{U}_2}{\argmin}\sum_{t=1}^{T}\|\bm{y}_t - \bm{U}_1^{(i+1)}\cm{S}_{(1)}^{(i)} \{\bm{Z}_t^\prime(\bm{\omega}^{(i+1)}) \otimes \bm{I}_{R_2}\}\vect(\bm{U}_2^\prime)\|_2^2$\\\vspace{2mm}
	\hspace*{5mm}$\cm{S}^{(i+1)}\leftarrow \underset{\cmt{S}}{\argmin}\sum_{t=1}^{T}\|\bm{y}_t-[\{\bm{z}_t^\prime(\bm{\omega}^{(i+1)})(\bm{I}_d\otimes \bm{U}_2^{(i+1)})\}\otimes\bm{U}_1^{(i+1)} ] \vect(\cm{S}_{(1)})\|_2^2$\\
	\hspace*{5mm}$\cm{G}^{(i+1)} = \cm{S}^{(i+1)}\times_1\bm{U}_1^{(i+1)}\times_2\bm{U}_2^{(i+1)}$.\\
	\textbf{until convergence}
\end{algorithm}

\begin{algorithm}[H]\small
	\caption{ADMM algorithm for the SLTR estimator}
	\label{alg:sparse-sarma}
	\textbf{Input:} ranks $(\pazocal{R}_1,\pazocal{R}_2)$, model orders $(p,r,s)$,  initialization $\bm{\omega}^{(0)}, \bm{U}_1^{(0)},  \bm{U}_2^{(0)}, \cm{S}^{(0)}$ and $\cm{G}^{(0)}$, hyperparameters $(\lambda, \varrho_1, \varrho_2)$.\\
	\textbf{repeat} $i=0,1,2,\dots$\\\vspace{2mm}
	\hspace*{5mm}\textbf{for} $k=1,\dots,r$:\\\vspace{2mm}
	\hspace*{8mm}$\displaystyle \lambda_k^{(i+1)} \leftarrow \underset{\lambda \in (-1,1)}{\argmin} \; \widetilde{\mathbb{L}}_T(\lambda_1^{(i+1)},\dots,\lambda_{k-1}^{(i+1)},\lambda,\lambda_{k+1}^{(i)},\dots,\lambda_{r}^{(i)}, \bm{\eta}^{(i)},\cm{G}^{(i)})	$\\\vspace{2mm}
	\hspace*{5mm}\textbf{for} $k=1,\dots,s$:\\\vspace{2mm}
	\hspace*{8mm}	$\displaystyle\bm{\eta}_k^{(i+1)} \leftarrow \underset{\bm{\eta}\in [0,1)\times(0,\pi)}{\argmin} \widetilde{\mathbb{L}}_T(\bm{\lambda}^{(i+1)}, \bm{\eta}_1^{(i+1)},\dots, \bm{\eta}_{k-1}^{(i+1)}, \bm{\eta}, \bm{\eta}_{k+1}^{(i)}, \dots, \bm{\eta}_{s}^{(i)},\cm{G}^{(i)})$\\\vspace{2mm}
	\hspace*{5mm}$\bm{U}_1^{(i+1)}\leftarrow\underset{\bm{U}_1^\prime\bm{U}_1 = \bm{I}_{\pazocal{R}_1}}{\argmin}\sum_{t=1}^{T}\|\bm{y}_t- [ \{\bm{z}_t^\prime(\bm{\omega}^{(i+1)}) (\bm{I}_d\otimes \bm{U}_2^{(i)}) \cm{S}_{(1)}^{(i)\prime}\}\otimes\bm{I}_N  ] \vect(\bm{U}_1)\|_2^2+\lambda\|\bm{U}_1\|_1$\\\vspace{2mm}
	\hspace*{5mm}$\bm{U}_2^{(i+1)} \leftarrow \underset{\bm{U}_2^\prime\bm{U}_2 = \bm{I}_{\pazocal{R}_2}}{\argmin}\sum_{t=1}^{T}\|\bm{y}_t - \bm{U}_1^{(i+1)}\cm{S}_{(1)}^{(i)} \{\bm{Z}_t^\prime(\bm{\omega}^{(i+1)}) \otimes \bm{I}_{R_2}\}\vect(\bm{U}_2^\prime)\|_2^2+\lambda\|\bm{U}_2\|_1$\\\vspace{2mm}
	\hspace*{5mm}$\cm{S}^{(i+1)}\leftarrow \underset{\cmt{S}}{\argmin}\sum_{t=1}^{T}\|\bm{y}_t-[\{\bm{z}_t^\prime(\bm{\omega}^{(i+1)})(\bm{I}_d\otimes \bm{U}_2^{(i+1)})\}\otimes\bm{U}_1^{(i+1)} ] \vect(\cm{S}_{(1)})\|_2^2$\\
	\hspace*{45mm}$+\sum_{j=1}^{2}\varrho_j\|\cm{S}_{(j)} - \bm{D}_j^{(i)}\bm{V}_j^{(i)\prime}+(\cm{C}_j^{(i)})_{(j)}\|_{\Fr}^2$\\
	\hspace*{5mm}\textbf{for }$j\in\{1,2\}$ \textbf{do}\\
	\hspace*{15mm}$\bm{D}_j^{(i+1)} \leftarrow \argmin_{\bm{D}_i = \diag(\bm{d}_i)}\| \cm{S}_{(j)}^{(i+1)} - \bm{D}_j\bm{V}_j^{(i)\prime}+(\cm{C}_j^{(i)})_{(j)}\ \|_{\Fr}^2$\\
	\hspace*{15mm}$\bm{V}_j^{(i+1)} \leftarrow \argmin_{\bm{V}_j^\prime \bm{V}_j = \bm{I}_{\pazocal{R}_j}}\| \cm{S}_{(j)}^{(i+1)} - \bm{D}_j^{(i+1)}\bm{V}_j^{\prime}+(\cm{C}_j^{(i)})_{(j)}\ \|_{\Fr}^2$\\
	\hspace*{15mm}$(\cm{C}_j^{(i+1)})_{(j)} \leftarrow (\cm{C}_j^{(i)})_{(j)} + \cm{S}^{(i+1)}_{(j)} - \bm{D}_j^{(i+1)}\bm{V}_j^{(i+1)\prime}$\\\vspace*{2mm}
	\hspace*{5mm}$\cm{G}^{(i+1)} = \cm{S}^{(i+1)}\times_1\bm{U}_1^{(i+1)}\times_2\bm{U}_2^{(i+1)}$.\\
	\textbf{until convergence}
\end{algorithm}

\begin{algorithm}[H]
	\caption{ADMM subroutine for sparse and orthogonal regression}
	\label{alg:sub}
	\textbf{Initialize:} $\bm{B}^{(0)} = \bm{W}^{(0)}, \bm{M}^{(0)} = \bm{0}$\\
	\textbf{repeat} $k=0,1,2,\dots$\\\vspace{2mm}
	\hspace*{5mm} $\bm{B}^{(k+1)}\leftarrow \arg\min_{\bm{B}^\prime \bm{B} = \bm{I}} \{ \|\bm{y} - \bm{X}\vect(\bm{B})\|_2^2 + \kappa\|\bm{B} - \bm{W}^{(k)} +\bm{M}^{(k)} \|_{\Fr}^2\}$\\
	\hspace*{5mm} $\bm{W}^{(k+1)}\leftarrow \arg\min_{\bm{W}} \{ \kappa\|\bm{B}^{(k+1)} - \bm{W} +\bm{M}^{(k)} \|_{\Fr}^2 + \lambda\|\bm{W}\|_1\}$\\
	\hspace*{5mm} $\bm{M}^{(k+1)} \leftarrow \bm{M}^{(k)} + \bm{B}^{(k+1)} - \bm{W}^{(k+1)}$\\
	\textbf{until convergence}
\end{algorithm}

\subsection{Algorithm for the SLTR estimator}\label{sec:ADMM}

For the SLTR estimator, we adopt an alternating direction methods of multipliers (ADMM) algorithm \citep{Boyd2011}, where  lines 7--9 in Algorithm \ref{alg:sarma} are revised to incorporate the $\ell_1$-penalties and orthogonality constraints. A similar approach is employed in \cite{Wang2021High}.

Developing an efficient algorithm for the SLTR estimator involves two major challenges. The first challenge arises from the row-orthogonal constraint imposed on the mode-1 and mode-2 unfolding of the core tensor $\cm{S}$, i.e. $\cm{S}_{(1)}$ and $\cm{S}_{(2)}$ in Assumption \ref{assum:para_add}. This constraint cannot be handled in a straightforward manner. The second challenge is related to the joint imposition of $l_1$-regularization and orthogonality constraints on $\bm{U}_i$'s, as specified by the same assumption. The $l_1$-regularization introduces non-smoothness to the algorithm, while the orthogonality constraints increase its nonconvexity.
To address these challenges, we employ the ADMM algorithm which updates the variables $\bm{U}_i$'s and $\cm{S}$ alternately. For a detailed step-by-step procedure, refer to Algorithm \ref{alg:sparse-sarma}.

Firstly, to address the row-orthogonal constraint of $\cm{S}_{(j)}$ (where $j=1$ or $2$), we decompose it using the equation $\cm{S}_{(j)} = \bm{D}_j\bm{V}_j^\prime$. Here, $\bm{D}_j\in\mathbb{R}^{\pazocal{R}_j\times \pazocal{R}_j}$ represents a diagonal matrix, while $\bm{V}_1\in\mathbb{R}^{\pazocal{R}_2\pazocal{R}_3 \times \pazocal{R}_1}$ and $\bm{V}_2\in\mathbb{R}^{\pazocal{R}_1\pazocal{R}_3 \times \pazocal{R}_2}$ are orthogonormal matrices. These matrices satisfy the condition $\bm{V}_j^\prime \bm{V}_j = \bm{I}_{\pazocal{R}_j}$ for $j=1, 2$.
By introducing these decompositions, we can then express the augmented Lagrangian corresponding to the objective functions given in \eqref{eq:sparse_lse}  as follows:

\begin{align*}
	\pazocal{L}_{\varrho}(\cm{S},\{\bm{U}_i\},\bm{\omega}, \{\bm{D}_j\}, \{\bm{V}_j\}; \{\cm{C}_j\}) &= \mathbb{\widetilde{L}}_T(\bm{\omega}, \cm{S}, \bm{U}_1, \bm{U}_2) + \lambda\sum_{i=1}^{2}\|\bm{U}_i\|_1\\
	&\hspace{5mm}+2\sum_{j=1}^{2}\varrho_j\langle(\cm{C}_j)_{(j)}, \cm{S}_{(j)} - \bm{D}_j\bm{V}_j^\prime\rangle + \sum_{j=1}^{2}\varrho_j\|\cm{S}_{(j)} - \bm{D}_j\bm{V}_j^\prime\|_{\Fr}^2
\end{align*}
where $\cm{C}_j\in\mathbb{R}^{\pazocal{R}_1 \times \pazocal{R}_2 \times \pazocal{R}_3}$ are the tensor-valued dual variables, and $\bm{\varrho} = (\varrho_1, \varrho_2)^\prime$ is the set of regularization parameters, which in practice can be selected together with $\lambda$ by a fine grid search with information criterion such as the BIC or its high-dimensional extensions, where the total number of nonzero parameters could be used as proxies for the degree of freedom. This brings us to Algorithm \ref{alg:sparse-sarma}. It is important to note that the row-orthogonal constraint originally imposed on $\cm{S}_{(1)}$ and $\cm{S}_{(2)}$ has been effectively transferred to the matrices $\bm{V}_j$'s in line 13. As a result, no explicit constraint is required for updating $\cm{S}$ in lines 9-10 of Algorithm \ref{alg:sparse-sarma}.

Secondly, we discuss the update process for $\bm{U}_i$. In \eqref{eq:ls}, $\mathbb{\widetilde{L}}_T(\bm{\omega}, \cm{S}, \bm{U}_1, \bm{U}_2)$ represents a least squares loss function with respect to each $\bm{U}_i$. Therefore, in lines 7-8 of Algorithm \ref{alg:sparse-sarma}, the update steps for $\bm{U}_i$ involve solving $\ell_1$-regularized least squares problems under an orthogonality constraint. This can be expressed in a general form as follows:
\begin{equation} \label{eq:algo-ortho}
	\min_{\bm{B}} \{ \|\bm{y} - \bm{X}\vect(\bm{B})\|_2^2 + \lambda\|\bm{B}\|_1\}, \quad \text{s.t. } \bm{B}^\prime \bm{B} = \bm{I}.
\end{equation}
Since handling the $\ell_1$-regularization and the orthogonality constraint for $\bm{B}$ together is challenging, we employ an ADMM subroutine to separate them into two steps. To achieve this, we introduce a dummy variable $\bm{W}$ as a surrogate for $\bm{B}$ and rewrite the problem \eqref{eq:algo-ortho} equivalently as:
\[
\min_{\bm{B}, \bm{W}} \{ \|\bm{y} - \bm{X}\vect(\bm{B})\|_2^2 + \lambda\|\bm{W}\|_1\}, \quad \text{s.t. } \bm{B}^\prime \bm{B} = \bm{I}\text{ and }\bm{B} = \bm{W}.
\]
Then, the corresponding Lagrangian formulation is:
\begin{equation}\label{eq:admm-ortho}
	\min_{\bm{B}, \bm{W}} \{ \|\bm{y} - \bm{X}\vect(\bm{B})\|_2^2 + \lambda\|\bm{W}\|_1 + 2\kappa\langle  \bm{M},\bm{B} - \bm{W} \rangle + \kappa\|\bm{B} - \bm{W}\|_{\Fr}^2\},
\end{equation}
where $\bm{M}$ represents the dual variable and $\kappa$ is the regularization parameter. Algorithm 3 in \cite{WZL21} presents the ADMM subroutine for solving \eqref{eq:admm-ortho}, which provides solutions for the $\bm{U}_i$-update subproblems in Algorithm \ref{alg:sparse-sarma}. We also include it as Algorithm \ref{alg:sub} here for sake of completeness.

Note that both the $\bm{B}$-update step in Algorithm \ref{alg:sub} and the $\bm{V}_i$-update step in line 13 of Algorithm \ref{alg:sparse-sarma} involve solving least squares problems subject to an orthogonality constraint. These problems can be efficiently solved using the splitting orthogonality constraint (SOC) method \citep{lai2014splitting}. On the other hand, the $\bm{W}$-update step in Algorithm \ref{alg:sub} corresponds to an $\ell_1$-regularized minimization, which can be effectively addressed through explicit soft-thresholding. As for the $\cm{S}$- and $\bm{D}_i$-update steps in lines 9 and 12 of Algorithm \ref{alg:sparse-sarma}, they simply entail solving straightforward least squares problems.

\subsection{Initialization for the algorithms}\label{sec:init}

The proposed algorithms require suitable initial values for the parameters $\bm{\omega}$, $\bm{U}_1$, $\bm{U}_2$, $\cm{S}$, and $\cm{G}$. In this section, we provide an easy-to-implement method for initializing these values.

Given the initial value  $\cm{A}^{(0)}=\cm{\widehat{A}}^{\text{init}}$ and pre-selected $(\pazocal{R}_1,\pazocal{R}_2, p,r,s)$, we are ready to  initialize $\bm{\omega}, \bm{U}_1, \bm{U}_2, \cm{S}$ and $\cm{G}$ for the proposed algorithms as follows:
\begin{itemize}
	\item Conduct the  HOSVD of $\cm{A}^{(0)}$ for the first two modes, $\cm{A}^{(0)} = \cm{H}^{(0)} \times_1 \bm{U}_1^{(0)} \times_2 \bm{U}_2^{(0)}$, to obtain $\cm{H}^{(0)}\in\mathbb{R}^{\pazocal{R}_1 \times \pazocal{R}_2\times \infty}$ and $\bm{U}_i^{(0)} \in\mathbb{R}^{N\times \pazocal{R}_i}$ for $i=1,2$.
	\item Next we aim to obtain $\cm{S}^{(0)}$ and $\bm{\omega}^{(0)}$ such that $\cm{S}^{(0)} \times_3 \bm{L}(\bm{\omega}^{(0)}) \approx \cm{H}^{(0)}$:
	\begin{itemize}
		\item[(i)] To determine $\bm{\omega}^{(0)}$, note that  $\bm{\omega}$ lies in the bounded parameter space $\bm{\Omega}$ defined in \eqref{eq:Omega}. Thus, we consider a grid of initial values for each element of $\bm{\omega}$ within the parameter space; e.g., 
		$\lambda_k \in \{-0.75,-0.5,-0.25,0.25,0.5,0.75\}$, $\gamma_h\in\{0.25,0.5,0.75\}$, and $\theta_h\in\{\pi/4,4\pi/3\}$, for any $1\leq k\leq r$ and $1\leq h\leq s$. To ensure identifiability, we only consider the combinations with distinct $\lambda_k$'s and $(\gamma_h,\theta_h)$'s. 
		\item[(ii)] For each choice of $\bm{\omega}^{(0)}$, we get the corresponding $\cm{S}^{(0)}=\cm{H}^{(0)}\times_3 \bm{L}^{\dagger}(\bm{\omega}^{(0)})$, where
		$\bm{L}^{\dagger}(\bm{\omega}^{(0)})$ is the left pseudo-inverse of $\bm{L}(\bm{\omega}^{(0)})$.  
	\end{itemize}
	\item Then, let $\cm{G}^{(0)} = \cm{S}^{(0)} \times_1 \bm{U}_1^{(0)} \times_2 \bm{U}_2^{(0)}$. 
	\item Finally, among all choices of initial values, we select the one leading to the smallest value for the loss function.
\end{itemize}

\section{Simulation studies}\label{sec:sim}

In this section, we present simulation experiments to examine finite-sample performance of the proposed  methods for the SARMA model with non-sparse or sparse  factor matrices. 

We consider the following two VARMA models as the data generating processes (DGPs),
\begin{itemize}
	\item DGP1:  the VMA(1) model $\bm{y}_t = \bm{\varepsilon}_t - \bm{\Theta}\bm{\varepsilon}_{t-1}$, and
	\item DGP2:  the VARMA($1,1$) model $\bm{y}_t = \bm{\Phi}\bm{y}_{t-1}+\bm{\varepsilon}_t - \bm{\Theta}\bm{\varepsilon}_{t-1}$,
\end{itemize}
which correspond to $p=0$ and 1, respectively. For both DGPs, $\{\bm{\varepsilon}_t\}$ are $i.i.d.$ $N(\bm{0}, \bm{I}_N)$, and we set  $\bm{\Theta}=\bm{B}\bm{J}\bm{B}^{-1}$, where $\bm{J} = \diag\{\lambda_1, \dots, \lambda_r, \bm{C}(\gamma_1,\theta_1), \dots, \bm{C}(\gamma_s,\theta_s), \bm{0}\}$ is the real Jordan normal form, with each $\bm{C}(\gamma, \theta)$ being the $2\times 2$ block defined as
\[
\bm{C}(\gamma, \theta)=
\gamma\cdot \left(\begin{matrix}
	\cos \theta& \sin \theta\\
	- \sin \theta& \cos \theta
\end{matrix}\right),
\]
and $\bm{B}$ is generated by a method to be specified below. 
For DGP2, we set
$\bm{\Phi}=\bm{B}\bm{K}\bm{B}^{-1}$,  where $\bm{K} = \diag\{\delta,0,\dots, 0\}$, with the entry $\delta\neq0$. It is noteworthy that both DGPs can be written in the form of the  SARMA model with orders ($p,r,s$) and Tucker ranks $\pazocal{R}_1=\pazocal{R}_2=r+2s$. Moreover, to produce non-sparse and sparse factor matrices, we generate $\bm{B}$ as follows:
\begin{itemize}
	\item Non-sparse case: $\bm{B}\in\mathbb{R}^{N\times N}$ is a randomly generated orthogonal matrix.
	\item Sparse case: $\bm{B}$ is obtained by inserting $N-\pazocal{S}$ zero rows into the randomly generated orthogonal matrix $\bm{B}_{\pazocal{S}}\in\mathbb{R}^{\pazocal{S} \times \pazocal{S}}$ and then concatenating the resulting $N\times \pazocal{S}$ matrix on the right with a $N\times(N-\pazocal{S})$ zero matrix. As a result, $(s_1, s_2) = (\pazocal{S}, \pazocal{S})$.
\end{itemize}
The non-sparse and sparse cases are fitted by the rank-constrained and SLTR estimators, respectively.

\begin{figure}[!t]
	\centering
	\includegraphics[width=\columnwidth]{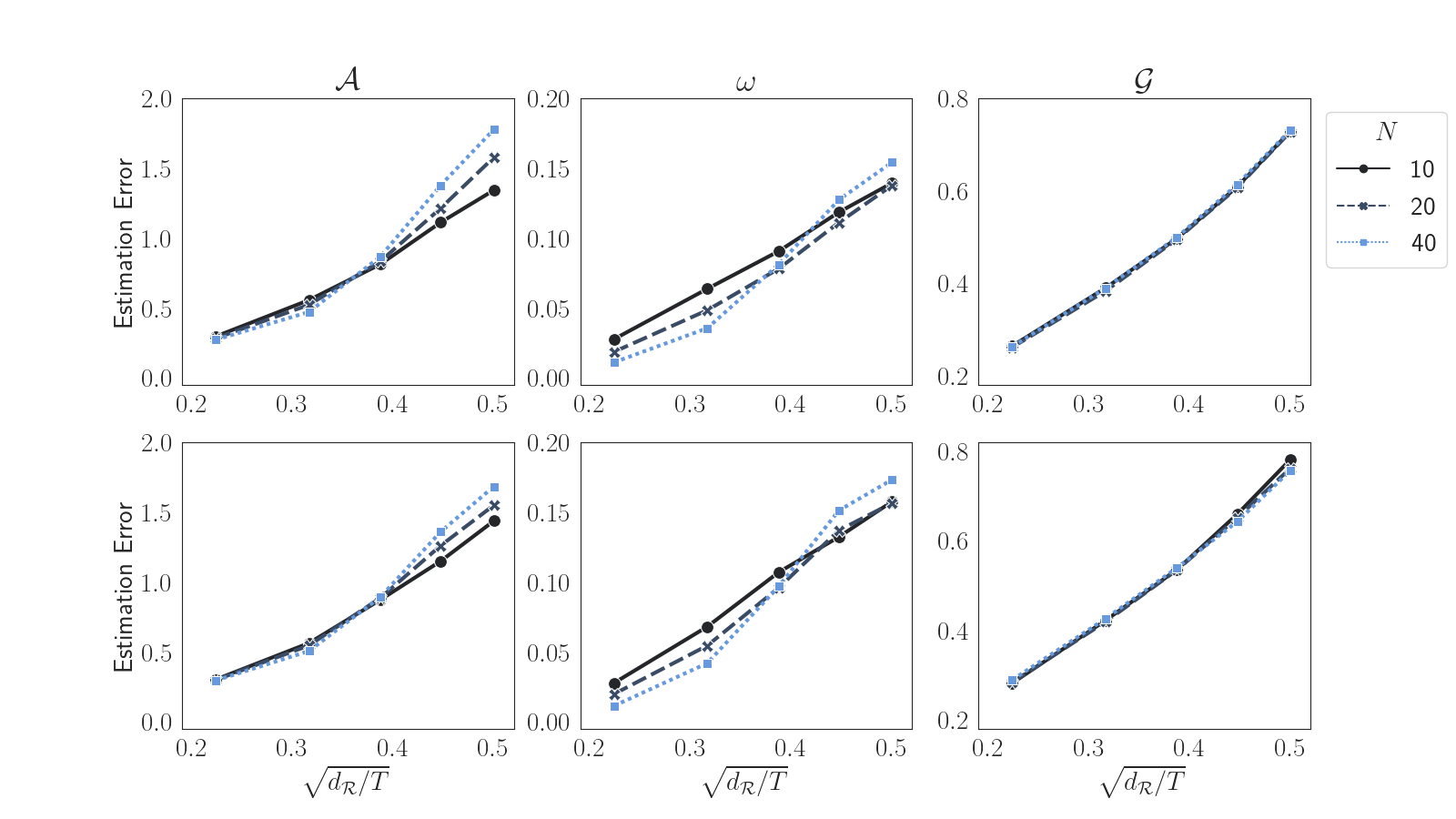}
	\caption{Plots of estimation errors $\|\cm{\widehat{A}} - \cm{A}^*\|_{\Fr}$ (left panel), $\|\bm{\widehat{\omega}} - \bm{\omega}^*\|_{2}$ (middle panel) and $\|\cm{\widehat{G}} - \cm{G}^*\|_{\Fr}$ (right panel) against $\sqrt{d_{\pazocal{R}}/{T}}$ for the rank-constrained estimator, where $(\pazocal{R}_1,\pazocal{R}_2, p, r, s) = (1,1,0,1,0)$ (top panel) or $(\pazocal{R}_1,\pazocal{R}_2, p, r, s) = (1,1,1,1,0)$ (bottom panel), and  $N=10$ (\protect\Lrateten), 20 (\protect\Lratetwenty) or 40 (\protect\Lratefourty).
	}
	\label{fig:rate_LTR}
\end{figure}

\begin{figure}[!t]
	\centering
	\includegraphics[width=\columnwidth]{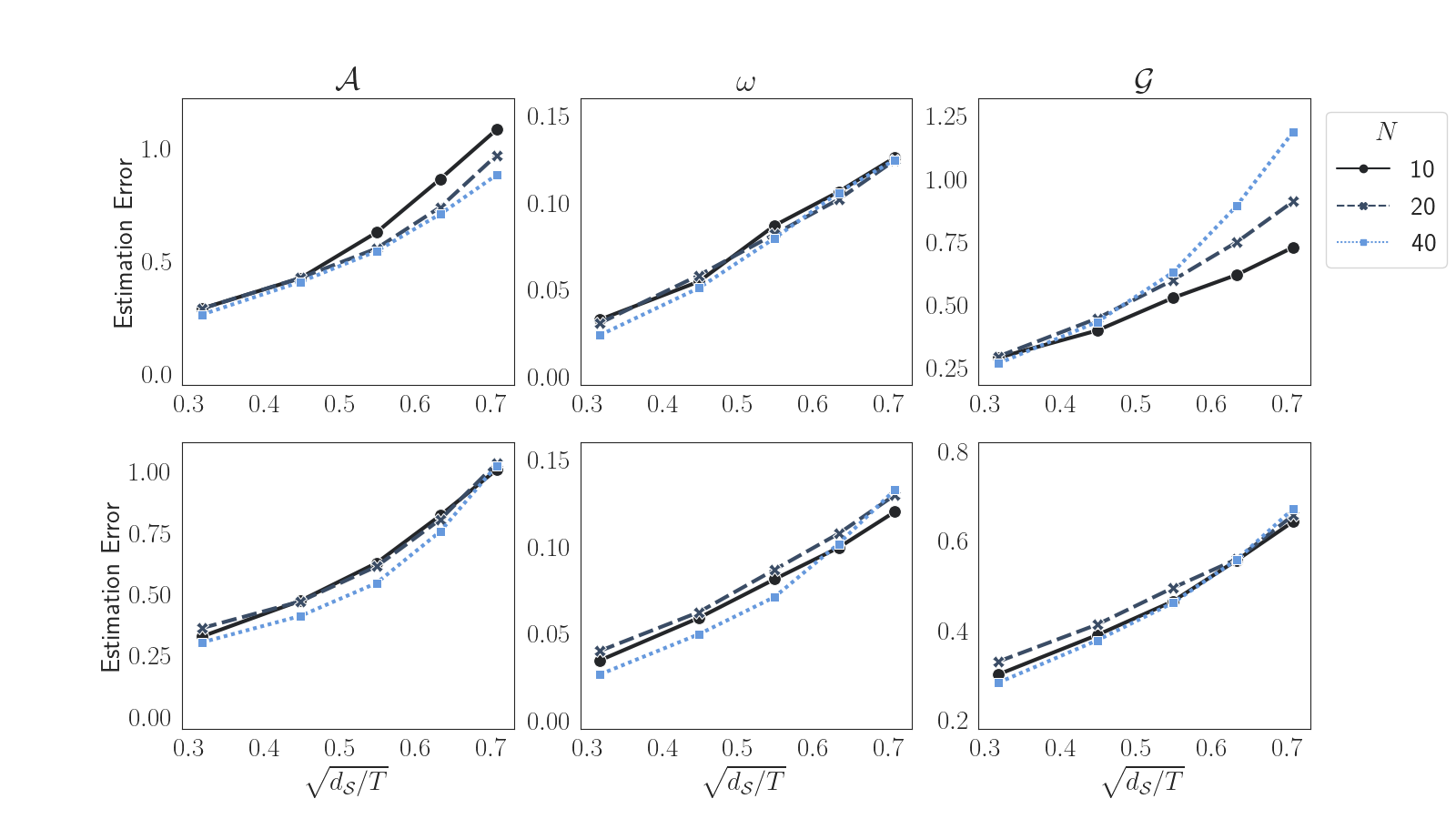}
	\caption{Plots of estimation errors $\|\cm{\widetilde{A}} - \cm{A}^*\|_{\Fr}$ (left panel), $\|\bm{\widetilde{\omega}} - \bm{\omega}^*\|_{2}$ (middle panel) and $\|\cm{\widetilde{G}} - \cm{G}^*\|_{\Fr}$ (right panel) against $\sqrt{d_{\pazocal{S}}/{T}}$ for the SLTR estimator, where $(\pazocal{R}_1,\pazocal{R}_2, p, r, s) = (1,1,0,1,0)$ (top panel) or $(\pazocal{R}_1,\pazocal{R}_2, p, r, s) = (1,1,1,1,0)$ (bottom panel), and  $N=10$ (\protect\Lrateten), 20 (\protect\Lratetwenty) or 40 (\protect\Lratefourty).
	}
	\label{fig:rate_SHORR}
\end{figure}




In the first experiment, we aim to verify the estimation error rates of the proposed estimators derived in Theorems \ref{thm:parametric} and \ref{thm:sparse}. 
We set $(r, s)=(1, 0)$ and $\lambda_1 = -0.7$ for both DGPs,   $\delta = 0.5$ for DGP2, and $N=10$, 20 or 40. The estimation is conducted via the algorithm in  Section \ref{section:algo} or the ADMM Algorithm \ref{alg:sparse-sarma} in the supplementary file given the true ranks and model orders. For the non-sparse case,  $T$ is chosen such that $d_{\pazocal{R}}/T \in \{0.05, 0.1,0.15,0.2, 0.25\}$. 
Figure~\ref{fig:rate_LTR} plots the estimation errors averaged over 500 replications against  $\sqrt{d_{\pazocal{R}}/T}$.  In all settings, it can be observed that there exists a roughly linear relationship between the estimation errors and the theoretical rate, which confirms our theoretical results.
For the sparse case, we set  $\pazocal{S} = 5$ for both DGPs  and choose $T$ such that $d_{\pazocal{S}}/T\in\{0.1,0.2,0.3,0.4,0.5\}$.  Figure~\ref{fig:rate_SHORR} plots the estimation errors averaged over 500 replications against  $\sqrt{d_{\pazocal{S}}/T}$. Similar to the non-sparse case, we observe an approximately linear relationship between the estimation errors and the theoretical rate across all settings, although the estimation error for $\cm{G}$ might be influenced by algorithmic errors when $N$ is large.

%

\begin{figure}[t]
	\centering
	\includegraphics[width=\columnwidth]{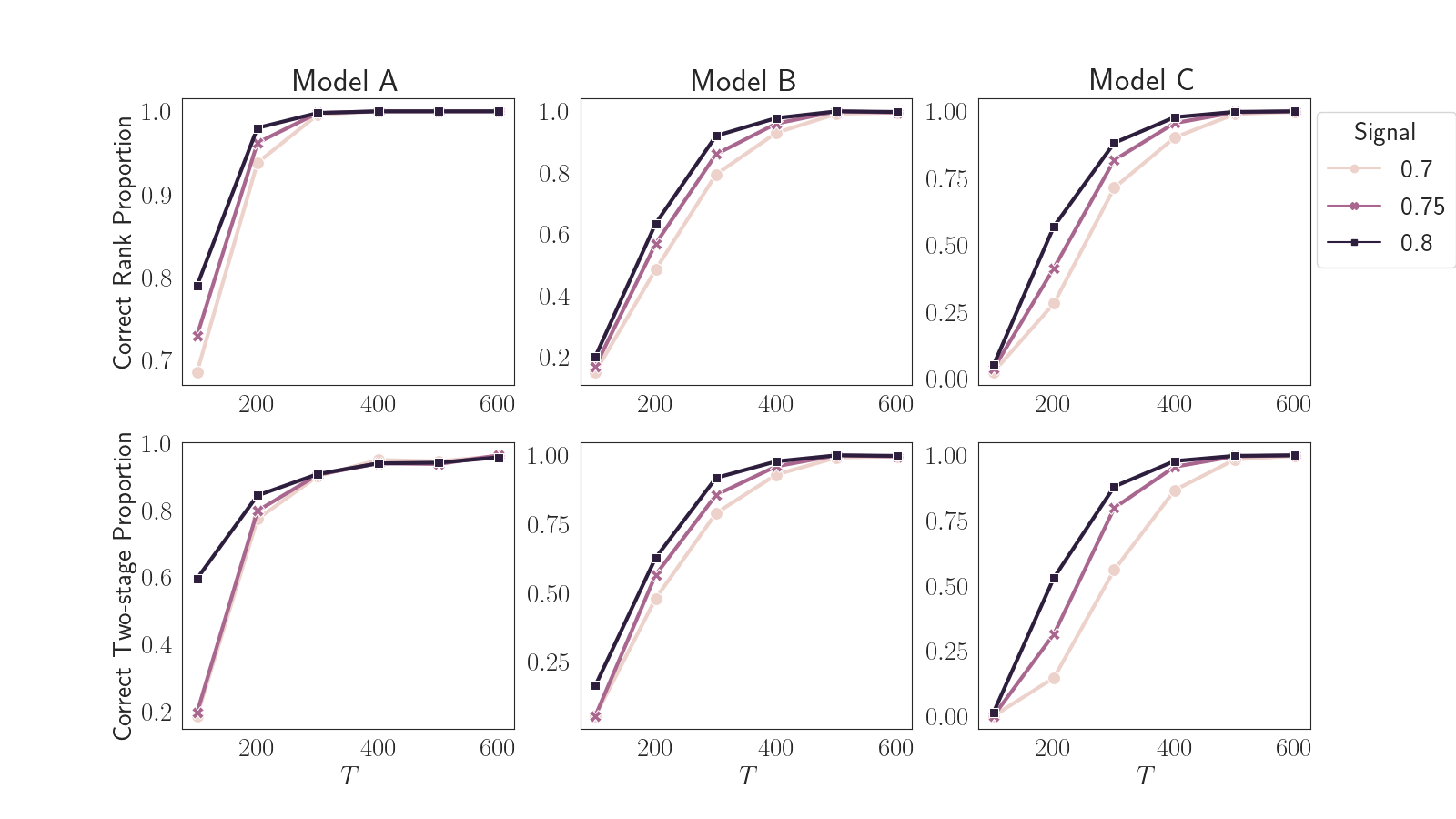}
	\caption{Proportions of correct rank selection (top panel) and two-stage selection (bottom panel) for models A (left panel), B (middle panel) and C (right panel) in the non-sparse case,  where the signal strength is 0.7 (\protect\Lorderone), 0.75 (\protect\Lordertwo) or 0.8 (\protect\Lorderthree).}
	\label{fig:twostep_LTR}
\end{figure}

\begin{figure}[t]
	\centering
	\includegraphics[width=\columnwidth]{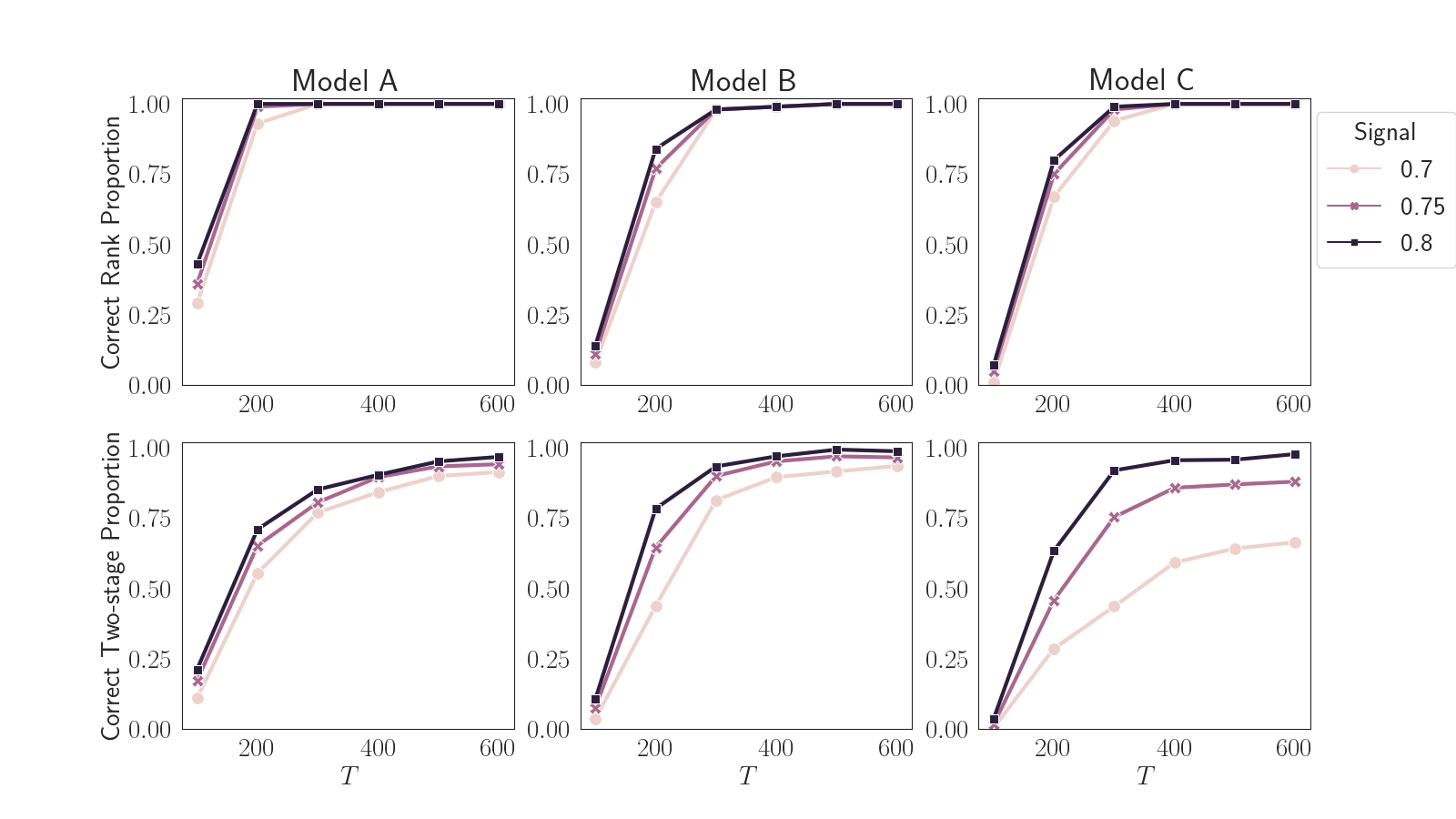}
	\caption{Proportions of correct rank selection (top panel) and two-stage selection (bottom panel) for models A (left panel), B (middle panel) and C (right panel) in the sparse case, where the signal strength is 0.7 (\protect\Lorderone), 0.75 (\protect\Lordertwo) or 0.8 (\protect\Lorderthree).}
	\label{fig:twostep_SHORR}
\end{figure}


The second experiment examines the performance of the rank selection method in Section \ref{sec:selection} and the model order selection criterion in Remark \ref{remark:order}. Almost identical settings apply to both the non-sparse case and the sparse case. Specifically, we consider three cases under DGP1: $(\pazocal{R}_1, \pazocal{R}_2, r, s)$=(1, 1, 1, 0) (model A),  (2, 2, 0, 1) (model B), and (3, 3, 1, 1) (model C). The results for DGP2 are similar and hence are omitted for brevity. For models B and C, we set $\theta_1=\pi/4$.  Note that   $\cm{A}_{(1)}$ and  $\cm{A}_{(2)}$ have the same singular values  under DGP1.  Moreover, when $r\leq 1$ and $s\leq1$, the magnitude of the nonzero singular values are directly determined by $|\lambda_1|$ and $\gamma_{1}$, which control the signal strength for the rank selection. We consider three levels of signal strength $\{ 0.7, 0.75, 0.8\}$, and set $-\lambda_{1}$ in model A, $\gamma_1$ in model B, and $-\lambda_{1}=\gamma_1$ in model C to these values. In addition, we consider $N=10$ and $T\in[100,600]$. The initial estimator $\cm{\widehat{A}}^{\text{init}}$ is obtained by  the nuclear norm or  lag group lasso regularized method in Remark \ref{remark:init} for the non-sparse or sparse case, respectively. For the model order selection, we minimize the BIC in Remark \ref{remark:order} with $c = 0.1$.

For the non-sparse case, the proportion of correct rank selection, $\{(\widehat{\pazocal{R}}_1, \widehat{\pazocal{R}}_2) = (\pazocal{R}_1^*, \pazocal{R}_2^*)\}$, and that of correct rank and model order selection,  $\{(\widehat{\pazocal{R}}_1, \widehat{\pazocal{R}}_2, \widehat{p}, \widehat{r}, \widehat{s}) = (\pazocal{R}_1^*, \pazocal{R}_2^*, p^*, r^*, s^*)\}$, based on the two-stage procedure are reported in Figure~\ref{fig:twostep_LTR}. 
It can be clearly seen that both proportions increase to one as $T$ and the signal strength increases. For all models, the proportion that the ranks and model orders are correctly selected simultaneously is fairly close to one when $T\geq 400$ across all settings. 

For the sparse case, we utilize the same data generation settings, with the only difference being that $\bm{B}$ is produced using a row sparsity of $\pazocal{S}=5$.  
The results are presented in Figure \ref{fig:twostep_SHORR}. Generally speaking, the patterns are similar to those in Figure \ref{fig:twostep_LTR}. However, it is more evident that when the signal strength $\leq 0.7$, model C requires a larger $T$   to achieve comparable proportions of simultaneously correct ranks and model orders selection,  since the model is more complex.
Nevertheless, although not shown in the figure, the accuracy of two-stage selections for model C will continue to increase as  $T$ grows.

\section{Descriptions of datasets} \label{sec5}

We provide more detailed descriptions of the variables and their transformations for the two datasets in Section \ref{sec:empirical} of the main paper through two tables.  Table  \ref{tab:macro-description} is for the quarterly macroeconomic dataset, and Table \ref{tab:rv-description} is for the daily realized volatilities dataset. 
%

\begin{table}[t]
	\caption{Forty six selected S\&P 500 stocks. CODE: stock code in the New York Stock Exchange. NAME: name of  company. G: group code, where 1 = communication service, 2 = information technology, 3 = consumer, 4 = financials, 5 = healthcare, 6 = materials and industrials, and 7 = energy and utilities.}
	\label{tab:rv-description}
	\centering
	\resizebox{.8\columnwidth}{!}{
		\begin{tabular}{@{}llllll@{}}
			\hline
			CODE & NAME                                 & G & CODE & NAME                              & G \\
			\hline
			T    & AT\&T Inc.                           & 1 & JPM  & JPMorgan Chase \& Co.             & 4 \\
			NWSA & News   Corp                          & 1 & WFC  & Wells Fargo \& Company            & 4 \\
			FTR  & Frontier   Communications Parent Inc & 1 & MS   & Morgan Stanley                    & 4 \\
			VZ   & Verizon Communications Inc.          & 1 & AIG  & American International Group Inc. & 4 \\
			IPG  & Interpublic   Group of Companies Inc & 1 & MET  & MetLife Inc.                      & 4 \\
			MSFT & Microsoft Corporation                & 2 & RF   & Regions   Financial Corp          & 4 \\
			HPQ  & HP   Inc                             & 2 & PGR  & Progressive Corporation           & 4 \\
			INTC & Intel Corporation                    & 2 & SCHW & Charles Schwab Corporation        & 4 \\
			EMC  & EMC   Instytut Medyczny SA           & 2 & FITB & Fifth   Third Bancorp             & 4 \\
			ORCL & Oracle Corporation                   & 2 & PFE  & Pfizer Inc.                       & 5 \\
			MU   & Micron Technology Inc.               & 2 & ABT  & Abbott Laboratories               & 5 \\
			AMD  & Advanced Micro Devices Inc.          & 2 & MRK  & Merck \& Co. Inc.                 & 5 \\
			AAPL & Apple Inc.                           & 2 & RAD  & Rite   Aid Corporation            & 5 \\
			YHOO & Yahoo! Inc.                          & 2 & JNJ  & Johnson \& Johnson                & 5 \\
			QCOM & Qualcomm Inc                         & 2 & AA   & Alcoa   Corp                      & 6 \\
			GLW  & Corning   Incorporated               & 2 & FCX  & Freeport-McMoRan Inc.             & 6 \\
			AMAT & Applied Materials Inc.               & 2 & X    & United   States Steel Corporation & 6 \\
			F    & Ford Motor Company                   & 3 & GE   & General Electric Company          & 6 \\
			LVS  & Las   Vegas Sands Corp.              & 3 & CSX  & CSX Corporation                   & 6 \\
			EBAY & eBay Inc.                            & 3 & ANR  & Alpha Natural Resources           & 7 \\
			KO   & Coca-Cola Company                    & 3 & XOM  & Exxon Mobil Corporation           & 7 \\
			BAC  & Bank of America Corp                 & 4 & CHK  & Chesapeake Energy                 & 7 \\
			C    & Citigroup Inc.                       & 4 & EXC  & Exelon Corporation                & 7\\
			\hline
		\end{tabular}
	}
\end{table}
\begin{landscape}
	\begin{table}[ht]
		\caption{Twenty quarterly macroeconomic variables. FRED MNEMONIC: mnemonic for data in FRED-QD. SW MNEMONIC: mnemonic in \cite{stock2012disentangling}. T: data transformation, where 1 = no transformation, 2 = first difference, and 3 = first difference of log series.  DESCRIPTION: brief definition of the data. G: Group code, where 1 = interest rate, 2 = money and credit, 3 = exchange rate, and 4 = stock market. }
		\label{tab:macro-description}
		\centering
		\resizebox{\columnwidth}{!}{
			\begin{tabular}{@{}llclc@{}}
				\hline
				FRED MNEMONIC  & SW MNEMONIC       & T & DESCRIPTION                                                                                                        & G \\
				\hline
				FEDFUNDS       & FedFunds          & 2     & Effective Federal Funds Rate (Percent)                                                                             & 1     \\
				TB3MS          & TB-3Mth           & 2     & 3-Month Treasury Bill: Secondary Market Rate (Percent)                                                             & 1     \\
				BAA10YM        & BAA\_GS10         & 1     & Moody's Seasoned Baa Corporate Bond Yield Relative to Yield on 10-Year   Treasury Constant Maturity (Percent)      & 1     \\
				TB6M3Mx        & tb6m\_tb3m        & 1     & 6-Month Treasury Bill Minus 3-Month Treasury Bill, secondary market   (Percent)                                    & 1     \\
				GS1TB3Mx       & GS1\_tb3m         & 1     & 1-Year Treasury Constant Maturity Minus 3-Month Treasury Bill, secondary   market (Percent)                        & 1     \\
				GS10TB3Mx      & GS10\_tb3m        & 1     & 10-Year Treasury Constant Maturity Minus 3-Month Treasury Bill, secondary   market (Percent)                       & 1     \\
				CPF3MTB3Mx     & CP\_Tbill Spread  & 1     & 3-Month Commercial Paper Minus 3-Month Treasury Bill, secondary market   (Percent)                                 & 1     \\
				BUSLOANSx      & Real C\&Lloand    & 3     & Real Commercial and Industrial Loans, All Commercial Banks (Billions of   2009 U.S. Dollars), deflated by Core PCE & 2     \\
				CONSUMERx      & Real ConsLoans    & 3     & Real Consumer Loans at All Commercial Banks (Billions of 2009 U.S.   Dollars), deflated by Core PCE                & 2     \\
				NONREVSLx      & Real NonRevCredit & 3     & Total Real Nonrevolving Credit Owned and Securitized, Outstanding   (Billions of Dollars), deflated by Core PCE    & 2     \\
				REALLNx        & Real LoansRealEst & 3     & Real Real Estate Loans, All Commercial Banks (Billions of 2009 U.S.   Dollars), deflated by Core PCE               & 2     \\
				EXSZUSx        & Ex rate:Switz     & 3     & Switzerland / U.S. Foreign Exchange Rate                                                                           & 3     \\
				EXJPUSx        & Ex rate:Japan     & 3     & Japan / U.S. Foreign Exchange Rate                                                                                 & 3     \\
				EXUSUKx        & Ex rate:UK        & 3     & U.S. / U.K. Foreign Exchange Rate                                                                                  & 3     \\
				EXCAUSx        & EX rate:Canada    & 3     & Canada / U.S. Foreign Exchange Rate                                                                                & 3     \\
				NIKKEI225      &                   & 3     & Nikkei Stock Average                                                                                               & 4     \\
				S\&P 500       &                   & 3     & S\&P's Common Stock Price Index: Composite                                                                         & 4     \\
				S\&P: indust   &                   & 3     & S\&P's Common Stock Price Index: Industrials                                                                       & 4     \\
				S\&P div yield &                   & 2     & S\&P's Composite Common Stock: Dividend Yield                                                                      & 4     \\
				S\&P PE ratio  &                   & 3     & S\&P's Composite Common Stock: Price-Earnings Ratio                                                                & 4    \\
				\hline
		\end{tabular}}
	\end{table}
\end{landscape}

\section{Proofs for Section 2 in the main paper} \label{sec:proof1}
\subsection{Proof of Proposition  \ref{prop:VARMAgen}}

\begin{proof}[Proof of Proposition \ref{prop:VARMAgen}]
	Consider the general VARMA$(p,q)$ model
	\[
	\bm{y}_t =\sum_{i=1}^{p} \bm{\Phi}_i\bm{y}_{t-i}+\bm{\varepsilon}_t - \sum_{j=1}^{q}\bm{\Theta}_j\bm{\varepsilon}_{t-j},\hspace{5mm} t\in\mathbb{Z}.
	\] 
	Note that it  can be written equivalently as
	\begin{equation} \label{aeq:VARMA}
		\bm{\varepsilon}_t = \bm{\Theta}_1\bm{\varepsilon}_{t-1}-\cdots-\bm{\Theta}_q\bm{\varepsilon}_{t-q}+\bm{\Phi}(B) \bm{y}_t,
	\end{equation} 
	where 
	$\bm{\Phi}(B) = \bm{I}-\sum_{i=1}^{p}\bm{\Phi}_i B^i=-\sum_{i=0}^{p}\bm{\Phi}_i B^i$, with $\bm{\Phi}_0=-\bm{I}$. Then we have
	\begin{align*}
		\underbrace{\left(\begin{matrix}
				\bm{\varepsilon}_t\\\bm{\varepsilon}_{t-1}\\\bm{\varepsilon}_{t-2}\\\vdots\\\bm{\varepsilon}_{t-q+1}\\
			\end{matrix} \right ) }_{\underline{\bm{\varepsilon}}_t}	
		= \underbrace{\left(\begin{matrix}
				\bm{\Theta}_1&\bm{\Theta}_2&\cdots&\bm{\Theta}_{q-1}&\bm{\Theta}_q\\
				\bm{I}&\bm{0}&\cdots&\bm{0}&\bm{0}\\
				\bm{0}&\bm{I}&\cdots&\bm{0}&\bm{0}\\
				\vdots&\vdots&\ddots&\vdots&\vdots\\
				\bm{0}&\bm{0}&\cdots&\bm{I}&\bm{0}
			\end{matrix} \right )}_{\underline{\bm{\Theta}}}  \underbrace{\left(\begin{matrix}
				\bm{\varepsilon}_{t-1}\\\bm{\varepsilon}_{t-2}\\\bm{\varepsilon}_{t-3}\\\vdots\\\bm{\varepsilon}_{t-q}\\
			\end{matrix} \right ) }_{\underline{\bm{\varepsilon}}_{t-1}} + \underbrace{\left(\begin{matrix}
				\bm{\Phi}(B) \bm{y}_t\\\bm{0}\\\bm{0}\\\vdots\\\bm{0}\\
			\end{matrix} \right ) }_{\underline{\bm{y}}_t},
	\end{align*}
	where $\underline{\bm{\Theta}}\in\mathbb{R}^{Nq\times Nq}$ is the MA companion matrix. By recursion, we have $\underline{\bm{\varepsilon}}_t=\sum_{j=0}^{\infty} \underline{\bm{\Theta}}^j \underline{\bm{y}}_{t-j}$. Let $\bm{P} = (\bm{I}_{N}, \bm{0}_{N\times N(q-1)})$. 
	Note that $\bm{P}\underline{\bm{\varepsilon}}_t=\bm{\varepsilon}_t$, and  $\underline{\bm{y}}_t=\bm{P}^\prime \bm{\Phi}(B)\bm{y}_t$. Thus,
	\begin{equation}\label{eq:VARMA_inf0}
		\bm{\varepsilon}_t=\sum_{j=0}^{\infty}\bm{P} \underline{\bm{\Theta}}^j \bm{P}^\prime \bm{\Phi}(B)\bm{y}_{t-j}
		=-\sum_{j=0}^{\infty} \bm{P} \underline{\bm{\Theta}}^j \bm{P}^\prime \sum_{i=0}^{p}\bm{\Phi}_i\bm{y}_{t-j-i} 
		=-\sum_{k=0}^{\infty} \left ( \sum_{i=0}^{p\wedge k}\bm{P} \underline{\bm{\Theta}}^{k-i} \bm{P}^\prime \bm{\Phi}_i \right )\bm{y}_{t-k}. 
	\end{equation} 
	Since $\bm{P} \bm{P}^\prime = \bm{I}_N$, it follows from \eqref{eq:VARMA_inf0} that the VAR($\infty$) representation of the  VARMA($p,q$)  model can be written as
	\begin{equation}\label{eq:VARMA_infa}
		\bm{y}_t =  \sum_{k=1}^{\infty} \underbrace{ \left ( \sum_{i=0}^{p\wedge k} \bm{P} \underline{\bm{\Theta}}^{k-i} \bm{P}^\prime \bm{\Phi}_i \right )}_{\bm{A}_k}\bm{y}_{t-k}+\bm{\varepsilon}_t.
	\end{equation}
	First, we simply set 
	\begin{align} \label{eq:init}
		\bm{G}_j=\sum_{i=0}^{j}\bm{P} \underline{\bm{\Theta}}^{j-i} \bm{P}^\prime \bm{\Phi}_i=\bm{A}_j,
		\hspace{5mm}\text{for } 1\leq j\leq p,
	\end{align}
	and then  we only need to focus on the reparameterization of $\bm{A}_k$ for $k>p$. By \eqref{eq:VARMA_infa}, for $j\geq 1$, we have
	\begin{equation} \label{eq:VARMA-Am}
		\bm{A}_{p+j} = \bm{P} \underline{\bm{\Theta}}^{j}\left (\sum_{i=0}^{p} \underline{\bm{\Theta}}^{p-i} \bm{P}^\prime \bm{\Phi}_i \right ).
	\end{equation}
	Next we  derive an alternative parameterization for  $\bm{A}_{p+j}$ with $j\geq 1$.
	
	Let $K=R+2S$, where $R=\sum_{k=1}^{r} n_k$ and  $S=\sum_{k=1}^{s}m_k$. Under the conditions of this proposition, the  real Jordan form \citep[][Chap. 3]{HJ12} of $\underline{\bm{\Theta}}$ can be written as
	\begin{align} \label{eq:real-Jordan-form}
		\underline{\bm{\Theta}} = \bm{B}\bm{J}\bm{B}^{-1} = \bm{B}\left(\begin{array}{ccccccc}
			\bm{J}_1&&&&&&\\
			&\ddots&&&&&\\
			&&\bm{J}_{r}&&&&\\
			&&&\bm{J}_{r+1}&&&\\
			&&&&\ddots&&\\
			&&&&&\bm{J}_{r+s}&\\
			&&&&&&\bm{0}_{(Nq-K)\times (Nq-K)}\\
		\end{array}\right)\bm{B}^{-1},
	\end{align}
	where $\bm{B}\in\mathbb{R}^{Nq \times Nq}$ is an invertible matrix, each $\bm{J}_k$ with $1\leq k\leq r$ is the $n_k\times n_k$ Jordan block corresponding to $\lambda_{k}$,
	\[
	\bm{J}_{k}=\left(\begin{matrix}
		\lambda_{k} & 1 & &\\
		& \lambda_{k}   & \ddots &\\
		&            & \ddots & 1\\
		&&& \lambda_{k}\\
	\end{matrix}\right), \hspace{5mm} 1\leq k\leq r,
	\]
	and each $\bm{J}_{r+k}$ for $1\leq k\leq s$ is the $2m_k\times 2m_k$ real Jordan block corresponding to the conjugate pair $(\lambda_{r+2k-1}, \lambda_{r+2k})$, 
				\[
				\bm{J}_{r+k}=\left(\begin{matrix}
					\bm{C}_k & \bm{I}_{2} & &\\
					& \bm{C}_k   & \ddots &\\
					&            & \ddots & \bm{I}_2\\
					&&& \bm{C}_k\\
				\end{matrix}\right)\hspace{5mm}\text{with}\quad
				\bm{C}_k=
				\gamma_k\cdot \left(\begin{matrix}
					\cos (\theta_k)& \sin (\theta_k)\\
					- \sin (\theta_k)& \cos (\theta_k)
				\end{matrix}\right), \hspace{5mm} 1\leq k\leq s.
				\]

				Denote $\bm{\widetilde{B}}=\bm{P} \bm{B}$
				and
				$\bm{\widetilde{B}}_-=\bm{B}^{-1}\left (\sum_{i=0}^{p} \underline{\bm{\Theta}}^{p-i} \bm{P}^\prime \bm{\Phi}_i \right )$. 
				Note that when $p=q=1$, $\bm{\widetilde{B}}=\bm{B}$ and $\bm{\widetilde{B}}_-=\bm{B}^{-1}(\bm{\Phi}_1-\bm{\Theta}_1)$.
				Then by \eqref{eq:VARMA-Am} and \eqref{eq:real-Jordan-form}, for $j\geq1$, we have 
				\begin{equation}\label{eq:Ajordan}
					\bm{A}_{p+j} = \bm{\widetilde{B}} \bm{J}^{j} \bm{\widetilde{B}}_-.
				\end{equation}
				According to the block form of $\bm{J}$ in \eqref{eq:real-Jordan-form}, we can partition the $N\times Nq$ matrix $\bm{\widetilde{B}}$ vertically and the $Nq\times N$ matrix $\bm{\widetilde{B}}_-$ horizontally as
				\[
				\bm{\widetilde{B}}=(\bm{\widetilde{B}}_1,\dots, \bm{\widetilde{B}}_{r+s}, \bm{\widetilde{B}}_{r+s+1})\quad\text{and}\quad
				\bm{\widetilde{B}}_-=(\bm{\widetilde{B}}_{-1},\dots, \bm{\widetilde{B}}_{-(r+s)}, \bm{\widetilde{B}}_{-(r+s+1)})^\prime,
				\]
				where $\bm{\widetilde{B}}_k$ and $\bm{\widetilde{B}}_{-k}$ are $N\times n_k$ matrices for $1\leq k\leq r$, $\bm{\widetilde{B}}_{r+k}$ and $\bm{\widetilde{B}}_{-(r+k)}$ are $N\times 2m_k$ matrices for $1\leq k\leq s$, and $\bm{\widetilde{B}}_{r+s+1}$ and 
				$\bm{\widetilde{B}}_{-(r+s+1)}$ are $N\times (Nq-K)$ matrices.
				Notice that for any  $j\geq1$,
				\[
				\bm{J}_k^{j} = \left(
				\begin{matrix}
					\lambda_k^{j}&\binom{j}{1}\lambda_k^{j-1}&\binom{j}{2}\lambda_k^{j-2}&\cdots&\binom{j}{n_k-1}\lambda_k^{j-n_k+1}\\	0&\lambda_k^{j}&\binom{j}{1}\lambda_k^{j-1}&\cdots&\binom{j}{n_k-2}\lambda_k^{j-n_k+2}\\
					\vdots&\vdots&\vdots&\ddots&\vdots\\
					0&0&0&\cdots&\lambda_k^{j}\\
				\end{matrix}
				\right), \quad 1\leq k\leq r,
				\]
				and 
				\begin{equation*}
					\bm{J}_{r+k}^{j} = \left(\begin{matrix}
						\bm{C}_k^{j}&\binom{{j}}{1}\bm{C}_k^{{j}-1}&\binom{{j}}{2}\bm{C}_k^{{j}-2}&\cdots&\binom{{j}}{m_k-1}\bm{C}_k^{{j}-m_k+1}\\
						\bm{0}&\bm{C}_k^{{j}}&\binom{{j}}{1}\bm{C}_k^{{j}-1}&\cdots&\binom{{j}}{m_k-2}\bm{C}_k^{{j}-m_k+2}\\
						\vdots&\vdots&\vdots&\ddots&\vdots\\
						\bm{0}&\bm{0}&\bm{0}&\cdots&\bm{C}_k^{{j}}\\
					\end{matrix}\right), \quad 1\leq k\leq s, 
				\end{equation*}
				with 
				\[
				\bm{C}_k^{j} = \gamma_k^{j} \cdot \left(\begin{matrix}
					\cos(j\theta_k)& \sin(j\theta_k)\\
					-\sin(j\theta_k)&\cos(j\theta_k)
				\end{matrix}\right).
				\]
				
				Let $\bm{\widetilde{b}}_k^{(i)}$ and $\bm{\widetilde{b}}_{-k}^{(i)}$ be the $i$-th columns of $\bm{\widetilde{B}}_k$ and $\bm{\widetilde{B}}_{-k}$, respectively. In addition, denote $\bm{\eta}_k=(\gamma_k,\theta_k)^\prime$ for $1\leq k \leq s$.
				Then by \eqref{eq:Ajordan}, for $j\geq 1$, we can show that 
				\begin{equation}\label{eq:linearcomb1}
					\bm{A}_{p+j} = \sum_{k=1}^{r+s} \bm{\widetilde{B}}_{k} \bm{J}_k^{j} \bm{\widetilde{B}}_{-k}^\prime
					=\sum_{k=1}^{r}\sum_{i=1}^{n_k}\ell_{i,j}^{I}(\lambda_k)\bm{G}_{k,i}^{I}
					+\sum_{k=1}^{s}\sum_{i=1}^{m_k}\left \{\ell_{i,j}^{II,1}(\bm{\eta}_k)\bm{G}_{k,i}^{II,1}+\ell_{i,j}^{II,2}(\bm{\eta}_k)\bm{G}_{k,i}^{II,2}\right \},
				\end{equation}
				where  $\ell_{i,j}^{I}(\cdot)$, $\ell_{i,j}^{II,1}(\cdot)$, and $\ell_{i,j}^{II,2}(\cdot)$ are real-valued functions defined as
				\begin{align}\label{eq:lij}
					\begin{split}
						\ell_{i,j}^{I}(\lambda) & =  \lambda^{j - i + 1} \binom{ j}{ i - 1} \mathbb{I}_{\{j\geq i-1\}},\\
						\ell_{i,j}^{II,1}(\bm{\eta}) &= \gamma^{j-i+1} \binom{ j }{i-1}  \cos \{ (j-i+1) \theta\}\mathbb{I}_{\{j\geq i-1\}},\\
						\ell_{i,j}^{II,2}(\bm{\eta}) &=
						\gamma^{j-i+1} \binom{ j }{i- 1 }  \sin \{( j-i+1) \theta\}\mathbb{I}_{\{j\geq i-1\}},
					\end{split}
				\end{align}
				for any  $\lambda\in (-1,0)\cap(0,1)$ and 
				$\bm{\eta}=(\gamma,\theta)^\prime\in(0,1)\times (-\pi/2, \pi/2)$, 
				and
				\begin{align}\label{eq:Gmat}
					\begin{split}
						\bm{G}_{k,i}^{I} &= \sum_{h=i}^{n_k} \bm{\widetilde{b}}_{k}^{(h-i+1)}\bm{\widetilde{b}}_{-k}^{(h)^\prime}, \quad 1\leq k\leq r,\; 1 \leq i \leq n_k,\\
						\bm{G}_{k,i}^{II,1} &= 	\sum_{h=i}^{ m_k} \left ( \bm{\widetilde{b}}_{r+k}^{(2h - 2i+1)} \bm{\widetilde{b}}_{-(r+k)}^{(2h-1)\prime} +  
						\bm{\widetilde{b}}_{r+k}^{(2h-2i+2)} \bm{\widetilde{b}}_{-(r+k)}^{(2h)\prime} \right ), \quad 1\leq k\leq s,\; 1 \leq i \leq m_k,\\
						\bm{G}_{k,i}^{II,2} &= 	\sum_{h=i}^{ m_k}\left ( \bm{\widetilde{b}}_{r+k}^{(2h - 2i+1)} \bm{\widetilde{b}}_{-(r+k)}^{(2h)\prime} -  
						\bm{\widetilde{b}}_{r+k}^{(2h-2i+2)} \bm{\widetilde{b}}_{-(r+k)}^{(2h-1)\prime} \right ), \quad 1\leq k\leq s,\; 1 \leq i \leq m_k.
					\end{split}
				\end{align}
				Note that for any fixed $k$ and $i$, $\bm{G}_{k,i}^{II,h}$ for $h=1,2$ have the same row and column spaces. Moreover,
				$\rank(\bm{G}_{j,l}^{I})\leq n_k$ and $\rank(\bm{G}_{k,i}^{II,h})\leq 2m_k$
				for all $1\leq j\leq r$, $1\leq k\leq s$, $1\leq l\leq n_k$, $1\leq i\leq m_k$, and $h=1,2$.
				Finally, combining \eqref{eq:init} and \eqref{eq:linearcomb1}, we accomplish the proof of this proposition.
			\end{proof}
			
			\subsection{Proof of Theorem \ref{thm:sol}}
			We first prove the existence. Using iteratively model \eqref{eq:model-scalar}, we get 
			\begin{align}\label{eq:sol}
				\bm{y}_t = \bm{\varepsilon}_t + \sum_{k=1}^{\infty}\sum_{j_1,\dots, j_k\geq1}\bm{A}_{j_1}\cdots \bm{A}_{j_k}\bm{\varepsilon}_{t-j_1-\dots-j_k},  \hspace{5mm} 
				\bm{A}_j=\sum_{k=1}^{d}\ell_{j,k}(\bm{\omega}) \bm{G}_k,
			\end{align}
			where the first equation can be written as $\bm{y}_t = \bm{\varepsilon}_t+ \sum_{j=1}^{\infty}\bm{\Psi}_j\bm{\varepsilon}_{t-j}$ with $\bm{\Psi}_j=\sum_{k=1}^{\infty}\sum_{j_1+\cdots+ j_k=j}\bm{A}_{j_1}\cdots \bm{A}_{j_k}$. Note that  $\bm{y}_t$ takes value in $[-\infty, \infty]^N$. 
			
			
			By the triangle inequality, we have
			\[
			\|\bm{y}_t\|_2 \leq \|\bm{\varepsilon}_t\|_2 + \sum_{k=1}^{\infty}\sum_{j_1,\dots, j_k\geq1}\|\bm{A}_{j_1}\|_{\op} \cdots \|\bm{A}_{j_k}\|_{\op} \| \bm{\varepsilon}_{t-j_1-\dots-j_k}\|_2.
			\]
			Denote $S = \sum_{j=1}^{\infty}\|\bm{A}_j\|_{\op}$. Since $\bm{A}_j=\bm{G}_j$ for $1\leq j\leq p$ and $\bm{A}_j=\sum_{k=p+1}^{d}\ell_{j,k}(\bm{\omega}) \bm{G}_k$ for $j\geq p+1$, we have 
			\begin{equation}\label{eq:S}
				S \leq \sum_{k=1}^{p}\|\bm{G}_j\|_{\op} + \sum_{j=1}^{\infty}\sum_{k=p+1}^{d} |\ell_{j+p,k}(\bm{\omega})| \|\bm{G}_k\|_{\op} \leq \sum_{k=1}^{p}\|\bm{G}_j\|_{\op}+ \sum_{j=1}^{\infty}\sum_{k=p+1}^{d} \rho^j \|\bm{G}_k\|_{\op}< 1.
			\end{equation}
			Since $\bm{\varepsilon}_t$ are $i.i.d.$ with $E(\|\bm{\varepsilon}_t\|_2)<\infty$, this leads to
			\begin{equation}\label{eq:finitey}
				\mathbb{E}(\| \bm{y}_t \|_2) \leq \mathbb{E}(\| \bm{\varepsilon}_t \|_2) (1 + \sum_{k=1}^{\infty} S^k) = \frac{  \mathbb{E}(\| \bm{\varepsilon}_t \|_2) }{ 1- S } < \infty.
			\end{equation}
			Thus, the VMA($\infty$) process  $\{\bm{y}_t\}$    is weakly stationary. 
			This proves the existence of a weakly stationary solution to model \eqref{eq:model-scalar}.
			
			To prove the uniqueness, suppose that $\{\bm{y}_t, t \in \mathbb{Z}\}$ is a weakly stationary and causal solution to model  \eqref{eq:model-scalar}. Then, applying recurrence relation \eqref{eq:model-scalar} $m$ times, we obtain
			\[
			\bm{y}_t = \bm{\varepsilon}_t + \sum_{k=1}^{m}\sum_{j_1,\dots, j_k\geq1}\bm{A}_{j_1}\cdots \bm{A}_{j_k}\bm{\varepsilon}_{t-j_1-\dots-j_k} \bm{\varepsilon}_{t-j_1-\dots-j_k} + \bm{r}_{t,m},  
			\]
			where
			\[
			\bm{r}_{t,m} = \sum_{j_1,\dots, j_{m+1} \geq1}\bm{A}_{j_1}\cdots \bm{A}_{j_k}\bm{y}_{t-j_1-\dots-j_{m+1}}.
			\]
			As shown in  \eqref{eq:S} and \eqref{eq:finitey}, under the conditions of this theorem, $0\leq S<1$ and $\mathbb{E}(\|\bm{y}_t\|_2) <\infty$. As a result,
			\[
			\mathbb{E}(\|\bm{r}_{t,m}\|_2) \leq \sum_{j_1,\dots, j_{m+1} \geq1} \|\bm{A}_{j_1}\|_{\op} \cdots \|\bm{A}_{j_{m+1}}\|_{\op} \mathbb{E}(\| \bm{y}_{t-j_1-\dots-j_{m+1}}\|_2) \leq S^m \mathbb{E}(\|\bm{y}_t\|_2)\rightarrow 0, 
			\]
			as $m\rightarrow\infty$. By the Borel-Cantelli Lemma, as $m\rightarrow\infty$, $\|\bm{r}_{t,m}\|_2 \to 0$ almost surely, that is, $\bm{r}_{t,m}\to \bm{0}$ almost surely.  Thus, $\bm{y}_t$ satisfies \eqref{eq:sol} almost surely, and the uniqueness is verified.
			
			\subsection{Proof of Theorem \ref{thm:identifiable}}
			Let $\pazocal{L}(\cdot):\bm{\Omega}\to\mathbb{R}^{(r+2s)\times(r+2s)}$ and $\pazocal{T}(\cdot):\bm{\Omega}\to\mathbb{C}^{(r+2s)\times(r+2s)}$ be two matrix-valued functions such that for any $\bm{\omega}\in\bm{\Omega}$, $\pazocal{L}(\bm{\omega}) = (\pazocal{L}_1(\bm{\omega}), \dots, \pazocal{L}_{r+2s}(\bm{\omega}))^\prime$ and $\pazocal{T}(\bm{\omega}) = (\pazocal{T}_1(\bm{\omega}), \dots, \pazocal{T}_{r+2s}(\bm{\omega}))^\prime$, where for all $1\leq j \leq r+2s$,
			\begin{align}
				&\pazocal{L}_j(\bm{\omega}) = (\lambda_1^j, \dots,\lambda_r^j, \gamma_1^j\cos(j\theta_1),\gamma_1^j\sin(j\theta_1),\dots, \gamma_s^j\cos(j\theta_s),\gamma_s^j\sin(j\theta_s))^\prime \hspace{3mm}\text{and}\notag\\
				&\pazocal{T}_j(\bm{\omega}) = (\lambda_1^j, \dots,\lambda_r^j,( \gamma_1e^{i\theta_1})^j, (\gamma_1e^{-i\theta_1})^j,\dots,(\gamma_se^{i\theta_s})^j, (\gamma_se^{-i\theta_s})^j)^\prime.\label{eq:Lj-map}
			\end{align}
			Then, let $\pazocal{F}: \pazocal{L}(\cdot) \to \pazocal{T}(\cdot)$ be a functional mapping such that for any $\bm{\omega}\in\bm{\Omega}$, 
			\begin{equation} \label{eq:F-map}
				\pazocal{T}(\bm{\omega}) = 	\pazocal{F}(\pazocal{L}(\bm{\omega})) =  \pazocal{L}(\bm{\omega})\bm{F},\quad\text{with}\quad \bm{F} = \left(\begin{matrix}\bm{I}_r &\\ &\bm{I}_s\otimes\bm{F}_c\end{matrix}\right)\quad\text{and}\quad\bm{F}_c = \left(\begin{matrix}1&1\\i&-i\end{matrix}\right).
			\end{equation}
			Since $\bm{F}$ is invertible, $\pazocal{F}(\cdot)$ is a bijective map. 
			Set $(x_1,\dots, x_{r+2s}) = \pazocal{T}_1(\bm{\omega})$ such that $x_k = \lambda_k$ for $1\leq k \leq r$, while $x_{r+2k-1}=\gamma_ke^{i\theta_k}$ and $x_{r+2k}=\gamma_ke^{-i\theta_k}$ for $1\leq k\leq s$, where we suppress $x_k$'s dependence on $\bm{\omega}$ for notation simplicity.
			
			For any $x\in\mathbb{C}$, we next define a vector-valued function $\bm{v}(x) = (x, x^2,\dots,x^{r+2s})^\prime\in\mathbb{C}^{r+2s}$. 
			For any $\bm{\omega}\in\bm{\Omega}$ satisfying the conditions of this lemma, we first show that $\pazocal{T}(\bm{\omega})$ is invertible. It holds trivially $\pazocal{T}(\bm{\omega}) = (\bm{v}(x_1), \dots, \bm{v}(x_{r+2s}))$. Suppose that there exists $\bm{c}\in\mathbb{C}^{r+2s}$ such that $\pazocal{T}(\bm{\omega})\bm{c} = \bm{0}$. By the condition that $x_k\neq x_\ell$ for all $1\leq k\neq \ell \leq r+2s$, this implies that the $(r+2s)$-order polynomial
			\begin{equation}\label{eq:poly}
				\text{poly}(x) = c_1x + c_2x^2 + \dots c_{r+2s}x^{r+2s} = x(c_1 + c_2x + \dots c_{r+2s}x^{r+2s-1})
			\end{equation}
			has $(r+2s)$ non-zero, distinct roots. Since the polynomial at \eqref{eq:poly} can have at most $(r+2s-1)$ non-zero, distinct roots, it must be a zero polynomial, i.e. $\bm{c} = \bm{0}$. As a result, $\pazocal{T}(\bm{\omega})$ has linearly independent columns, i.e. it is invertible. Moreover, by \eqref{eq:F-map}, $\pazocal{L}(\bm{\omega})$ is also invertible. 
			
			Then, let $\bm{L}_{[1:d]}(\bm{\omega})$ be a square matrix consisting of the first $d = p+r+2s$ rows of $\bm{L}(\bm{\omega})$, and it follows that
			\[
			\bm{L}_{[1:d]}(\bm{\omega}) =  \left(\begin{matrix}\bm{I}_p &\\ &\bm{\pazocal{L}}(\bm{\omega})\end{matrix}\right)\in\mathbb{R}^{d\times d}
			\]
			is an invertible matrix.
			Subsequently, the matrix $(\bm{G}_1,\dots,\bm{G}_d)$ can be uniquely defined as
			\[
			(\bm{G}_1,\dots,\bm{G}_d) = (\bm{A}_1,\dots,\bm{A}_d) \left( [\bm{L}_{[1:d]}(\bm{\omega})^\prime]^{-1} \otimes \bm{I}_N\right),
			\]
			where the inverse of $\bm{L}_{[1:d]}(\bm{\omega})$ is well-defined.
			
			It remains show that $\bm{\omega}$ is unique, i.e. there does not exist $\bm{\widetilde{\omega}} \neq \bm{\omega}$ such that $\bm{A}_j(\bm{\omega}, \cm{G}) = \bm{A}_j(\bm{\widetilde{\omega}}, \cm{\widetilde{G}})$ for all $j\geq 1$, where $\cm{\widetilde{G}}$ may be different from $\cm{G}$ but still satisfies the condition that all $\bm{\widetilde{G}}_k$'s are non-zero matrices. Suppose that such $\bm{\widetilde{\omega}}$ does exist and further set $(\widetilde{x}_1,\dots, \widetilde{x}_{r+2s}) = \pazocal{T}_1(\bm{\widetilde{\omega}})$ where $\pazocal{T}_1(\cdot)$ is defined in \eqref{eq:Lj-map}.
			There must exists some non-zero $\widetilde{x}_k\notin\{x_1,\dots,x_{r+2s}\}$. Suppose that there exists $\bm{c}\in\mathbb{C}^{r+2s+1}$ such that 
			\[
			(\bm{v}(x_1), \dots, \bm{v}(x_{r+2s}), \bm{v}(\widetilde{x}_k))\bm{c} = \bm{0}.
			\]
			This implies that the polynomial at \eqref{eq:poly} has $(r+2s+1)$ non-zero distinct roots, which only holds if $\bm{c}=0$, i.e. $\bm{v}(\widetilde{x}_k)$ is linearly independent of $\bm{v}(x_\ell)$ for all $1\leq \ell \leq r+2s$. By \eqref{eq:F-map}, this implies that the $k$-th column of $\bm{\pazocal{L}}(\bm{\widetilde{\omega}})$ does not belong to the column space of $\bm{\pazocal{L}}(\bm{\omega})$.
			Then, it is implied from
			\[
			(\bm{A}_{p+1}, \dots, \bm{A}_d) = (\bm{G}_r,\dots,\bm{G}_d)\left( \bm{\pazocal{L}}(\bm{\omega})^\prime \otimes \bm{I}_N\right)=(\bm{\widetilde{G}}_r,\dots,\bm{\widetilde{G}}_d)\left( \bm{L}(\bm{\widetilde{\omega}})^\prime \otimes \bm{I}_N\right)
			\] 
			that $\bm{\widetilde{G}}_k = \bm{0}$, leading to a contradiction. 
			
			Therefore, if the parameters $\bm{\omega}$ and $\cm{G}$ satisfy the conditions of the lemma, they are uniquely identified for any $\bm{A}_1,\bm{A}_2,\dots$.

			
			\renewcommand{\thelemma}{S.\arabic{lemma}}
			\setcounter{theorem}{0}	
			
			\section{Proofs for Section 3 in the main paper}\label{sec:proof2}
			
			\subsection{Useful properties of $\bm{L}(\bm{\omega})$}
			According to the definition of $\bm{L}(\bm{\omega})$, for $j\geq1$, denote the $j$th entry of  $\bm{\ell}^{I}(\lambda_i)$ by $\ell_{j}^{I}(\lambda_i)=\lambda_i^j$ with $1\leq i\leq r$,
			and denote the transpose of the  $j$th row of the $\infty\times 2$ matrix $\bm{\ell}^{II}(\bm{\eta}_k)$  by $\ell_{j}^{II}(\bm{\eta}_k):=(\ell_{j}^{II,1}(\bm{\eta}_k),\ell_{j}^{II,2}(\bm{\eta}_k))^\prime=(\gamma_k^j\cos(j\theta_k),\gamma_k^j\sin(j\theta_k) )^\prime$ with $1\leq k\leq s$. 
			Let $\bm{L}^{I}(\bm{\lambda})=(\bm{\ell}^{I}(\lambda_1), \cdots, \bm{\ell}^{I}(\lambda_r))$ and $\bm{L}^{II}(\bm{\eta})=(\bm{\ell}^{II}(\bm{\eta}_1), \cdots, \bm{\ell}^{II}(\bm{\eta}_s))$. In addition, define the following matrix by augmenting $\bm{L}(\bm{\omega})$ with $(r+2s)$ extra columns consisting of  first-order derivatives: 
			\begin{equation}\label{eq:Lstk}
				\bm{L}_{\rm{stack}}(\bm{\omega})
				=\left ( 
				\begin{array}{ccccc}
					\bm{I}_p& & & &\\
					&\bm{L}^{I}(\bm{\lambda})&\bm{L}^{II}(\bm{\eta})&\nabla \bm{L}^{I}(\bm{\lambda})&\nabla_{\theta}  \bm{L}^{II}(\bm{\eta})
				\end{array}
				\right ),
			\end{equation}
			where  $\nabla \bm{L}^{I}(\bm{\lambda})=(	\nabla \bm{\ell}^{I}(\lambda_1), \cdots, \nabla \bm{\ell}^{I}(\lambda_r))$ and $\nabla_{\theta}\bm{L}^{II}(\bm{\eta}) =(\nabla_{\theta}\bm{\ell}^{II}(\bm{\eta}_1), \cdots, \nabla_{\theta}\bm{\ell}^{II}(\bm{\eta}_s))$. 
			We can similarly define $\nabla_{\gamma}\bm{L}^{II}(\bm{\eta})$. Note that from \eqref{eq:eqsp}, it holds $\colsp\{\nabla_{\gamma}\bm{L}^{II}(\bm{\eta})\}=\colsp\{\nabla_{\theta}\bm{L}^{II}(\bm{\eta})\}$, which is why $\nabla_{\gamma}\bm{L}^{II}(\bm{\eta})$ is not included in $\bm{L}_{\rm{stack}}(\bm{\omega})$. Denote  
			\[
			\sigma_{\min, L}=\sigma_{\min}(\bm{L}_{\rm{stack}}(\bm{\omega}^*))\quad\text{and}\quad
			\sigma_{\max, L}=\sigma_{\max}(\bm{L}_{\rm{stack}}(\bm{\omega}^*)),
			\]
			where $\bm{\omega}^*$ is the true value of $\bm{\omega}$.
			Lemma \ref{cor1} below gives some exponential decay properties induced by the parametric form of $\bm{L}(\bm{\omega})$, which will be used repeatedly in our theoretical analysis. Then, based on Lemma \ref{cor1}(ii), we can show that $\sigma_{\min, L}\asymp1$ and $\sigma_{\max, L}\asymp 1$; this is stated in Lemma \ref{lemma:fullrank}. 
			
			\begin{lemma}\label{cor1}
				Suppose that Assumption \ref{assum:statn}(i) holds. Then  (i) there exists an absolute constant $C_L>0$ such that for all $\bm{\omega}\in\bm{\Omega}$ and $j\geq 1$,
				\[\max_{1\leq i\leq r, 1\leq k\leq s, 1\leq h \leq 2}\{|\nabla\ell_{j}^{I}(\lambda_i)|, \|\nabla \ell_{j}^{II,h}(\bm{\eta}_k)\|_2, |\nabla^2\ell_{j}^{I}(\lambda_i)|, \|\nabla^2 \ell_{j}^{II,h}(\bm{\eta}_k)\|_{\Fr}\}\leq C_{L}\bar{\rho}^{j};
				\]
				and (ii) there exists an absolute constant $C_*>0$ such that $\|\bm{A}_{j}^*\|_{\op}\leq C_* \bar{\rho}^j$ for all $j\geq 1$ if Assumption \ref{assum:statn}(iii) further holds.
			\end{lemma}

			\begin{lemma} \label{lemma:fullrank}
				Let $J=2(r+2s)$. Denote $x_k^* = \lambda_k^*$ for $1\leq k \leq r$ and $x^*_{r+2k-1} = \gamma_k^* e^{i\theta_k^*}, x^*_{r+2k} = \gamma_k^*e^{-i\theta_k^*}$ for $ 1\leq k\leq s$, and let $\nu_1^*=\min\{|x^*_k|,1\leq k\leq r+2s\}$ and $\nu_2^*=\min\{|x^*_j-x^*_k|, 1\leq j<k\leq r+2s\}$. Under Assumptions \ref{assum:statn}(i) and \ref{assum:statn}(ii),  the matrix $\bm{L}_{\rm{stack}}(\bm{\omega}^*)$ has full rank, and its maximum and minimum singular values satisfy 
				\begin{equation*}
					\min\{1, c_{\bar{\rho}}\}\leq \sigma_{\min, L}\leq 	\sigma_{\max, L}\leq \max\{1, C_{\bar{\rho}}\}.
				\end{equation*}
				where $	C_{\bar{\rho}}=C_1\bar{\rho}\sqrt{J}(1-\bar{\rho})^{-1}\asymp 1$ and $c_{\bar{\rho}}=0.25^s (\nu_1^*)^{3J/2}(\nu_2^*)^{J(J/2-1)}/C_{\bar{\rho}}^{J-1}\asymp 1$.
			\end{lemma}
			
			\subsection{Notations and main idea: linearization of parametric structure}\label{section:notation}
			
			For  simplicity, denote the perturbations of $\bm{\omega}^*$, $\cm{G}^*$ and $\cm{A}^*$ by $\delta_{\bm{\omega}}=\|\bm{\omega} - \bm{\omega}^*\|_2$, $\delta_{\cmtt{G}}=\|\cm{G}-\cm{G}^*\|_{\Fr}$ and $\delta_{\cmtt{A}}=\|\cm{A}-\cm{A}^*\|_{\Fr}= \|\bm{\Delta}\|_{\Fr}$, respectively. Let
			\begin{equation*}
				\bm{\Upsilon} = \left \{\bm{\Delta} = \cm{A}-\cm{A}^*\in\mathbb{R}^{N\times N\times \infty} \mid \cm{A}=\cm{G}\times_3\bm{L}(\bm{\omega}), \cm{G}\in\bm{\Gamma}(\pazocal{R}_1, \pazocal{R}_2), \bm{\omega}\in\bm{\Omega}, \delta_{\bm{\omega}} \leq c_{\bm{\omega}} \right \},
			\end{equation*} 
			where 	$\bm{\Gamma}(\pazocal{R}_1,\pazocal{R}_2) = \{\cm{G}\in\mathbb{R}^{N\times N\times d}\mid \rank(\cm{G}_{(1)})\leq \pazocal{R}_1, \rank(\cm{G}_{(2)})\leq \pazocal{R}_2 \}$. It is noteworthy that under the conditions of Theorem \ref{thm:parametric}, $\bm{\widehat{\Delta}}:= \cm{\widehat{A}}-\cm{A}^* \in \bm{\Upsilon}$.

			A crucial intermediate step for our theoretical analysis is to establish the following linear approximation within a fixed local neighborhood of $\bm{\omega}^*$,
			\begin{equation} \label{eq:linearize}
				\bm{\Delta}(\bm{\omega}, \cm{G}) =\cm{{A}}(\bm{\omega}, \cm{G})-\cm{A}^* \approx \cm{M}(\bm{\omega}-\bm{\omega}^*, \cm{G}-\cm{G}^*) \times_3\bm{L}_{\rm{stack}}(\bm{\omega}^*),
			\end{equation}
			where  $\cm{M}:  \mathbb{R}^{r+2s} \times \mathbb{R}^{N\times N\times d} \to \mathbb{R}^{N \times N\times (d+r+2s)}$ is a bilinear function defined as follows:
			\begin{align}\label{eq:Gfunc}
				\begin{split}
					& \cm{M}(\bm{a},\cm{B}) = \stk \Bigg ( \cm{B},\hspace{2mm} \big \{a_i \bm{G}_i^{I*} \big\}_{1\leq i\leq r}, \\
					&\Big\{ a_{r+2k} \bm{G}^{II,1*}_{k} - \frac{a_{r+2k-1}}{\gamma_k^*}  \bm{G}^{II,2*}_{k},\hspace{1mm} a_{r+2k-1} \bm{G}^{II,2*}_{k} + \frac{a_{r+2k}}{\gamma_k^*} \bm{G}^{II,1*}_{k} \Big \}_{1 \leq k \leq s} 
					\Bigg),
				\end{split}
			\end{align}
			for any $\bm{a} = (a_1,\dots,a_{r+2s})^\prime \in \mathbb{R}^{r+2s}$ and $\cm{B}\in\mathbb{R}^{N\times N\times d}$, with the true values $\bm{\omega}^*$ and  $\cm{G}^*$ fixed. 
			The linear approximation in \eqref{eq:linearize} will be formalized in the proof of Lemma \ref{lemma:delnorm}; in particular, see \eqref{eq:stackH} and \eqref{eq:Delta} for the linear form and the remainder term, respectively. 
			
			In addition, the following notations will be used in the proof of Lemma \ref{lemma:delnorm}.
			First, for the convenience of notation in the proof, according to the block form of $\bm{L}(\bm{\omega})$, we partition $\cm{G}\in\mathbb{R}^{N\times N\times d}$ as  
			$\cm{G}=\stk(\cm{G}^{\textrm{AR}}, \cm{G}^{\textrm{MA}})=(\cm{G}^{\textrm{AR}},\cm{G}^{I},\cm{G}^{II})$,
			where $\cm{G}^{\textrm{AR}}=\stk(\bm{G}_1,\dots, \bm{G}_p)$, $\cm{G}^{I}=\stk(\bm{G}^{I}_1,\ldots,\bm{G}^{I}_r)$, and $\cm{G}^{II}=\stk(\bm{G}^{II,1}_{1}, \bm{G}^{II,2}_{1},\ldots,\bm{G}^{II,1}_{s}, \bm{G}^{II,2}_{s})$ are $N\times N\times p$, $N\times N\times r$, and $N\times N\times 2s$ tensors, respectively.
			Here, $\bm{G}^{I}_i=\bm{G}_{p+i}$  for $1\leq i\leq r$, and $\bm{G}^{II,1}_{k}=\bm{G}_{p+r+2k-1}$ and $\bm{G}^{II,2}_{k}=\bm{G}_{p+r+2k}$ for $1\leq k\leq s$.
			Then, for any $\cm{A}=\cm{G}\times_3\bm{L}(\bm{\omega})$, we have $\bm{A}_k=\bm{G}_k$ for $1\leq k\leq p$, and
			\begin{align}\label{eq:linear}
				\bm{A}_{p+j} =  \sum_{i=1}^{r}\ell_{j}^{I}(\lambda_i)\bm{G}_i^{I}+\sum_{k=1}^{s}\left\{\ell_{j}^{II,1}(\bm{\eta}_k)\bm{G}_{k}^{II,1}+\ell_{j}^{II,2}(\bm{\eta}_k)\bm{G}_{k}^{II,2}\right\}, \hspace{2mm}\text{for}\hspace{2mm} j\geq 1.
			\end{align}
			Moreover, for simplicity, let 
			\begin{equation*}
				\cm{G}_{\rm{stack}} =\cm{M}(\bm{\omega}-\bm{\omega}^*, \cm{G}-\cm{G}^*).
			\end{equation*}
			Equivalently, we can express $\cm{G}_{\rm{stack}}=\stk\left (\cm{G}-\cm{G}^*, \cm{D}(\bm{\omega})\right )$ as the $N\times N\times (d+r+2s)$ tensor formed by augmenting $\cm{G}-\cm{G}^*$ with the $N\times N\times (r+2s)$ tensor
			\begin{align*}
				\cm{D}(\bm{\omega}) &= \stk \bigg ( \big \{(\lambda_i - \lambda_i^*)\bm{G}_i^{I*}\big\}_{1\leq i\leq r}, \notag\\
				&\big\{(\theta_k - \theta_k^*) \bm{G}^{II,1*}_{k} - \frac{\gamma_k -\gamma_k^*}{\gamma_k^*}\bm{G}^{II,2*}_{k}, \; (\theta_k - \theta_k^*) \bm{G}^{II,2*}_{k} + \frac{\gamma_k -\gamma_k^*}{\gamma_k^*}\bm{G}^{II,1*}_{k} \big \}_{1\leq k\leq s} \bigg).
			\end{align*}
			Lastly, note that for every $\bm{\Delta}(\bm{\omega}, \cm{G})\in \bm{\Upsilon}$, its corresponding $\cm{G}_{\rm{stack}}\in\bm{\Xi}$, where
			\begin{equation}\label{eq:Xi}
				\bm{\Xi} = \left\{ \cm{M}(\bm{a} , \cm{B} ) \in  \mathbb{R}^{N \times N \times (d+r+2s)} \mid   \bm{a}  \in \mathbb{R}^{r+2s}, \cm{B} \in \bm{\Gamma}(2\pazocal{R}_1, 2\pazocal{R}_2) \right\}.
			\end{equation}

			\subsection{Proof of Lemma \ref{lemma:delnorm}}
			By Assumption \ref{assum:statn}(iii), without loss of generality, let $\max_{p+1\leq k\leq d}\|\bm{G}_k^*\|_{\Fr}=\alpha$ and $\min_{p+1\leq k\leq d}\|\bm{G}_k^*\|_{\Fr}=c_{\cmtt{G}}\alpha$, where $c_{\cmtt{G}}>0$ is an absolute constant.
			
			Let $\bm{\Delta}=\cm{A}-\cm{A}^* = \cm{G} \times_3 \bm{L}(\bm{\omega}) - \cm{G}^* \times_3 \bm{L}(\bm{\omega}^*)$. Denote by $\bm{\Delta}_j$ with $j\geq 1$ the frontal slices of $\bm{\Delta}$, i.e. $\bm{\Delta}_{(1)}=(\bm{\Delta}_1,\bm{\Delta}_2,\dots)$.
			Then $\bm{\Delta}_j=\bm{G}_j-\bm{G}_j^*$ for $1\leq j\leq p$.
			For $j\geq 1$, by \eqref{eq:linear} and the Taylor expansion,
			\begin{align}\label{eq:delta}
				\bm{\Delta}_{p+j}&=\bm{A}_{p+j}-\bm{A}_{p+j}^* \notag \\
				&=\sum_{k=1}^{r}\Bigg \{\ell_{j}^{I}(\lambda_k^*) +\nabla\ell_{j}^{I}(\lambda_k^*) (\lambda_k-\lambda_k^*) +\frac{1}{2}\nabla^2\ell_{j}^{I}(\widetilde{\lambda}_k) (\lambda_k-\lambda_k^*)^2 \Bigg \}\bm{G}_k^{I}\notag \\
				&\hspace{5mm} +\sum_{k=1}^{s}\Bigg \{\ell_{j}^{II,1}(\bm{\eta}_k^*) +(\bm{\eta}_k-\bm{\eta}_k^*)^\prime \nabla \ell_{j}^{II,1}(\bm{\eta}_k^*) \notag\\
				&\hspace{33mm} +\frac{1}{2}(\bm{\eta}_k-\bm{\eta}_k^*)^{\prime}\nabla^2 \ell_{j}^{II,1}(\widetilde{\bm{\eta}}_k)(\bm{\eta}_k-\bm{\eta}_k^*)\Bigg \}\bm{G}_{k}^{II,1}\notag \\
				&\hspace{5mm} +\sum_{k=1}^{s}\Bigg \{\ell_{j}^{II,2}(\bm{\eta}_k^*) + (\bm{\eta}_k-\bm{\eta}_k^*)^\prime \nabla \ell_{j}^{II,2}(\bm{\eta}_k^*) \notag\\ &\hspace{33mm}+\frac{1}{2}(\bm{\eta}_k-\bm{\eta}_k^*)^{\prime}\nabla^2 \ell_{j}^{II,2}(\widetilde{\bm{\eta}}_k) (\bm{\eta}_k-\bm{\eta}_k^*)\Bigg \}\bm{G}_{k}^{II,2} -\bm{A}_{p+j}^* \notag\\
				&:=\bm{H}_j+\bm{R}_j,
			\end{align}
			where $\widetilde{\lambda}_k$ lies between $\lambda_k^*$ and $\lambda_k$ for $1 \leq k \leq r$, $\widetilde{\bm{\eta}}_k$ lies between $\bm{\eta}^*_k$ and $\bm{\eta}_k$ for $1 \leq k \leq s$, 
			\begin{align}
				\begin{split}\label{eq:Hj}
					\bm{H}_j &=\sum_{k=1}^{r}\ell_{j}^{I}(\lambda_k^*) (\bm{G}_k^{I}-\bm{G}_k^{I*}) +\sum_{k=1}^{s}\sum_{h=1}^2\ell_{j}^{II,h}(\bm{\eta}_k^*) (\bm{G}_{k}^{II,h}-\bm{G}_{k}^{II,h*}) \\
					&\hspace{5mm} +\sum_{k=1}^{r}(\lambda_k-\lambda_k^*)\nabla\ell_{j}^{I}(\lambda_k^*) \bm{G}_k^{I*} +\sum_{k=1}^{s}\sum_{h=1}^2(\bm{\eta}_k-\bm{\eta}_k^*)^\prime \nabla \ell_{j}^{II,h}(\bm{\eta}_k^*)\bm{G}_{k}^{II,h*},
				\end{split}
			\end{align}
			and
			\begin{align} 
				\begin{split}\label{eq:Rj}
					\bm{R}_j &= \sum_{i=1}^{r}\nabla\ell_{j}^{I}(\lambda_k^*) (\lambda_k-\lambda_k^*) (\bm{G}_k^{I} - \bm{G}_k^{I*}) \\
					&\hspace{5mm}+ \sum_{k=1}^{s}\sum_{h=1}^2(\bm{\eta}_k-\bm{\eta}_k^*)^\prime \nabla \ell_{j}^{II,h}(\bm{\eta}_k^*)(\bm{G}_{k}^{II,h} - \bm{G}_{k}^{II,h*}) \\
					&\hspace{5mm} +\frac{1}{2} \sum_{k=1}^{r}\nabla^2\ell_{j}^{I}(\widetilde{\lambda}_k) (\lambda_k-\lambda_k^*)^2 \bm{G}_k^{I} \\
					&\hspace{5mm} 
					+\frac{1}{2} \sum_{k=1}^{s}\sum_{h=1}^2(\bm{\eta}_k-\bm{\eta}_k^*)^{\prime}\nabla^2 \ell_{j}^{II,h}(\widetilde{\bm{\eta}}_k)(\bm{\eta}_k-\bm{\eta}_k^*)\bm{G}_{k}^{II,h}.	
				\end{split}
			\end{align}
			
			We first handle the terms in $\bm{R}_j$, and denote $\bm{R}_j=\bm{R}_{1j} +\bm{R}_{2j}+\bm{R}_{3j}$, where
			\begin{align}\label{eq:Rjs}
				\begin{split}
					\bm{R}_{1j}=& \sum_{k=1}^{r}\nabla\ell_{j}^{I}(\lambda_k^*)  (\lambda_k-\lambda_k^*) (\bm{G}_k^{I} - \bm{G}_k^{I*}) \\
					&
					+ \sum_{k=1}^{s}\sum_{h=1}^2(\bm{\eta}_k-\bm{\eta}_k^*)^\prime \nabla \ell_{j}^{II,h}(\bm{\eta}_k^*)(\bm{G}_{k}^{II,h} - \bm{G}_{k}^{II,h*}),\\
					\bm{R}_{2j} =&\frac{1}{2}\sum_{k=1}^{r}\nabla^2\ell_{j}^{I}(\widetilde{\lambda}_k) (\lambda_k-\lambda_k^*)^2 (\bm{G}_k^{I} - \bm{G}_k^{I*} ) 
					\\&+\frac{1}{2} \sum_{k=1}^{s}\sum_{h=1}^2(\bm{\eta}_k-\bm{\eta}_k^*)^{\prime}\nabla^2 \ell_{j}^{II,h}(\widetilde{\bm{\eta}}_k)(\bm{\eta}_k-\bm{\eta}_k^*)(\bm{G}_{k}^{II,h} - \bm{G}_{k}^{II,h*}), \\
					\bm{R}_{3j} =&\frac{1}{2}\sum_{k=1}^{r}\nabla^2\ell_{j}^{I}(\widetilde{\lambda}_k) (\lambda_k-\lambda_k^*)^2 \bm{G}_k^{I*} \\
					&+ \frac{1}{2} \sum_{k=1}^{s}\sum_{h=1}^2(\bm{\eta}_k-\bm{\eta}_k^*)^{\prime}\nabla^2 \ell_{j}^{II,h}(\widetilde{\bm{\eta}}_k)(\bm{\eta}_k-\bm{\eta}_k^*) \bm{G}_{k}^{II,h*}.
				\end{split}
			\end{align}
			Note that for any matrix $\bm{Y} = \sum_{k=1}^{d}a_k\bm{X}_k$, it holds 
			\[
			\|\bm{Y}\|_{\op}\leq \|\bm{Y}\|_{\Fr}\leq  (\sum_{k=1}^{d}\|\bm{X}_k\|_{\Fr}^2)^{1/2}(\sum_{k=1}^{d}a_k^2)^{1/2}=\|\cm{X}\|_{\Fr}\|\bm{a}\|_2 ,
			\] 
			and $\sum_{k=1}^da_k^4\leq (\sum_{k=1}^da_k^2)^2$, where $\bm{a} = (a_1,\dots,a_d)^\prime\in\mathbb{R}^d$, and $\cm{X}$ is a tensor with frontal slices $\bm{X}_k$'s such that $\cm{X}_{(1)}=(\bm{X}_1,\dots,\bm{X}_{d})$. 
			Then, by Lemma \ref{cor1}(i),
			\begin{align*}
				\| \bm{R}_{1j}\|_{\Fr}
				&\leq C_L\bar{\rho}^j \sqrt{\|\bm{\lambda} - \bm{\lambda}^*\|_2^2 + 2\|\bm{\eta} - \bm{\eta}^*\|_2^2} \cdot
				\|\cm{G}^{\textrm{MA}} - \cm{G}^{\textrm{MA}*}\|_{\Fr} 	 \notag \\
				&\leq \sqrt{2}C_L\bar{\rho}^j \delta_{\bm{\omega}} \cdot \|\cm{G}^{\textrm{MA}} - \cm{G}^{\textrm{MA}*}\|_{\Fr} \leq \sqrt{2}C_L\bar{\rho}^j \delta_{\bm{\omega}} \delta_{\cmtt{G}},
			\end{align*}
			and
			\begin{align*}
				\| \bm{R}_{2j}\|_{\Fr}&\leq \frac{\sqrt{2}}{2}C_L\bar{\rho}^j \delta_{\bm{\omega}}^2\sqrt{ \sum_{k=1}^{r}\|\bm{G}_k^{I} - \bm{G}_k^{I*}\|_{\Fr}^2 + \sum_{k=1}^{s}\sum_{h=1}^2 \|\bm{G}_{k}^{II,h} - \bm{G}_{k}^{II,h*}\|_{\Fr}^2 } \notag\\
				&\leq \frac{\sqrt{2}}{2} C_L\bar{\rho}^j \delta_{\bm{\omega}}^2 \cdot \|\cm{G}^{\textrm{MA}} - \cm{G}^{\textrm{MA}*}\|_{\Fr} \leq \frac{\sqrt{2}}{2} C_L\bar{\rho}^j \delta_{\bm{\omega}}^2 \delta_{\cmtt{G}}.
			\end{align*}
			Moreover, by Assumption \ref{assum:statn}(iii) and Lemma \ref{cor1}(i), we can show that
			\begin{equation*}
				\| \bm{R}_{3j}\|_{\Fr} \leq C_L \alpha \bar{\rho}^j\delta_{\bm{\omega}}^2.
			\end{equation*}
			As a result,
			\begin{align} \label{eq:Rnorm1}
				\|\bm{R}_j\|_{\Fr} &\leq \|\bm{R}_{1j}\|_{\Fr} + \|\bm{R}_{2j}\|_{\Fr} + \|\bm{R}_{3j}\|_{\Fr} \notag\\ &\leq    C_L \bar{\rho}^j \delta_{\bm{\omega}} \left ( \sqrt{2}  \delta_{\cmtt{G}} + \frac{\sqrt{2}}{2} \delta_{\bm{\omega}}\delta_{\cmtt{G}}+  \alpha\delta_{\bm{\omega}} \right ).
			\end{align}
			
			Now consider $\bm{H}_j$ in \eqref{eq:Hj}. 
			Notice that for any $j\geq1$ and $1\leq k\leq s$, 
			\begin{align}\label{eq:eqsp}
				\begin{split}
					&\nabla_\gamma\ell_{j}^{II,1}(\bm{\eta}_k)=j\gamma_k^{j-1}\cos(j\theta_k)=\frac{1}{\gamma_{k}}\nabla_\theta\ell_{j}^{II,2}(\bm{\eta}_k),\\
					&\nabla_\gamma\ell_{j}^{II,2}(\bm{\eta}_k)=j\gamma_k^{j-1}\sin(j\theta_k)=-\frac{1}{\gamma_{k}}\nabla_\theta\ell_{j}^{II,1}(\bm{\eta}_k).
				\end{split}
			\end{align}
			Thus, the last term on the right side of  \eqref{eq:Hj} can be simplified to 
			\begin{align} \label{eq:linearcomb}
				\begin{split}
					&\sum_{k=1}^{s}\sum_{h=1}^2(\bm{\eta}_k-\bm{\eta}_k^*)^\prime \nabla \ell_{j}^{II,h}(\bm{\eta}_k^*)\bm{G}_{k}^{II,h*}\\
					&\hspace{5mm}= \sum_{k=1}^{s}\left[ (\theta_k - \theta_k^*) \bm{G}^{II,1*}_{k} - \frac{1}{\gamma_k^*}(\gamma_k -\gamma_k^*)\bm{G}^{II,2*}_{k}\right]\nabla_{\theta} \ell_{j}^{II,1}(\bm{\eta}_k^*)\\
					&\hspace{10mm}+\sum_{k=1}^{s} \left[ (\theta_k - \theta_k^*) \bm{G}^{II,2*}_{k} + \frac{1}{\gamma_k^*}(\gamma_k -\gamma_k^*)\bm{G}^{II,1*}_{k}\right]\nabla_{\theta} \ell_{j}^{II,2}(\bm{\eta}_k^*).
				\end{split}
			\end{align}
			Let $\cm{H}=\stk(\bm{H}_1, \bm{H}_2, \dots)$ and  $\cm{R}=\stk(\bm{R}_1, \bm{R}_2, \dots)$. Then by \eqref{eq:Hj} and \eqref{eq:linearcomb}, it can be verified that 
			\begin{align}\label{eq:stackH}
				\cm{\widetilde{H}}&:= \stk(\bm{G}_1-\bm{G}_1^*,\cdots,\bm{G}_p-\bm{G}_p^*,\cm{H})\notag\\&=(\cm{G}-\cm{G}^*)\times_3  \bm{L}(\bm{\omega}^*) +\cm{D}(\bm{\omega}) \times_3 \left (\nabla \bm{L}^{I}(\bm{\lambda}^*), \nabla_{\theta}  \bm{L}^{II}(\bm{\eta}^*)\right ) \notag\\
				&= \cm{G}_{\rm{stack}}\times_3\bm{L}_{\rm{stack}}(\bm{\omega}^*),
			\end{align}
			where $\cm{D}(\bm{\omega})\in\mathbb{R}^{N\times N\times (r+2s)}$  
			and $\cm{G}_{\rm{stack}}\in\mathbb{R}^{N\times N\times (d+r+2s)}$ are  defined in Appendix \ref{section:notation}.
			Note that 
			\begin{equation}\label{eq:Delta}
				\bm{\Delta}=\cm{\widetilde{H}}+\stk(\bm{0}_{N\times N\times p}, \cm{R}).
			\end{equation}
			
			Moreover,
			\begin{align}\label{eq:DFr2}
				\|\cm{D}(\bm{\omega})\|_{\Fr}^2 &= \sum_{i=1}^{r}(\lambda_i - \lambda_i^*)^2\|\bm{G}_i^{I*}\|_{\Fr}^2 + \sum_{k=1}^{s} \left\| (\theta_k - \theta_k^*) \bm{G}^{II,1*}_{k} - \frac{\gamma_k -\gamma_k^*}{\gamma_k^*}\bm{G}^{II,2*}_{k}\right\|_{\Fr}^2 \notag\\
				&\hspace{5mm}+ \sum_{k=1}^{s}\left\| (\theta_k - \theta_k^*) \bm{G}^{II,2*}_{k} + \frac{\gamma_k -\gamma_k^*}{\gamma_k^*}\bm{G}^{II,1*}_{k}\right\|_{\Fr}^2 \notag\\
				&= \sum_{i=1}^{r}(\lambda_i - \lambda_i^*)^2\|\bm{G}_i^{I*}\|_{\Fr}^2 + \sum_{k=1}^{s} (\theta_k - \theta_k^*)^2(\|\bm{G}^{II,1*}_{k}\|_{\Fr}^2 + \|\bm{G}^{II,2*}_{k}\|_{\Fr}^2 ) \notag\\
				&\hspace{5mm}+ \sum_{k=1}^{s}\frac{(\gamma_k - \gamma_k^*)^2}{\gamma_k^{*2}} (\|\bm{G}^{II,1*}_{k}\|_{\Fr}^2 + \|\bm{G}^{II,2*}_{k}\|_{\Fr}^2 ),
			\end{align}
			which, together with Assumption \ref{assum:statn}(iii), leads to
			\begin{equation}\label{eq:DFr}
				\sqrt{2}c_{\cmtt{G}} \alpha \delta_{\bm{\omega}}  \leq \|\cm{D}(\bm{\omega})\|_{\Fr} \leq \frac{\sqrt{2}\alpha \delta_{\bm{\omega}}}{\min_{1\leq k\leq s}\gamma_{k}^*}.
			\end{equation}
			By the simple inequalities $ (|x| + |y|) / 2 \leq \sqrt{x^2 + y^2} \leq |x| + |y|$, we have
			$0.5(\delta_{\cmtt{G}} + \|\cm{D}(\bm{\omega})\|_{\Fr}) \leq  \|\cm{G}_{\rm{stack}}\|_{\Fr}	\leq \delta_{\cmtt{G}} + \|\cm{D}(\bm{\omega})\|_{\Fr}$,
			and thus in view of \eqref{eq:DFr} we further have
			\begin{equation}\label{eq:Gstacknorm}
				0.5\left (\delta_{\cmtt{G}} +\sqrt{2}c_{\cmtt{G}}\alpha \delta_{\bm{\omega}}\right ) \leq  \|\cm{G}_{\rm{stack}}\|_{\Fr}	\leq \delta_{\cmtt{G}} + \frac{\sqrt{2}\alpha \delta_{\bm{\omega}}}{\min_{1\leq k\leq s}\gamma_{k}^*}, 
			\end{equation} 
			where $\delta_{\cmtt{G}} = \|\cm{G} - \cm{G}^*\|_{\Fr}$.
			By Lemma \ref{lemma:fullrank}, $\sigma_{\min, L}=\sigma_{\min}(\bm{L}_{\rm{stack}}(\bm{\omega}^*))>0$. 
			Then  it follows from \eqref{eq:Gstacknorm} that
			\[
			0.5\sigma_{\min, L}\left (\delta_{\cmtt{G}} +\sqrt{2}c_{\cmtt{G}}\alpha \delta_{\bm{\omega}}\right )\leq \|\cm{\widetilde{H}}\|_{\Fr} \leq  \sigma_{\max, L}\left (\delta_{\cmtt{G}} + \frac{\sqrt{2}\alpha \delta_{\bm{\omega}}}{\min_{1\leq k\leq s}\gamma_{k}^*}\right ).
			\]
			Combining this with \eqref{eq:Rnorm1}, \eqref{eq:Delta}, \eqref{eq:DFr}, as well as the fact that $\|\cm{G}^{\textrm{MA}} - \cm{G}^{\textrm{MA}*}\|_{\Fr}\leq \delta_{\cmtt{G}}$, we have
			\begin{align*}
				\|\bm{\Delta}\|_{\Fr} &\leq \|\cm{\widetilde{H}}\|_{\Fr} + \|\cm{R}\|_{\Fr} \\
				&\leq \left\{\sigma_{\max, L} + \frac{\sqrt{2}C_L}{1-\bar{\rho}} \left (\delta_{\bm{\omega}}+\frac{\delta_{\bm{\omega}}^2}{2}\right )\right\} \delta_{\cmtt{G}}
				+\left(\frac{\sqrt{2} \sigma_{\max, L}}{\min_{1\leq k\leq s}\gamma_{k}^*} + \frac{C_L }{1-\bar{\rho}} \delta_{\bm{\omega}}\right)\alpha
			\end{align*}
			and
			\begin{align*}
				\|\bm{\Delta}\|_{\Fr} &\geq \|\cm{\widetilde{H}}\|_{\Fr} - \|\cm{R}\|_{\Fr}\\
				&\geq \left\{0.5\sigma_{\min, L}-\frac{\sqrt{2}C_L}{1-\bar{\rho}}  \left (\delta_{\bm{\omega}}+\frac{\delta_{\bm{\omega}}^2}{2}\right )\right\} \delta_{\cmtt{G}} +\left( \frac{c_{\cmtt{G}} \sigma_{\min, L}}{\sqrt{2}}- \frac{C_L }{1-\bar{\rho}} \delta_{\bm{\omega}}\right) \alpha. 
			\end{align*}
			Thus, by taking
			\[
			\delta_{\bm{\omega}}\leq c_{\bm{\omega}}= \min\left \{2, \frac{c_{\cmtt{G}}(1-\bar{\rho})\sigma_{\min, L} }{8\sqrt{2}C_L  }\right \},
			\]
			we can show that 
			\begin{equation*}
				c_{\Delta}	\left(\delta_{\cmtt{G}} + \alpha \delta_{\bm{\omega}}\right) 
				\leq \|\bm{\Delta}\|_{\Fr} \leq         
				C_{\Delta}	\left(\delta_{\cmtt{G}} +  \alpha \delta_{\bm{\omega}}\right),
			\end{equation*}
			where
			\[
			c_{\Delta}= c_l\cdot\sigma_{\min, L}  \asymp 1 \quad \text{and}\quad
			C_{\Delta}= c_u\cdot \max\left \{\sigma_{\max, L},(1-\bar{\rho})^{-1} \right \} \asymp 1,
			\]
			with $c_l=0.25\min\{1,\sqrt{2}c_{\cmtt{G}}\}$ and $c_u=1+\sqrt{2}(\min_{1\leq k\leq s}\gamma_{k}^*)^{-1}+(4\sqrt{2}+2)C_L$. By Lemma \ref{lemma:fullrank}, we have $c_{\bm{\omega}}\asymp 1$,  $c_{\Delta}\asymp 1$, and $C_{\Delta}\asymp 1$. The proof of this lemma is complete.
			
			
			\subsection{Proof of Theorem \ref{thm:parametric}}\label{section:proofThm2}
			
			Let $\bm{x}_t=(\bm{y}_{t-1}^\prime, \bm{y}_{t-2}^\prime, \dots)^\prime$.
			Denote by $\bm{\Delta}_j$ with $j\geq 1$ the frontal slices of $\bm{\Delta}$, i.e., $\bm{\Delta}_{(1)}=(\bm{\Delta}_1,\bm{\Delta}_2,\dots)$. Denote 
			\begin{align}\label{eq:notation_init}
				\begin{split}
					& S_1(\bm{\Delta}) = \frac{2}{T}\sum_{t=1}^{T}\langle \sum_{j=1}^{\infty}\bm{\Delta}_j\bm{y}_{t-j} , \sum_{k=t}^{\infty}\bm{\Delta}_k\bm{y}_{t-k}\rangle,\\
					&S_2(\bm{\Delta})  = \frac{2}{T}\sum_{t=1}^{T}\langle \sum_{j=t}^{\infty}\bm{A}_j^* \bm{y}_{t-j}, \sum_{k=1}^{t-1}\bm{\Delta}_k \bm{y}_{t-k} \rangle, \\
					&S_3(\bm{\Delta}) = \frac{2}{T}\sum_{t=1}^{T}\langle \bm{\varepsilon}_t, \sum_{j=t}^{\infty}\bm{\Delta}_j \bm{y}_{t-j} \rangle.
				\end{split}
			\end{align}
			The following three lemmas are sufficient for the proof of 
			Theorem \ref{thm:parametric}.
			
			\begin{lemma}[Strong convexity and smoothness properties] \label{lemma:rsc}
				Under Assumptions \ref{assum:error} and  \ref{assum:statn}, if $T \gtrsim (\kappa_2 / \kappa_1)^2  d_{\pazocal{R}} \log(\kappa_2/\kappa_1)$, then 	with probability at least $1  -2e^{-cd_{\pazocal{R}}\log(\kappa_2/\kappa_1)}-3e^{-cN}$, 
				\begin{equation*}
					\kappa_1  \|\bm{\Delta}\|_{\Fr}^2 \lesssim	\frac{1}{T}\sum_{t=1}^{T}\|\bm{\Delta}_{(1)}\bm{x}_t\|_2^2 \lesssim  \kappa_2 \|\bm{\Delta}\|_{\Fr}^2, \quad \forall \bm{\Delta}\in\bm{\Upsilon}.
				\end{equation*}
			\end{lemma}
			
			\begin{lemma}[Deviation bound]\label{lemma:dev}
				Under the conditions of Lemma \ref{lemma:rsc}, with probability at least $1-2e^{-cd_{\pazocal{R}}\log(\kappa_2/\kappa_1)}-5e^{-cN}$,
				\[
				\frac{1}{T}\left |\sum_{t=1}^{T}\langle \bm{\varepsilon}_t, \bm{\Delta}_{(1)}\bm{x}_t \rangle \right | \lesssim  \sqrt{\frac{\kappa_2 \lambda_{\max}(\bm{\Sigma}_{\varepsilon}) d_{\pazocal{R}}}{T}} \|\bm{\Delta}\|_{\Fr}, \quad \forall \bm{\Delta}\in\bm{\Upsilon}.
				\]
			\end{lemma}

			\begin{lemma}[Effects of initial values]\label{lemma:init}
				Under Assumptions \ref{assum:error} and  \ref{assum:statn}, if $T \gtrsim (\kappa_2 / \kappa_1)  d_{\pazocal{R}}$, then with probability at least $1-\{ 2+\sqrt{\kappa_2/\lambda_{\max}(\bm{\Sigma}_{\varepsilon})} \} \sqrt{N/\{ (\pazocal{R}_1+ \pazocal{R}_2) T\}}$,
				\[
				|S_1(\bm{\Delta})| \lesssim \kappa_1 \|\bm{\Delta}\|_{\Fr}^2, \quad |S_i(\bm{\Delta})| \lesssim \sqrt{\frac{\kappa_2 \lambda_{\max}(\bm{\Sigma}_{\varepsilon}) d_{\pazocal{R}}}{T}}  \|\bm{\Delta}\|_{\Fr},\quad i=2,3, \quad \forall \bm{\Delta}\in\bm{\Upsilon}.
				\]
			\end{lemma}

			Now we prove Theorem \ref{thm:parametric}. 
			Note that $\sum_{j=1}^{t-1}\bm{A}_j\bm{y}_{t-j}= \cm{A}_{(1)}\bm{\widetilde{x}}_{t}$. 
			Due to the optimality of $\cm{\widehat{A}}$, we have
			\[
			\sum_{t=1}^{T} \| \bm{y}_t - \cm{A}^*_{(1)}\bm{\widetilde{x}}_{t} - \widehat{\bm{\Delta}}_{(1)} \bm{\widetilde{x}}_{t}\|_2^2
			\leq \sum_{t=1}^{T} \| \bm{y}_t - \cm{A}^*_{(1)}\bm{\widetilde{x}}_{t}\|_2^2,
			\]
			Then, since $\bm{y}_t - \cm{A}^*_{(1)}\bm{\widetilde{x}}_{t}=\bm{\varepsilon}_t +   \sum_{j=t}^{\infty}\bm{A}_j^* \bm{y}_{t-j}$, it follows that
			\begin{align}\label{eq:thm1eq}
				\frac{1}{T}\sum_{t=1}^{T}\|\widehat{\bm{\Delta}}_{(1)}\bm{\widetilde{x}}_{t}\|_2^2 &\leq 
				\frac{2}{T}\sum_{t=1}^{T}\langle \bm{\varepsilon}_t, \widehat{\bm{\Delta}}_{(1)}\bm{\widetilde{x}}_{t} \rangle + \underbrace{\frac{2}{T}\sum_{t=1}^{T}\langle \sum_{j=t}^{\infty}\bm{A}_j^* \bm{y}_{t-j}, \widehat{\bm{\Delta}}_{(1)}\bm{\widetilde{x}}_{t}  \rangle}_{S_2(\bm{\widehat{\Delta}})} \notag \\
				&= \frac{2}{T}\sum_{t=1}^{T}\langle \bm{\varepsilon}_t, \bm{\widehat{\Delta}}_{(1)}\bm{x}_t \rangle + S_2(\bm{\widehat{\Delta}}) - S_3(\bm{\widehat{\Delta}}),
			\end{align}
			where  $S_2(\cdot)$ and $S_3(\cdot)$ are defined as in \eqref{eq:notation_init}, $\widehat{\bm{\Delta}}_{(1)}\bm{\widetilde{x}}_{t}= \sum_{k=1}^{t-1}\widehat{\bm{\Delta}}_k \bm{y}_{t-k}$, and $\widehat{\bm{\Delta}}_{(1)}\bm{x}_{t}= \sum_{k=1}^{\infty}\widehat{\bm{\Delta}}_k \bm{y}_{t-k}$.
			
			Moreover, applying the inequality $\|\bm{a}-\bm{b}\|_2^2\geq\|\bm{a}\|_2^2-2\langle\bm{a},\bm{b}\rangle$ with $\bm{a}=\widehat{\bm{\Delta}}_{(1)}\bm{x}_{t}=\sum_{j=1}^{\infty}\widehat{\bm{\Delta}}_j\bm{y}_{t-j}$ and $\bm{b}=\sum_{k=t}^{\infty}\widehat{\bm{\Delta}}_k\bm{y}_{t-k}$, we can lower bound the left-hand side of  \eqref{eq:thm1eq} to further obtain that
			\begin{align} \label{eq:thm1eq1}
				\frac{1}{T}\sum_{t=1}^{T}\|\bm{\widehat{\Delta}}_{(1)}\bm{x}_t\|_{2}^2 - S_1(\bm{\widehat{\Delta}})
				&\leq \frac{2}{T}\sum_{t=1}^{T}\langle \bm{\varepsilon}_t, \bm{\widehat{\Delta}}_{(1)}\bm{x}_t \rangle + S_2(\bm{\widehat{\Delta}}) - S_3(\bm{\widehat{\Delta}}),
			\end{align}
			where $S_1(\cdot)$ is defined as in \eqref{eq:notation_init}.
			It is worth pointing out that  $S_i(\bm{\widehat{\Delta}})$ for $1\leq i\leq 3$ capture the initialization effect of $\bm{y}_s=\bm{0}$ for $s\leq 0$ on the estimation.
			
			Note that  $\widehat{\bm{\Delta}}=\cm{\widehat{A}}-\cm{A}^* \in\bm{\Upsilon}$ and $\kappa_2\geq \kappa_1$. Suppose that the high probability events in Lemmas  \ref{lemma:rsc}--\ref{lemma:init} hold. Then we can derive the estimation error bound from \eqref{eq:thm1eq1}:
			\[
			\kappa_1 \| \bm{\widehat{\Delta}}\|_{\Fr}^2 \lesssim \sqrt{\frac{\kappa_2 \lambda_{\max}(\bm{\Sigma}_{\varepsilon}) d_{\pazocal{R}}}{T}}  \|\bm{\widehat{\Delta}}\|_{\Fr}, \quad \text{or} \quad \| \bm{\widehat{\Delta}}\|_{\Fr} \lesssim  \sqrt{\frac{\kappa_2 \lambda_{\max}(\bm{\Sigma}_{\varepsilon}) d_{\pazocal{R}}}{\kappa_1^2 T}}.
			\]
			Furthermore, applying Lemma \ref{lemma:init} again, we can  derive  the prediction error bound from \eqref{eq:thm1eq} and the above result as follows: 
			\begin{align*}
				\frac{1}{T}\sum_{t=1}^{T} \| \bm{\widehat{\Delta}}_{(1)}\bm{\widetilde{x}}_{t}\|_2^2 
				\lesssim \frac{\kappa_2 \lambda_{\max}(\bm{\Sigma}_{\varepsilon})  d_{\pazocal{R}} }{\kappa_1 T}.
			\end{align*}
			The proof of this theorem is complete.

			
			\subsection{Proof of Lemma \ref{cor1}}
			
			
			\noindent\textbf{Proof of (i):} By definition,  $\ell_{j}^{I}(\lambda_i)=\lambda_i^j$ for $1\leq i\leq r$, and $\ell_{j}^{II,1}(\bm{\eta}_k)=\gamma_k^j\cos(j\theta_k)$ and  $\ell_{j}^{II,2}(\bm{\eta}_k)=\gamma_k^j\sin(j\theta_k)$ for $1\leq k\leq s$. 
			Then the first-order derivatives are $\nabla\ell_{j}^{I}(\lambda_i)=j\lambda_i^{j-1}$, $\nabla_\gamma\ell_{j}^{II,1}(\bm{\eta}_k)=j\gamma_k^{j-1}\cos(j\theta_k)$, 
			$\nabla_\theta\ell_{j}^{II,1}(\bm{\eta}_k)=-j\gamma_k^{j}\sin(j\theta_k)$, $\nabla_\gamma\ell_{j}^{II,2}(\bm{\eta}_k)=j\gamma_k^{j-1}\sin(j\theta_k)$, and
			$\nabla_\theta\ell_{j}^{II,2}(\bm{\eta}_k)=j\gamma_k^{j}\cos(j\theta_k)$. The second-order derivatives are
			$\nabla^2\ell_{j}^{I}(\lambda_i)=j(j-1)\lambda_i^{j-2}$, 
			$\nabla^2_\gamma\ell_{j}^{II,1}(\bm{\eta}_k)=j(j-1)\gamma_k^{j-2}\cos(j\theta_k)$,
			$\nabla^2_{\gamma\theta}\ell_{j}^{II,1}(\bm{\eta}_k)=-j^2\gamma_k^{j-1}\sin(j\theta_k)$, $\nabla^2_\theta\ell_{j}^{II,1}(\bm{\eta}_k)=-j^2\gamma_k^{j}\cos(j\theta_k)$, 
			$\nabla^2_\gamma\ell_{j}^{II,2}(\bm{\eta}_k)=j(j-1)\gamma_k^{j-2}\sin(j\theta_k)$, 
			$\nabla^2_{\gamma\theta}\ell_{j}^{II,2}(\bm{\eta}_k)=j^2\gamma_k^{j-1}\cos(j\theta_k)$, and $\nabla^2_\theta\ell_{j}^{II,2}(\bm{\eta}_k)=-j^2\gamma_k^{j}\sin(j\theta_k)$. By Assumption \ref{assum:statn}(i), there exists $\rho_1>0$ such that $\max\{|\lambda_1|,\ldots,|\lambda_r|, \gamma_1,\ldots,\gamma_s\} \leq \rho_1< \bar{\rho}$. Thus,
			\[\max_{1\leq i\leq r, 1\leq k\leq s, 1\leq h \leq 2}\{|\nabla\ell_{j}^{I}(\lambda_i)|, \|\nabla \ell_{j}^{II,h}(\bm{\eta}_k)\|_2, |\nabla^2\ell_{j}^{I}(\lambda_i)|, \|\nabla^2 \ell_{j}^{II,h}(\bm{\eta}_k)\|_{\Fr}\}\leq C_{L}\bar{\rho}^{j}.
			\]
			by choosing $C_L$ dependent on $\rho_1$ and $\bar{\rho}$ such that $C_L\geq 2j^2(\rho_1/\bar{\rho})^{j-2}\bar{\rho}^{-2}$ for all $j\geq 1$. Note that $C_L$ exists and is an absolute constant.
			
			\bigskip
			\noindent\textbf{Proof of (ii):} 
			By Assumption \ref{assum:statn},
			$\max\{|\lambda_1^*|,\ldots,|\lambda_r^*|, \gamma_1^*,\ldots,\gamma_s^*\} \leq \bar{\rho}$, and  there exists an absolute constant $C_{\cmtt{G}}>0$ such that 
			$\|\bm{G}_k^*\|_{\op}\leq C_{\cmtt{G}}$ for all $1\leq k\leq p$. Then, by a method similar to \eqref{eq:S},
			we can show that
			$\|\bm{A}_j^*\|_{\op}=\|\bm{G}_j^*\|_{\op}\leq C_{\cmtt{G}}$ for $1\leq j\leq p$, and
			\[
			\|\bm{A}_j^*\|_{\op}\leq \sum_{k=p+1}^{d} |\ell_{j,k}(\bm{\omega}^*)| \|\bm{G}_k^*\|_{\op} \leq \bar{\rho}^{j-p}\sum_{k=p+1}^{d}\|\bm{G}_j^*\|_{\op}\leq 
			\bar{\rho}^{j-p} (r+2s) C_{\cmtt{G}} 
			\]
			for  $j\geq p+1$.
			Then, taking $C_* = \bar{\rho}^{-p} \max\{(r+2s) C_{\cmtt{G}}, 1\}$, we accomplish the proof of (ii).
			

			\subsection{Proof of Lemma \ref{lemma:fullrank}} 
			
			
			Let $J=2(r+2s)$. Consider the following partitions of the $\infty \times (p+J)$ matrix $\bm{L}_{\rm{stack}}(\bm{\omega})$:
			\[
			\bm{L}_{\rm{stack}}(\bm{\omega})=\left ( 
			\begin{array}{cc}
				\bm{I}_p& \\
				&\bm{L}_{\rm{stack}}^{\ma}(\bm{\omega})
			\end{array}
			\right )
			=\left (\begin{matrix}
				\bm{I}_p&\\
				&\bm{L}_{[1:J]}(\bm{\omega})\\
				&\bm{L}_{\rm{Rem}}(\bm{\omega})
			\end{matrix}\right ),
			\]
			where $\bm{L}_{\rm{stack}}^{\ma}(\bm{\omega})=\left (\bm{L}^{I}(\bm{\lambda}), \bm{L}^{II}(\bm{\eta}), \nabla \bm{L}^{I}(\bm{\lambda}), \nabla_{\theta}  \bm{L}^{II}(\bm{\eta})\right )$ is further partitioned into two blocks, the $J\times J$ block
			$\bm{L}_{[1:J]}(\bm{\omega})$ and  the $\infty\times J$ remainder block $\bm{L}_{\rm{Rem}}(\bm{\omega})$. Note that for $1\leq j\leq J$, the $j$th row of $\bm{L}_{[1:J]}(\bm{\omega})$ is 
			\[
			\bm{L}_{j}(\bm{\omega}):=\left (\left ( \bm{\ell}_j^{I}(\bm{\lambda}) \right )^\prime,  \left (\bm{\ell}_j^{II}(\bm{\eta})\right )^\prime, \left ( \nabla \bm{\ell}_j^{I}(\bm{\lambda}) \right )^\prime,   \left (\nabla_{\theta}\bm{\ell}_j^{II}(\bm{\eta})\right )^\prime \right ),
			\]
			where $\bm{\ell}_j^{I}(\bm{\lambda}) =(\lambda_1^j,\dots, \lambda_r^j)^\prime$, $\nabla \bm{\ell}_j^{I}(\bm{\lambda}) =(j\lambda_1^{j-1},\dots, j\lambda_r^{j-1})^\prime$, and
			\begin{align*}
				\bm{\ell}_j^{II}(\bm{\eta}) &=\left (\gamma_1^j\cos(j\theta_1),\gamma_1^j\sin(j\theta_1), \dots, \gamma_s^j\cos(j\theta_s),\gamma_s^j\sin(j\theta_s) \right )^\prime,\\
				\nabla_{\theta}\bm{\ell}_j^{II}(\bm{\eta}) &= \left (-j\gamma_1^j\sin(j\theta_1),j\gamma_1^j\cos(j\theta_1), \dots, -j\gamma_s^j\sin(j\theta_s),j\gamma_s^j\cos(j\theta_s) \right )^\prime.
			\end{align*}
			For $j\geq 1$, the $j$th row of $\bm{L}_{\rm{Rem}}(\bm{\omega})$ is $\bm{L}_{J+j}(\bm{\omega})$.
			
			By Lemma \ref{cor1}(i), we have $\|\bm{L}_{\rm{stack}}^{\ma}(\bm{\omega})\|_{\Fr}\leq \sqrt{J\sum_{j=1}^{\infty} C_L^2\bar{\rho}^{2j}}\leq  C_L\sqrt{J}\bar{\rho}(1-\bar{\rho})^{-1}=C_{\bar{\rho}}$. Then
			\begin{align}\label{eq:sigmax}
				\sigma_{\max}(\bm{L}_{\rm{stack}}(\bm{\omega}))\leq \max\left \{1,\sigma_{\max}(\bm{L}_{\rm{stack}}^{\ma}(\bm{\omega}))\right \}\leq  \max\left \{1,\|\bm{L}_{\rm{stack}}^{\ma}(\bm{\omega})\|_{\Fr}\right \}
				\leq \max \{1, C_{\bar{\rho}} \}
			\end{align}
			and
			\begin{align}\label{eq:maxsigmasub}
				\sigma_{\max}(\bm{L}_{[1:J]}(\bm{\omega}))\leq \|\bm{L}_{[1:J]}(\bm{\omega})\|_{\Fr}
				\leq\|\bm{L}_{\rm{stack}}^{\ma}(\bm{\omega})\|_{\Fr}	\leq  C_{\bar{\rho}}.
			\end{align}
			
			It remains to derive a lower bound of $\sigma_{\min}(\bm{L}_{\rm{stack}}(\bm{\omega}))$. To this end, we first derive a lower bound of $\sigma_{\min}(\bm{L}_{[1:J]}(\bm{\omega}))$ by lower bounding the determinant of $\bm{L}_{[1:J]}(\bm{\omega})$.
			For any $(\gamma, \theta)\in[0,1)\times(-\pi/2, \pi/2)$, it can be verified that 
			\[
			\left (\gamma^j\cos(j\theta),\gamma^j\sin(j\theta) \right ) \underbrace{\left(\begin{matrix}1&1\\i&-i\end{matrix}\right)}_{:=\bm{C}_1} = \left ((\gamma e^{i\theta})^j,(\gamma e^{-i\theta})^j\right )
			\]
			and 
			\[
			\left (-j\gamma^j\sin(j\theta),j\gamma^j\cos(j\theta) \right ) \underbrace{\left(\begin{matrix}-i&i\\1&1\end{matrix}\right)}_{:=\bm{C}_2} =\left (j(\gamma e^{i\theta})^j,j(\gamma e^{-i\theta})^j\right ).
			\]
			
			Let $\bm{P}_1 = \diag( \bm{I}_r, \bm{C}_1, \dots, \bm{C}_1, \bm{I}_r, \bm{C}_2, \dots, \bm{C}_2)$ be a $J\times J$ block diagonal matrix consisting of two identity matrices $\bm{I}_r$ and $s$ repeated blocks of $\bm{C}_1$ and $\bm{C}_2$.
			We then have $\det(\bm{P}_1)=(-2i)^{2s}=4^s$, and
			\begin{align*}
				\bm{L}_{[1:J]}(\bm{\omega})\bm{P}_1 = \left(\begin{matrix}
					x_1&x_2&\cdots&x_{r+2s}&x_1&x_2&\cdots&x_{r+2s}\\
					x_1^2&x_2^2&\cdots&x_{r+2s}^2&2x_1^2&2x_2^2&\cdots&2x_{r+2s}^2\\
					\vdots&\vdots&\ddots&\vdots&\vdots&\vdots&\ddots&\vdots\\
					x_1^{J}&x_2^{J}&\cdots&x_{r+2s}^{J}&Jx_1^{J}&Jx_2^{J}&\cdots&Jx_{r+2s}^{J}
				\end{matrix}\right) := \bm{P}_2 \in \mathbb{R}^{J \times J},
			\end{align*}
			where $x_i = \lambda_i$ for $1\leq i\leq r$, while $x_{r+2k-1} = \gamma_k e^{i\theta_k}$ and $x_{r+2k} = \gamma_ke^{-i\theta_k}$ for $ 1\leq k\leq s$, and $i$ is the imaginary unit.
			
			We subtract the $j$th column of $\bm{P}_2$ from its  $(r+2s+j)$th column, for all $1\leq j \leq r+2s$, and obtain a matrix with the same determinant as $\bm{P}_2$ as follows,
			\begin{align*}
				\bm{P}_3=\left(\begin{matrix}
					x_1&x_2&\cdots&x_{r+2s}&0&0&\cdots&0\\
					x_1^2&x_2^2&\cdots&x_{r+2s}^2&x_1^2&x_2^2&\cdots&x_{r+2s}^2\\
					\vdots&\vdots&\ddots&\vdots&\vdots&\vdots&\ddots&\vdots\\
					x_1^{J}&x_2^{J}&\cdots&x_{r+2s}^{J}&(J-1)x_1^{J}&(J-1)x_2^{J}&\cdots&(J-1)x_{r+2s}^{J}
				\end{matrix}\right).
			\end{align*}
			Note that $\bm{P}_3=\bm{P}_4\bm{P}_5$, where 
			\[
			\bm{P}_4=\left(\begin{matrix}
				1&1&\cdots&1&0&0&\cdots&0\\
				x_1&x_2&\cdots&x_{r+2s}&x_1&x_2&\cdots&x_{r+2s}\\
				x_1^2&x_2^2&\cdots&x_{r+2s}^2&2x_1^2&2x_2^2&\cdots&2x_{r+2s}^2\\
				\vdots&\vdots&\ddots&\vdots&\vdots&\vdots&\ddots&\vdots\\
				x_1^{J-1}&x_2^{J-1}&\cdots&x_{r+2s}^{J-1}&(J-1)x_1^{J-1}&(J-1)x_2^{J-1}&\cdots&(J-1)x_{r+2s}^{J-1}
			\end{matrix}\right)
			\] 
			is a generalized Vandermonde matrix \citep{li2008special}, and $\bm{P}_5=\diag\{x_1, \dots, x_{r+2s},x_1,\dots,x_{r+2s}\}$. By \cite{li2008special}, 
			$|\det(\bm{P}_4)|=\prod_{i=1}^{r+2s}x_i\prod_{1\leq k<j\leq r+2s}(x_j - x_k)^4$. 
			As a result,
			\begin{equation*}
				|\det(\bm{P}_2)|=|\det(\bm{P}_3)|=|\det(\bm{P}_4)||\det(\bm{P}_5)|= \prod_{j=1}^{r+2s}|x_j|^3\prod_{1\leq j<k\leq r+2s}(x_j - x_k)^4\geq \nu_1^{3J/2}\nu_2^{J(J/2-1)},
			\end{equation*}
			where $\nu_1=\min\{|x_k|,1\leq k\leq r+2s\}$, $\nu_2=\min\{|x_j-x_k|, 1\leq j<k\leq r+2s\}$, and $J=2(r+2s)$. It follows that
			\begin{align}\label{eq:det2}
				|\det( \bm{L}_{[1:J]}(\bm{\omega}) )| =\frac{|\det(\bm{P}_2)|}{|\det(\bm{P}_1)|} \geq 0.25^s \nu_1^{3J/2}\nu_2^{J(J/2-1)}>0,
			\end{align}
			and hence $\bm{L}_{[1:J]}(\bm{\omega})$ is full-rank. Moreover, combining \eqref{eq:maxsigmasub} and \eqref{eq:det2}, we have
			\begin{align}\label{eq:minsigmasub}
				\sigma_{\min}(\bm{L}_{[1:J]}(\bm{\omega})) \geq \frac{	|\det( \bm{L}_{[1:J]}(\bm{\omega}) )| }{\sigma_{\max}^{J-1}(\bm{L}_{[1:J]}(\bm{\omega}))} \geq \frac{0.25^s  \nu_1^{3J/2}\nu_2^{J(J/2-1)}}{C_{\bar{\rho}}^{J-1}}.
			\end{align}
			Note that the right side of \eqref{eq:minsigmasub} when $\bm{\omega}=\bm{\omega}^*$ is $c_{\bar{\rho}}$, and Assumption \ref{assum:statn}(ii) ensures that  $c_{\bar{\rho}}>0$ is an absolute constant.

			Finally, similar to \eqref{eq:sigmax}, by the Courant–Fischer theorem, it can be shown that
			\begin{align*}
				\sigma_{\min}(\bm{L}_{\rm{stack}}(\bm{\omega}))\geq \min\left \{1,\sigma_{\min}(\bm{L}_{\rm{stack}}^{\ma}(\bm{\omega}))\right \}\geq  \min\left \{1, \sigma_{\min}(\bm{L}_{[1:J]}(\bm{\omega}))\right \},
			\end{align*}
			which, together with \eqref{eq:minsigmasub}, leads to a lower bound of 	$\sigma_{\min}(\bm{L}_{\rm{stack}}(\bm{\omega}))$. Applying this and  the upper bound in \eqref{eq:sigmax} to $\bm{\omega}=\bm{\omega}^*$ under Assumption \ref{assum:statn}(ii), we accomplish the proof of this lemma.
			

			\subsection{Proof of Lemma \ref{lemma:rsc}}\label{sec:proof_RSC}
			It suffices  to show that the event stated in this lemma holds uniformly over the intersection of $\bm{\Upsilon}$ and the sphere $\pazocal{S}(\delta):= \{\bm{\Delta} \in \mathbb{R}^{N\times N\times \infty} \mid\|\bm{\Delta}\|_{\Fr}=\delta\}$, for some radius $\delta>0$ to be chosen such that  $\bm{\Upsilon}\cap\pazocal{S}(\delta)\neq \emptyset$, since  the event will remain true if we multiply $\bm{\Delta}$ by any nonzero real number. 
			
			Recall that any $\bm{\Delta}\in\bm{\Upsilon}$ can be written as $\bm{\Delta}=\cm{A}-\cm{A}^*=\cm{G}\times_3\bm{L}(\bm{\omega})-\cm{A}^*:=\bm{\Delta}(\bm{\omega},\cm{G})$,
			for some $\cm{G}\in\mathbb{R}^{N\times N\times d}$ and $\bm{\omega}\in\bm{\Omega}$ dependent on $\bm{\Delta}$ with  $\|\bm{\omega} - \bm{\omega}^*\|_2 \leq 
			c_{\bm{\omega}}$.
			In addition, by  \eqref{eq:stackH} and \eqref{eq:Delta}, we can further write 
			\[
			\bm{\Delta}=\cm{G}_{\rm{stack}}\times_3\bm{L}_{\rm{stack}}(\bm{\omega}^*)+\stk(\bm{0}_{N\times N\times p}, \cm{R}).
			\]  
			
			Note that throughout the  proofs of Lemmas \ref{lemma:rsc}--\ref{lemma:init}, we will suppress the dependence of $\cm{G}_{\rm{stack}}$ and $\cm{R}$ on $\bm{\Delta}$.
			Thus,	
			\begin{equation*}
				\bm{\Delta}_{(1)}\bm{x}_t=\{\cm{G}_{\rm{stack}}\times_3\bm{L}_{\rm{stack}}(\bm{\omega}^*)\}_{(1)}\bm{x}_t+\cm{R}_{(1)}\bm{x}_{t-p}
				=(\cm{G}_{\rm{stack}})_{(1)} \bm{z}_t +\cm{R}_{(1)}\bm{x}_{t-p},
			\end{equation*}
			where 
			\begin{equation}\label{eq:ts-z}
				\bm{x}_t=(\bm{y}_{t-1}^\prime, \bm{y}_{t-2}^\prime, \dots)^\prime \quad\text{and}\quad	\bm{z}_t =\left \{\bm{L}_{\rm{stack}}^\prime(\bm{\omega}^*)\otimes \bm{I}_N \right \}\bm{x}_{t}.
			\end{equation}
			For simplicity, denote $\bm{X}=(\bm{x}_T, \dots, \bm{x}_1)$, $\bm{X}_{-p}=(\bm{x}_{T-p}, \dots, \bm{x}_{1-p})$ and $\bm{Z}=(\bm{z}_T,\dots, \bm{z}_1)$. Then $\bm{\Delta}_{(1)}\bm{X}=(\cm{G}_{\rm{stack}})_{(1)} \bm{Z} +\cm{R}_{(1)}\bm{X}_{-p}$, and hence
			$\|(\cm{G}_{\rm{stack}})_{(1)} \bm{Z}\|_{\Fr}-\|\cm{R}_{(1)}\bm{X}_{-p}\|_{\Fr} \leq \|\bm{\Delta}_{(1)}\bm{X}\|_{\Fr} \leq \|(\cm{G}_{\rm{stack}})_{(1)} \bm{Z}\|_{\Fr}+\|\cm{R}_{(1)}\bm{X}_{-p}\|_{\Fr}$,
			i.e.,
			\begin{equation}\label{eq:split1}
				\left (  \sum_{t=1}^{T}\|\bm{\Delta}_{(1)}\bm{x}_t\|_2^2 \right )^{1/2}\geq
				\left (\sum_{t=1}^{T}\|(\cm{G}_{\rm{stack}})_{(1)}\bm{z}_t\|_2^2 \right )^{1/2}-
				\left (\sum_{t=1}^{T}\|\cm{R}_{(1)}\bm{x}_{t-p}\|_2^2\right )^{1/2}
			\end{equation}
			and 
			\begin{equation}\label{eq:split2}
				\left (  \sum_{t=1}^{T}\|\bm{\Delta}_{(1)}\bm{x}_t\|_2^2 \right )^{1/2}\leq
				\left (\sum_{t=1}^{T}\|(\cm{G}_{\rm{stack}})_{(1)}\bm{z}_t\|_2^2 \right )^{1/2}+
				\left (\sum_{t=1}^{T}\|\cm{R}_{(1)}\bm{x}_{t-p}\|_2^2\right )^{1/2}.
			\end{equation}
			
			Now we restrict our attention to $\bm{\Delta}\in\bm{\Upsilon}\cap\pazocal{S}(\delta)$, where $\delta>0$ will be specified later. If $\bm{\Delta}\in\bm{\Upsilon}\cap\pazocal{S}(\delta)$, since $\|\bm{\Delta} \|_{\Fr}= \delta$, then it follows from Lemma \ref{lemma:delnorm} that
			\begin{equation}\label{eq:rsc-norm}
				\delta C_{\Delta}^{-1} \leq \delta_{\cmtt{G}} \leq \delta c_{\Delta}^{-1} 
				\quad\text{and} \quad
				\frac{\delta C_{\Delta}^{-1}}{\alpha}  \leq \delta_{\bm{\omega}}  \leq \frac{\delta c_{\Delta}^{-1}}{\alpha},
			\end{equation}
			where $\delta_{\bm{\omega}}=\|\bm{\omega} - \bm{\omega}^*\|_2$ and $\delta_{\cmtt{G}}=\|\cm{G}-\cm{G}^*\|_{\Fr}$.
			To guarantee that $\bm{\Upsilon}\cap\pazocal{S}(\delta)\neq \emptyset$, it is sufficient to choose $\delta>0$ such that 
			\begin{equation}\label{eq:condition1}
				\frac{\delta c_{\Delta}^{-1}}{\alpha} \leq  c_{\bm{\omega}}. 
			\end{equation}   
			Furthermore, by   \eqref{eq:Gstacknorm} and \eqref{eq:rsc-norm}, we can obtain the following bounds of $\|\cm{G}_{\rm{stack}}\|_{\Fr}$   by restricting the corresponding $\bm{\Delta}\in\bm{\Upsilon}\cap\pazocal{S}(\delta)$:
			\begin{align}\label{eq:Gnorm}
				\begin{split}
					\sup_{\bm{\Delta}\in\bm{\Upsilon}\cap\pazocal{S}(\delta)}	\|\cm{G}_{\rm{stack}}\|_{\Fr}&\leq	\delta c_{\Delta}^{-1} \left (1+\frac{\sqrt{2}}{\min_{1\leq k\leq s}\gamma_{k}^*}\right ) \asymp \delta,\\
					\inf_{\bm{\Delta}\in\bm{\Upsilon}\cap\pazocal{S}(\delta)}	\|\cm{G}_{\rm{stack}}\|_{\Fr}&\geq
					0.5 \delta C_{\Delta}^{-1} (1+ \sqrt{2}c_{\cmtt{G}}) \asymp \delta.
				\end{split}
			\end{align}               
			
			Next we establish the following union bounds that hold for all $\bm{\Delta}\in\bm{\Upsilon}\cap\pazocal{S}(\delta)$:  
			\begin{itemize}
				\item [(i)] If $T \gtrsim (\kappa_2 / \kappa_1)^2  d_{\pazocal{R}}\log(\kappa_2/\kappa_1)$, then
				\begin{equation*}
					\mathbb{P}\left \{\forall \bm{\Delta}\in \bm{\Upsilon}\cap\pazocal{S}(\delta): \frac{ c_{\cmtt{M}} \delta^2 \kappa_1}{8} \lesssim \frac{1}{T}\sum_{t=1}^{T}\|(\cm{G}_{\rm{stack}})_{(1)}\bm{z}_t\|_2^2 \lesssim 6C_{\cmtt{M}} \delta^2 \kappa_2\right \} \geq 1-2e^{-cd_{\pazocal{R}}\log(\kappa_2/\kappa_1)}. 
				\end{equation*}
				\item [(ii)] if $T\gtrsim N$, then
				\begin{equation*}
					\mathbb{P}\left \{ \sup_{\bm{\Delta}\in\bm{\Upsilon}\cap\pazocal{S}(\delta)}\frac{1}{T}\sum_{t=1}^{T}\|\cm{R}_{(1)}\bm{x}_{t-p}\|_2^2 \lesssim \delta^2  \delta_{\bm{\omega}}^2 \lambda_{\max}(\bm{\Sigma}_\varepsilon) \mu_{\max}(\bm{\Psi}_*)\right \}	\geq 1-3e^{-N\log9}.
				\end{equation*}
			\end{itemize}
			
			\noindent
			\textbf{Proof of (i):} Note that by \eqref{eq:Gnorm}, for every $\bm{\Delta}\in \bm{\Upsilon}\cap\pazocal{S}(\delta)$,  its corresponding $\cm{G}_{\rm{stack}}$ has a bounded Frobenius norm in the order of $\delta$.  This relationship between $\bm{\Delta}$ and $\cm{G}_{\rm{stack}}$ allows us to convert the problem of finding  union bounds over all $\bm{\Delta}\in \bm{\Upsilon}\cap\pazocal{S}(\delta)$ to that over all $\cm{G}_{\rm{stack}}$ with  Frobenius norm in the order of $\delta$.
			
			Moreover, note that for all $\bm{\Delta}\in \bm{\Upsilon}$, it holds $\cm{G}_{\rm{stack}}\in \bm{\Xi}$, where $\bm{\Xi}$ is defined as in \eqref{eq:Xi}, and hence
			\begin{equation}\label{eq:GstackXi}
				\frac{\cm{G}_{\rm{stack}}}{\|\cm{G}_{\rm{stack}}\|_{\Fr}}\in \bm{\Xi}_1=\{\cm{M} \in \bm{\Xi}\mid \|\cm{M}  \|_{\Fr} = 1\}.
			\end{equation}  
			Then, by taking $\cm{M}=\cm{G}_{\rm{stack}}/\|\cm{G}_{\rm{stack}}\|_{\Fr}$, it follows directly from  \eqref{eq:rscG1} in Lemma \ref{lemma:rscG1} that, if $T \gtrsim (\kappa_2 / \kappa_1)^2  d_{\pazocal{R}}\log(\kappa_2/\kappa_1)$, 
			\begin{align*}
				\mathbb{P} \left \{ \forall \bm{\Delta}\in \bm{\Upsilon}:	\frac{c_{\cmtt{M}}\kappa_1}{8}  \|\cm{G}_{\rm{stack}}\|_{\Fr}^2\leq \frac{1}{T}\sum_{t=1}^{T}\|(\cm{G}_{\rm{stack}})_{(1)}\bm{z}_t\|_2^2 \leq 6C_{\cmtt{M}}\kappa_2 \|\cm{G}_{\rm{stack}}\|_{\Fr}^2 \right \} \geq 1-2e^{-cd_{\pazocal{R}}\log(\kappa_2/\kappa_1)}.
			\end{align*}
			Thus, by restricting our attention to $\bm{\Delta}\in \bm{\Upsilon}\cap\pazocal{S}(\delta)$ and combining \eqref{eq:Gnorm} with the above result, we can obtain (i).
			
			\bigskip\noindent
			\textbf{Proof of (ii):}  First note that $\cm{R}_{(1)}\bm{x}_{t-p}=\sum_{j=1}^{\infty}\bm{R}_j\bm{y}_{t-p-j}$, where  $\cm{R}=\stk(\bm{R}_1, \bm{R}_2, \dots)$.
			For all $\bm{\Delta}\in \bm{\Upsilon}\cap\pazocal{S}(\delta)$, it follows from \eqref{eq:rsc-norm} and the choice of $\delta$ in \eqref{eq:condition1} that
			\begin{equation}\label{eq:rsc-norm1}
				\delta_{\cmtt{G}} \leq \delta c_{\Delta}^{-1}, \quad
				\alpha \delta_{\bm{\omega}}  \leq \delta c_{\Delta}^{-1},  
				\quad\text{and} \quad \delta_{\bm{\omega}} \leq c_{\bm{\omega}}.
			\end{equation}
			Combining the above bounds with \eqref{eq:Rnorm1}, we have
			\begin{equation} \label{eq:Rnorm}
				\|\bm{R}_j\|_{\Fr} \leq \|\bm{R}_{1j}\|_{\Fr} + \|\bm{R}_{2j}\|_{\Fr} + \|\bm{R}_{3j}\|_{\Fr} \leq   \delta  C_{\cmtt{R}}  \bar{\rho}^j \delta_{\bm{\omega}},
			\end{equation}
			where
			$C_{\cmtt{R}} = C_L c_{\Delta}^{-1} (\sqrt{2}+\sqrt{2}c_{\bm{\omega}}/2 + 1)\asymp1$. Then
			\begin{align}\label{eq:Rx1}
				\begin{split}
					\frac{1}{T}\sum_{t=1}^{T}\|\cm{R}_{(1)}\bm{x}_{t-p}\|_2^2&=\frac{1}{T}\sum_{t=1}^{T}\sum_{i=1}^{\infty}\sum_{j=1}^{\infty}\bm{y}_{t-p-i}^\prime \bm{R}_i^\prime\bm{R}_j\bm{y}_{t-p-j}\\
					&=\frac{1}{T}\sum_{i=1}^{\infty}\sum_{j=1}^{\infty}\left( \sum_{t=1}^{T}\langle \bm{R}_i\bm{y}_{t-p-i},\bm{R}_j\bm{y}_{t-p-j}\rangle \right )\\
					&\leq \sum_{i=1}^{\infty}\sum_{j=1}^{\infty}\left (\frac{1}{T}\sum_{t=1}^{T}\|\bm{R}_j\bm{y}_{t-p-i}\|_2^2\right )^{1/2} \left (\frac{1}{T}\sum_{t=1}^{T}\|\bm{R}_j\bm{y}_{t-p-j}\|_2^2\right )^{1/2}\\
					&=\left \{\sum_{j=1}^{\infty} \left (\frac{1}{T}\sum_{t=1}^{T}\|\bm{R}_j\bm{y}_{t-p-j}\|_2^2\right )^{1/2}\right \}^2.
				\end{split}
			\end{align}
			In addition, we can show that
			\[
			\frac{1}{T}\sum_{t=1}^{T}\|\bm{R}_j\bm{y}_{t-p-j}\|_2^2   
			=\trace \left\{\bm{R}_j\left( \frac{1}{T} \sum_{t=1}^{T}\bm{y}_{t-p-j}\bm{y}_{t-p-j}^\prime \right) \bm{R}_j^\prime\right\} \leq \|\bm{R}_j\|_{\Fr}^2 \left\| \frac{1}{T} \sum_{t=1}^{T}\bm{y}_{t-p-j}\bm{y}_{t-p-j}^\prime \right\|_{\op}.
			\]
			As a result,
			\begin{equation}\label{eq:Rx}
				\frac{1}{T}\sum_{t=1}^{T}\|\cm{R}_{(1)}\bm{x}_{t-p}\|_2^2 \leq \left (\sum_{j=1}^{\infty} \|\bm{R}_j\|_{\Fr} \left\| \frac{1}{T} \sum_{t=1}^{T}\bm{y}_{t-p-j}\bm{y}_{t-p-j}^\prime \right\|_{\op}^{1/2}\right )^2.
			\end{equation}
			Then, by \eqref{eq:Rnorm}, \eqref{eq:Rx} and \eqref{eq:Rjy1} in Lemma \ref{lemma:Rjy}, if $T \gtrsim N$, with probability at least $1-3e^{-N\log9}$, we have 
			\begin{align*}
				\sup_{\bm{\Delta}\in\bm{\Upsilon}\cap\pazocal{S}(\delta)} \frac{1}{T}\sum_{t=1}^{T}\|\cm{R}_{(1)}\bm{x}_{t-p}\|_2^2 
				\lesssim
				\delta^2 \delta_{\bm{\omega}}^2  \lambda_{\max}(\bm{\Sigma}_\varepsilon)\mu_{\max}(\bm{\Psi}_*)
				\left (\sum_{j=1}^{\infty}  \bar{\rho}^j  \sqrt{j \sigma^2+1} \right )^2
				\lesssim \delta^2  \delta_{\bm{\omega}}^2 \lambda_{\max}(\bm{\Sigma}_\varepsilon)\mu_{\max}(\bm{\Psi}_*),
			\end{align*}
			where the second inequality follows from the fact that $\sum_{j=1}^{\infty}  \bar{\rho}^j  \sqrt{j \sigma^2+1}\asymp1$. Thus (ii) is verified.
			
			Finally, it remains to combine \eqref{eq:split1} and \eqref{eq:split2} with the results in (i) and (ii).
			To ensure that the lower bound of $T^{-1}\sum_{t=1}^{T}\|(\cm{G}_{\rm{stack}})_{(1)}\bm{z}_t\|_2^2$ in (i) dominates the upper bound of $T^{-1}\sum_{t=1}^{T}\|\cm{R}_{(1)}\bm{x}_{t-p}\|_2^2$  in (ii), we only need $\delta_{\bm{\omega}} \lesssim \sqrt{c_{\cmtt{M}}  \kappa_1 /\{\lambda_{\max}(\bm{\Sigma}_\varepsilon)\mu_{\max}(\bm{\Psi}_*)\}}$. In view of \eqref{eq:rsc-norm}, this can be guaranteed by  choosing $\delta>0$ such that
			\begin{equation}\label{eq:condition2}
				\frac{\delta c_{\Delta}^{-1}}{\alpha} \lesssim  \sqrt{\frac{c_{\cmtt{M}}  \kappa_1}{\lambda_{\max}(\bm{\Sigma}_\varepsilon) \mu_{\max}(\bm{\Psi}_*)}}.
			\end{equation}
			Combining the above condition with \eqref{eq:condition1}, we can obtain the desirable $\delta$. Then we have
			\begin{equation*}
				\mathbb{P}\left \{\forall \bm{\Delta}\in \bm{\Upsilon}\cap\pazocal{S}(\delta): c_{\cmtt{M}} \delta^2 \kappa_1 \lesssim \frac{1}{T}\sum_{t=1}^{T}\|\bm{\Delta}_{(1)}\bm{x}_t\|_2^2   \lesssim C_{\cmtt{M}} \delta^2 \kappa_2\right \} \geq 1-2e^{-cd_{\pazocal{R}}\log(\kappa_2/\kappa_1)}-2e^{-N\log9}.
			\end{equation*}
			Since $c_{\cmtt{M}}$ and $C_{\cmtt{M}}$  are absolute constants, and  $\|\bm{\Delta} \|_{\Fr}= \delta$ for all  $\bm{\Delta}\in\bm{\Upsilon}\cap\pazocal{S}(\delta)$, the proof of this lemma is complete.
			
			\subsection{Proof of Lemma \ref{lemma:dev}}
			Similar to  Lemma \ref{lemma:rsc}, it suffices  to prove that the result of Lemma \ref{lemma:dev} holds uniformly over the intersection of $\bm{\Upsilon}$ and the sphere $\pazocal{S}(\delta):= \{\bm{\Delta} \in \mathbb{R}^{N\times N\times \infty} \mid\|\bm{\Delta}\|_{\Fr}=\delta\}$, 
			where the radius $\delta>0$ is chosen to satisfy  condition \eqref{eq:condition1} to ensure that  $\bm{\Upsilon}\cap\pazocal{S}(\delta)\neq \emptyset$. Specifically, we will prove that 
			\begin{equation}\label{eq:devbd}
				\mathbb{P}\left \{
				\sup_{\bm{\Delta}\in\bm{\Upsilon}\cap\pazocal{S}(\delta)} \frac{1}{T}\left |\sum_{t=1}^{T}\langle \bm{\Delta}_{(1)}\bm{x}_t, \bm{\varepsilon}_t \rangle \right | \lesssim  \delta  \sqrt{\frac{\kappa_2   \lambda_{\max}(\bm{\Sigma}_{\varepsilon})d_{\pazocal{R}}}{T}}  \right \} \geq 1-2e^{-cd_{\pazocal{R}}\log(\kappa_2/\kappa_1)}-5e^{-cN}.
			\end{equation} 
			
			From the proof of Lemma \ref{lemma:rsc},
			for every $\bm{\Delta}=\bm{\Delta}(\bm{\omega},\cm{G})\in\bm{\Upsilon}\cap\pazocal{S}(\delta)$, the corresponding $\cm{G}_{\rm{stack}}$ and $\cm{R}=\stk(\bm{R}_1, \bm{R}_2, \dots)$ satisfy 
			\begin{equation} \label{eq:GRnorm}
				\|\cm{G}_{\rm{stack}}\|_{\Fr} \lesssim\delta,\hspace{5mm} \frac{\cm{G}_{\rm{stack}}}{\|\cm{G}_{\rm{stack}}\|_{\Fr}}\in \bm{\Xi}_1,
				\hspace{5mm}\text{and}\hspace{5mm}
				\|\bm{R}_j\|_{\Fr}  \leq   \delta  c_{\bm{\omega}}  C_{\cmtt{R}}  \bar{\rho}^j \quad\text{for all } j\geq1; 
			\end{equation}
			see  \eqref{eq:condition1}--\eqref{eq:Rnorm}. 
			Moreover, by the definition of $\bm{R}_j$ in \eqref{eq:Rjs}, it can be verified that, for all $j \geq 1$, 
			\begin{align*}
				\colsp(\bm{R}_j) \subseteq \colsp(\{\cm{G}_{(1)}, \cm{G}_{(1)}^*\}) \hspace{5mm}\text{and}\hspace{5mm} \colsp(\bm{R}_j^\prime) \subseteq \colsp(\{\cm{G}_{(2)}, \cm{G}_{(2)}^{*}\}),
			\end{align*}
			where $\colsp(\bm{U})$ represents the column space of a matrix $\bm{U}$, and	$\colsp(\{\bm{U}, \bm{V}\})$ represents the  span  of all column vectors of the matrices $\bm{U}$ and $\bm{V}$. Since for any  $\bm{\Delta}=\bm{\Delta}(\bm{\omega},\cm{G})\in\bm{\Upsilon}$ it holds
			$\textrm{rank}(\cm{G}_{(i)})\leq \pazocal{R}_i$ and  $\textrm{rank}(\cm{G}_{(i)}^*)\leq \pazocal{R}_i$ for $i=1,2$, we then have 
			\begin{equation}\label{eq:Rjlr}
				\textrm{rank}(\bm{R}_j)\leq 2 (\pazocal{R}_1 \wedge \pazocal{R}_2):=2\pazocal{R}_{\wedge}, \quad\text{for all } j\geq1.
			\end{equation}
			Note that
			$\bm{\Delta}_{(1)}\bm{x}_t=(\cm{G}_{\rm{stack}})_{(1)} \bm{z}_t +\cm{R}_{(1)}\bm{x}_{t-p}$ and $\cm{R}_{(1)}\bm{x}_{t-p}=\sum_{j=1}^{\infty}\bm{R}_j\bm{y}_{t-p-j}$. As a result, by \eqref{eq:GRnorm} and \eqref{eq:Rjlr}, we have
			\begin{align*}
				\sup_{\bm{\Delta}\in\bm{\Upsilon}\cap\pazocal{S}(\delta)} \frac{1}{T}\left |\sum_{t=1}^{T}\langle \bm{\Delta}_{(1)}\bm{x}_t, \bm{\varepsilon}_t \rangle \right | 
				&\leq \sup_{\bm{\Delta}\in\bm{\Upsilon}\cap\pazocal{S}(\delta)} \frac{1}{T}\left |\sum_{t=1}^{T}\langle (\cm{G}_{\rm{stack}})_{(1)} \bm{z}_t, \bm{\varepsilon}_t \rangle \right | + \sum_{j=1}^{\infty} \sup_{\bm{\Delta}\in\bm{\Upsilon}\cap\pazocal{S}(\delta)} \frac{1}{T}\left |\sum_{t=1}^{T}\langle \bm{R}_j\bm{y}_{t-p-j}, \bm{\varepsilon}_t \rangle \right |\\
				&\lesssim \delta \sup_{\cmtt{M} \in \bm{\Xi}_1 } \frac{1}{T}\sum_{t=1}^{T}\langle \cm{M}_{(1)} \bm{z}_t, \bm{\varepsilon}_t \rangle + \delta  \sum_{j=1}^{\infty}  \bar{\rho}^j  \sup_{\bm{M} \in \bm{\Pi}(2\pazocal{R}_{\wedge})}\frac{1}{T}\sum_{t=1}^{T}\langle \bm{M} \bm{y}_{t-p-j}, \bm{\varepsilon}_t \rangle.
			\end{align*}
			Applying \eqref{eq:devG} and \eqref{eq:Rjy2} in Lemmas \ref{lemma:rscG1} and \ref{lemma:Rjy}, respectively, we can show that
			\begin{align*}
				\sup_{\bm{\Delta}\in\bm{\Upsilon}\cap\pazocal{S}(\delta)} \frac{1}{T}\left |\sum_{t=1}^{T}\langle \bm{\Delta}_{(1)}\bm{x}_t, \bm{\varepsilon}_t \rangle \right | 
				& \lesssim \delta  \sqrt{\frac{\kappa_2   \lambda_{\max}(\bm{\Sigma}_{\varepsilon})d_{\pazocal{R}}}{T}} + \delta   \lambda_{\max}(\bm{\Sigma}_\varepsilon)\sqrt{\frac{2 \mu_{\max}(\bm{\Psi}_*) N \pazocal{R}_{\wedge}}{T}} \sum_{j=1}^{\infty}  \bar{\rho}^j (2j \sigma^2+1)\\
				& \lesssim \delta  \sqrt{\frac{\kappa_2   \lambda_{\max}(\bm{\Sigma}_{\varepsilon})d_{\pazocal{R}}}{T}} 
			\end{align*}
			with probability at least $1- e^{-c d_{\pazocal{R}}}- 2e^{-cd_{\pazocal{R}}\log(\kappa_2/\kappa_1)}-4e^{-N\log9}$, where the second inequality follows from the fact that $d_{\pazocal{R}}\gtrsim N \pazocal{R}_{\wedge}$, $\kappa_2\gtrsim \lambda_{\max}(\bm{\Sigma}_\varepsilon)\mu_{\max}(\bm{\Psi}_*)$ and $\sum_{j=1}^{\infty}  \bar{\rho}^j (2j \sigma^2+1)\asymp1$. Since $ e^{-c d_{\pazocal{R}}} \leq e^{-cN}$, \eqref{eq:devbd} holds and the proof of this lemma is complete.

			\subsection{Proof of Lemma \ref{lemma:init}}
			Similar to  Lemmas \ref{lemma:rsc} and \ref{lemma:rsc}, it suffices  to prove that the result of Lemma \ref{lemma:init} holds uniformly over the intersection of $\bm{\Upsilon}$ and the sphere $\pazocal{S}(\delta):= \{\bm{\Delta} \in \mathbb{R}^{N\times N\times \infty} \mid\|\bm{\Delta}\|_{\Fr}=\delta\}$, 
			where the radius $\delta>0$ is chosen to satisfy  condition \eqref{eq:condition1} to ensure that  $\bm{\Upsilon}\cap\pazocal{S}(\delta)\neq \emptyset$.	
			
			Recall from the proof of Lemma \ref{lemma:rsc} that for every $\bm{\Delta}=\bm{\Delta}(\bm{\omega},\cm{G})\in\bm{\Upsilon}\cap\pazocal{S}(\delta)$,
			\begin{equation*}
				\delta_{\cmtt{G}} \leq \delta c_{\Delta}^{-1}, \quad	
				\alpha \delta_{\bm{\omega}}  \leq \delta c_{\Delta}^{-1}, \quad\text{and} \quad \|\bm{R}_j\|_{\Fr} \leq   \delta c_{\bm{\omega}} C_{\cmtt{R}}  \bar{\rho}^j,
			\end{equation*}
			where $\delta_{\bm{\omega}}=\|\bm{\omega} - \bm{\omega}^*\|_2$ and $\delta_{\cmtt{G}}=\|\cm{G}-\cm{G}^*\|_{\Fr}$; see \eqref{eq:rsc-norm1} and \eqref{eq:Rnorm}. Combining this with  \eqref{eq:delta}, \eqref{eq:Hj},  Lemma \ref{cor1}(i) and Assumption \ref{assum:statn}(iii), we can show that
			\[
			\|\bm{H}_j\|_{\Fr} \leq \|\bm{\ell}^{(j)}(\bm{\omega}^*)\|_2 \delta_{\cmtt{G}} + \|\nabla\bm{\ell}^{(j)}(\bm{\omega}^*)\|_2 \alpha \delta_{\bm{\omega}}\leq   \delta c_{\Delta}^{-1}(1+C_L)\sqrt{r+2s} \bar{\rho}^j,
			\]
			where $\bm{\ell}^{(j)}(\bm{\omega}^*)$ represents the $j$th row of the matrix $\bm{L}^{\ma}(\bm{\omega})=\left ( \bm{\ell}^{I}(\lambda_1), \dots, \bm{\ell}^{I}(\lambda_r),  \bm{\ell}^{II}(\bm{\eta}_1), \dots, \bm{\ell}^{II}(\bm{\eta}_s)\right )$, and further that
			\[
			\|\bm{\Delta}_{j+p}\|_{\Fr} \leq \|\bm{H}_j\|_{\Fr} + \|\bm{R}_j\|_{\Fr} \leq \delta  \left \{c_{\Delta}^{-1}(1+C_L)\sqrt{r+2s}+c_{\bm{\omega}} C_{\cmtt{R}} \right \} \bar{\rho}^{j}, \quad j\geq 1.
			\]
			Note that $\|\bm{\Delta}_j\|_{\Fr} = \|\bm{G}_j - \bm{G}_j^*\|_{\Fr} \leq \delta c_{\Delta}^{-1}$ for $1\leq j\leq p$. Then, we have 
			\begin{align}\label{eq:delta-decay}
				\|\bm{\Delta}_j\|_{\op} \leq \|\bm{\Delta}_j\|_{\Fr} \leq \delta  C_{1} \bar{\rho}^j, \quad \text{for all }j \geq 1, 
			\end{align}
			where  $C_1=\left \{c_{\Delta}^{-1}(1+C_L)\sqrt{r+2s}+c_{\bm{\omega}} C_{\cmtt{R}} \right \} \bar{\rho}^{-p} \asymp 1$.
			
			In addition, by Lemma \ref{lemma:Wcov}, 
			\begin{align*}
				\mathbb{E}(\|\bm{y}_t\|_2^2)=\trace(\bm{\Sigma}_y) \leq N\lambda_{\max}(\bm{\Sigma}_y)\leq N \lambda_{\max}(\bm{\Sigma}_\varepsilon)\mu_{\max}(\bm{\Psi}_*).
			\end{align*}
			Combining this with Lemma \ref{cor1}(ii) and \eqref{eq:delta-decay}, for all $j\geq 1$, we have 
			\begin{equation}\label{eq:init1}
				\mathbb{E}(\|\bm{A}_j^*\bm{y}_{t-j}\|_2) \leq \left \{\mathbb{E}(\|\bm{A}_j^*\bm{y}_{t-j}\|_2^2) \right \}^{1/2}\leq C_* \bar{\rho}^{j} \sqrt{\lambda_{\max} (\bm{\Sigma}_\varepsilon) \mu_{\max}(\bm{\Psi}_*) N} 
			\end{equation}
			and
			\begin{align}\label{eq:init2}
				\mathbb{E} \left (\sup_{\small{\bm{\Delta}\in \bm{\Upsilon}\cap\pazocal{S}(\delta)}} \|\bm{\Delta}_j\bm{y}_{t-j}\|_2 \right ) \leq \left \{\mathbb{E} \left (\sup_{\small{\bm{\Delta}\in \bm{\Upsilon}\cap\pazocal{S}(\delta)}} \|\bm{\Delta}_j\bm{y}_{t-j}\|_2^2 \right )\right \}^{1/2} \leq \delta  C_{1}\bar{\rho}^{j} \sqrt{\lambda_{\max} (\bm{\Sigma}_\varepsilon) \mu_{\max}(\bm{\Psi}_*) N}.
			\end{align}
			
			By the Cauchy-Schwarz inequality and \eqref{eq:init2},
			\[
			\mathbb{E}\left \{\sup_{\small{\bm{\Delta}\in \bm{\Upsilon}\cap\pazocal{S}(\delta)}} |S_1(\bm{\Delta})|\right \} \leq \frac{2}{T}\sum_{t=1}^{T}\sum_{j=1}^{\infty}\sum_{k=t}^{\infty} \delta ^2  C_{1}^2 \bar{\rho}^{j+k} \lambda_{\max} (\bm{\Sigma}_\varepsilon) \mu_{\max}(\bm{\Psi}_*)N \leq \frac{\delta ^2C_{2}\kappa_2 N}{T},
			\]
			where $C_{2} = 2C_{1}^2\bar{\rho}^2/(1-\bar{\rho})^3 \asymp1$. Similarly, by \eqref{eq:init1} and \eqref{eq:init2}, 
			\[
			\mathbb{E} \left \{\sup_{\small{\bm{\Delta}\in \bm{\Upsilon}\cap\pazocal{S}(\delta)}} |S_2(\bm{\Delta})| \right \} \leq \frac{2}{T}\sum_{t=1}^{T} \sum_{j=t}^{\infty} \sum_{k=1}^{t-1} \delta C_* C_1 \bar{\rho}^{j+k} \lambda_{\max} (\bm{\Sigma}_\varepsilon) \mu_{\max}(\bm{\Psi}_*) N 
			\leq\frac{ \delta C_{3}\kappa_2 N}{T},
			\]
			where $C_3= 2 C_*C_{1}\bar{\rho}^2/(1-\bar{\rho})^3 \asymp1$. Moreover, note that $\mathbb{E}(\|\bm{\varepsilon}_t\|_2) \leq \sqrt{\mathbb{E}(\|\bm{\varepsilon}_t\|_2^2)} \leq \sqrt{\lambda_{\max}(\bm{\Sigma}_\varepsilon) N}$. Then by \eqref{eq:init2} and a method similar to the above, 
			\[
			\mathbb{E}\left \{ \sup_{\small{\bm{\Delta}\in \bm{\Upsilon}\cap\pazocal{S}(\delta)}} |S_3(\bm{\Delta})| \right \} \leq \frac{2}{T}\sum_{t=1}^{T} \sum_{j=t}^{\infty} \delta   C_{1} \bar{\rho}^{j} \lambda_{\max} (\bm{\Sigma}_\varepsilon) \sqrt{\mu_{\max}(\bm{\Psi}_*)} N \leq \frac{\delta C_{4}  \sqrt{\kappa_2 \lambda_{\max} (\bm{\Sigma}_\varepsilon)} N}{T},
			\]
			where  $C_{4} = 2C_{1}\bar{\rho}/(1-\bar{\rho})^2 \asymp1$.
			By Markov's inequality, we can show that 
			\begin{equation}\label{eq:S1Delta}
				\mathbb{P}\left \{\sup_{\small{\bm{\Delta}\in \bm{\Upsilon}\cap\pazocal{S}(\delta)}} |S_1(\bm{\Delta})| \geq \delta^2 C_2 \kappa_1 \right \} \leq \frac{\mathbb{E}\{\sup_{\small{\bm{\Delta}\in \bm{\Upsilon}\cap\pazocal{S}(\delta)}} |S_1(\bm{\Delta})|\}}{\delta^2 C_2 \kappa_1 } \leq \frac{\kappa_2 N}{\kappa_1 T}\leq  \sqrt{\frac{N}{ (\pazocal{R}_1+ \pazocal{R}_2) T}},
			\end{equation}
			\[
			\mathbb{P}\left \{\sup_{\small{\bm{\Delta}\in \bm{\Upsilon}\cap\pazocal{S}(\delta)}} |S_2(\bm{\Delta})| \geq \delta C_3 \sqrt{\frac{\kappa_2 \lambda_{\max}(\bm{\Sigma}_{\varepsilon}) d_{\pazocal{R}}}{T}}  \right \} \leq  \sqrt{\frac{\kappa_2 N^2}{\lambda_{\max}(\bm{\Sigma}_{\varepsilon}) T d_{\pazocal{R}}}} \leq \sqrt{\frac{\kappa_2 N}{\lambda_{\max}(\bm{\Sigma}_{\varepsilon})  (\pazocal{R}_1+ \pazocal{R}_2) T}},
			\]	
			and 
			\[
			\mathbb{P}\left \{\sup_{\small{\bm{\Delta}\in \bm{\Upsilon}\cap\pazocal{S}(\delta)}} |S_3(\bm{\Delta})| \geq \delta C_4 \sqrt{\frac{\kappa_2 \lambda_{\max}(\bm{\Sigma}_{\varepsilon}) d_{\pazocal{R}}}{T}}  \right \} \leq  \sqrt{\frac{ N^2}{T d_{\pazocal{R}}}} \leq \sqrt{\frac{N}{ (\pazocal{R}_1+ \pazocal{R}_2) T}},
			\]
			where the last inequality in \eqref{eq:S1Delta} uses the condition that  $T \gtrsim (\kappa_2 / \kappa_1)  d_{\pazocal{R}}$, and $d_{\pazocal{R}}=
			\pazocal{R}_1\pazocal{R}_2 d  + (\pazocal{R}_1+ \pazocal{R}_2)N$. 
			Then the sum of the above three tail probabilities is $(2+\sqrt{\kappa_2 /\lambda_{\max}(\bm{\Sigma}_{\varepsilon})}) \sqrt{N/\{ (\pazocal{R}_1+ \pazocal{R}_2) T\}}$.
			Thus, we have proved the result of this lemma for all $\bm{\Delta}\in \bm{\Upsilon}\cap\pazocal{S}(\delta)$. Replacing $\delta$ by $\|\bm{\Delta}\|_{\Fr}$ in the above inequalities, we accomplish the proof for all $\bm{\Delta}\in \bm{\Upsilon}$.

			
			
			\subsection{Auxiliary lemmas for the proofs of Lemmas \ref{lemma:rsc} and \ref{lemma:dev}}
			
			Lemmas \ref{lemma:rsc} and \ref{lemma:dev}  are established based on the auxiliary results below. In particular, Lemmas \ref{lemma:rscG1} and \ref{lemma:Rjy} are both used directly in the proofs of Lemmas \ref{lemma:rsc} and \ref{lemma:dev}. The  covering and discretization results in Lemma \ref{lemma:epsilon-net} play an essential role in the proof of  Lemma \ref{lemma:rscG1}.  
			The last three lemmas are useful results for the proofs of both Lemmas \ref{lemma:rscG1} and \ref{lemma:Rjy}:
			Lemmas \ref{lemma:martgl} and \ref{lemma:hansonw} give high-probabilitity concentration and Hanson-Wright inequalities for stationary time series, respectively; and Lemma \ref{lemma:Wcov} provides deterministic bounds for covariance matrices of stationary time series. As in \eqref{eq:GstackXi}, let  $\bm{\Xi}_1=\{\cm{M} \in \bm{\Xi}\mid \|\cm{M}  \|_{\Fr} = 1\}$. Note that the following definition is used in Lemma \ref{lemma:epsilon-net}.
			
			\begin{definition}[Generalized $\epsilon$-net of $\bm{\Xi}_1$]\label{def}
				For any $\epsilon>0$, we say that
				$\bar{\bm{\Xi}}(\epsilon)$ is a generalized $\epsilon$-net of $\bm{\Xi}_1$ if $\bar{\bm{\Xi}}(\epsilon)\subset\bm{\Xi}$, and for any $\cm{M}(\bm{a}, \cm{B}) \in \bm{\Xi}_1$, there exists $\cm{M}(\bar{\bm{a}}, \bar{\cm{B}})\in \bar{\bm{\Xi}}(\epsilon)$ such that $\| \cm{M}(\bm{a}, \cm{B})- \cm{M}(\bar{\bm{a}}, \bar{\cm{B}})\|_{\Fr} \leq \epsilon$. However,  $\bar{\bm{\Xi}}(\epsilon)$ is not required to be a subset of $\bm{\Xi}_1$; that is, $\bar{\bm{\Xi}}(\epsilon)$ may not be an $\epsilon$-net of $\bm{\Xi}_1$.
			\end{definition}
			
			
			\begin{lemma}\label{lemma:rscG1} 
				Suppose that Assumptions  \ref{assum:statn} and \ref{assum:error} hold and  $T\gtrsim(\kappa_2/\kappa_1)^2d_{\pazocal{R}}\log(\kappa_2/\kappa_1)$. Let $\bm{z}_t =\left \{\bm{L}_{\rm{stack}}^\prime(\bm{\omega}^*)\otimes \bm{I}_N \right \}\bm{x}_{t}$ be  defined as in  \eqref{eq:ts-z}. Then
				\begin{equation}\label{eq:rscG1}
					\mathbb{P}\left (  \frac{c_{\cmtt{M}}\kappa_1}{8} \leq \inf_{\cmtt{M}\in\bm{\Xi}_1}\frac{1}{T}\sum_{t=1}^{T}\|\cm{M}_{(1)}\bm{z}_t\|_2^2  \leq 
					\sup_{\cmtt{M}\in\bm{\Xi}_1}\frac{1}{T}\sum_{t=1}^{T}\|\cm{M}_{(1)}\bm{z}_t\|_2^2  \leq 6C_{\cmtt{M}}\kappa_2 \right ) \geq 1- 2e^{-cd_{\pazocal{R}}\log(\kappa_2/\kappa_1)}.
				\end{equation}
				and
				\begin{equation}\label{eq:devG}
					\mathbb{P}\left \{ \sup_{\small{ \cmt{M} \in \bm{\Xi}_1 }} \frac{1}{T}\sum_{t=1}^{T}\langle \cm{M}_{(1)} \bm{z}_t, \bm{\varepsilon}_t \rangle  \lesssim  \sqrt{\frac{\kappa_2   \lambda_{\max}(\bm{\Sigma}_{\varepsilon})d_{\pazocal{R}}}{T}} \right \} \geq 1- e^{-c d_{\pazocal{R}}}- 2e^{-cd_{\pazocal{R}}\log(\kappa_2/\kappa_1)}.
				\end{equation}
			\end{lemma}

			
			\begin{lemma}\label{lemma:Rjy}
				Suppose that Assumptions  \ref{assum:statn} and \ref{assum:error} hold  and 	 $T \gtrsim N$. Then
				\begin{equation}\label{eq:Rjy1}
					\mathbb{P}\left \{\forall j\geq 1: \Big\| \frac{1}{T} \sum_{t=1}^{T}\bm{y}_{t-p-j}\bm{y}_{t-p-j}^\prime \Big\|_{\op}  \leq  2\lambda_{\max}(\bm{\Sigma}_\varepsilon) \mu_{\max}(\bm{\Psi}_*) (j \sigma^2+1)\right \} \geq 1-3e^{-N\log9}
				\end{equation}
				and moreover, if $T \gtrsim N\pazocal{R}$,
				\begin{equation}\label{eq:Rjy2}
					\mathbb{P}\left \{\forall j\geq 1: \sup_{\bm{M} \in \bm{\Pi}(\pazocal{R})}\frac{1}{T}\sum_{t=1}^{T}\langle \bm{M} \bm{y}_{t-p-j}, \bm{\varepsilon}_t \rangle \leq 24\lambda_{\max}(\bm{\Sigma}_\varepsilon)(2j \sigma^2+1) \sqrt{\frac{\mu_{\max}(\bm{\Psi}_*) N \pazocal{R}}{T}} \right \} \geq 1-4e^{-N\pazocal{R}\log9}.  
				\end{equation}	
			\end{lemma}
			
			\begin{lemma}[Covering number and discretization for $\bm{\Xi}_1$]\label{lemma:epsilon-net}
				For any $0<\epsilon<2/3$, let $\bar{\bm{\Xi}}(\epsilon)$ be a minimal generalized $\epsilon$-net of $\bm{\Xi}_1$. 
				\begin{itemize}
					\item [(i)] The cardinality of $\bar{\bm{\Xi}}(\epsilon)$ satisfies
					\[
					\log |\bar{\bm{\Xi}}(\epsilon)| \lesssim d_{\pazocal{R}} \log(1/\epsilon),
					\]
					where 	$d_{\pazocal{R}}=
					\pazocal{R}_1\pazocal{R}_2 d  + (\pazocal{R}_1+ \pazocal{R}_2)N$.
					
					\item[(ii)]  There exist absolute constants  $c_{\cmtt{M}}, C_{\cmtt{M}}>0$ independent of $\epsilon$ such that for any $\cm{M}\in \bar{\bm{\Xi}}(\epsilon)$, it holds $c_{\cmtt{M}}\leq \|\cm{M}\|_{\Fr}\leq C_{\cmtt{M}}$.
					
					\item[(iii)] For any $\bm{X} \in \mathbb{R}^{N \times N (d+r+2s)}$ and $\bm{Z}\in\mathbb{R}^{N(d+r+2s)\times T}$, it holds 
					\begin{align}
						\sup_{ \cmt{M}\in \bm{\Xi}_1} \langle \cm{M}_{(1)}, \bm{X}\rangle &\leq (1- 1.5\epsilon)^{-1} \max_{ \cmt{M} \in \bar{\bm{\Xi}}(\epsilon)} \langle \cm{M}_{(1)}, \bm{X}\rangle,\label{eq:disc1}\\
						\sup_{ \cmt{M}\in \bm{\Xi}_1} \| \cm{M}_{(1)} \bm{Z}\|_{\Fr} & \leq (1- 1.5\epsilon)^{-1} \max_{ \cmt{M} \in \bar{\bm{\Xi}}(\epsilon)} \| \cm{M}_{(1)} \bm{Z}\|_{\Fr}. \label{eq:disc2}
					\end{align}	 
				\end{itemize} 
			\end{lemma}	
			
			
			In Lemmas \ref{lemma:martgl}--\ref{lemma:Wcov} below, we adopt notations as follows.  Let $\{\bm{w}_t\}$ be a time series taking values in $\mathbb{R}^M$, 
			where $M$ is an arbitrary positive integer. If $\{\bm{w}_t\}$ is stationary with mean zero, then we denote the  covariance matrix of $\bm{w}_t$ by
			$\bm{\Sigma}_w = \mathbb{E}(\bm{w}_t\bm{w}_t^\prime)$.
			In addition,  let $\underline{\bm{w}}_T=(\bm{w}_T^\prime, \dots, \bm{w}_1^\prime)^\prime$, and denote its covariance matrix by
			\[
			\underline{\bm{\Sigma}}_w=\mathbb{E}(\underline{\bm{w}}_T\underline{\bm{w}}_T^\prime)=\left (\bm{\Sigma}_w(j-i)\right )_{1\leq i,j\leq T},
			\]
			where  $\bm{\Sigma}_w(\ell) = \mathbb{E}(\bm{w}_t\bm{w}_{t-\ell}^\prime)$ is the lag-$\ell$ autocovariance matrix of $\bm{w}_t$ for $\ell\in\mathbb{Z}$, and $\bm{\Sigma}_w(0)=\bm{\Sigma}_w$. 
			For the time series $\{\bm{y}_t\}$, accordingly we define $\bm{\Sigma}_y=\mathbb{E}(\bm{y}_t\bm{y}_t^\prime)$ and  $\underline{\bm{\Sigma}}_y=\mathbb{E}(\underline{\bm{y}}_T\underline{\bm{y}}_T^\prime)=\left (\bm{\Sigma}_y(j-i)\right )_{1\leq i,j\leq T}$, where  $\underline{\bm{y}}_T=(\bm{y}_{T}^\prime, \dots, \bm{y}_{1}^\prime)^\prime$,  $\bm{\Sigma}_y (\ell) = \mathbb{E}(\bm{y}_t\bm{y}_{t-\ell}^\prime)$  is the lag-$\ell$ covariance matrix of $\bm{y}_t$ for $\ell\in \mathbb{Z}$, and $\bm{\Sigma}_y=\bm{\Sigma}_y(0)$.

			\begin{lemma}[Concentration bound for stationary time series]\label{lemma:martgl} 
				Suppose that Assumption \ref{assum:error}  holds for $\{\bm{\varepsilon}_t\}$, and $\{\bm{w}_t\}$ is a zero-mean stationary time series. Assume that  $\bm{w}_t$ is $\mathcal{F}_{t-1}$-measurable, where
				$\mathcal{F}_t=\sigma\{\bm{\varepsilon}_t, \bm{\varepsilon}_{t-1}, \dots\}$ for $t\in\mathbb{Z}$ is a filtration. Then, for any $a,b>0$, we have 
				\begin{equation*}
					\mathbb{P}\left \{ \sum_{t=1}^{T}\langle\bm{w}_t, \bm{\varepsilon}_t \rangle\geq a, \; \sum_{t=1}^{T}\lVert \bm{w}_t \rVert^2 \leq b\right \} \leq  \exp\left \{-\frac{a^2}{2\sigma^2 \lambda_{\max}(\bm{\Sigma}_{\varepsilon})b}\right \}.
				\end{equation*}
			\end{lemma}
			
			\begin{lemma}[Hanson-Wright inequalities for stationary time series]\label{lemma:hansonw}
				Suppose that Assumption \ref{assum:error}  holds for $\{\bm{\varepsilon}_t\}$, and $\{\bm{w}_t\}$ has the vector MA($\infty$) representation
				\[
				\bm{w}_t= \sum_{j=1}^\infty \bm{\Psi}_j^w  \bm{\varepsilon}_{t-j},
				\]
				where 
				$\bm{\Psi}_j^w\in\mathbb{R}^{M\times N}$ for all $j$, and $\sum_{j=1}^{\infty}\|\bm{\Psi}_j^w \|_{\op}<\infty$. Let  $T_0$ be a fixed integer.
				\begin{itemize}
					\item [(i)]  Then, $\{\bm{w}_t\}$  is a zero-mean stationary time series, and
					\[
					\mathbb{P}\left \{\left |\frac{1}{T}\sum_{t=T_0+1}^{T_0+T}\|\bm{w}_t\|_2^2 - \mathbb{E}\left (\|\bm{w}_t\|_2^2 \right )\right | \geq \frac{M}{\sqrt{T}}\sigma^2\lambda_{\max}(\underline{\bm{\Sigma}}_w)\right \} \leq 2e^{-cM}.
					\]
					\item [(ii)] For any $\bm{M}\in\mathbb{R}^{Q\times M}$ with $Q\geq 1$ and any $\delta>0$, it holds
					\begin{equation*}
						\mathbb{P}\left \{\left |\frac{1}{T}\sum_{t=T_0+1}^{T_0+T}\|\bm{M}\bm{w}_t\|_2^2 - \mathbb{E}\left (\|\bm{M}\bm{w}_t\|_2^2\right )\right | \geq \delta \sigma^2 \lambda_{\max}(\underline{\bm{\Sigma}}_w)\|\bm{M}\|_{\Fr}^2\right \} \leq 2e^{-c\min(\delta, \delta^2)T}.
					\end{equation*}
				\end{itemize}
			\end{lemma}
			
			\begin{lemma}[Deterministic bounds for covariance matrices of stationary time series]\label{lemma:Wcov}
				Suppose that Assumptions \ref{assum:statn} and \ref{assum:error} hold, and  the vector MA($\infty$) representation of $\{\bm{y}_t\}$  is $\bm{y}_t =\bm{\Psi}_*(B)\bm{\varepsilon}_{t}$, where $\bm{\Psi}_*(B) = \bm{I}_N+\sum_{j=1}^{\infty}\bm{\Psi}_j^* B^j$.   
				Let
				$\mu_{\min}(\bm{\Psi}_*) = \min_{|z|=1}\lambda_{\min}(\bm{\Psi}_*(z)\bm{\Psi}_*^{\HH}(z))$ and $\mu_{\max}(\bm{\Psi}_*) = \max_{|z|=1}\lambda_{\max}(\bm{\Psi}_*(z)\bm{\Psi}_*^{\HH}(z))$, where $\bm{\Psi}_*^{\HH}(z)$ is the conjugate transpose of $\bm{\Psi}_*(z)$.
				\begin{itemize}
					\item [(i)] It holds
					\begin{equation*}
						\lambda_{\min}(\bm{\Sigma}_\varepsilon)\mu_{\min}(\bm{\Psi}_*) \leq \lambda_{\min}(\underline{\bm{\Sigma}}_y)\leq \lambda_{\max}(\underline{\bm{\Sigma}}_y) \leq  \lambda_{\max}(\bm{\Sigma}_\varepsilon)\mu_{\max}(\bm{\Psi}_*)
					\end{equation*}
					and 
					\[
					\lambda_{\min}(\bm{\Sigma}_\varepsilon)\mu_{\min}(\bm{\Psi}_*) \leq \lambda_{\min}(\bm{\Sigma}_y)\leq \lambda_{\max}(\bm{\Sigma}_y) \leq   \lambda_{\max}(\bm{\Sigma}_\varepsilon)\mu_{\max}(\bm{\Psi}_*).
					\]
					\item [(ii)]  Define  the time series $\{\bm{w}_t\}$ by
					$\bm{w}_t=\bm{W}\bm{x}_t=\sum_{i=1}^\infty \bm{W}_i\bm{y}_{t-i}$,
					where $\bm{x}_t=(\bm{y}_{t-1}^\prime, \bm{y}_{t-2}^\prime, \dots)^\prime$, $\bm{W}=(\bm{W}_1, \bm{W}_2, \dots)\in\mathbb{R}^{M\times \infty}$, and $\bm{W}_i$'s are $M\times N$ blocks such that $\sum_{i=1}^{\infty}\|\bm{W}_i\|_{\op}<\infty$. Then, $\{\bm{w}_t\}$ is a zero-mean stationary time series. Moreover,
					\begin{equation}\label{eq:sigmaw0}
						\lambda_{\min}(\bm{\Sigma}_\varepsilon)\mu_{\min}(\bm{\Psi}_*) \sigma_{\min}^2(\bm{W})   \leq \lambda_{\min}(\bm{\Sigma}_w) \leq  \lambda_{\max}(\bm{\Sigma}_w)  \leq \lambda_{\max}(\bm{\Sigma}_\varepsilon)\mu_{\max}(\bm{\Psi}_*)\sigma_{\max}^2(\bm{W})
					\end{equation}
					and 
					\begin{equation}\label{eq:sigmaw}
						\lambda_{\max}(\underline{\bm{\Sigma}}_w)\leq \lambda_{\max}(\bm{\Sigma}_\varepsilon) \mu_{\max}(\bm{\Psi}_*) \left (\sum_{i=1}^\infty\|\bm{W}_i\|_{\op}\right )^2.
					\end{equation}
				\end{itemize}
			\end{lemma}

			\subsection{Proof of Lemma \ref{lemma:rscG1}}
			
			\noindent\textbf{Proof of \eqref{eq:rscG1}:}
			Denote $\bm{W}=\bm{L}_{\rm{stack}}^\prime(\bm{\omega}^*)\otimes \bm{I}_N$, and let $\bm{L}_{\rm{stack}}^{(j)}(\bm{\omega}^*)$ be the $j$th row of $\bm{L}_{\rm{stack}}(\bm{\omega}^*)$ for $j\geq1$. Then 
			$\bm{z}_t =\bm{W}\bm{x}_{t}=\sum_{j=1}^\infty \bm{W}_j\bm{y}_{t-j}$ and $\bm{W}=(\bm{W}_1, \bm{W}_2, \dots)$,
			where $\bm{x}_t=(\bm{y}_{t-1}^\prime, \bm{y}_{t-2}^\prime, \dots)^\prime$ and  $\bm{W}_j=\bm{L}_{\rm{stack}}^{(j)}(\bm{\omega}^*)\otimes \bm{I}_N$  for $j\geq1$. By the definition of $\bm{L}_{\rm{stack}}(\bm{\omega}^*)$ 
			and Lemma \ref{cor1}(i), we have
			\[
			\|\bm{L}_{\rm{stack}}^{(j)}(\bm{\omega}^*)\|_{2}=1 \quad\text{for } 1\leq j\leq p \quad\text{and}\quad \|\bm{L}_{\rm{stack}}^{(j)}(\bm{\omega}^*)\|_{2}\leq C_L\sqrt{J}\bar{\rho}^j \quad\text{for } j\geq p+1,
			\]
			where $J=2(r+2s)$. Thus,
			\[
			\sum_{j=1}^{\infty}\|\bm{W}_j\|_{\op}=\sum_{j=1}^{\infty}\|\bm{L}_{\rm{stack}}^{(j)}(\bm{\omega}^*)\|_{2}\leq C_{\bar{\rho}},
			\]
			where $C_{\bar{\rho}}=C_L\sqrt{J}\bar{\rho}(1-\bar{\rho})^{-1}$. In addition,  by  Lemma \ref{lemma:fullrank},
			\[ 
			\min\{1, c_{\bar{\rho}}\} \leq  \sigma_{\min, L}=\sigma_{\min}(\bm{W})\leq \sigma_{\max}(\bm{W})=\sigma_{\max, L}\leq \max\{1, C_{\bar{\rho}}\}.
			\]
			Then, by setting $\bm{w}_t=\bm{z}_t$, it follows from Lemma \ref{lemma:Wcov}(ii) that
			\begin{equation}\label{eq:sigmaz0}
				\kappa_1 \leq \lambda_{\min}(\bm{\Sigma}_w) \leq  \lambda_{\max}(\bm{\Sigma}_w)  \leq  \kappa_2
			\end{equation}
			and 
			\begin{equation}\label{eq:sigmaz}
				\lambda_{\max}(\underline{\bm{\Sigma}}_w)\leq \lambda_{\max}(\bm{\Sigma}_\varepsilon) \mu_{\min}(\bm{\Psi}_*) C_{\bar{\rho}}^2	\leq \kappa_2,
			\end{equation}
			where $\kappa_1=	\lambda_{\min}(\bm{\Sigma}_\varepsilon)\mu_{\min}(\bm{\Psi}_*) 	\min\{1, c_{\bar{\rho}}^2\}$ and $ \kappa_2=\lambda_{\max}(\bm{\Sigma}_\varepsilon)\mu_{\max}(\bm{\Psi}_*)\max\{1, C_{\bar{\rho}}^2\}$. 
			
			Furthermore, since $\bm{z}_t=\mathcal{W}(B)\bm{y}_t=\mathcal{W}(B)\bm{\Psi}_*(B)\bm{\varepsilon}_{t}$ is a zero-mean and stationary time series, where $\mathcal{W}(B) = \sum_{i=1}^{\infty}\bm{W}_i B^i$, we can  apply Lemma \ref{lemma:hansonw}(ii) with $T_0=0$, $\bm{w}_t=\bm{z}_t$,  and  $\delta=\kappa_1/(2\sigma^2\kappa_2)$, in conjunction with \eqref{eq:sigmaz0} and \eqref{eq:sigmaz}, 
			to obtain 
			\begin{equation*}
				\mathbb{P}\left \{\left |\frac{1}{T}\sum_{t=1}^{T}\|\cm{M}_{(1)}\bm{z}_t\|_2^2 - \mathbb{E}\left (\|\cm{M}_{(1)}\bm{z}_t\|_2^2\right )\right | \geq 0.5\kappa_1\|\cm{M}\|_{\Fr}^2\right \} \leq 2e^{-c_\sigma(\kappa_1/\kappa_2)^2 T}, 
			\end{equation*}
			where $c_\sigma=c\min\{0.5\sigma^{-2}, 0.25\sigma^{-4}\}$.
			Note that by \eqref{eq:sigmaz0} and $\bm{w}_t=\bm{z}_t$, we can show that
			\[
			\kappa_1 \|\cm{M}\|_{\Fr}^2 \leq \lambda_{\min}(\bm{\Sigma}_w)\|\cm{M}\|_{\Fr}^2\leq
			\mathbb{E}\left (\|\cm{M}_{(1)}\bm{z}_t\|_2^2 \right )
			\leq\lambda_{\max}(\bm{\Sigma}_w) \|\cm{M}\|_{\Fr}^2\leq \kappa_2 \|\cm{M}\|_{\Fr}^2.
			\]
			As a result, for any   $\cm{M}\in\mathbb{R}^{N\times N\times (d+r+2s)}$,  we have
			\begin{equation}\label{eq:hansonw}
				\mathbb{P}\left(  0.5\kappa_1 \|\cm{M}\|_{\Fr}^2\leq \frac{1}{T}\sum_{t=1}^{T}\|\cm{M}_{(1)}\bm{z}_t\|_2^2 \leq 1.5 \kappa_2 \|\cm{M}\|_{\Fr}^2  \right)  
				\geq 1- 2e^{-c_\sigma(\kappa_1/\kappa_2)^2 T}.
			\end{equation}
			
			Next we strengthen \eqref{eq:hansonw} to  union bounds that hold for all $\cm{M}\in\bm{\Xi}_1$. For simplicity, denote $\bm{Z}=(\bm{z}_T,\dots, \bm{z}_1)$, and then
			\[
			\frac{1}{T}\sum_{t=1}^{T}\|\cm{M}_{(1)}\bm{z}_t\|_2^2 = \frac{1}{T}\|\cm{M}_{(1)} \bm{Z}\|_{\Fr}^2.
			\]
			We consider a minimal generalized $\epsilon_0$-net $\bar{\bm{\Xi}}(\epsilon_0)$ of $\bm{\Xi}_1$, where $0<\epsilon_0<2/3$ will be chosen later. By  Lemma \ref{lemma:epsilon-net}(ii), any $\cm{M}\in\bar{\bm{\Xi}}(\epsilon_0)$ satisfies $c_{\cmtt{M}}\leq \|\cm{M}\|_{\Fr}\leq C_{\cmtt{M}}$. Define the event 
			\[
			\mathcal{E}(\epsilon_0)=\left \{\forall \cm{M}\in\bar{\bm{\Xi}}(\epsilon_0): \sqrt{0.5c_{\cmtt{M}}\kappa_1}<\frac{1}{\sqrt{T}}\|\cm{M}_{(1)} \bm{Z}\|_{\Fr}<\sqrt{1.5C_{\cmtt{M}}\kappa_2}\right \}.
			\]
			Then, by the pointwise bounds in \eqref{eq:hansonw} and the covering number in Lemma \ref{lemma:epsilon-net}(i), 
			we have
			\begin{align}\label{eq:event}
				\mathbb{P}\{\stcomp{\mathcal{E}}(\epsilon_0)\}&\leq e^{Cd_{\pazocal{R}}\log(1/\epsilon_0)} \max_{\cmtt{M}\in\bar{\bm{\Xi}}(\epsilon_0)} \mathbb{P} \left [\stcomp{\left \{ 0.5c_{\cmtt{M}}\kappa_1 \leq \frac{1}{T}\|\cm{M}_{(1)}\bm{Z}\|_{\Fr}^2 \leq 1.5C_{\cmtt{M}}\kappa_2  \right \} }\right ] \notag\\
				&\leq  2\exp\left\{ - c_\sigma(\kappa_1/\kappa_2)^2 T + Cd_{\pazocal{R}}\log(1/\epsilon_0)\right\}.
			\end{align}
			
			By  Lemma \ref{lemma:epsilon-net}(iii), it holds
			\begin{equation}\label{eq:upper}
				\mathcal{E}(\epsilon_0)\subset  \left \{ \max_{\cmtt{M}\in\bar{\bm{\Xi}}(\epsilon_0)}\frac{1}{\sqrt{T}}\|\cm{M}_{(1)} \bm{Z}\|_{\Fr}\leq \sqrt{1.5C_{\cmtt{M}}\kappa_2}\right \}\subset \left \{ \sup_{\cmtt{M}\in\bm{\Xi}_1}\frac{1}{\sqrt{T}}\|\cm{M}_{(1)} \bm{Z}\|_{\Fr}\leq \frac{\sqrt{1.5C_{\cmtt{M}}\kappa_2}}{1-1.5\epsilon_0}\right \}.
			\end{equation}
			Moreover,  for any $\cm{M}\in\bm{\Xi}_1$ and its corresponding $\bar{\cm{M}}\in\bar{\bm{\Xi}}(\epsilon_0)$ defined as in  the proof of Lemma \ref{lemma:epsilon-net}(iii), similarly to \eqref{eq:disc3}, we can show that
			\begin{align*}
				\frac{1}{\sqrt{T}}	\| \cm{M}_{(1)} \bm{Z}\|_{\Fr}
				&	\geq \frac{1}{\sqrt{T}}\| \bar{\cm{M}}_{(1)} \bm{Z}\|_{\Fr} - \frac{1}{\sqrt{T}}\| (\cm{M}-\bar{\cm{M}})_{(1)}\bm{Z}\|_{\Fr} \\
				&\geq \min_{\bar{\cmt{M}} \in \bar{\bm{\Xi}}(\epsilon)} \frac{1}{\sqrt{T}} \| \bar{\cm{M}}_{(1)} \bm{Z}\|_{\Fr} - \frac{1}{\sqrt{T}} \sum_{i=1}^{4} \| (\cm{M}_i)_{(1)}\bm{Z}\|_{\Fr}\\
				&\geq \min_{\bar{\cmt{M}} \in \bar{\bm{\Xi}}(\epsilon)} \frac{1}{\sqrt{T}}\| \bar{\cm{M}}_{(1)} \bm{Z}\|_{\Fr}  -  \sum_{i=1}^{4} \|\cm{M}_i\|_{\Fr} 	\sup_{\cmtt{M} \in \bm{\Xi}_1} \frac{1}{\sqrt{T}} \| \cm{M}_{(1)} \bm{Z}\|_{\Fr}  \\
				&\geq \min_{\bar{\cmt{M}} \in \bar{\bm{\Xi}}(\epsilon)} \frac{1}{\sqrt{T}} \| \bar{\cm{M}}_{(1)} \bm{Z}\|_{\Fr}  - 1.5\epsilon_0 \sup_{\cmtt{M} \in \bm{\Xi}_1} \frac{1}{\sqrt{T}}\| \cm{M}_{(1)} \bm{Z}\|_{\Fr}.
			\end{align*}
			Taking the infimum over all  $\cm{M}\in\bm{\Xi}_1$ and combining the result with \eqref{eq:upper},  we can show that on the event $\mathcal{E}(\epsilon_0)$, it holds
			\begin{align*}
				\inf_{\cmtt{M}\in\bm{\Xi}_1}\frac{1}{\sqrt{T}}	\| \cm{M}_{(1)} \bm{Z}\|_{\Fr} &\geq \sqrt{0.5c_{\cmtt{M}}\kappa_1} - 1.5\epsilon_0\cdot \frac{\sqrt{1.5C_{\cmtt{M}}\kappa_2}}{1-1.5\epsilon_0} 
				\geq \sqrt{0.5c_{\cmtt{M}}\kappa_1}-3 \epsilon_0 \sqrt{1.5C_{\cmtt{M}}\kappa_2}
			\end{align*}
			if $0<\epsilon_0\leq1/3$. Thus, by setting 
			\[
			\epsilon_0= \min\left\{\frac{1}{6}\sqrt{\frac{c_{\cmtt{M}}\kappa_1}{3 C_{\cmtt{M}}\kappa_2}}, \frac{1}{3} \right \},
			\] 
			we have
			\begin{equation}\label{eq:lower}
				\mathcal{E}(\epsilon_0)\subset \left \{ \inf_{\cmtt{M}\in\bm{\Xi}_1}\frac{1}{\sqrt{T}}\|\cm{M}_{(1)} \bm{Z}\|_{\Fr}\geq \frac{\sqrt{0.5c_{\cmtt{M}}\kappa_1}}{2}\right \}.
			\end{equation} 
			As a result, with the above choice of $\epsilon_0$, we have
			\[
			\mathcal{E}(\epsilon_0) \subset \left \{ \frac{c_{\cmtt{M}}\kappa_1}{8} \leq \inf_{\cmtt{M}\in\bm{\Xi}_1}\frac{1}{T}\sum_{t=1}^{T}\|\cm{M}_{(1)}\bm{z}_t\|_2^2  \leq 
			\sup_{\cmtt{M}\in\bm{\Xi}_1}\frac{1}{T}\sum_{t=1}^{T}\|\cm{M}_{(1)}\bm{z}_t\|_2^2  \leq 6C_{\cmtt{M}}\kappa_2  \right \},
			\]
			which, together with \eqref{eq:event} and the condition that  $T\gtrsim(\kappa_2/\kappa_1)^2d_{\pazocal{R}}\log(\kappa_2/\kappa_1)$, leads to \eqref{eq:rscG1}.
			
			\bigskip
			\noindent\textbf{Proof of \eqref{eq:devG}:} Consider a minimal generalized $1/3$-net $\bar{\bm{\Xi}}(1/3)$ of $\bm{\Xi}_1$.
			By \eqref{eq:event} and \eqref{eq:upper}, we have 
			\[
			\mathbb{P}\left \{ \max_{\cmtt{M}\in\bar{\bm{\Xi}}(1/3)}\frac{1}{\sqrt{T}}\|\cm{M}_{(1)} \bm{Z}\|_{\Fr} > \sqrt{1.5C_{\cmtt{M}}\kappa_2}\right \} \leq 2e^{-cd_{\pazocal{R}}\log(\kappa_2/\kappa_1)},
			\]
			under the condition that  $T\gtrsim(\kappa_2/\kappa_1)^2d_{\pazocal{R}}\log(\kappa_2/\kappa_1)$. Note that $\sum_{t=1}^{T}\langle \cm{M}_{(1)} \bm{z}_t, \bm{\varepsilon}_t \rangle =\langle \cm{M}_{(1)}, \sum_{t=1}^{T} \bm{\varepsilon}_t  \bm{z}_t^\prime \rangle $. Then  by  Lemma \ref{lemma:epsilon-net}, for any $K>0$, we have 
			\begin{align*}
				& \mathbb{P}\left \{ \sup_{\cmtt{M} \in \bm{\Xi}_1 } \frac{1}{T}\sum_{t=1}^{T}\langle \cm{M}_{(1)} \bm{z}_t, \bm{\varepsilon}_t \rangle  \geq K \right \} \\
				& \hspace{5mm}\leq 
				\mathbb{P}\left \{ \max_{\cmtt{M} \in \bar{\bm{\Xi}}(1/3)} \frac{1}{T}\sum_{t=1}^{T}\langle \cm{M}_{(1)}\bm{z}_t, \bm{\varepsilon}_t \rangle  \geq \frac{K}{2} \right \}\\
				& \hspace{5mm}\leq  \mathbb{P}\left \{  \max_{\cmtt{M} \in \bar{\bm{\Xi}}(1/3)} \frac{1}{T}\sum_{t=1}^{T}\langle \cm{M}_{(1)}\bm{z}_t, \bm{\varepsilon}_t \rangle  \geq \frac{K}{2}, \;  \max_{\cmtt{M} \in \bar{\bm{\Xi}}(1/3)}\frac{1}{T}\sum_{t=1}^{T}\|\cm{M}_{(1)}\bm{z}_t\|_2^2  \leq 1.5C_{\cmtt{M}}\kappa_2  \right \} + 2e^{-cd_{\pazocal{R}}\log(\kappa_2/\kappa_1)}\\
				& \hspace{5mm}\leq  e^{C d_{\pazocal{R}} \log 3} \max_{\cmtt{M} \in \bar{\bm{\Xi}}(1/3)}  \mathbb{P}\left \{  \frac{1}{T}\sum_{t=1}^{T}\langle \cm{M}_{(1)}\bm{z}_t, \bm{\varepsilon}_t \rangle  \geq \frac{K}{2}, \;\frac{1}{T}\sum_{t=1}^{T}\|\cm{M}_{(1)}\bm{z}_t\|_2^2  \leq 1.5C_{\cmtt{M}}\kappa_2  \right \} + 2e^{-cd_{\pazocal{R}}\log(\kappa_2/\kappa_1)},
			\end{align*}
			where the first inequality follows from \eqref{eq:disc1}, and the last from the covering number in Lemma \ref{lemma:epsilon-net}(i).  For any   $\cm{M}\in\mathbb{R}^{N\times N\times (d+r+2s)}$, we can apply Lemma \ref{lemma:martgl} with $\bm{w}_t=\cm{M}_{(1)}\bm{z}_t$ to obtain the following pointwise bound:
			\[
			\mathbb{P}\left \{  \frac{1}{T}\sum_{t=1}^{T}\langle \cm{M}_{(1)}\bm{z}_t, \bm{\varepsilon}_t \rangle  \geq \frac{K}{2}, \;\frac{1}{T}\sum_{t=1}^{T}\|\cm{M}_{(1)}\bm{z}_t\|_2^2  \leq 1.5C_{\cmtt{M}}\kappa_2  \right \} \leq \exp\left \{-\frac{K^2 T}{12\sigma^2 C_{\cmtt{M}}\kappa_2   \lambda_{\max}(\bm{\Sigma}_{\varepsilon})}\right \}.
			\]
			Then, choosing 
			$K$ such that ${K^2 T}/\{12\sigma^2 C_{\cmtt{M}}\kappa_2   \lambda_{\max}(\bm{\Sigma}_{\varepsilon})\}\gtrsim  d_{\pazocal{R}}$, i.e., 
			$K\asymp \sqrt{\kappa_2   \lambda_{\max}(\bm{\Sigma}_{\varepsilon})d_{\pazocal{R}}/T}$, we have
			\[
			\mathbb{P}\left \{ \sup_{\cmtt{M} \in \bm{\Xi}_1 } \frac{1}{T}\sum_{t=1}^{T}\langle \cm{M}_{(1)} \bm{z}_t, \bm{\varepsilon}_t \rangle  \geq K \right \} \leq e^{-c d_{\pazocal{R}}}+  2e^{-cd_{\pazocal{R}}\log(\kappa_2/\kappa_1)},
			\] 
			and hence \eqref{eq:devG}. The proof  of Lemma \ref{lemma:rscG1} is complete.
			
			\subsection{Proof of Lemma \ref{lemma:Rjy}}
			
			The following result for low-rank matrices is used in the proof of Lemma \ref{lemma:Rjy}.
			
			\begin{lemma}[Covering number and discretization for low-rank matrices]\label{lemma:coverLR}
				Let $\bm{\Pi}(\pazocal{R})= \{ \bm{M} \in \mathbb{R}^{N \times N} \mid \|\bm{M}\|_{\Fr} =1, \rank(\bm{M}) \leq \pazocal{R} \}$, and let $\bar{\bm{\Pi}}(\pazocal{R})$ be a minimal $1/2$-net of $\bm{\Pi}(\pazocal{R})$ in the Frobenius norm. Then the cardinality of $\bar{\bm{\Pi}}(\pazocal{R})$ satisfies 
				\[
				\log|\bar{\bm{\Pi}}(\pazocal{R})|\leq (2N+1)\pazocal{R} \log 18.
				\]
				Moreover, for any $\bm{X}\in\mathbb{R}^{N\times N}$, it holds
				\[
				\sup_{\bm{M} \in  \bm{\Pi}(\pazocal{R})}\langle \bm{M}, \bm{X} \rangle \leq 4 \max_{\bm{M} \in  \bar{\bm{\Pi}}(\pazocal{R})} \langle \bm{M}, \bm{X} \rangle.
				\]
			\end{lemma}
			
			\begin{proof}[Proof of Lemma \ref{lemma:coverLR}]
				The covering number is given by Lemma 3.1 in \cite{candes2011tight}.
				For any $\bm{M} \in \bm{\Pi}(\pazocal{R})$, there exists $\bar{\bm{M}} \in \bar{\bm{\Pi}}(\pazocal{R})$ satisfying $\|\bm{M} - \bar{\bm{M}}\|_{\Fr} \leq 1/2$. Note that the rank of $\bm{M} - \bar{\bm{M}}$ is at most $2\pazocal{R}$. Based on the singular value decomposition of $\bm{M} - \bar{\bm{M}}$, we can find two matrices $\bm{M}^{(1)}$ and $\bm{M}^{(2)}$ with rank at most $\pazocal{R}$ such that
				$\bm{M} - \bar{\bm{M}} = \bm{M}^{(1)} + \bm{M}^{(2)}$ and $\langle \bm{M}^{(1)} , \bm{M}^{(2)}\rangle = 0$. 
				Then it holds $\| \bm{M}^{(1)} \|_{\Fr}  + \| \bm{M}^{(2)} \|_{\Fr} \leq \sqrt{2} \|\bm{M} - \bar{\bm{M}}\|_{\Fr} \leq \sqrt{2}/2$. 
				Hence, 
				for any $\bm{X}\in\mathbb{R}^{N\times N}$, we have
				\begin{align*}
					\langle \bm{M}, \bm{X} \rangle=  \langle\bar{\bm{M}}, \bm{X}\rangle + \sum_{i=1}^{2}\langle\bm{M}^{(i)},\bm{X}\rangle  
					& \leq \max_{\small{\bar{\bm{M}} \in  \bar{\bm{\Pi}}(\pazocal{R})}}  \langle\bar{\bm{M}}, \bm{X}\rangle+ \sum_{i=1}^{2} \|\bm{M}^{(i)}\|_{\Fr}\sup_{\small{ \bm{M} \in  \bm{\Pi}(\pazocal{R})}} \langle \bm{M}, \bm{X} \rangle \\
					& \leq \max_{\small{\bar{\bm{M}} \in  \bar{\bm{\Pi}}(\pazocal{R})}} \langle\bar{\bm{M}}, \bm{X}\rangle + \frac{\sqrt{2}}{2} \sup_{\small{ \bm{M} \in  \bm{\Pi}(\pazocal{R})}} \langle \bm{M}, \bm{X} \rangle .
				\end{align*}
				Taking supremum with respect to $  \bm{M} \in  \bm{\Pi}(\pazocal{R})$ on both sides of the last inequality, we accomplish the proof of this lemma.
			\end{proof}
			
			
			\noindent\textbf{Proof of \eqref{eq:Rjy1}:}
			Denote $S^{N-1} = \{\bm{u} \in \mathbb{R}^{N} \mid \|\bm{u}\|_2=1\}$, and let $\bar{S}^{N-1}$ be a minimal $(1/4)$-net of $S^{N-1}$ in the Euclidean norm. Fix $j\geq1$. By Lemma 5.4 in \cite{vershynin2010introduction}, 
			\begin{equation*}
				\Big\| \frac{1}{T} \sum_{t=1}^{T}\bm{y}_{t-p-j}\bm{y}_{t-p-j}^\prime \Big\|_{\op} 
				\leq 
				2 \max_{\bm{u}\in\bar{S}^{N-1}}\bm{u}^\prime \left(  \frac{1}{T} \sum_{t=1}^{T}\bm{y}_{t-p-j}\bm{y}_{t-p-j}^\prime \right) \bm{u}=2 \max_{\bm{u}\in\bar{S}^{N-1}}\frac{1}{T} \sum_{t=1}^{T} (\bm{u}^\prime\bm{y}_{t-p-j})^2.
			\end{equation*}
			Then for any $K>0$, 
			\begin{align}\label{eq:Ry2}
				\mathbb{P}\left (\Big\| \frac{1}{T} \sum_{t=1}^{T}\bm{y}_{t-p-j}\bm{y}_{t-p-j}^\prime \Big\|_{\op}  \geq K\right ) &\leq 	\mathbb{P}\left \{ \max_{\bm{u}\in\bar{S}^{N-1}}\frac{1}{T} \sum_{t=1}^{T} (\bm{u}^\prime\bm{y}_{t-p-j})^2 \geq K/2\right \} \notag \\
				&\leq 9^N \max_{\bm{u}\in S^{N-1}} \mathbb{P}\left \{\frac{1}{T} \sum_{t=1}^{T} (\bm{u}^\prime\bm{y}_{t-p-j})^2 \geq K/2\right \},
			\end{align}
			where we used the fact that the cardinality of $\bar{S}^{N-1}$ satisfies $|\bar{S}^{N-1}|\leq 9^N$. For any $\bm{u}\in S^{N-1}$, applying Lemma \ref{lemma:hansonw}(ii) with $\bm{M}=\bm{u}^\prime$, $\bm{w}_t=\bm{y}_{t-1}$, and $T_0=1-p-j$, together with the result
			\[
			\lambda_{\max}(\underline{\bm{\Sigma}}_w)=\lambda_{\max}(\underline{\bm{\Sigma}}_y)\leq  \lambda_{\max}(\bm{\Sigma}_\varepsilon)\mu_{\max}(\bm{\Psi}_*) 
			\]
			as implied by Lemma \ref{lemma:Wcov}(i), we can show that 
			\begin{equation*}
				\mathbb{P}\left \{\left |\frac{1}{T} \sum_{t=1}^{T} (\bm{u}^\prime\bm{y}_{t-p-j})^2 -\mathbb{E}\{(\bm{u}^\prime\bm{y}_{t-p-j})^2\}\right | \geq  \delta \sigma^2\lambda_{\max}(\bm{\Sigma}_\varepsilon) \mu_{\max}(\bm{\Psi}_*) \right \} \leq 2e^{-c\min(\delta, \delta^2)T}.
			\end{equation*}
			holds for any $\delta>0$. In addition, by Lemma \ref{lemma:Wcov}(i),
			\[
			\mathbb{E}\{(\bm{u}^\prime\bm{y}_{t-p-j})^2\}\leq \lambda_{\max}(\bm{\Sigma}_y)\leq  \lambda_{\max}(\bm{\Sigma}_\varepsilon)\mu_{\max}(\bm{\Psi}_*).
			\]
			In view of the above result and taking $\delta=j$, we further have
			\begin{equation}\label{eq:upperj}
				\mathbb{P}\left \{\frac{1}{T} \sum_{t=1}^{T} (\bm{u}^\prime\bm{y}_{t-p-j})^2 \geq  \lambda_{\max}(\bm{\Sigma}_\varepsilon)\mu_{\max}(\bm{\Psi}_*)(j \sigma^2+1)\right \} \leq 2e^{-cjT}.
			\end{equation}
			Combining \eqref{eq:Ry2} and \eqref{eq:upperj}, if $T\geq 2N\log9/c$, then
			\begin{equation}\label{eq:yop}
				\mathbb{P}\left \{\Big\| \frac{1}{T} \sum_{t=1}^{T}\bm{y}_{t-p-j}\bm{y}_{t-p-j}^\prime \Big\|_{\op}  \geq  2\lambda_{\max}(\bm{\Sigma}_\varepsilon)\mu_{\max}(\bm{\Psi}_*)(j \sigma^2+1) \right \} \leq 2e^{-jN\log9}.
			\end{equation}
			By considering the union bound over all $j\geq1$, we have
			\[
			\mathbb{P}\left \{\exists j\geq 1: \Big\| \frac{1}{T} \sum_{t=1}^{T}\bm{y}_{t-p-j}\bm{y}_{t-p-j}^\prime \Big\|_{\op}  \geq  2\lambda_{\max}(\bm{\Sigma}_\varepsilon)\mu_{\max}(\bm{\Psi}_*)(j \sigma^2+1) \right \} \leq \sum_{j=1}^{\infty}2e^{-jN\log9}\leq 3e^{-N\log9}.
			\]
			and hence \eqref{eq:Rjy1}.
			
			\bigskip
			\noindent\textbf{Proof of \eqref{eq:Rjy2}:} We first fix $j\geq1$.
			Note that $\sum_{t=1}^{T}\langle \bm{M} \bm{y}_{t-p-j}, \bm{\varepsilon}_t \rangle=\langle \bm{M}, \sum_{t=1}^{T} \bm{\varepsilon}_t \bm{y}_{t-p-j}^{\prime} \rangle$. Moreover, it can be verified that 
			\[
			\frac{1}{T}\sum_{t=1}^{T}\|\bm{M} \bm{y}_{t-p-j}\|_{2}^2 \leq \|\bm{M} \|_{\Fr}^2 \Big\| \frac{1}{T} \sum_{t=1}^{T}\bm{y}_{t-p-j}\bm{y}_{t-p-j}^\prime \Big\|_{\op}.
			\]
			Thus, by \eqref{eq:yop} and Lemma \ref{lemma:coverLR}, for any $K>0$, we have
			\begin{align*}
				& \mathbb{P}\left \{\sup_{\bm{M} \in \bm{\Pi}(\pazocal{R})}\frac{1}{T}\sum_{t=1}^{T}\langle \bm{M} \bm{y}_{t-p-j}, \bm{\varepsilon}_t \rangle \geq K\right \}\\
				& \leq  \mathbb{P}\left \{\max_{\bm{M} \in \bar{\bm{\Pi}}(\pazocal{R})}\frac{1}{T}\sum_{t=1}^{T}\langle \bm{M} \bm{y}_{t-p-j}, \bm{\varepsilon}_t \rangle \geq \frac{K}{4}, \,\Big\| \frac{1}{T} \sum_{t=1}^{T}\bm{y}_{t-p-j}\bm{y}_{t-p-j}^\prime \Big\|_{\op}  \leq  2\lambda_{\max}(\bm{\Sigma}_\varepsilon)\mu_{\max}(\bm{\Psi}_*) (j \sigma^2+1)\right \} \\
				&\hspace{5mm}+ 2e^{-jN\log9}\\
				& \leq e^{9N\pazocal{R}} \max_{\bm{M} \in \bar{\bm{\Pi}}(\pazocal{R})} \mathbb{P}\left \{\frac{1}{T}\sum_{t=1}^{T}\langle \bm{M} \bm{y}_{t-p-j}, \bm{\varepsilon}_t \rangle \geq \frac{K}{4}, \,\frac{1}{T}\sum_{t=1}^{T}\|\bm{M} \bm{y}_{t-p-j}\|_{2}^2   \leq  2\lambda_{\max}(\bm{\Sigma}_\varepsilon)\mu_{\max}(\bm{\Psi}_*) (j \sigma^2+1) \right \} \\
				&\hspace{5mm}+ 2e^{-jN\log9}.
			\end{align*}
			Then, applying Lemma \ref{lemma:martgl} with $\bm{w}_t=\bm{M} \bm{y}_{t-p-j}$, we have the pointwise bound for any $\bm{M}\in\mathbb{R}^{N\times N}$ as follows:
			\begin{align*}
				&	\mathbb{P}\left \{\frac{1}{T}\sum_{t=1}^{T}\langle \bm{M} \bm{y}_{t-p-j}, \bm{\varepsilon}_t \rangle \geq \frac{K}{4}, \,\frac{1}{T}\sum_{t=1}^{T}\|\bm{M} \bm{y}_{t-p-j}\|_{2}^2   \leq   2\lambda_{\max}(\bm{\Sigma}_\varepsilon)\mu_{\max}(\bm{\Psi}_*) (j \sigma^2+1) \right \}\\
				&\hspace{5mm} \leq \exp\left \{- \frac{K^2 T}{64    \lambda_{\max}^2(\bm{\Sigma}_\varepsilon) \mu_{\max}(\bm{\Psi}_*) (j \sigma^2+1)\sigma^2} \right \} =\exp\{-9N \pazocal{R} (j +1)\},
			\end{align*}
			if we choose 
			\[
			K=24\lambda_{\max}(\bm{\Sigma}_\varepsilon) \sqrt{\frac{\mu_{\max}(\bm{\Psi}_*)(j \sigma^2+1)(j\sigma^2+\sigma^2)N \pazocal{R}}{T}}.
			\]
			Consequently, we have
			\begin{align*}
				&\mathbb{P}\left \{\sup_{\bm{M} \in \bm{\Pi}(\pazocal{R})}\frac{1}{T}\sum_{t=1}^{T}\langle \bm{M} \bm{y}_{t-p-j}, \bm{\varepsilon}_t \rangle \geq 24\lambda_{\max}(\bm{\Sigma}_\varepsilon) \sqrt{\frac{\mu_{\max}(\bm{\Psi}_*) (j \sigma^2+1) (j \sigma^2+\sigma^2)N \pazocal{R}}{T}}\right \}\\
				&\hspace{5mm} \leq e^{-9  j N \pazocal{R}}+ 2e^{-jN\log9} \leq 3 e^{-  j N\log9}.
			\end{align*}
			Taking the union bound over all $j\geq1$ as in the proof of \eqref{eq:Rjy1} and noting that $(j \sigma^2+1) (j \sigma^2+\sigma^2) \leq (2j \sigma^2+1)^2$, we can verify \eqref{eq:Rjy2}.
			

			\subsection{Proof of Lemma \ref{lemma:epsilon-net}}
			
			The following covering result for low-Tucker-rank tensors is used in the proof of Lemma \ref{lemma:epsilon-net}. 
			\begin{lemma}[Covering number for low-Tucker-rank tensors] \label{lemma:covering}
				Let $\bm{\Pi}( \pazocal{R}_1,\pazocal{R}_2)=\{\cm{T}\in\mathbb{R}^{p_1\times p_2\times p_3}:\|\cm{T}\|_{\textup{F}}\leq 1,\textup{rank}(\cm{T}_{(i)})\leq \pazocal{R}_i,i=1, 2\}$. For any $\epsilon>0$, let $\bar{\bm{\Pi}}(\epsilon; \pazocal{R}_1,\pazocal{R}_2)$ be a minimal $\epsilon$-net  for $\bm{\Pi}(\pazocal{R}_1,\pazocal{R}_2)$ in the Frobenius norm. Then the cardinality of $\bar{\bm{\Pi}}(\epsilon; \pazocal{R}_1,\pazocal{R}_2)$ satisfies
				\begin{equation*}
					|\bar{\bm{\Pi}}(\epsilon; \pazocal{R}_1,\pazocal{R}_2)|\leq(9/\epsilon)^{\pazocal{R}_1\pazocal{R}_2p_3 + p_1\pazocal{R}_1+p_2\pazocal{R}_2}.
				\end{equation*}
			\end{lemma}
			\begin{proof}[Proof of Lemma \ref{lemma:covering}]
				The proof of this lemma is straightforward given the proof of Lemma 2 in \cite{RSS17}.
			\end{proof}
			
			
			\noindent\textbf{Proof of (i):} Notice that the results for $\cm{G}_{\rm{stack}}$ in \eqref{eq:Gstacknorm} can be generalized to any $\cm{M}(\bm{a}, \cm{B})$ with  $\bm{a}\in \mathbb{R}^{r+2s}$ and $\cm{B}\in\mathbb{R}^{N\times N\times (r+2s)}$, where $\cm{G}_{\rm{stack}}=\cm{M}(\bm{\omega}-\bm{\omega}^*, \cm{G}-\cm{G}^*)$.  That is,  by a method similar to that for \eqref{eq:Gstacknorm}, we can show that for any $\bm{a}\in \mathbb{R}^{r+2s}$ and $\cm{B}\in\mathbb{R}^{N\times N\times (r+2s)}$,
			\begin{align}\label{eq:EFr1}
				0.5 ( \|\cm{B}\|_{\Fr} + \varpi_1\|\bm{a}\|_2) \leq \|\cm{M} (\bm{a} , \cm{B} )\|_{\Fr} \leq \|\cm{B}\|_{\Fr} + \varpi_2\|\bm{a}\|_2,
			\end{align}
			where 
			\[
			\varpi_1=\sqrt{2}c_{\cmtt{G}} \alpha\quad\text{and}\quad \varpi_2=\frac{\sqrt{2}\alpha}{\min_{1\leq k\leq s}\gamma_{k}^*}.
			\]
			Thus, if $\|\cm{M}(\bm{a} , \cm{B} )\|_{\Fr} = 1$, then $ 1 \leq \|\cm{B}\|_{\Fr} \leq 2$ and $\varpi_2^{-1} \leq \|\bm{a}\|_2 \leq \varpi_1^{-1}$. As a result, 
			\[\bm{\Xi}_1\subset \{\cm{M}(\bm{a}, \cm{B}) \mid \bm{a}\in \bm{\Pi}^{(1)}, \cm{B}\in \bm{\Pi}^{(2)} \},\]
			where
			\[\bm{\Pi}^{(1)} = \left \{ \bm{a} \in \mathbb{R}^{r+2s} \mid \varpi_2^{-1} \leq \|\bm{a}\|_2  \leq \varpi_1^{-1}\right \}\]
			and
			\[\bm{\Pi}^{(2)} = \{ \cm{B} \in \mathbb{R}^{N \times N \times d} \mid \cm{B}\in\bm{\Gamma}(2\pazocal{R}_1, 2\pazocal{R}_2), 1  \leq \|\cm{B}\|_{\Fr} \leq 2 \}.\]
			Hence, the problem of covering $\bm{\Xi}_1$ can be converted into that of covering $\bm{\Pi}^{(1)}$ and $\bm{\Pi}^{(2)}$.
			
			For any fixed $\epsilon >0 $,  let $\bar{\bm{\Pi}}^{(1)}(\epsilon)$ be a minimal $\epsilon/(2\varpi_2)$-net  for $\bm{\Pi}^{(1)}$ in the Euclidean norm, and let $\bar{\bm{\Pi}}^{(2)}(\epsilon)$ be a minimal $\epsilon/2$-net for $\bm{\Pi}^{(2)}$ in the Frobenius norm. Then denote
			\[
			\bar{\bm{\Xi}}(\epsilon) = \left\{ \cm{M}(\bm{a} , \cm{B} ) \in  \mathbb{R}^{N \times N \times (d+r+2s)} \mid  \bm{a}  \in \bar{\bm{\Pi}}^{(1)}(\epsilon),  \cm{B} \in \bar{\bm{\Pi}}^{(2)}(\epsilon) \right\}.
			\]
			Thus,  for every $\cm{M}(\bm{a}, \cm{B}) \in \bm{\Xi}_1$, there exists $\cm{M}( \bar{\bm{a}}, \bar{\cm{B}})\in \bar{\bm{\Xi}}(\epsilon)$ with    $\bar{\bm{a}} \in \bar{\bm{\Pi}}^{(1)}(\epsilon)$ and $\bar{\cm{B}}\in \bar{\bm{\Pi}}^{(2)}(\epsilon)$ such that 
			\begin{equation}\label{eq:e-net}
				\|\bm{a} - \bar{\bm{a}}\|_2\leq \epsilon/(2\varpi_2) \quad\text{and}\quad	\|\cm{B} - \bar{\cm{B}}\|_{\Fr} \leq\epsilon/2.
			\end{equation}
			Since $\cm{M}(\bm{a}, \cm{B})- \cm{M}( \bar{\bm{a}}, \bar{\cm{B}})=\cm{M}(\bm{a} - \bar{\bm{a}}, \cm{B} - \bar{ \cm{B}})$, it follows from \eqref{eq:EFr1} and \eqref{eq:e-net} that
			\begin{align*}
				\|\cm{M}(\bm{a}, \cm{B})- \cm{M}( \bar{\bm{a}}, \bar{\cm{B}})\|_{\Fr} = \|\cm{M}(\bm{a} - \bar{\bm{a}}, \cm{B} - \bar{ \cm{B}}) \|_{\Fr} \leq \|\cm{B} - \bar{\cm{B}}\|_{\Fr} + \varpi_2 \|\bm{a} - \bar{\bm{a}}\|_2 \leq \epsilon.
			\end{align*}
			In addition, note that $\bar{\bm{\Xi}}(\epsilon)\subset \bm{\Xi}$. Therefore, $\bar{\bm{\Xi}}(\epsilon)$ is a generalized $\epsilon$-net of $\bm{\Xi}_1$.
			Moreover, by a standard volumetric argument (see also Corollary 4.2.13 in \cite{Vershynin2018} for details) and Lemma \ref{lemma:covering}, the cardinalities of $\bar{\bm{\Pi}}^{(1)}(\epsilon)$ and $\bar{\bm{\Pi}}^{(2)}(\epsilon)$ satisfy
			\begin{align*}
				\log|\bar{\bm{\Pi}}^{(1)}(\epsilon)| \leq (r+2s)\log\{6\varpi_2/(\varpi_1 \epsilon)\} \quad\text{and}\quad 	\log|\bar{\bm{\Pi}}^{(2)}(\epsilon)| \leq \{ 4\pazocal{R}_1\pazocal{R}_2 d + 2(\pazocal{R}_1+\pazocal{R}_2)N  \}\log(18/\epsilon).
			\end{align*}
			Noting that  $\varpi_1\varpi_2^{-1}\asymp1$ is independent of $\epsilon$, we have
			\begin{align*}
				\log|\bar{\bm{\Xi}}(\epsilon)| \leq \log|\bar{\bm{\Pi}}^{(1)}(\epsilon)| + \log|\bar{\bm{\Pi}}^{(2)}(\epsilon)| \lesssim
				\left \{\pazocal{R}_1\pazocal{R}_2d + (\pazocal{R}_1 + \pazocal{R}_2) N\right \} \log(1/\epsilon).
			\end{align*}
			
			\bigskip\noindent
			\textbf{Proof of (ii):} Since 
			$\bar{\bm{\Pi}}^{(1)}(\epsilon)\subset\bm{\Pi}^{(1)}$ and $\bar{\bm{\Pi}}^{(2)}(\epsilon)\subset\bm{\Pi}^{(2)}$, we have
			\[\bar{\bm{\Xi}}(\epsilon)\subset \left\{ \cm{M}(\bm{a} , \cm{B} ) \in  \mathbb{R}^{N \times N \times (d+r+2s)} \mid  \bm{a}  \in \bm{\Pi}^{(1)},  \cm{B} \in \bm{\Pi}^{(2)} \right\}.\] 
			Thus, by \eqref{eq:EFr1}, for any $\cm{M}\in \bar{\bm{\Xi}}(\epsilon)$, it holds
			\[
			c_{\cmtt{M}}:=0.5 ( 1 + \varpi_1\varpi_2^{-1}) \leq \|\cm{M}\|_{\Fr} \leq 2 + \varpi_2\varpi_1^{-1}:=C_{\cmtt{M}}.
			\]
			Since $\varpi_1\varpi_2^{-1}=c_{\cmtt{G}}\min_{1\leq k\leq s}\gamma_{k}^*\asymp1$  is independent of $\epsilon$, (ii) is proved.
			
			\bigskip\noindent
			\textbf{Proof of (iii):}
			From the proof of (i), for every $\cm{M}:=\cm{M}(\bm{a} , \cm{B} )\in\bm{\Xi}_1$, there exists $\bar{\cm{M}}:=\cm{M}( \bar{\bm{a}}, \bar{\cm{B}})\in \bar{\bm{\Xi}}(\epsilon)$ such that    $\bar{\bm{a}} \in \bar{\bm{\Pi}}^{(1)}(\epsilon)$ and $\bar{\cm{B}}\in \bar{\bm{\Pi}}^{(2)}(\epsilon)$  satisfy  \eqref{eq:e-net}. Since $\bar{\bm{\Pi}}^{(2)}(\epsilon)\subset \bm{\Pi}^{(2)}$, we have $\cm{B} - \bar{\cm{B}} \in \bm{\Gamma}(4\pazocal{R}_1, 4\pazocal{R}_2)$.
			Then by considering the higher-order singular value decomposition for $\cm{B} - \bar{\cm{B}}$, we can find four tensors $\cm{B}_i \in \bm{\Gamma}(2\pazocal{R}_1, 2\pazocal{R}_2)$ with $1\leq i\leq 4$ such that $\cm{B} - \bar{\cm{B}} = \sum_{i=1}^{4}\cm{B}_i$ and $\langle\cm{B}_i, \cm{B}_j\rangle = 0$ for all $i \neq j$.
			As a result, we can show that
			\[
			\cm{M}-\bar{\cm{M}}=\cm{M}(\bm{a}, \cm{B})- \cm{M}( \bar{\bm{a}}, \bar{\cm{B}})=\cm{M}(\bm{a} - \bar{\bm{a}}, \cm{B} - \bar{ \cm{B}})
			=\cm{M}\left (\sum_{i=1}^{4}\frac{\bm{a} - \bar{\bm{a}}}{4}, \sum_{i=1}^{4}\cm{B}_i\right )= \sum_{i=1}^{4} \cm{M}_i,\]
			where $\cm{M}_i=\cm{M} \left ((\bm{a} - \bar{\bm{a}})/4, \cm{B}_i \right )\in\bm{\Xi}$.
			Moreover,  by \eqref{eq:EFr1}, \eqref{eq:e-net} and the Cauchy-Schwarz inequality, it holds 
			\begin{align*}
				\sum_{i=1}^{4} \|\cm{M}_i\|_{\Fr} \leq \sum_{i=1}^{4} \left\{ \|\cm{B}_i\|_{\Fr} + \frac{\varpi_2}{4} \|\bm{a} - \bar{\bm{a}}\|_2 \right\} \leq 2 \| \cm{B} - \bar{\cm{B}} \|_{\Fr} + \varpi_2 \|\bm{a} - \bar{\bm{a}}\|_2 \leq 1.5\epsilon.
			\end{align*}
			Therefore, for any $\cm{M}\in\bm{\Xi}_1$, we can show that
			\begin{align*}
				\langle \cm{M}_{(1)}, \bm{X}\rangle  
				= \langle \bar{\cm{M}}_{(1)}, \bm{X}\rangle + \langle (\cm{M}-\bar{\cm{M}})_{(1)}, \bm{X}\rangle 
				&\leq \max_{\bar{\cmt{M}} \in \bar{\bm{\Xi}}(\epsilon)} \langle \bar{\cm{M}}_{(1)}, \bm{X}\rangle + \sum_{i=1}^{4} \langle (\cm{M}_i)_{(1)}, \bm{X}\rangle\\
				&\leq \max_{\bar{\cmt{M}} \in \bar{\bm{\Xi}}(\epsilon)} \langle \bar{\cm{M}}_{(1)}, \bm{X}\rangle + \sum_{i=1}^{4} \|\cm{M}_i\|_{\Fr} 	\sup_{\cmtt{M} \in \bm{\Xi}_1} \langle \cm{M}_{(1)}, \bm{X}\rangle\\
				&\leq \max_{\bar{\cmt{M}} \in \bar{\bm{\Xi}}(\epsilon)} \langle \bar{\cm{M}}_{(1)}, \bm{X}\rangle + 1.5\epsilon	\sup_{ \cmt{M} \in \bm{\Xi}_1} \langle \cm{M}_{(1)}, \bm{X}\rangle,
			\end{align*}
			Taking supremum over all  $\cm{M} \in \bm{\Xi}_1$ on both sides, we accomplish the proof of \eqref{eq:disc1}.
			The proof of \eqref{eq:disc2} follows the same arguments as those for \eqref{eq:disc1} except that the above inequalities are revised to
			\begin{align}\label{eq:disc3}
				\| \cm{M}_{(1)} \bm{Z}\|_{\Fr}  
				\leq \| \bar{\cm{M}}_{(1)} \bm{Z}\|_{\Fr} + \| (\cm{M}-\bar{\cm{M}})_{(1)}\bm{Z}\|_{\Fr} 
				&\leq \max_{\bar{\cmt{M}} \in \bar{\bm{\Xi}}(\epsilon)} \| \bar{\cm{M}}_{(1)} \bm{Z}\|_{\Fr} + \sum_{i=1}^{4} \| (\cm{M}_i)_{(1)}\bm{Z}\|_{\Fr}\notag\\
				&\leq \max_{\bar{\cmt{M}} \in \bar{\bm{\Xi}}(\epsilon)} \| \bar{\cm{M}}_{(1)} \bm{Z}\|_{\Fr}  + \sum_{i=1}^{4} \|\cm{M}_i\|_{\Fr} 	\sup_{\cmtt{M} \in \bm{\Xi}_1} \| \cm{M}_{(1)} \bm{Z}\|_{\Fr} \notag \\
				&\leq \max_{\bar{\cmt{M}} \in \bar{\bm{\Xi}}(\epsilon)} \| \bar{\cm{M}}_{(1)} \bm{Z}\|_{\Fr}  + 1.5\epsilon \sup_{\cmtt{M} \in \bm{\Xi}_1} \| \cm{M}_{(1)} \bm{Z}\|_{\Fr}.
			\end{align}
			The proof of Lemma \ref{lemma:epsilon-net} is complete.
			
			\subsection{Proof of Lemma \ref{lemma:martgl}}
			Note that $\langle\bm{w}_t, \bm{\varepsilon}_t \rangle=\langle \bm{\Sigma}_{\epsilon}^{1/2}\bm{w}_t, \bm{\xi}_t \rangle$, where $\bm{\xi}_t$ is mean-zero and $\sigma^2$-sub-Gaussian.  Moreover, it holds $\|\bm{\Sigma}_{\varepsilon}^{1/2}\bm{w}_t\|_2^2\leq \lambda_{\max}(\bm{\Sigma}_{\varepsilon})\|\bm{w}_t\|_2^2$. Then by a straightforward multivariate generalization of Lemma 4.2 in \cite{simchowitz2018learning}, we can show that
			\begin{align*}
				\mathbb{P}\left \{ \sum_{t=1}^{T}\langle\bm{w}_t, \bm{\varepsilon}_t \rangle\geq a, \; \sum_{t=1}^{T}\lVert \bm{w}_t \rVert^2 \leq b\right \} & \leq  \mathbb{P}\left \{ \sum_{t=1}^{T}\langle \bm{\Sigma}_{\varepsilon}^{1/2}\bm{w}_t, \bm{\xi}_t \rangle\geq a, \; \sum_{t=1}^{T}\lVert \bm{\Sigma}_{\varepsilon}^{1/2} \bm{w}_t \rVert^2 \leq \lambda_{\max}(\bm{\Sigma}_{\varepsilon}) b\right \} \\
				& \leq \exp\left \{-\frac{a^2}{2\sigma^2 \lambda_{\max}(\bm{\Sigma}_{\varepsilon}) b}\right \}.
			\end{align*}
			The proof is complete.
			
			\subsection{Proof of Lemma \ref{lemma:hansonw}}
			\noindent
			\textbf{Proof of (i): } 
			First it is obvious that $\{\bm{w}_t\}$  is a zero-mean stationary time series. Without loss of generality, we let $T_0=0$ in what follows. 
			
			It is worth noting that under Assumption \ref{assum:error}, $\bm{\varepsilon}_{t}=	\bm{\Sigma}_\varepsilon^{1/2}\bm{\xi}_t$, and  all coordinates of the vector $\bm{\xi} = (\bm{\xi}_{T-1}^\prime, \bm{\xi}_{T-2}^\prime, \dots)^\prime$ are independent and $\sigma^2$-sub-Gaussian with mean zero and variance one.
			In addition, by the vector MA($\infty$) representation of $\bm{w}_t$, we have  $\underline{\bm{w}}_T = \underline{\bm{\Psi}}^w\bm{\xi}$, 
			where
			\begin{align} \label{eq:HD}
				\underset{TM \times \infty }{\underline{\bm{\Psi}}^w}= \left( \begin{matrix}
					\bm{\Psi}_1^w \bm{\Sigma}_\varepsilon^{1/2}&\bm{\Psi}_2^w \bm{\Sigma}_\varepsilon^{1/2} &\bm{\Psi}_3^w \bm{\Sigma}_\varepsilon^{1/2}&\cdots&\bm{\Psi}_{T}^w \bm{\Sigma}_\varepsilon^{1/2}&\cdots\\
					&\bm{\Psi}_1^w \bm{\Sigma}_\varepsilon^{1/2}&\bm{\Psi}_2^w \bm{\Sigma}_\varepsilon^{1/2}&\cdots&\bm{\Psi}_{T-1}^w \bm{\Sigma}_\varepsilon^{1/2}&\cdots\\
					&&\ddots&&&&\\
					&&&&\bm{\Psi}_1^w \bm{\Sigma}_\varepsilon^{1/2} &\cdots
				\end{matrix}\right).
			\end{align}
			Then, it holds
			\begin{equation}\label{eq:varw}
				\underline{\bm{\Sigma}}_w=\mathbb{E}(\underline{\bm{w}}_T\underline{\bm{w}}_T^\prime)=\underline{\bm{\Psi}}^w(\underline{\bm{\Psi}}^{w})^\prime.
			\end{equation}
			Observe that
			$\sum_{t=1}^{T} \|\bm{w}_t\|_2^2 = \underline{\bm{w}}_T^\prime \underline{\bm{w}}_T  =  \bm{\xi}^\prime (\underline{\bm{\Psi}}^{w})^\prime \underline{\bm{\Psi}}^w \bm{\xi}$.
			Since $\bm{\xi}$ is a vector with independent, zero-mean and sub-Gaussian coordinates, we can apply the Hanson-Wright inequality \citep{Vershynin2018} to obtain that for any $\iota>0$,
			\begin{equation}\label{eq:hansonw1}
				\mathbb{P}\left ( \left |\sum_{t=1}^{T} \|\bm{w}_t\|_2^2 - T\mathbb{E}\left (\|\bm{w}_t\|_2^2\right )\right | \geq \iota\right ) \leq 2\exp\left\{ - c \min\left( \frac{\iota}{\sigma^2\|(\underline{\bm{\Psi}}^{w})^\prime \underline{\bm{\Psi}}^w\|_{\op}}, \frac{\iota^2}{\sigma^4 \|(\underline{\bm{\Psi}}^{w})^\prime \underline{\bm{\Psi}}^w\|_{\Fr}^2}\right)\right\}.
			\end{equation}
			Note that by \eqref{eq:varw},
			$\|(\underline{\bm{\Psi}}^{w})^\prime \underline{\bm{\Psi}}^w\|_{\op}=\| \underline{\bm{\Psi}}^w(\underline{\bm{\Psi}}^{w})^\prime\|_{\op}=\lambda_{\max}(\underline{\bm{\Sigma}}_w)$, and $\underline{\bm{\Sigma}}_w$ is a $TM\times TM$ matrix. Then 
			\[
			\|(\underline{\bm{\Psi}}^{w})^\prime \underline{\bm{\Psi}}^w\|_{\Fr}
			=\|\underline{\bm{\Psi}}^w(\underline{\bm{\Psi}}^{w})^\prime \|_{\Fr}=\|\underline{\bm{\Sigma}}_w\|_{\Fr}\leq \sqrt{TM} \lambda_{\max}(\underline{\bm{\Sigma}}_w).
			\]
			Taking $\iota=\sigma^2\sqrt{T}M\lambda_{\max}(\underline{\bm{\Sigma}}_w)$ in \eqref{eq:hansonw1}, the proof of (i) is complete.
			
			\bigskip
			\noindent\textbf{Proof of (ii): }
			Define the vector $\underline{\bm{m}}_T=((\bm{M}\bm{w}_T)^\prime, \dots,  (\bm{M}\bm{w}_1)^\prime)^\prime=(\bm{I}_T\otimes \bm{M})\underline{\bm{w}}_T$.  Then $\underline{\bm{m}}_T=\bm{P}\bm{\xi}$, where   $\bm{P}=(\bm{I}_T\otimes \bm{M})  \underline{\bm{\Psi}}^w$. As a result,
			$\sum_{t=1}^{T}\|\bm{M}\bm{w}_t\|_2^2=\underline{\bm{m}}_T^\prime \underline{\bm{m}}_T =\bm{\xi}^\prime\bm{P}^\prime\bm{P}\bm{\xi}$. Similar to \eqref{eq:varw}, it follows from the Hanson-Wright inequality that
			for any $\iota>0$,
			\begin{equation}\label{eq:hansonw2}
				\mathbb{P}\left ( \left |\sum_{t=1}^{T} \|\bm{M} \bm{w}_t\|_2^2 - T\mathbb{E}\left (\|\bm{M}\bm{w}_t\|_2^2\right )\right | \geq \iota\right ) \leq 2\exp\left\{ - c \min\left( \frac{\iota}{\sigma^2\|\bm{P}^\prime\bm{P}\|_{\op}}, \frac{\iota^2}{\sigma^4 \|\bm{P}^\prime\bm{P}\|_{\Fr}^2}\right)\right\}.
			\end{equation}
			By \eqref{eq:varw}, we have
			$\|\bm{P}^\prime\bm{P}\|_{\op}=\|\bm{P}\bm{P}^\prime\|_{\op}\leq 
			\|\bm{M}\bm{M}^\prime\|_{\op} \| \underline{\bm{\Psi}}^w(\underline{\bm{\Psi}}^{w})^\prime\|_{\op}
			\leq \lambda_{\max}(\underline{\bm{\Sigma}}_w)\|\bm{M}\|_{\Fr}^2$. Moreover, 
			\begin{align*}
				\trace(\bm{P}^\prime\bm{P})=\trace(\bm{P}\bm{P}^\prime)&=\trace\{(\bm{I}_T\otimes \bm{M})\underline{\bm{\Sigma}}_w (\bm{I}_T\otimes \bm{M}^\prime)\}\\
				&= \textrm{vec}(\bm{I}_T\otimes \bm{M})^\prime(\underline{\bm{\Sigma}}_w\otimes\bm{I}_{TQ})\textrm{vec}(\bm{I}_T\otimes \bm{M})
				\leq T\lambda_{\max}(\underline{\bm{\Sigma}}_w) \|\bm{M}\|_{\Fr}^2,
			\end{align*}
			where the second equality follows from \eqref{eq:varw}. As a result,
			\[
			\|\bm{P}^\prime\bm{P}\|_{\Fr}\leq \sqrt{\|\bm{P}^\prime\bm{P}\|_{\op}\trace(\bm{P}^\prime\bm{P})}
			\leq 
			\sqrt{\|\bm{P}\bm{P}^\prime\|_{\op}\trace(\bm{P}\bm{P}^\prime)} \leq  \sqrt{T} \lambda_{\max}(\underline{\bm{\Sigma}}_w) \|\bm{M}\|_{\Fr}^2.
			\]
			Taking $\iota=\delta\sigma^2T\lambda_{\max}(\underline{\bm{\Sigma}}_w) \|\bm{M}\|_{\Fr}^2$ in \eqref{eq:hansonw2}, the proof of (ii) is complete.
			
			\subsection{Proof of Lemma \ref{lemma:Wcov}}
			\noindent\textbf{Proof of (i):}  Consider the spectral density of $\{\bm{y}_t\}$,
			\begin{equation*}
				\bm{f}_y(\theta) = (2\pi)^{-1} \bm{\Psi}_*(e^{-i\theta})\bm{\Sigma}_{\varepsilon}\bm{\Psi}_*^{\HH}(e^{-i\theta}), \hspace{5mm}\theta\in [-\pi, \pi].
			\end{equation*}
			Let 
			\[
			\mathpzc{M}(\bm{f}_y) = \max_{\theta\in[-\pi,\pi]} \lambda_{\max}(\bm{f}_y(\theta))
			\quad\text{and}\quad
			\mathpzc{m}(\bm{f}_y) = \min_{\theta\in[-\pi,\pi]} \lambda_{\min}(\bm{f}_y(\theta))
			\]
			Along the lines of \cite{basu2015regularized},  it holds 
			\[
			2\pi \mathpzc{m}(\bm{f}_y) \leq \lambda_{\min}(\underline{\bm{\Sigma}}_y)\leq \lambda_{\max}(\underline{\bm{\Sigma}}_y) \leq  2\pi \mathpzc{M}(\bm{f}_y),
			\]
			\[
			2\pi \mathpzc{m}(\bm{f}_y) \leq \lambda_{\min}(\bm{\Sigma}_y)\leq \lambda_{\max}(\bm{\Sigma}_y) \leq  2\pi \mathpzc{M}(\bm{f}_y),
			\]
			and
			\begin{equation}\label{eq:fy}
				\lambda_{\min}(\bm{\Sigma}_\varepsilon)\mu_{\min}(\bm{\Psi}_*) \leq 2\pi \mathpzc{m}(\bm{f}_y)\leq 2\pi \mathpzc{M}(\bm{f}_y) \leq \lambda_{\max}(\bm{\Sigma}_\varepsilon)\mu_{\max}(\bm{\Psi}_*);
			\end{equation}
			see Proposition 2.3 therein.
			Thus, (i) is proved.
			
			\bigskip
			\noindent\textbf{Proof of (ii):} First, since $\sum_{i=1}^{\infty}\|\bm{W}_i\|_{\op}<\infty$ and $\{\bm{y}_t\}$ is   stationary with mean zero, the time series $\bm{w}_t=\mathcal{W}(B)\bm{y}_t=\mathcal{W}(B)\bm{\Psi}_*(B)\bm{\varepsilon}_{t}$ is also zero-mean and stationary, where $\mathcal{W}(B) = \sum_{i=1}^{\infty}\bm{W}_i B^i$.
			
			For any $\ell\in \mathbb{Z}$, denote by $\bm{\Sigma}_y (\ell) = \mathbb{E}(\bm{y}_t\bm{y}_{t-\ell}^\prime)$ the lag-$\ell$ covariance matrix of $\bm{y}_t$, and then $\bm{\Sigma}_y(\ell) = \int_{-\pi}^{\pi} \bm{f}_y(\theta) e^{i \ell \theta} d\theta$. 
			For any fixed $\bm{u} \in \mathbb{R}^{N}$ with $\|\bm{u}\|_2 = 1$, 
			\begin{align}\label{eq:eigenW}
				\bm{u}^\prime \bm{\Sigma}_w \bm{u} &= \bm{u}^\prime \mathbb{E}\left (\sum_{j=1}^{\infty} \bm{W}_{j}\bm{y}_{t-j} \sum_{k=1}^{\infty} \bm{W}^\prime_{k}\bm{y}_{t-k} \right )\bm{u}  \notag\\
				& =\bm{u}^\prime \sum_{j=1}^{\infty} \sum_{k=1}^{\infty} \bm{W}_{j}\bm{\Sigma}_y(k-j) \bm{W}^\prime_{k} \bm{u} \notag\\
				& = \int_{-\pi}^{\pi} \sum_{j=1}^{\infty} \sum_{k=1}^{\infty} \bm{u}^\prime  \bm{W}_j\bm{f}_y(\theta)  e^{-i(j-k)\theta} \bm{W}_k^\prime \bm{u} \, d\theta \notag\\
				& = \int_{-\pi}^{\pi} \bm{u}^\prime \mathcal{W}(e^{-i\theta}) \bm{f}_y(\theta)  \mathcal{W}^{\HH}(e^{-i\theta}) \bm{u} \, d\theta,
			\end{align}
			where  $\mathcal{W}(z) = \sum_{j=1}^{\infty}\bm{W}_jz^j$ for $z\in \mathbb{C}$, and $\mathcal{W}^{\HH}(e^{-i\theta})=\big\{ \mathcal{W}(e^{i\theta}) \big\}^\prime$ is the  conjugate transpose of $\mathcal{W}(e^{-i\theta})$.
			Since $\bm{f}_y(\theta)$ is Hermitian, $\bm{u}^\prime \mathcal{W}(e^{-i\theta}) \bm{f}_y(\theta)  \mathcal{W}^{\HH}(e^{-i\theta})\bm{u}$ is  real for all $\theta \in [-\pi, \pi]$. Then it is easy to see that
			\[
			\mathpzc{m}(\bm{f}_y) \cdot  \bm{u}^\prime \mathcal{W}(e^{-i\theta}) \mathcal{W}^{\HH}(e^{-i\theta}) \bm{u}  \leq
			\bm{u}^\prime \mathcal{W}(e^{-i\theta}) \bm{f}_y(\theta) \mathcal{W}^{\HH}(e^{-i\theta}) \bm{u}
			\leq \mathpzc{M}(\bm{f}_y) \cdot \bm{u}^\prime \mathcal{W}(e^{-i\theta}) \mathcal{W}^{\HH}(e^{-i\theta}) \bm{u}.
			\]
			Moreover, since $\int_{-\pi}^{\pi} e^{i\ell\theta} d \theta = 0$ for any $\ell\neq0$, we can show that
			\begin{align*}
				\int_{-\pi}^{\pi}\bm{u}^\prime \mathcal{W}(e^{-i\theta}) \mathcal{W}^{\HH}(e^{-i\theta}) \bm{u} \, d\theta 
				&	=\int_{-\pi}^{\pi} \sum_{j=1}^{\infty} \sum_{k=1}^{\infty} \bm{u}^\prime  \bm{W}_j e^{-i(j-k)\theta} \bm{W}_k^\prime \bm{u} \, d\theta \\
				&	= 2\pi \bm{u}^\prime  \bm{W} \bm{W}^\prime \bm{u}.
			\end{align*}
			which, together with the fact of $\|\bm{u}\|_2 = 1$, implies that
			\begin{align} \label{eq:WW}
				2\pi \sigma_{\min}^2(\bm{W}) \leq \int_{-\pi}^{\pi}\bm{u}^\prime \mathcal{W}(e^{-i\theta}) \mathcal{W}^{\HH}(e^{-i\theta})  \bm{u} \, d\theta  \leq 2\pi \sigma_{\max}^2(\bm{W}).
			\end{align}
			In view of \eqref{eq:fy}--\eqref{eq:WW}, we accomplish the proof of \eqref{eq:sigmaw0}.
			
			To verify \eqref{eq:sigmaw}, note that   the spectral density of $\{\bm{w}_t\}$ is  
			\[
			\bm{f}_w(\theta) = \mathcal{W}(e^{-i\theta})\bm{f}_y(\theta)\mathcal{W}^{\HH}(e^{-i\theta}), \hspace{5mm}\theta\in [-\pi, \pi];\]  
			see Section 9.2 of \cite{Priestley81}. Then 
			\begin{align*}
				\mathpzc{M}(\bm{f}_w)=\max_{\theta\in[-\pi,\pi]} \lambda_{\max}(\bm{f}_w(\theta)) & \leq \mathpzc{M}(\bm{f}_y) \max_{\theta\in[-\pi,\pi]} \lambda_{\max}\{\mathcal{W}(e^{-i\theta})\mathcal{W}^{\HH}(e^{-i\theta})\}\\
				&= \mathpzc{M}(\bm{f}_y) \max_{\theta\in[-\pi,\pi]} \left \| \sum_{j=1}^{\infty}  \bm{W}_j e^{-ij\theta} \right \|_{\op}^2\\
				&\leq \mathpzc{M}(\bm{f}_y) \left (\sum_{j=1}^{\infty}  \|\bm{W}_j\|_{\op} \right )^2 
			\end{align*}
			In addition, by a method similar to the proof of Proposition 2.3 in \cite{basu2015regularized}, we can show that 
			\[
			\lambda_{\max}(\underline{\bm{\Sigma}}_w)\leq 2\pi \mathpzc{M}(\bm{f}_w).
			\]
			Combining the above results with \eqref{eq:fy}, the proof of \eqref{eq:sigmaw} is complete.
			
			\subsection{Proof of Theorem \ref{thm:sparse}}
			
			For a matrix $\bm{X}\in\mathbb{R}^{p_1\times p_2}$, we denote by $ \|\bm{X}\|_{0}$ the number of nonzero elements in $\bm{X}$, $\|\bm{X}\|_{2,0}$ the number of nonzero rows in $\bm{X}$, and define $\|\bm{X}\|_1 = \|\vect(\bm{X})\|_1$ and $\|\bm{X}\|_{2,1} = \sum_{i=1}^{p_1}\|\bm{x}_i\|_2$, where $\bm{x}_i$'s are the row vectors of $\bm{X}$.
			Let 
			\begin{equation*}
				\bm{\Upsilon}_{\mathbb{S}} = \left \{\bm{\Delta} = \cm{A}-\cm{A}^*\in\mathbb{R}^{N\times N\times \infty} \mid \cm{A}=\cm{G}\times_3\bm{L}(\bm{\omega}), \cm{G}\in\bm{\Gamma}_{\mathbb{S}}, \bm{\omega}\in\bm{\Omega}, \delta_{\bm{\omega}} \leq c_{\bm{\omega}} \right \},
			\end{equation*} 
			where 	$\bm{\Gamma}_{\mathbb{S}} = \{ \cm{G}=\cm{S}\times_1\bm{U}_1\times_2\bm{U}_2\mid \cm{S}\in\bm{\Omega}_{\cmtt{S}},\bm{U}_i\in\pazocal{U}_i, i=1\text{ or }2 \}$. It is noteworthy that under the conditions of Theorem \ref{thm:sparse}, $\bm{\widetilde{\Delta}}:= \cm{\widetilde{A}}-\cm{A}^* \in \bm{\Upsilon}_{\mathbb{S}}$. 
			For simplicity, we further denote the perturbation of $\cm{S}^*$ by $\delta_{\cmtt{S}}=\|\cm{S}-\cm{S}^*\|_{\Fr}$.
			Moreover, we denote
			\[
			\bm{\Gamma}_{\mathrm{S}}( s_1,s_2,\pazocal{R}_1,\pazocal{R}_2)=\{\cm{M}=\cm{S}\times_1\bm{U}_1\times_2\bm{U}_2\mid \cm{S}\in\mathbb{R}^{\pazocal{R}_1\times \pazocal{R}_2 \times d}, \bm{U}_i\in\pazocal{U}_{\mathrm{S},i}, i=1, 2\},
			\]
			where $\pazocal{U}_{\mathrm{S},i} = \{\bm{U}\in\mathbb{R}^{N\times \pazocal{R}_i} \mid \bm{U}^\prime \bm{U} = \bm{I}_{\pazocal{R}_i}, \|\bm{U}\|_0 \leq s_i\}$, and then define
			\begin{equation}\label{eq:Xi_S}
				\bm{\Xi}_{\mathrm{S}}(s_1,s_2,\pazocal{R}_1, \pazocal{R}_2) = \left\{ \cm{M}(\bm{a} , \cm{B} ) \in  \mathbb{R}^{N \times N \times (d+r+2s)} \mid   \bm{a}  \in \mathbb{R}^{r+2s}, \cm{B} \in \bm{\Gamma}_{\mathrm{S}}( s_1,s_2,\pazocal{R}_1,\pazocal{R}_2) \right\},
			\end{equation}
			and $\bm{\Xi}_{\mathrm{S},1}(s_1,s_2,\pazocal{R}_1, \pazocal{R}_2) = \bm{\Xi}_{\mathrm{S}}(s_1,s_2,\pazocal{R}_1, \pazocal{R}_2) \cap \{\cm{M}\in  \mathbb{R}^{N \times N \times (d+r+2s)} \mid \|\cm{M}\|_{\Fr}=1\}$. 
			
			The proof of Theorem \ref{thm:sparse} depends directly on the following three lemmas.
			
			\begin{lemma}[Strong convexity and smoothness properties for the sparse model] \label{lemma:sparsersc}
				Under Assumptions \ref{assum:error}--\ref{assum:spec_gap}, if $T \gtrsim (\kappa_2 / \kappa_1)^2  \bar{d}_{\pazocal{S}} \log(\kappa_2/\kappa_1)$, then with probability at least $1  -2e^{-c\bar{d}_{\pazocal{S}}\log(\kappa_2/\kappa_1)}-3e^{-c\bar{s}_2\log N}$, 
				\begin{equation*}
					\kappa_1  \|\bm{\Delta}\|_{\Fr}^2 \lesssim	\frac{1}{T}\sum_{t=1}^{T}\|\bm{\Delta}_{(1)}\bm{x}_t\|_2^2 \lesssim  \kappa_2 \|\bm{\Delta}\|_{\Fr}^2, \quad \forall \bm{\Delta}\in\bm{\Upsilon}_{\mathbb{S}},
				\end{equation*}
				where 
				$\bar{d}_{\pazocal{S}} = \pazocal{R}_1\pazocal{R}_2 d + \sum_{i=1}^{2}\bar{s}_i \pazocal{R}_i(1 + \log N\pazocal{R}_i) $ and
				$\bar{s}_i = (s_i+\underline{u}^{-1})\pazocal{R}_i$, with $i=1$ or $2$.
			\end{lemma}

			\begin{lemma}[Deviation bound for the sparse model]\label{lemma:sparsedev}
				Under the conditions of Lemma \ref{lemma:sparsersc} and if $T \gtrsim (\kappa_2 / \kappa_1)^2  d_{\pazocal{S}}\log(\kappa_2/\kappa_1)$, given that $\lambda \gtrsim \sqrt{\kappa_2 \lambda_{\max}(\bm{\Sigma}_{\varepsilon})d_{\pazocal{S}}/{T}}$,
				\[
				\frac{1}{T}\left |\sum_{t=1}^{T}\langle \bm{\varepsilon}_t, \bm{\Delta}_{(1)}\bm{x}_t \rangle \right | \lesssim \lambda \left(\delta_{\cmtt{S}} + \alpha \delta_{\bm{\omega}} + \sum_{i=1}^{2}\|\bm{\Delta}_{\bm{U}_i} \|_{1} \right)/4 +  \tau \|\bm{\Delta}_{\bm{U}_1}\|_{1}\|\bm{\Delta}_{\bm{U}_2}\|_{1}, \quad \forall \bm{\Delta}\in\bm{\Upsilon}_{\mathbb{S}}
				\]
				holds with probability at least $1-3e^{-cd_{\pazocal{S}}\log(\kappa_2/\kappa_1)}-4e^{-cs_2\log N(\pazocal{R}_1 \wedge \pazocal{R}_2)}$, where $d_{\pazocal{S}} = \pazocal{R}_1\pazocal{R}_2 d + \sum_{i=1}^{2}s_i \pazocal{R}_i(1 + \log N\pazocal{R}_i)$ and $\tau = \sqrt{(d+\log N)/T}$.
			\end{lemma}
			
			\begin{lemma}[Effects of initial values]\label{lemma:sparseinit}
				Under Assumptions \ref{assum:error}--\ref{assum:para_add}, if $T \gtrsim \bar{s}_2$, then with probability at least $1-c\sqrt{\bar{s}_2/T}(1+\sqrt{\bar{s}_1/d_{\pazocal{S}}})$,
				\[
				|S_1(\bm{\Delta})| \lesssim \kappa_1 \|\bm{\Delta}\|_{\Fr}^2/4, \quad |S_i(\bm{\Delta})| \lesssim \sqrt{\frac{\kappa_2 \lambda_{\max}(\bm{\Sigma}_{\varepsilon}) d_{\pazocal{S}}}{T}}  \|\bm{\Delta}\|_{\Fr}/4,\quad i=2,3, \quad \forall \bm{\Delta}\in\bm{\Upsilon}_{\mathbb{S}}\cap\pazocal{S}(\delta),
				\]
				where $\bar{d}_{\pazocal{S}}$ and $d_{\pazocal{S}}$ are defined in Lemmas \ref{lemma:sparsersc} and \ref{lemma:sparsedev}, and $\bar{s}_i = (s_i+\underline{u}^{-1})\pazocal{R}_i$ for $i=1$ or $2$. 
			\end{lemma}
			
			Now we give the proof of Theorem \ref{thm:sparse}. Denote $\bm{\widetilde{\Delta}} = \cm{\widetilde{A}} - \cm{A}^*$.
			Note that $\sum_{j=1}^{t-1}\bm{A}_j\bm{y}_{t-j}= \cm{A}_{(1)}\bm{\widetilde{x}}_{t}$. 
			Due to the optimality of $\cm{\widetilde{A}}$, we have
			\[
			\sum_{t=1}^{T} \| \bm{y}_t - \cm{A}^*_{(1)}\bm{\widetilde{x}}_{t} - \widetilde{\bm{\Delta}}_{(1)} \bm{\widetilde{x}}_{t}\|_2^2 + \lambda \sum_{i=1}^{2}\|\bm{\widetilde{U}}_i\|_{1}
			\leq \sum_{t=1}^{T} \| \bm{y}_t - \cm{A}^*_{(1)}\bm{\widetilde{x}}_{t}\|_2^2 + \lambda \sum_{i=1}^{2}\|\bm{U}_i^*\|_{1},
			\]
			Then, since $\bm{y}_t - \cm{A}^*_{(1)}\bm{\widetilde{x}}_{t}=\bm{\varepsilon}_t +   \sum_{j=t}^{\infty}\bm{A}_j^* \bm{y}_{t-j}$ and $\widetilde{\bm{\Delta}}_{(1)}(\bm{x}_t - \bm{\widetilde{x}}_{t})= \sum_{k=t}^{\infty}\widetilde{\bm{\Delta}}_k \bm{y}_{t-k}$, it follows from \eqref{eq:thm1eq} and \eqref{eq:thm1eq1} that
			\begin{align}\label{eq:s_thmeq}
				\frac{1}{T}\sum_{t=1}^{T}\|\widetilde{\bm{\Delta}}_{(1)}\bm{x}_{t}\|_2^2 \leq 
				\frac{2}{T}\sum_{t=1}^{T}\langle \bm{\varepsilon}_t, \widetilde{\bm{\Delta}}_{(1)}\bm{x}_{t} \rangle + \sum_{k=1}^{3} S_k(\bm{\widetilde{\Delta}}) + \lambda \sum_{i=1}^{2}\left(\|\bm{U}_i^*\|_{1} - \|\bm{\widetilde{U}}_i\|_{1}\right),
			\end{align}
			where  $S_k(\cdot)$ for $1\leq k \leq 3$ are the initialization error terms defined as in \eqref{eq:notation_init}, and $\widetilde{\bm{\Delta}}_{(1)}\bm{x}_{t}= \sum_{k=1}^{\infty}\widetilde{\bm{\Delta}}_k \bm{y}_{t-k}$.
			Let $\bm{\widetilde{\Delta}}_{\cmtt{S}} = \cm{\widetilde{S}}- \cm{S}^*$, $\bm{\widetilde{\Delta}}_{\omega} = \bm{\widetilde{\omega}} - \bm{\omega}^*$, and $\bm{\widetilde{\Delta}}_{\bm{U}_i} = \bm{\widetilde{U}}_i - \bm{U}_i^*$ for $i = 1$ or $2$. 
			On the right-hand side of \eqref{eq:s_thmeq}, denote the event that the first term are bounded by $\lambda(\|\bm{\widetilde{\Delta}}_{\cmtt{S}}\|_{\Fr} + \alpha \|\bm{\widetilde{\Delta}}_{\omega}\|_2+ \sum_{i=1}^{2}\|\bm{\widetilde{\Delta}}_{\bm{U}_i} \|_{1})$ and $\tau\|\bm{\widetilde{\Delta}}_{\bm{U}_1}\|_{1}\|\bm{\widetilde{\Delta}}_{\bm{U}_2}\|_{1}$ as $\pazocal{I}_1$,
			\begin{align*}
				\begin{split}
					\pazocal{I}_1 = \left\{ 	
					\frac{1}{T}\left |\sum_{t=1}^{T}\langle \bm{\varepsilon}_t, \bm{\widetilde{\Delta}}_{(1)}\bm{x}_t \rangle \right | \lesssim \lambda \left(\|\bm{\widetilde{\Delta}}_{\cmtt{S}}\|_{\Fr} + \alpha \|\bm{\widetilde{\Delta}}_{\omega}\|_2  + \sum_{i=1}^{2}\|\bm{\widetilde{\Delta}}_{\bm{U}_i} \|_{1} \right)/4 +  \tau \|\bm{\widetilde{\Delta}}_{\bm{U}_1}\|_{1}\|\bm{\widetilde{\Delta}}_{\bm{U}_2}\|_{1} \right\}.
				\end{split}
			\end{align*}
			On event $\pazocal{I}_1$, if we multiply $2$ to both sides of \eqref{eq:s_thmeq} we have
			\begin{align*}
				\begin{split}
					\frac{2}{T}\sum_{t=1}^{T}\|\widetilde{\bm{\Delta}}_{(1)}\bm{x}_{t}\|_2^2 &\leq \lambda \left(\|\bm{\widetilde{\Delta}}_{\cmtt{S}}\|_{\Fr} + \alpha \|\bm{\widetilde{\Delta}}_{\omega}\|_2 + \sum_{i=1}^{2}\|\bm{\widetilde{\Delta}}_{\bm{U}_i} \|_{1} + 2\sum_{i=1}^{2}\|\bm{U}_i^*\|_{1} - 2\sum_{i=1}^{2}\|\bm{\widetilde{U}}_i\|_{1}\right) \\
					&\hspace{10mm}+2\sum_{k=1}^{3} S_k(\bm{\widetilde{\Delta}})+4\tau \|\bm{\widetilde{\Delta}}_{\bm{U}_1}\|_{1}\|\bm{\widetilde{\Delta}}_{\bm{U}_2}\|_{1}.
				\end{split}
			\end{align*}
			
			Denote by $\mathbb{S}_i$ the index set of the nonzero entries of $\bm{U}_i^*$, and by $\mathbb{S}_i^C$ the complement of $\mathbb{S}_{\bm{U}_i}$, for $i = 1$ or $2$. By the elementwise sparsity of each
			$\bm{U}_i^*$ in Assumption \ref{assum:sparse}, the cardinality of the index set $|\mathbb{S}_{\bm{U}_i}| = \|\bm{U}_i^*\|_{0} \leq s_i\pazocal{R}_i$ for $i=1$ or $2$.
			Note that $\|\bm{\widetilde{U}}_i\|_{1} \geq \|(\bm{\widetilde{\Delta}}_{\bm{U}_i})_{\mathbb{S}_i^C} + (\bm{U}_i^*)_{\mathbb{S}_i}\|_{1} - \|(\bm{\widetilde{\Delta}}_{\bm{U}_i})_{\mathbb{S}_i}\|_{1} = \|(\bm{\widetilde{\Delta}}_{\bm{U}_i})_{\mathbb{S}_i^C}\|_{1} + \|(\bm{U}_i^*)_{\mathbb{S}_i}\|_{1} - \|(\bm{\widetilde{\Delta}}_{\bm{U}_i})_{\mathbb{S}_i}\|_{1}$, $(\bm{U}_i^*)_{\mathbb{S}_i} = \bm{U}_i^*$ and $\|\bm{\widetilde{\Delta}}_{\bm{U}_i} \|_{1} \leq \|(\bm{\widetilde{\Delta}}_{\bm{U}_i})_{\mathbb{S}_i} \|_{1} +  \|(\bm{\widetilde{\Delta}}_{\bm{U}_i})_{\mathbb{S}_i^C} \|_{1}$, so we have
			\begin{equation*}
				\begin{split}
					\frac{2}{T}\sum_{t=1}^{T}\|\widetilde{\bm{\Delta}}_{(1)}\bm{x}_{t}\|_2^2 &\leq \lambda(\|\bm{\widetilde{\Delta}}_{\cmtt{S}}\|_{\Fr} + \alpha \|\bm{\widetilde{\Delta}}_{\omega}\|_2) + 3\lambda \sum_{i=1}^{2}\|(\bm{\widetilde{\Delta}}_{\bm{U}_i})_{\mathbb{S}_i}\|_{1} - \lambda\sum_{i=1}^{2}\|(\bm{\widetilde{\Delta}}_{\bm{U}_i})_{\mathbb{S}_i^C}\|_{1} \\
					&\hspace{10mm}+2\sum_{k=1}^{3} S_k(\bm{\widetilde{\Delta}})+4\tau \|\bm{\widetilde{\Delta}}_{\bm{U}_1}\|_{1}\|\bm{\widetilde{\Delta}}_{\bm{U}_2}\|_{1}.
				\end{split}
			\end{equation*}
			
			Next, we assume that there is a lower bound for $T^{-1}\sum_{t=1}^{T}\|\widetilde{\bm{\Delta}}_{(1)}\bm{x}_{t}\|_2^2$ and then define the event $\pazocal{I}_2 = \{T^{-1}\sum_{t=1}^{T}\|\widetilde{\bm{\Delta}}_{(1)}\bm{x}_{t}\|_2^2 \gtrsim \kappa_1 \|\widetilde{\bm{\Delta}}\|_{\Fr}^2\}$, where $\kappa_1=	\lambda_{\min}(\bm{\Sigma}_\varepsilon)\mu_{\min}(\bm{\Psi}_*) 	\min\{1, c_{\bar{\rho}}^2\}$ and $c_{\bar{\rho}}$ is an absolute constant defined in Lemma \ref{lemma:fullrank} in the Appendix. Moreover, we assume that the initialization error terms have a upper bound, and denote the event
			\[
			\pazocal{I}_3 = \{ S_1(\bm{\widetilde{\Delta}}) \lesssim \kappa_1\|\bm{\widetilde{\Delta}}\|_{\Fr}^2/4, S_k(\bm{\widetilde{\Delta}}) \lesssim \lambda\|\bm{\widetilde{\Delta}}\|_{\Fr}/4, k=2\text{ or }3\}.
			\]
			By Assumptions \ref{assum:sparse} and \ref{assum:para_add}, let $\bar{s}_i = (s_i + \underline{u}^{-1})\pazocal{R}_i$, and then $\bm{\widetilde{\Delta}}_{\bm{U}_i}$ has an elementwise sparsity of at most $\bar{s}_i$ for $i=1$ or $2$. 
			Hence, by the perturbation bounds in Lemma \ref{lemma:sparse_perturb_svd} and Lemma \ref{lemma:delnorm}, 
			\begin{equation*}
				\|\bm{\widetilde{\Delta}}_{\bm{U}_i}\|_{1} \leq \sqrt{\bar{s}_i}\|\bm{\widetilde{\Delta}}_{\bm{U}_i}\|_{\Fr} \leq c_{\Delta}^{-1} {\beta}^{-1} {C\eta_i}\sqrt{\bar{s}_i} \|\bm{\widetilde{\Delta}}\|_{\Fr}.
			\end{equation*}
			Suppose that $c_{\Delta}^{-2} {\beta}^{-2} {C^2\eta_1\eta_2}\sqrt{\bar{s}_1\bar{s}_2}\tau \lesssim \kappa_1/8$, on the events $\pazocal{I}_1, \pazocal{I}_2$ and $\pazocal{I}_3$,
			\begin{align*}
				\kappa_1\|\bm{\widetilde{\Delta}}\|_{\Fr}^2 & \lesssim \lambda\|\bm{\widetilde{\Delta}}\|_{\Fr} + \lambda(\|\bm{\widetilde{\Delta}}_{\cmtt{S}}\|_{\Fr} + \alpha \|\bm{\widetilde{\Delta}}_{\omega}\|_2 )+ 3\lambda \sum_{i=1}^{2}\|(\bm{\widetilde{\Delta}}_{\bm{U}_i})_{\mathbb{S}_i}\|_{1}\\
				&\lesssim \lambda\|\bm{\widetilde{\Delta}}\|_{\Fr} + \lambda(\|\bm{\widetilde{\Delta}}_{\cmtt{S}}\|_{\Fr} + \alpha \|\bm{\widetilde{\Delta}}_{\omega}\|_2 ) + 3\lambda\sum_{i=1}^{2} \sqrt{s_i} \|\bm{\widetilde{\Delta}}_{\bm{U}_i}\|_{\Fr},
			\end{align*}
			where by the perturbation bounds in Lemma \ref{lemma:sparse_perturb_svd} and Lemma \ref{lemma:delnorm}, $\|\bm{\widetilde{\Delta}}_{\bm{U}_i}\|_{\Fr} \leq c_{\Delta}^{-1} {\beta}^{-1} {C\eta_i} \|\bm{\widetilde{\Delta}}\|_{\Fr}$.
			Similarly, we can show
			\begin{align*}
				\|\bm{\widetilde{\Delta}}_{\cmtt{S}}\|_{\Fr} + \alpha \|\bm{\widetilde{\Delta}}_{\omega}\|_2 
				&\overset{\text{(Lemma \ref{lemma:sparse_perturb_svd})}}{\leq}  {\beta}^{-1}{C(\eta_1 + \eta_2)} \|\cm{\widetilde{G}}-\cm{G}^*\|_{\Fr} + \alpha\|\bm{\widetilde{\omega}} - \bm{\omega}^*\|_2\\
				&\overset{\text{(Lemma \ref{lemma:delnorm})}}{\leq} c_{\Delta}^{-1}\max\left({\beta}^{-1}{C(\eta_1 + \eta_2)}, 1\right)\|\bm{\widetilde{\Delta}}\|_{\Fr},
			\end{align*}
			
			Since $\sqrt{x} + \sqrt{y} \leq 2\sqrt{x+y}$ for any $x,y\geq 0$, we have
			\[
			\kappa_1\|\bm{\widetilde{\Delta}}\|_{\Fr}^2  \lesssim c_{\Delta}^{-1}	{\beta}^{-1}{(\eta_1 + \eta_2)} \sqrt{s_1 + s_2} \lambda \|\bm{\widetilde{\Delta}}\|_{\Fr}.
			\]
			And the estimation error bound and in-sample prediction error bound are given by
			\[
			\|\bm{\widetilde{\Delta}}\|_{\Fr}\lesssim (c_{\Delta}	{\beta} \kappa_1)^{-1} {(\eta_1 + \eta_2)} \sqrt{s_1 + s_2} \lambda  \quad\text{and}\quad \frac{1}{T}\sum_{t=1}^{T}\|\widetilde{\bm{\Delta}}_{(1)}\bm{\widetilde{x}}_{t}\|_2^2 \lesssim (c_{\Delta}^2	{\beta}^2 \kappa_1)^{-1} (\eta_1 + \eta_2)^2 (s_1 + s_2) \lambda^2,
			\]
			respectively.
			
			In the second part, we show the conditions that events $\pazocal{I}_1, \pazocal{I}_2$ and $\pazocal{I}_3$ occur with high probability. 
			First, denote by  $ d_{\pazocal{S}} = \pazocal{R}_1\pazocal{R}_2 d + \sum_{i=1}^{2}s_i \pazocal{R}_i(1 + \log N\pazocal{R}_i)$ the sample complexity of the model. If $T \gtrsim (\kappa_2 / \kappa_1)^2  d_{\pazocal{S}}\log(\kappa_2/\kappa_1)$, given that $\lambda \gtrsim \sqrt{\kappa_2 \lambda_{\max}(\bm{\Sigma}_{\varepsilon})d_{\pazocal{S}}/{T}}$ and $\tau = \sqrt{(d+\log N)/T}$, it follows from Lemma \ref{lemma:sparsedev} that the event $\pazocal{I}_1$ holds with probability at least $1-3e^{-cd_{\pazocal{S}}\log(\kappa_2/\kappa_1)}-4e^{-cs_2\log N(\pazocal{R}_1 \wedge \pazocal{R}_2)}$. Moreover, if $T \gtrsim \beta^{-4}\eta_1^{-2}\eta_2^{-2}\bar{s}_1\bar{s}_2(d+\log N)$, the condition that $c_{\Delta}^{-2} {\beta}^{-2} {C^2\eta_1\eta_2}\sqrt{\bar{s}_1\bar{s}_2}\tau \lesssim \kappa_1/8$ holds. 
			
			Secondly, if $T \gtrsim (\kappa_2 / \kappa_1)^2  \bar{d}_{\pazocal{S}}\log(\kappa_2/\kappa_1)$, it follows from Lemma \ref{lemma:sparsersc} that the event $\pazocal{I}_2$ holds with probability at least $1  -2e^{-c\bar{d}_{\pazocal{S}}\log(\kappa_2/\kappa_1)}-3e^{-c\bar{s}_2\log N}$, where $\bar{d}_{\pazocal{S}} = \pazocal{R}_1\pazocal{R}_2 d + \sum_{i=1}^{2}\bar{s}_i \pazocal{R}_i(1 + \log N\pazocal{R}_i) $. 
			
			Finally, if $T \gtrsim \bar{s}_2$, then it follows from Lemma \ref{lemma:sparseinit} that with probability at least $1-c\sqrt{\bar{s}_2/T}(1+\sqrt{\bar{s}_1/d_{\pazocal{S}}})$, the event $\pazocal{I}_3$ holds.
			
			\subsection{Proof of Lemma \ref{lemma:sparsersc}}
			The proof of this lemma is largely based on some existing results in the proof of Lemma \ref{lemma:rsc} in Section \ref{sec:proof_RSC}. 
			First of all, under Assumptions \ref{assum:sparse} \& \ref{assum:para_add}, there are at most 
			$\underline{u}^{-1}\pazocal{R}_i$ nonzero rows in $\bm{U}_i$ and $s_i\pazocal{R}_i$ ones in $\bm{U}_i^*$, for all $\bm{U}_i \in \pazocal{U}_i$ with $i=1$ or $2$. Hence, for any $\bm{\Delta} = \cm{A} - \cm{A}^* \in\bm{\Upsilon}_{\mathbb{S}}$, we only need to consider those that satisfy
			$\cm{G}_{\rm{stack}}\in \bm{\Xi}_{\mathrm{S}}(\bar{s}_1\pazocal{R}_1,\bar{s}_2\pazocal{R}_2,2\pazocal{R}_1,2\pazocal{R}_2)$ and $\bm{R}_j \in \bm{\Pi}(\bar{s}_1(\pazocal{R}_1\wedge\pazocal{R}_2),\bar{s}_2(\pazocal{R}_1\wedge\pazocal{R}_2),2(\pazocal{R}_1\wedge\pazocal{R}_2))$ for all $j\geq 1$.
			
			Then, it remains to show that the following two results hold for all $\bm{\Delta}\in\bm{\Upsilon}_{\mathbb{S}}\cap\pazocal{S}(\delta)$ that satisfy the above sparsity conditions. 
			\begin{itemize}
				\item [(i)] If $T \gtrsim (\kappa_2 / \kappa_1)^2  \bar{d}_{\pazocal{S}}\log(\kappa_2/\kappa_1)$, then
				\begin{equation*}
					\mathbb{P}\left \{\forall \bm{\Delta}\in \bm{\Upsilon}_{\mathbb{S}}\cap\pazocal{S}(\delta): \frac{ c_{\cmtt{M}} \delta^2 \kappa_1}{8} \lesssim \frac{1}{T}\sum_{t=1}^{T}\|(\cm{G}_{\rm{stack}})_{(1)}\bm{z}_t\|_2^2 \lesssim 6C_{\cmtt{M}} \delta^2 \kappa_2\right \} \geq 1-2e^{-c\bar{d}_{\pazocal{S}}\log(\kappa_2/\kappa_1)}.
				\end{equation*}
				\item [(ii)] If $T\gtrsim \bar{s}_2 \log N$, then
				\begin{equation*}
					\mathbb{P}\left \{ \sup_{\bm{\Delta}\in\bm{\Upsilon}_{\mathbb{S}}\cap\pazocal{S}(\delta)}\frac{1}{T}\sum_{t=1}^{T}\|\cm{R}_{(1)}\bm{x}_{t-p}\|_2^2 \lesssim \delta^2  \delta_{\bm{\omega}}^2 \lambda_{\max}(\bm{\Sigma}_\varepsilon) \mu_{\max}(\bm{\Psi}_*)\right \}	\geq 1-3e^{-c\bar{s}_2 \log N}.
				\end{equation*}
			\end{itemize}
			
			The result in (i) can be jointly obtained from \eqref{eq:rsc-norm}, \eqref{eq:Gnorm} and Lemma \ref{lemma:sparserscG1}. While to obtain the result in (ii), we first denote by $\mathbb{K}(s) = \{\bm{v}\in\mathbb{R}^N:\|\bm{v}\|_{0}\leq s, \|\bm{v}\|_2\leq 1\}$ the set of $s$-sparse vectors. Note that for all $j\geq 1$,
			\begin{align*}
				\frac{1}{T}\sum_{t=1}^{T}\|\bm{R}_j\bm{y}_{t-p-j}\|_2^2   
				&=\trace \left\{\bm{R}_j\left( \frac{1}{T} \sum_{t=1}^{T}\bm{y}_{t-p-j}\bm{y}_{t-p-j}^\prime \right) \bm{R}_j^\prime\right\} \\
				&\leq \|\bm{R}_j\|_{\Fr}^2 \left(\sup_{\bm{v}\in\mathbb{K}(\bar{s}_2)} \bm{v}^\prime \frac{1}{T} \sum_{t=1}^{T}\bm{y}_{t-p-j}\bm{y}_{t-p-j}^\prime \bm{v} \right).
			\end{align*}
			Combining this with \eqref{eq:rsc-norm1} \& \eqref{eq:Rx1} and \eqref{eq:sparse_Rjy1}, if $T \gtrsim \bar{s}_2\log N$, with probability at least $1-3e^{-\bar{s}_2\log N \log9}$, we have 
			\begin{align*}
				\sup_{\bm{\Delta}\in\bm{\Upsilon}_{\mathbb{S}}\cap\pazocal{S}(\delta)} \frac{1}{T}\sum_{t=1}^{T}\|\cm{R}_{(1)}\bm{x}_{t-p}\|_2^2 
				\lesssim
				\delta^2 \delta_{\bm{\omega}}^2  \lambda_{\max}(\bm{\Sigma}_\varepsilon)\mu_{\max}(\bm{\Psi}_*)
				\left (\sum_{j=1}^{\infty}  \bar{\rho}^j  \sqrt{j \sigma^2+1} \right )^2
				\lesssim \delta^2  \delta_{\bm{\omega}}^2 \lambda_{\max}(\bm{\Sigma}_\varepsilon)\mu_{\max}(\bm{\Psi}_*),
			\end{align*}
			where the second inequality follows from the fact that $\sum_{j=1}^{\infty}  \bar{\rho}^j  \sqrt{j \sigma^2+1}\asymp1$. Thus (ii) is verified. The rest of the proof is same as the proof of Lemma \ref{lemma:rsc} and hence is omitted here.
			
			
			\subsection{Proof of Lemma \ref{lemma:sparsedev}}
			
			The proof of this lemma closely follows the proof of Lemma \ref{lemma:dev}. Essentially, we only need to show the following two intermediate results. 
			\begin{itemize}
				\item [(i)] If $T \gtrsim (\kappa_2 / \kappa_1)^2  d_{\pazocal{S}}\log(\kappa_2/\kappa_1)$, then
				\begin{align*}
					&\mathbb{P}\left \{ \sup_{\bm{\Delta}\in\bm{\Upsilon}_{\mathbb{S}}} \frac{1}{T}\left |\sum_{t=1}^{T}\langle (\cm{G}_{\rm{stack}})_{(1)} \bm{z}_t, \bm{\varepsilon}_t \rangle \right | \leq \lambda \left(\delta_{\cmtt{S}} + \alpha \delta_{\bm{\omega}} + \sum_{i=1}^{2}\|\bm{\Delta}_{\bm{U}_i} \|_{1} \right)/4 + \tau \|\bm{\Delta}_{\bm{U}_1}\|_{1}\|\bm{\Delta}_{\bm{U}_2}\|_{1} \right \}\\
					&\hspace{30mm} \geq 1- e^{-c d_{\pazocal{S}}}- 2e^{-cd_{\pazocal{S}}\log(\kappa_2/\kappa_1)}.
				\end{align*}
				\item [(ii)] If $T \gtrsim (s_1+s_2)(\pazocal{R}_1 \wedge \pazocal{R}_2 + \log N(\pazocal{R}_1 \wedge \pazocal{R}_2))$, then
				\begin{align*}
					&\mathbb{P}\left \{ \sum_{j=1}^{\infty} \sup_{\bm{\Delta}\in\bm{\Upsilon}_{\mathbb{S}}} \frac{1}{T}\left |\sum_{t=1}^{T}\langle \bm{R}_j\bm{y}_{t-p-j}, \bm{\varepsilon}_t \rangle \right | \leq \lambda \delta_{\bm{\omega}} \left(\delta_{\cmtt{S}} + \alpha \delta_{\bm{\omega}} + \sum_{i=1}^{2}\|\bm{\Delta}_{\bm{U}_i} \|_{1} \right)/4  +  \tau \delta_{\bm{\omega}} \|\bm{\Delta}_{\bm{U}_1}\|_{1}\|\bm{\Delta}_{\bm{U}_2}\|_{1}  \right \}\\
					&\hspace{35mm}	\geq  1-4e^{-s_2 \log N(\pazocal{R}_1 \wedge \pazocal{R}_2)\log 9}.
				\end{align*}
			\end{itemize}
			
			\bigskip\noindent
			\textbf{Proof of (i):} 
			Note that $\cm{G} - \cm{G}^* = \cm{S}\times_1 \bm{\Delta}_{\bm{U}_1} \times_2 \bm{U}_2^* + \cm{S}\times_1 \bm{U}_1^* \times_2 \bm{\Delta}_{\bm{U}_2} +  \cm{S}\times_1 \bm{\Delta}_{\bm{U}_1} \times_2\bm{\Delta}_{\bm{U}_2} + \bm{\Delta}_{\cmtt{S}} \times_1 \bm{U}_1^* \times_2 \bm{U}_2^*$, we then have $\cm{G}_{\rm{stack}}=\stk(\cm{G} - \cm{G}^*, \cm{D}(\bm{\omega})) = \sum_{i=1}^{4}\cm{M}_i$, where 
			\begin{align*}
				\begin{split}
					\cm{M}_1 &= \sum_{i=1}^{N}\sum_{m=1}^{\pazocal{R}_1}\cm{M}_{1,i,m}, \quad \cm{M}_2 = \sum_{i=1}^{N}\sum_{m=1}^{\pazocal{R}_2}\cm{M}_{2,i,m},\quad \cm{M}_3 = \sum_{i,j=1}^{N}\sum_{m=1}^{\pazocal{R}_1}\sum_{h=1}^{\pazocal{R}_2}\cm{M}_{3,i,j,m,h},\hspace{2mm}
					\text{and}\\
					\cm{M}_4 &= \stk(\bm{\Delta}_{\cmtt{S}} \times_1 \bm{U}_1^* \times_2 \bm{U}_2^*, \cm{D}(\bm{\omega})) =  \cm{M}(\bm{\omega} - \bm{\omega}^*, \bm{\Delta}_{\cmtt{S}} \times_1 \bm{U}_1^* \times_2 \bm{U}_2^*),
				\end{split}
			\end{align*}
			where for any $\bm{a} = (a_1,\dots,a_{r+2s})^\prime \in \mathbb{R}^{r+2s}$ and $\cm{B}\in\mathbb{R}^{N\times N\times d}$, the bilinear functional $\cm{M}(\bm{a}, \cm{B})$ is defined as in \eqref{eq:Gfunc}, and moreover, for $1\leq i \leq N$, the $N\times N\times(d+r+2s)$ tensors $\cm{M}_{1,i}$, $\cm{M}_{2,i}$ and $\cm{M}_{3,i,m}$ are defined respectively by
			\begin{align*}
				&\cm{M}_{1,i,m}=\cm{M}(\bm{0}, \cm{S} \times_1 (\bm{\Delta}_{\bm{U}_1}^{(i,m)}\bm{e}_{i}\bm{\bar{e}}_{ m}^\prime) \times_2 \bm{U}_2^{*\prime}),\quad\cm{M}_{2,i,m}=\cm{M}(\bm{0}, \cm{S} \times_1 \bm{U}_1^{*\prime} \times_2 (\bm{\Delta}_{\bm{U}_2}^{(i,m)}\bm{e}_{i}\bm{\widetilde{e}}_{m}^\prime)),\\
				&\hspace{25mm}\text{and}\hspace{2mm}\cm{M}_{3,i,j,m,h}=\cm{M}(\bm{0}, \cm{S} \times_1 (\bm{\Delta}_{\bm{U}_1}^{(i,m)}\bm{e}_{i}\bm{\bar{e}}_{m}^\prime) \times_2 (\bm{\Delta}_{\bm{U}_2}^{(j,h)}\bm{e}_{j}\bm{\widetilde{e}}_{h}^\prime))
			\end{align*}
			where $\bm{e}_{\ell}, \bm{\bar{e}}_{\ell}$ and $\bm{\widetilde{e}}_{\ell}$ are coordinate vectors whose $\ell$-th element is $1$ and the others are $0$ of dimensional $N, \pazocal{R}_1$ and $\pazocal{R}_2$ respectively, and $\bm{\Delta}_{\bm{U}_i}^{(k,\ell)}$ is the $(k,\ell)$-th element of $\bm{\Delta}_{\bm{U}_i}, i = 1, 2$ with $1\leq \ell \leq N$.
			Since $\|\bm{U}_i^*\|_{\op} = 1$ and $\|\cm{S}_{(i)}\|_{\op} = \|\cm{G}_{(i)}\|_{\op}$ for $i = 1$ or $2$, by Assumption \ref{assum:statn}, the norms of $\cm{M}_k$'s further satisfy
			\begin{align}\label{eq:sparse_devM}
				\begin{split}
					&\|\cm{M}_{1,i,m}\|_{\Fr} \leq \|\cm{S}_{(1)}\|_{\op} |\bm{\Delta}_{\bm{U}_1}^{(i,m)}| \leq C_{\cmtt{G}}|\bm{\Delta}_{\bm{U}_1}^{(i,m)}|, \quad \|\cm{M}_{2,i,m}\|_{\Fr} \leq \|\cm{S}_{(2)}\|_{\op} |\bm{\Delta}_{\bm{U}_2}^{(i,m)}|\leq C_{\cmtt{G}}|\bm{\Delta}_{\bm{U}_2}^{(i,m)}|,\\
					&\|\cm{M}_{3,i,j,m,h}\|_{\Fr} \leq  \|\cm{S}_{(1)}\|_{\op} |\bm{\Delta}_{\bm{U}_1}^{(i,m)}| |\bm{\Delta}_{\bm{U}_2}^{(j,h)}| \leq C_{\cmtt{G}} |\bm{\Delta}_{\bm{U}_1}^{(i,m)}| |\bm{\Delta}_{\bm{U}_2}^{(j,h)}|,\hspace{2mm}\text{and}\hspace{2mm}\|\cm{M}_4\|_{\Fr} \leq C_{\Delta} \left( \delta_{\cmtt{S}} + \alpha \delta_{\bm{\omega}} \right),
				\end{split}
			\end{align}
			where the norm bound on $\cm{M}_4$ follows from Lemma \ref{lemma:delnorm}. 
			
			As a result,
			\begin{align*}
				\frac{1}{T} \sum_{t=1}^{T} \langle \bm{\varepsilon}_t,(\cm{G}_{\rm{stack}})_{(1)} \bm{z}_t\rangle 
				&= \sum_{k=1}^{2} \sum_{i=1}^{N} \sum_{m=1}^{\pazocal{R}_k} 	\frac{1}{T}\sum_{t=1}^{T} \langle \bm{\varepsilon}_t, (\cm{M}_{k,i,m})_{(1)}\bm{z}_t\rangle + \frac{1}{T}\sum_{t=1}^{T} \langle \bm{\varepsilon}_t, (\cm{M}_4)_{(1)} \bm{z}_t \rangle\\
				&\hspace{5mm}+\sum_{i,j=1}^{N} \sum_{m=1}^{\pazocal{R}_1}	\sum_{h=1}^{\pazocal{R}_2} \frac{1}{T}\sum_{t=1}^{T} \langle \bm{\varepsilon}_t, (\cm{M}_{3,i,j,m,h})_{(1)}\bm{z}_t\rangle
				\\
				&\leq \sum_{i=1}^{N} \sum_{m=1}^{\pazocal{R}_1}\|\cm{M}_{1,i,m}\|_{\Fr} \sup_{\small{ \cmt{M} \in \bm{\Xi}_{\mathrm{S},1}(1,s_2\pazocal{R}_2,1,\pazocal{R}_2) }}\frac{1}{T}\sum_{t=1}^{T} \langle \bm{\varepsilon}_t, \cm{M}_{(1)} \bm{z}_t\rangle\\
				&\hspace{5mm} + \sum_{i=1}^{N} \sum_{m=1}^{\pazocal{R}_2} \|\cm{M}_{2,i,m}\|_{\Fr} \sup_{\small{ \cmt{M} \in \bm{\Xi}_{\mathrm{S},1}(s_1\pazocal{R}_1,1,\pazocal{R}_1,1) }}\frac{1}{T}\sum_{t=1}^{T} \langle \bm{\varepsilon}_t, \cm{M}_{(1)} \bm{z}_t\rangle\\
				&\hspace{5mm} + \sum_{i,j=1}^{N}\sum_{m=1}^{\pazocal{R}_1}\sum_{h=1}^{\pazocal{R}_2} \|\cm{M}_{3,i,j,m,h}\|_{\Fr} \sup_{\small{ \cmt{M} \in \bm{\Xi}_{\mathrm{S},1}(1,1,1,1) }}\frac{1}{T}\sum_{t=1}^{T} \langle \bm{\varepsilon}_t, \cm{M}_{(1)} \bm{z}_t\rangle\\
				&\hspace{5mm}+ \|\cm{M}_4\|_{\Fr} \sup_{\small{ \cmt{M} \in \bm{\Xi}_{\mathrm{S},1}(s_1\pazocal{R}_1,s_2\pazocal{R}_2,\pazocal{R}_1,\pazocal{R}_2) }}\frac{1}{T}\sum_{t=1}^{T} \langle \bm{\varepsilon}_t, \cm{M}_{(1)} \bm{z}_t\rangle.
			\end{align*}
			Combine this with \eqref{eq:sparse_devM} and \eqref{eq:devG} in Lemma \ref{lemma:sparserscG1}, we can show (i).
			
			\bigskip\noindent
			\textbf{Proof of (ii):}  
			Note that $\bm{R}_j=\bm{R}_{1j} +\bm{R}_{2j}+\bm{R}_{3j}$, where $\bm{R}_{kj}, 1\leq k\leq 3$ are defined in \eqref{eq:Rjs}. In $\bm{R}_{1j}$ and $\bm{R}_{2j}$, we have $\bm{G}_k - \bm{G}_k^* = \bm{\Delta}_{\bm{U}_1} \bm{S}_{k} \bm{U}_2^{*\prime} + \bm{U}_1^{*} \bm{S}_k \bm{\Delta}_{\bm{U}_2}^\prime + \bm{\Delta}_{\bm{U}_1} \bm{S}_{k} \bm{\Delta}_{\bm{U}_2}^\prime + \bm{U}_1^{*} (\bm{S}_{k} - \bm{S}_{k}^{*})  \bm{U}_2^{*\prime}$ for all $\bm{G}_k$-matrices, while in $\bm{R}_{3j}$, we have $\bm{G}_k^* = \bm{U}_1^*\bm{S}_k^*\bm{U}_2^{*\prime}$ for all $\bm{G}_k^*$-matrices. It is then possible to further break down $\bm{R}_{1j}$ into $\bm{M}_{1j} + \bm{M}_{3j} + \bm{M}_{5j} +  \bm{M}_{7j}$ and $\bm{R}_{2j}$ into $\bm{M}_{2j} + \bm{M}_{4j} +  \bm{M}_{6j} + \bm{M}_{8j}$, respectively, where
			\begin{align*}
				\begin{split}
					\bm{M}_{kj} &= \sum_{i=1}^{N}\sum_{m=1}^{\pazocal{R}_1}\bm{M}_{kj,i,m},\hspace{2mm}1\leq k\leq 2, \hspace{2mm} \bm{M}_{kj}= \sum_{i=1}^{N}\sum_{m=1}^{\pazocal{R}_2}\bm{M}_{kj,i,m},\hspace{2mm}3\leq k\leq 4,\\
					\bm{M}_{kj} &= \sum_{i,\ell=1}^{N}\sum_{m=1}^{\pazocal{R}_1}\sum_{h=1}^{\pazocal{R}_2}\bm{M}_{kj,i,\ell,m,h}, \hspace{2mm} 5\leq k\leq 6,\\
					\bm{M}_{7j} &= \sum_{k=1}^{r}\nabla\ell_{j}^{I}(\lambda_k^*)  (\lambda_k-\lambda_k^*) \bm{U}_1^*(\bm{S}_k^{I} - \bm{S}_k^{I*})\bm{U}_2^{*\prime} \\
					&\hspace{5mm}+ \sum_{k=1}^{s}\sum_{h=1}^2(\bm{\eta}_k-\bm{\eta}_k^*)^\prime \nabla \ell_{j}^{II,h}(\bm{\eta}_k^*) \bm{U}_1^* (\bm{S}_{k}^{II,h} - \bm{S}_{k}^{II,h*})\bm{U}_2^{*\prime}, \\
					\bm{M}_{8j} &= \frac{1}{2}\sum_{k=1}^{r}\nabla^2\ell_{j}^{I}(\widetilde{\lambda}_k) (\lambda_k-\lambda_k^*)^2 \bm{U}_1^{*}  (\bm{S}_k^{I} - \bm{S}_k^{I*} ) \bm{U}_2^{*\prime} 
					\\
					&\hspace{5mm}+\frac{1}{2} \sum_{k=1}^{s}\sum_{h=1}^2(\bm{\eta}_k-\bm{\eta}_k^*)^{\prime}\nabla^2 \ell_{j}^{II,h}(\widetilde{\bm{\eta}}_k)(\bm{\eta}_k-\bm{\eta}_k^*) \bm{U}_1^*(\bm{S}_{k}^{II,h} - \bm{S}_{k}^{II,h*})\bm{U}_2^{*\prime},
				\end{split}
			\end{align*}
			where $\bm{M}_{kj,i,m}=\bm{\Delta}_{\bm{U}_1}^{(i,m)} \bm{e}_i\bm{\bar{e}}_m^\prime  \bm{T}_{kj} \bm{U}_2^{*}$ for $1\leq k\leq 2$, $\bm{M}_{kj,i,m} = \bm{\Delta}_{\bm{U}_2}^{(i,m)}\bm{U}_1^{*} \bm{T}_{(k-2)j} \bm{\widetilde{e}}_m\bm{e}_i^\prime$ for $3\leq k\leq 4$ and $\bm{M}_{kj,i,\ell,m,h} = \bm{\Delta}_{\bm{U}_1}^{(i,m)}\bm{\Delta}_{\bm{U}_2}^{(\ell,h)}\bm{e}_i\bm{\bar{e}}_m^\prime \bm{T}_{(k-4)j} \bm{\widetilde{e}}_h\bm{e}_\ell^\prime$ for $5\leq k\leq 6$ are all $N\times N$ matrices, with
			\begin{align}
				\begin{split}
					\bm{T}_{1j}  &= \sum_{k=1}^{r}\nabla\ell_{j}^{I}(\lambda_k^*)  (\lambda_k-\lambda_k^*)\bm{S}_{k}^I + \sum_{k=1}^{s}\sum_{h=1}^2(\bm{\eta}_k-\bm{\eta}_k^*)^\prime \nabla \ell_{j}^{II,h}(\bm{\eta}_k^*) \bm{S}_{k}^{II,h}\hspace{2mm}\text{and}\\
					\bm{T}_{2j}  &= \frac{1}{2} \sum_{k=1}^{r}\nabla^2\ell_{j}^{I}(\widetilde{\lambda}_k) (\lambda_k-\lambda_k^*)^2 \bm{S}_{k}^I + \sum_{k=1}^{s}\sum_{h=1}^2(\bm{\eta}_k-\bm{\eta}_k^*)^{\prime}\nabla^2 \ell_{j}^{II,h}(\widetilde{\bm{\eta}}_k)(\bm{\eta}_k-\bm{\eta}_k^*)\bm{S}_{k}^{II,h}.
				\end{split}
			\end{align}
			For the tidiness of notation, we further let $\bm{M}_{9j} = \bm{R}_{3j}$, and then $\bm{R}_j = \sum_{k=1}^{9}\bm{M}_{kj}$. Note that the norms of $\bm{M}_{kj}$'s further satisfy 
			\begin{align} \label{eq:sparse_devRM}
				\begin{split}
					&\|\bm{M}_{1j,i,m}\|_{\Fr} \leq \sqrt{2}C_L C_{\cmtt{G}}\bar{\rho}^j \delta_{\bm{\omega}} |\bm{\Delta}_{\bm{U}_1}^{(i,m)}|,\hspace{19mm} \|\bm{M}_{2j,i,m}\|_{\Fr} \leq \frac{\sqrt{2}}{2}C_L C_{\cmtt{G}}\bar{\rho}^j \delta_{\bm{\omega}}^2 |\bm{\Delta}_{\bm{U}_1}^{(i,m)}|, \\
					&\|\bm{M}_{3j,i,m}\|_{\Fr} \leq \sqrt{2}C_L C_{\cmtt{G}}\bar{\rho}^j \delta_{\bm{\omega}} |\bm{\Delta}_{\bm{U}_2}^{(i,m)}|, \hspace{19mm} \|\bm{M}_{4j,i,m}\|_{\Fr} \leq \frac{\sqrt{2}}{2}C_L C_{\cmtt{G}}\bar{\rho}^j \delta_{\bm{\omega}}^2|\bm{\Delta}_{\bm{U}_2}^{(i,m)}|,\\
					&\|\bm{M}_{5j,i,\ell,m,h}\|_{\Fr} \leq \sqrt{2}C_L\bar{\rho}^j C_{\cmtt{G}} \delta_{\bm{\omega}} |\bm{\Delta}_{\bm{U}_1}^{(i,m)}||\bm{\Delta}_{\bm{U}_2}^{(\ell,h)}|, \hspace{2mm} \|\bm{M}_{6j,i,\ell,m,h}\|_{\Fr} \leq  \frac{\sqrt{2}}{2} C_L\bar{\rho}^j C_{\cmtt{G}} \delta_{\bm{\omega}}^2 |\bm{\Delta}_{\bm{U}_1}^{(i,m)}||\bm{\Delta}_{\bm{U}_2}^{(\ell,h)}|,\\
					&\|\bm{M}_{7j}\|_{\Fr} \leq \sqrt{2}C_L\bar{\rho}^j \delta_{\bm{\omega}} \delta_{\cmtt{S}} ,\hspace{8mm} \|\bm{M}_{8j}\|_{\Fr} \leq \frac{\sqrt{2}}{2}C_L\bar{\rho}^j \delta_{\bm{\omega}}^2 \delta_{\cmtt{S}}, \hspace{4mm}\text{and}\hspace{4mm}
					\|\bm{M}_{9j}\|_{\Fr} \leq C_LC_{\cmtt{G}} \bar{\rho}^j \alpha \delta_{\bm{\omega}}^2,
				\end{split}
			\end{align}
			
			As a result, 
			\begin{align}\label{eq:sparse_devR}
				\begin{split}
					\frac{1}{T}\sum_{t=1}^{T} \left| \langle \bm{\varepsilon}_t, \bm{R}_j\bm{y}_{t-p-j}\rangle \right|&=
					\sum_{k=1}^{2}\sum_{i=1}^{N} \sum_{m=1}^{\pazocal{R}_1}	\frac{1}{T}\sum_{t=1}^{T} \langle \bm{\varepsilon}_t, \bm{M}_{kj,i,m}\bm{y}_{t-p-j}\rangle +\sum_{k=3}^{4} \sum_{i=1}^{N} \sum_{m=1}^{\pazocal{R}_2}	\frac{1}{T}\sum_{t=1}^{T} \langle \bm{\varepsilon}_t, \bm{M}_{kj,i,m} \bm{y}_{t-p-j}\rangle\\
					&\hspace{5mm}+\sum_{k=5}^{6} \sum_{i,\ell=1}^{N} \sum_{m=1}^{\pazocal{R}_1} \sum_{h=1}^{\pazocal{R}_2}	\frac{1}{T}\sum_{t=1}^{T} \langle \bm{\varepsilon}_t, \bm{M}_{kj, i,\ell,m,h}\bm{y}_{t-p-j}\rangle + \sum_{k=7}^{9}\frac{1}{T}\sum_{t=1}^{T} \langle \bm{\varepsilon}_t, \bm{M}_{kj}\bm{y}_{t-p-j}\rangle\\
					& \leq 	\sum_{k=1}^{2}\sum_{i=1}^{N} \sum_{m=1}^{\pazocal{R}_1}	 \|\bm{M}_{kj,i,m}\|_{\Fr}\sup_{\bm{M}\in\bm{\Pi}_{\mathrm{S},1}(1,s_2\pazocal{R}_2,1)}	\frac{1}{T}\sum_{t=1}^{T} \langle \bm{\varepsilon}_t, \bm{M}\bm{y}_{t-p-j}\rangle \\
					&\hspace{5mm}+\sum_{k=3}^{4} \sum_{i=1}^{N}  \sum_{m=1}^{\pazocal{R}_2} \|\bm{M}_{kj,i,m}\|_{\Fr}	\sup_{\bm{M}\in\bm{\Pi}_{\mathrm{S},1}(s_1\pazocal{R}_1,1,1)}	\frac{1}{T}\sum_{t=1}^{T} \langle \bm{\varepsilon}_t, \bm{M}\bm{y}_{t-p-j}\rangle \\
					&\hspace{5mm}+\sum_{k=5}^{6} \sum_{i,\ell=1}^{N} \sum_{m=1}^{\pazocal{R}_1}\sum_{h=1}^{\pazocal{R}_2}	\frac{1}{T}\sum_{t=1}^{T}  \|\bm{M}_{kj,i,\ell,m,h}\|_{\Fr} \sup_{\bm{M}\in\bm{\Pi}_{\mathrm{S},1}(1,1,1)} \langle \bm{\varepsilon}_t, \bm{M}\bm{y}_{t-p-j}\rangle \\
					&\hspace{5mm}+ \sum_{k=7}^{9}\|\bm{M}_{kj}\|_{\Fr} \sup_{\bm{M}\in\bm{\Pi}_{\mathrm{S},1}(s_1\pazocal{R}_1,s_2\pazocal{R}_2,\pazocal{R}_1\wedge \pazocal{R}_2)} \frac{1}{T}\sum_{t=1}^{T} \langle \bm{\varepsilon}_t, \bm{M}\bm{y}_{t-p-j}\rangle
				\end{split}
			\end{align}
			
			Then by \eqref{eq:sparse_devRM}, \eqref{eq:sparse_devR} and \eqref{eq:sparse_Rjy2} in Lemma \ref{lemma:sparse_Rjy}, (ii) is verified. 
			
			
			\subsection{Proof of Lemma \ref{lemma:sparseinit}}
			The proof of this lemma follows closely from the proof of Lemma \ref{lemma:init}. The main difference lies in the sparsity conditions. Specifically, denote the index sets of the nonzero rows and columns in $\bm{A}_j^*$ by $\mathbb{S}_1$ and $\mathbb{S}_2$, and those in $\bm{\bm{\Delta}}_j$ by $\bar{\mathbb{S}}_1$ and $\bar{\mathbb{S}}_2$, respectively. By Assumptions \ref{assum:sparse} and \ref{assum:para_add}, it holds that the cardinality of $\mathbb{S}_i$ and $\bar{\mathbb{S}}_i$ satisfy $|\mathbb{S}_i|\leq s_i\pazocal{R}_i$ and $|\bar{\mathbb{S}}_i| \leq \bar{s}_i$ for $i=1$ or $2$. 
			Then, note that $\mathbb{E}(\|(\bm{y}_t)_{\mathbb{S}_2}\|_2^2)\leq s_2\pazocal{R}_2 \lambda_{\max}(\bm{\Sigma}_\varepsilon)\mu_{\max}(\bm{\Psi}_*)$ holds by Lemma \ref{lemma:Wcov}, and thus for all $j\geq 1$,
			\begin{equation}\label{eq:sparse_init1}
				\mathbb{E}(\|\bm{A}_j^*\bm{y}_{t-j}\|_2) \leq \left \{\mathbb{E}(\|\bm{A}_j^*\bm{y}_{t-j}\|_2^2) \right \}^{1/2}\leq  \|\bm{A}_j^*\|_{\op}\mathbb{E}(\|(\bm{y}_{t-j})_{\mathbb{S}_2}\|_2^2)^{1/2} \leq C_* \bar{\rho}^{j} \sqrt{\lambda_{\max} (\bm{\Sigma}_\varepsilon) \mu_{\max}(\bm{\Psi}_*) s_2\pazocal{R}_2}
			\end{equation}
			and
			\begin{align}\label{eq:sparse_init2}
				\begin{split}
					\mathbb{E} \left (\sup_{\small{\bm{\Delta}\in \bm{\Upsilon}_{\mathbb{S}}\cap\pazocal{S}(\delta)}} \|\bm{\Delta}_j\bm{y}_{t-j}\|_2 \right ) &\leq \left \{\mathbb{E} \left (\sup_{\small{\bm{\Delta}\in \bm{\Upsilon}_{\mathbb{S}}\cap\pazocal{S}(\delta)}} \|\bm{\Delta}_j\bm{y}_{t-j}\|_2^2 \right )\right \}^{1/2}\\ &\leq \delta  C_{1}\bar{\rho}^{j} \sqrt{\lambda_{\max} (\bm{\Sigma}_\varepsilon) \mu_{\max}(\bm{\Psi}_*) \bar{s}_2}.
				\end{split}
			\end{align}
			
			By the Cauchy-Schwarz inequality and \eqref{eq:sparse_init2},
			\[
			\mathbb{E}\left \{\sup_{\small{\bm{\Delta}\in \bm{\Upsilon}_{\mathbb{S}}\cap\pazocal{S}(\delta)}} |S_1(\bm{\Delta})|\right \} \leq \frac{2}{T}\sum_{t=1}^{T}\sum_{j=1}^{\infty}\sum_{k=t}^{\infty} \delta ^2  C_{1}^2 \bar{\rho}^{j+k} \lambda_{\max} (\bm{\Sigma}_\varepsilon) \mu_{\max}(\bm{\Psi}_*)\bar{s}_2 \leq \frac{\delta ^2C_{2}\kappa_2 \bar{s}_2}{T},
			\]
			where $C_{2} = 2C_{1}^2\bar{\rho}^2/(1-\bar{\rho})^3 \asymp1$. Similarly, by \eqref{eq:sparse_init1} and \eqref{eq:sparse_init2}, 
			\[
			\mathbb{E} \left \{\sup_{\small{\bm{\Delta}\in \bm{\Upsilon}_{\mathbb{S}}\cap\pazocal{S}(\delta)}} |S_2(\bm{\Delta})| \right \} \leq \frac{2}{T}\sum_{t=1}^{T} \sum_{j=t}^{\infty} \sum_{k=1}^{t-1} \delta C_* C_1 \bar{\rho}^{j+k} \lambda_{\max} (\bm{\Sigma}_\varepsilon) \mu_{\max}(\bm{\Psi}_*) \sqrt{s_2\bar{s}_2\pazocal{R}_2} 
			\leq\frac{ \delta C_{3}\kappa_2  \sqrt{s_2\bar{s}_2\pazocal{R}_2}  }{T},
			\]
			where $C_3= 2 C_*C_{1}\bar{\rho}^2/(1-\bar{\rho})^3 \asymp1$. Moreover, note that $\mathbb{E}(\|(\bm{\varepsilon}_t)_{\bar{\mathbb{S}}_1}\|_2) \leq \sqrt{\mathbb{E}(\|(\bm{\varepsilon}_t)_{\bar{\mathbb{S}}_1}\|_2^2)} \leq \sqrt{\lambda_{\max}(\bm{\Sigma}_\varepsilon) \bar{s}_1}$. Then by \eqref{eq:sparse_init2} and a method similar to the above, 
			\begin{align*}
				\mathbb{E}\left \{ \sup_{\small{\bm{\Delta}\in \bm{\Upsilon}_{\mathbb{S}}\cap\pazocal{S}(\delta)}} |S_3(\bm{\Delta})| \right \} &\leq \frac{2}{T}\sum_{t=1}^{T} \sum_{j=t}^{\infty} \delta   C_{1} \bar{\rho}^{j} \lambda_{\max} (\bm{\Sigma}_\varepsilon) \sqrt{\mu_{\max}(\bm{\Psi}_*)} \sqrt{\bar{s}_1\bar{s}_2} \\
				&\leq \frac{\delta C_{4}  \sqrt{\kappa_2 \lambda_{\max} (\bm{\Sigma}_\varepsilon)} \sqrt{\bar{s}_1\bar{s}_2}}{T},
			\end{align*}
			where  $C_{4} = 2C_{1}\bar{\rho}/(1-\bar{\rho})^2 \asymp1$.
			By Markov's inequality, we can show that 
			\begin{equation*}
				\mathbb{P}\left \{\sup_{\small{\bm{\Delta}\in \bm{\Upsilon}_{\mathbb{S}}\cap\pazocal{S}(\delta)}} |S_1(\bm{\Delta})| \geq \delta^2 C_2 \kappa_1 \right \} \leq \frac{\mathbb{E}\{\sup_{\small{\bm{\Delta}\in \bm{\Upsilon}_{\mathbb{S}}\cap\pazocal{S}(\delta)}} |S_1(\bm{\Delta})|\}}{\delta^2 C_2 \kappa_1 } \leq \frac{\kappa_2 \bar{s}_2}{\kappa_1 T}\leq  \sqrt{\frac{\bar{s}_2}{  T}},
			\end{equation*}
			\[
			\mathbb{P}\left \{\sup_{\small{\bm{\Delta}\in \bm{\Upsilon}_{\mathbb{S}}\cap\pazocal{S}(\delta)}} |S_2(\bm{\Delta})| \geq \delta C_3 \sqrt{\frac{\kappa_2 \lambda_{\max}(\bm{\Sigma}_{\varepsilon}) d_{\pazocal{S}}}{T}}  \right \} \leq \sqrt{\frac{\kappa_2 s_2\bar{s}_2\pazocal{R}_2}{\lambda_{\max}(\bm{\Sigma}_{\varepsilon}) T d_{\pazocal{S}}}} \leq \sqrt{\frac{\kappa_2 \bar{s}_2}{\lambda_{\max}(\bm{\Sigma}_{\varepsilon}) T}},
			\]	
			and 
			\[
			\mathbb{P}\left \{\sup_{\small{\bm{\Delta}\in \bm{\Upsilon}_{\mathbb{S}}\cap\pazocal{S}(\delta)}} |S_3(\bm{\Delta})| \geq \delta C_4 \sqrt{\frac{\kappa_2 \lambda_{\max}(\bm{\Sigma}_{\varepsilon}) d_{\pazocal{S}}}{T}}  \right \} \leq  \sqrt{\frac{ \bar{s}_1\bar{s}_2}{T d_{\pazocal{S}}}},
			\]
			where the last inequality in \eqref{eq:S1Delta} uses the condition that  $T \gtrsim \bar{s}_2$.
			Then the sum of the above three tail probabilities is $(1+\sqrt{\kappa_2/\lambda_{\max}(\bm{\Sigma}_{\varepsilon})})\sqrt{\bar{s}_2/T}(1+\sqrt{\bar{s}_1/d_{\pazocal{S}}})$.
			\subsection{Auxiliary lemmas for the proofs of Lemmas \ref{lemma:sparsersc} and \ref{lemma:sparsedev}}
			
			The proofs of Lemmas \ref{lemma:sparsersc} and \ref{lemma:sparsedev} rely on the following auxiliary results.
			
			\begin{lemma}[HOSVD perturbation bound]\label{lemma:sparse_perturb_svd}
				Suppose that $\cm{G} = \cm{S}\times\bm{U}_1\times\bm{U}_2$ and $\cm{\widetilde{G}} = \cm{\widetilde{S}}\times\bm{\widetilde{U}}_1\times\bm{\widetilde{U}}_2$ are two HOSVD for $\cm{G}$ and $\cm{\widetilde{G}}$, with the same multilinear ranks $(\pazocal{R}_1, \pazocal{R}_2)$ along the first and second modes. Under Assumptions \ref{assum:para_add} and \ref{assum:spec_gap}, we have
				\begin{align}
					\begin{split}
						\|\cm{\widetilde{S}} - \cm{S}\|_{\Fr} \leq \frac{C(\eta_1 + \eta_2)}{\beta}&\|\cm{\widetilde{G}} - \cm{G}\|_{\Fr}\hspace{3mm}\text{and}\\
						\|\bm{\widetilde{U}}_i - \bm{U}_i\|_{\Fr} \leq \frac{C\eta_i}{\beta}&\|\cm{\widetilde{G}} - \cm{G}\|_{\Fr},
					\end{split}
				\end{align}
				where $\eta_i = \sum_{j=1}^{\pazocal{R}_i}\sigma_{1}^2(\cm{G}_{(i)}) / \sigma_{j}^2(\cm{G}_{(i)})$ for $i=1$ or $2$. Moreover,
				\begin{equation}
					\|\cm{\widetilde{G}} - \cm{G}\|_{\Fr} \leq \|\cm{\widetilde{S}} - \cm{S}\|_{\Fr} + C_{\cmtt{S}}\sum_{i=1}^{2}	\|\bm{\widetilde{U}}_i - \bm{U}_i\|_{\Fr}.
				\end{equation}	
			\end{lemma}
			\begin{proof}[Proof of Lemma \ref{lemma:sparse_perturb_svd}]
				The proof of this lemma follows trivially from Lemma 1 in \cite{Wang2021High}.
			\end{proof}
			
			\begin{lemma}[Covering number for sparse-and-low-Tucker-rank tensors] \label{lemma:sparse_covering}
				Let $$\bm{\Pi}_{\mathrm{S}}( s_1,s_2,\pazocal{R}_1,\pazocal{R}_2)=\{\cm{M}=\cm{S}\times_1\bm{U}_1\times_2\bm{U}_2:\|\cm{M}\|_{\textup{F}}\leq 1,\cm{S}\in\mathbb{R}^{\pazocal{R}_1\times \pazocal{R}_2 \times d}, \bm{U}_i\in\pazocal{U}_{\mathrm{S},i}, i=1, 2\},$$ 
				where $\pazocal{U}_{\mathrm{S},i} = \{\bm{U}\in\mathbb{R}^{N\times \pazocal{R}_i} \mid \bm{U}^\prime \bm{U} = \bm{I}_{\pazocal{R}_i}, \|\bm{U}\|_0 \leq s_i\}$. For any $\epsilon>0$, let $\bar{\bm{\Pi}}_{\mathrm{S}}(\epsilon; s_1,s_2,\pazocal{R}_1,\pazocal{R}_2)$ be a minimal $\epsilon$-net  for $\bm{\Pi}_{\mathrm{S}}(s_1,s_2,\pazocal{R}_1,\pazocal{R}_2)$ in the Frobenius norm. Then $\bar{\bm{\Pi}}_{\mathrm{S}}(\epsilon; s_1,s_2,\pazocal{R}_1,\pazocal{R}_2)$ has cardinality satisfying
				\begin{equation*}
					|\bar{\bm{\Pi}}_{\mathrm{S}}(\epsilon; s_1,s_2,\pazocal{R}_1,\pazocal{R}_2)|\leq \binom{N\pazocal{R}_1}{s_1}\binom{N\pazocal{R}_2}{s_2}\left(\frac{9}{\epsilon}\right)^{\pazocal{R}_1\pazocal{R}_2d + s_1+s_2}.
				\end{equation*}
			\end{lemma}
			
			
			\begin{proof}[Proof of Lemma \ref{lemma:sparse_covering}]
				The proof of this lemma follows trivially from Lemma \ref{lemma:covering}.
			\end{proof}

			Recall from \eqref{eq:Xi_S} that
			\begin{equation*}
				\bm{\Xi}_{\mathrm{S}}(s_1,s_2,\pazocal{R}_1, \pazocal{R}_2) = \left\{ \cm{M}(\bm{a} , \cm{B} ) \in  \mathbb{R}^{N \times N \times (d+r+2s)} \mid   \bm{a}  \in \mathbb{R}^{r+2s}, \cm{B} \in \bm{\Gamma}_{\mathrm{S}}( s_1,s_2,\pazocal{R}_1,\pazocal{R}_2) \right\},
			\end{equation*}
			and $\bm{\Xi}_{\mathrm{S},1}(s_1,s_2,\pazocal{R}_1, \pazocal{R}_2) = \bm{\Xi}_{\mathrm{S}}(s_1,s_2,\pazocal{R}_1, \pazocal{R}_2) \cap \{\cm{M}\in  \mathbb{R}^{N \times N \times (d+r+2s)} \mid \|\cm{M}\|_{\Fr}=1\}$. 
			\begin{lemma}[Covering number and discretization for $\bm{\Xi}_{\mathrm{S},1}$]\label{lemma:sparse_epsilon-net}
				For any $0<\epsilon<2/3$, let $\bar{\bm{\Xi}}_{\mathrm{S}}(\epsilon; s_1,s_2,\pazocal{R}_1,\pazocal{R}_2)$ be a minimal generalized $\epsilon$-net of $\bm{\Xi}_{\mathrm{S},1}(s_1,s_2,\pazocal{R}_1,\pazocal{R}_2)$. In (ii) -- (iii), denote $\bm{\Xi}_{\mathrm{S},1}(s_1,s_2,\pazocal{R}_1,\pazocal{R}_2)$ and $\bm{\bar{\Xi}}_{\mathrm{S}}(s_1,s_2,\pazocal{R}_1,\pazocal{R}_2)$ by $\bm{\Xi}_1$ and $\bar{\bm{\Xi}}$, respectively.
				\begin{itemize}
					\item [(i)] The cardinality of $\bar{\bm{\Xi}}_{\mathrm{S}}(\epsilon; s_1,s_2,\pazocal{R}_1,\pazocal{R}_2)$ satisfies
					\[
					\log |\bar{\bm{\Xi}}_{\mathrm{S}}(\epsilon; s_1,s_2,\pazocal{R}_1,\pazocal{R}_2)| \lesssim (\pazocal{R}_1\pazocal{R}_2 d  +  s_1+ s_2) \log(1/\epsilon) + \sum_{i=1}^{2}s_i\log N\pazocal{R}_i,
					\]
					
					\item[(ii)]  There exist absolute constants  $c_{\cmtt{M}}, C_{\cmtt{M}}>0$ such that for any $\cm{M}\in \bm{\bar{\Xi}}_{\mathrm{S}}(s_1,s_2,\pazocal{R}_1,\pazocal{R}_2)$, it holds $c_{\cmtt{M}}\leq \|\cm{M}\|_{\Fr}\leq C_{\cmtt{M}}$.
					
					\item[(iii)] For any $\bm{X} \in \mathbb{R}^{N \times N (d+r+2s)}$ and $\bm{Z}\in\mathbb{R}^{N(d+r+2s)\times T}$, it holds 
					\begin{align*}
						\sup_{ \cmt{M}\in \bm{\Xi}_{\mathrm{S},1}(s_1,s_2,\pazocal{R}_1,\pazocal{R}_2)} \langle \cm{M}_{(1)}, \bm{X}\rangle &\leq (1- 1.5\epsilon)^{-1} \max_{ \cmt{M} \in \bm{\bar{\Xi}}_{\mathrm{S}}(s_1,s_2,\pazocal{R}_1,\pazocal{R}_2)} \langle \cm{M}_{(1)}, \bm{X}\rangle,\\
						\sup_{ \cmt{M}\in \bm{\Xi}_{\mathrm{S},1}(s_1,s_2,\pazocal{R}_1,\pazocal{R}_2)} \| \cm{M}_{(1)} \bm{Z}\|_{\Fr} & \leq (1- 1.5\epsilon)^{-1} \max_{ \cmt{M} \in \bm{\bar{\Xi}}_{\mathrm{S}}(s_1,s_2,\pazocal{R}_1,\pazocal{R}_2)} \| \cm{M}_{(1)} \bm{Z}\|_{\Fr}. 
					\end{align*}	 
				\end{itemize} 
			\end{lemma}

			\begin{proof}[Proof of Lemma \ref{lemma:sparse_epsilon-net}]
				The proof of this lemma follows trivially from Lemma \ref{lemma:sparse_covering} and the proof of Lemma \ref{lemma:epsilon-net}.
			\end{proof}
			
			
			\begin{lemma}\label{lemma:sparserscG1} 
				Suppose that Assumptions  \ref{assum:error} and \ref{assum:statn} hold and  $T\gtrsim(\kappa_2/\kappa_1)^2d_{1}\log(\kappa_2/\kappa_1)$. Let $\bm{z}_t =\left \{\bm{L}_{\rm{stack}}^\prime(\bm{\omega}^*)\otimes \bm{I}_N \right \}\bm{x}_{t}$ be  defined as in  \eqref{eq:ts-z}. Then
				\begin{align}\label{eq:sparse_rscG1}
					\begin{split}
						&\mathbb{P}\left (  \frac{c_{\cmtt{M}}\kappa_1}{8} \leq \inf_{\cmtt{M}\in\bm{\Xi}_{\mathrm{S},1}(s_1,s_2,\pazocal{R}_1,\pazocal{R}_2)}\frac{1}{T}\sum_{t=1}^{T}\|\cm{M}_{(1)}\bm{z}_t\|_2^2  \leq 
						\sup_{\cmtt{M}\in\bm{\Xi}_{\mathrm{S},1}(s_1,s_2,\pazocal{R}_1,\pazocal{R}_2)}\frac{1}{T}\sum_{t=1}^{T}\|\cm{M}_{(1)}\bm{z}_t\|_2^2  \leq 6C_{\cmtt{M}}\kappa_2 \right ) \\
						&\hspace{55mm}\geq 1- 2e^{-cd_1\log(\kappa_2/\kappa_1)}.
					\end{split}
				\end{align}
				and
				\begin{equation}\label{eq:devG}
					\mathbb{P}\left \{ \sup_{\small{ \cmt{M} \in \bm{\Xi}_{\mathrm{S},1}(s_1,s_2,\pazocal{R}_1,\pazocal{R}_2)}} \frac{1}{T}\sum_{t=1}^{T}\langle \cm{M}_{(1)} \bm{z}_t, \bm{\varepsilon}_t \rangle  \lesssim  \sqrt{\frac{\kappa_2   \lambda_{\max}(\bm{\Sigma}_{\varepsilon})d_1}{T}} \right \} \geq 1- e^{-c d_1}- 2e^{-cd_1\log(\kappa_2/\kappa_1)},
				\end{equation}
				where $d_1 = \pazocal{R}_1\pazocal{R}_2 d  + \sum_{i=1}^{2} s_i  (1 + \log N\pazocal{R}_i)$.
			\end{lemma}
			
			\begin{proof}[Proof of Lemma \ref{lemma:sparserscG1}]
				The proof of this lemma follows trivially from Lemma \ref{lemma:sparse_epsilon-net} and the proof of Lemma \ref{lemma:rscG1}.
			\end{proof}
			
			\begin{lemma}[Covering number and discretization for sparse low-rank matrices]\label{lemma:sparse_coverLR}
				Let $\bm{\Pi}_{\mathrm{S},1}(s_1,s_2,\pazocal{R})= \{ \bm{M} = \bm{U}_1\bm{S}\bm{U}_2^\prime \mid \|\bm{M}\|_{\Fr} =1, \bm{S}\in\mathbb{R}^{\pazocal{R}\times\pazocal{R}}, \bm{U}_i\in\mathbb{R}^{N\times \pazocal{R}}, \|\bm{U}_i\|_{0} \leq s_i, i=1\textrm{ or }2 \}$, and let $\bar{\bm{\Pi}}_{\mathrm{S}}(s_1,s_2,\pazocal{R})$ be a minimal $1/2$-net of $\bm{\Pi}_{\mathrm{S},1}(s_1,s_2,\pazocal{R})$ in the Frobenius norm. Then the cardinality of $\bar{\bm{\Pi}}_{\mathrm{S}}(s_1,s_2,\pazocal{R})$ satisfies 
				\[
				\log|\bar{\bm{\Pi}}_{\mathrm{S}}(s_1,s_2,\pazocal{R})|\leq [\pazocal{R} + (s_1+s_2)(1 + \log N\pazocal{R})] \log 18.
				\]
				Moreover, for any $\bm{X}\in\mathbb{R}^{N\times N}$, it holds
				\[
				\sup_{\bm{M} \in  \bm{\Pi}_{\mathrm{S},1}(s_1,s_2,\pazocal{R})}\langle \bm{M}, \bm{X} \rangle \leq 4 \max_{\bm{M} \in  \bar{\bm{\Pi}}_{\mathrm{S}}(s_1,s_2,\pazocal{R})} \langle \bm{M}, \bm{X} \rangle.
				\]
			\end{lemma}
			
			\begin{proof}[Proof of Lemma \ref{lemma:sparse_coverLR}]
				The proof of this lemma follows trivially from Lemma \ref{lemma:coverLR}.
			\end{proof}
			

			\begin{lemma}\label{lemma:sparse_Rjy}
				Suppose that Assumptions \ref{assum:error} and \ref{assum:statn} hold. 
				\begin{itemize}
					\item[(i)] Let $\mathbb{K}(s) = \{\bm{v}\in\mathbb{R}^N:\|\bm{v}\|_{0}\leq s, \|\bm{v}\|_2\leq 1\}$ the set of $s$-sparse vectors. 
					If $T \gtrsim s\log N$, then
					\begin{equation}\label{eq:sparse_Rjy1}
						\begin{split}
							&\mathbb{P}\left \{\forall j\geq 1: \sup_{\bm{v}\in\mathbb{K}(s)} \bm{v}^\prime \frac{1}{T} \sum_{t=1}^{T}\bm{y}_{t-p-j}\bm{y}_{t-p-j}^\prime \bm{v}  \leq  2\lambda_{\max}(\bm{\Sigma}_\varepsilon) \mu_{\max}(\bm{\Psi}_*) (j \sigma^2+1)\right \} \\
							&\hspace{40mm}\geq 1-3e^{-s\log N\log9}.
						\end{split}
					\end{equation}
					\item[(ii)] If $T \gtrsim (s_1+s_2)(\pazocal{R} + \log N\pazocal{R})$, then
					\begin{align}\label{eq:sparse_Rjy2}
						\begin{split}
							&\mathbb{P}\Bigg \{\forall j\geq 1: \sup_{\bm{M} \in \bm{\Pi}_{\mathrm{S}}(s_1,s_2,\pazocal{R})}\frac{1}{T}\sum_{t=1}^{T}\langle \bm{M} \bm{y}_{t-p-j}, \bm{\varepsilon}_t \rangle \leq  24\lambda_{\max}(\bm{\Sigma}_\varepsilon) (2j \sigma^2+1) \\
							&\hspace{28mm} \cdot \sqrt{\frac{\mu_{\max}(\bm{\Psi}_*)[\pazocal{R} + (s_1+s_2)(1 + \log N\pazocal{R})]}{T}} \Bigg \} \geq 1-4e^{-s_2\log N\pazocal{R}\log9}.  
						\end{split}
					\end{align}	
				\end{itemize}
			\end{lemma}
			
			\begin{proof}[Proof of Lemma \ref{lemma:sparse_Rjy}]
				The proof of this lemma follows trivially from Lemma \ref{lemma:sparse_coverLR}, Lemma F.2 in \cite{basu2015regularized} and the proof of Lemma \ref{lemma:Rjy}.
			\end{proof}
			
\putbib[SARMA]
\end{bibunit}
\end{document}